\documentclass[12pt,a4paper,twoside]{scrartcl}
\usepackage[utf8]{inputenc}
\usepackage[english]{babel}
\usepackage{amsmath}
\usepackage{amsfonts}
\usepackage{amssymb}
\usepackage{graphicx}
\usepackage[left=3.5cm,right=2cm,top=2.5cm,bottom=2.5cm]{geometry}

\usepackage[scr=rsfso, cal=cm]{mathalfa} 

\title{Quantization of the Proca field on curved spacetimes}
\author{Maximilian Schambach}
\date{\today}

\usepackage{titlesec,titletoc}
\usepackage{scrtime}
\usepackage{textcomp}
\usepackage{upgreek}
\usepackage{framed}
\usepackage{tabularx}
\usepackage{xcolor}
\substitutecolormodel{rgb}{cmyk} 
\usepackage{cases}

\usepackage{enumitem} 
\setlist[enumerate]{noitemsep,label=(\roman*),itemsep=1ex}

\usepackage{float} 	

\definecolor{linkblue}{cmyk}{1,.7,0,0}
\definecolor{linkred}{cmyk}{0,.92,.92,.27}
\definecolor{ownwhite}{cmyk}{1,1,1,1}

\usepackage{subfigure}
\usepackage{setspace} 
\usepackage{bbm} 	
\usepackage[normalem]{ulem}	
\usepackage{marvosym} 

\usepackage[all]{xy} 
\usepackage{tikz} 
\usepackage{parskip}
\usepackage{mathtools} 

\usepackage{braket}  

\usepackage{bm}		

\usepackage{slashed}  

\usepackage{tensor}

\numberwithin{equation}{section}

\numberwithin{figure}{section}	
\numberwithin{table}{section}	

\usepackage[margin=0pt,font=small,labelformat=simple,labelfont=bf,labelsep=space, justification=justified, hypcap]{caption}
\usepackage[outercaption]{sidecap}

\usepackage{fancyhdr}
\titleformat*{\section}{\sffamily \bfseries \Large}
\titleformat*{\subsection}{\sffamily \bfseries \large}
\titleformat*{\subsubsection}{\sffamily \bfseries \normalsize}

\usepackage{mdframed}
\usepackage[amsmath,thmmarks, framed]{ntheorem}
\usepackage{lipsum}
\usepackage{amssymb}

\usepackage{etoolbox}

\usepackage{lipsum}

\usepackage[hidelinks,
linkcolor=linkblue,
citecolor=linkblue,
urlcolor=black,
pdftitle={Quantization of the Proca field on curved spacetimes},
pdfauthor={Maximilian Schambach}]{hyperref}

\usepackage{ellipsis} 

\makeatletter
\let\nobreakitem\item
\let\@nobreakitem\@item
\patchcmd{\nobreakitem}{\@item}{\@nobreakitem}{}{}
\patchcmd{\nobreakitem}{\@item}{\@nobreakitem}{}{}
\patchcmd{\@nobreakitem}{\@itempenalty}{\@M}{}{}
\patchcmd{\@xthm}{\ignorespaces}{\nobreak\ignorespaces}{}{}
\patchcmd{\@ythm}{\ignorespaces}{\nobreak\ignorespaces}{}{}

\renewtheoremstyle{break}%
{\item{\theorem@headerfont
		##1\ ##2\theorem@separator}\hskip\labelsep\relax\nobreakitem}%
{\item{\theorem@headerfont
		##1\ ##2\ (##3)\theorem@separator}\hskip\labelsep\relax\nobreakitem}
\makeatother
\theoremheaderfont{\kern-0cm\normalfont\bfseries}

\theoremframepreskip{0.6cm}
\theoremframepostskip{0.6cm}
\theorempreskip{0.6cm}
\theorempostskip{0.6cm}

\theoremstyle{break}
\theoremheaderfont{\normalfont \bfseries}
\theorembodyfont{\itshape}

\theoreminframepreskip{0.3cm}
\theoreminframepostskip{0.3cm}
\newframedtheorem{theorem}{Theorem}[section]

\AtBeginEnvironment {theorem}{
	
}

\theoreminframepreskip{0.4cm}
\theoreminframepostskip{0.4cm}
\usepackage{etoolbox}
\theorembodyfont{\normalfont}
\theoremheaderfont{\normalfont \bfseries}
\newtheorem{definition}[theorem]{Definition}
\newframedtheorem{lemma}[theorem]{Lemma}
\AtBeginEnvironment {lemma}{}

\theorembodyfont{\normalfont}
\theoremheaderfont{\normalfont \bfseries}
\newframedtheorem{assumption}[theorem]{Assumption}
\AtBeginEnvironment {assumption}{}

\theorembodyfont{\normalfont}
\theoremheaderfont{\normalfont \bfseries}
\newframedtheorem{conjecture}[theorem]{Conjecture}
\AtBeginEnvironment {conjecture}{}

\theorempreskip{0cm}
\theorempostskip{0.1cm}
\theoremstyle{nonumberplain}
\theoremheaderfont{\itshape}
\theorembodyfont{\normalfont}
\theoremsymbol{$\square$}
\theorempostwork{\hrule\bigskip}
\newtheorem{proof}{Proof:}[section]
\newtheorem{proof-idea}{Proof idea:}[section]

\definecolor{grey}{rgb}{0.925,0.925,0.925}

\DeclarePairedDelimiter\abs{\lvert}{\rvert}%
\DeclarePairedDelimiter\norm{\lVert}{\rVert}%

\makeatletter
\let\oldabs\abs
\def\abs{\@ifstar{\oldabs}{\oldabs*}}
\let\oldnorm\norm
\def\norm{\@ifstar{\oldnorm}{\oldnorm*}}
\makeatother

\newcommand{\largewedge}{\mbox{\Large $\wedge$}}

\newcommand{\e}{\ensuremath{\mathrm{e}}}

\newcommand{\IC}{\mathbb{C}} 
\newcommand{\IR}{\mathbb{R}} 
\newcommand{\IN}{\mathbb{N}} 

\newcommand{\M}{\ensuremath{\mathcal{M}}}    
\newcommand{\N}{\ensuremath{\mathcal{N}}}    
\newcommand{\F}{\ensuremath{\mathcal{F}}}    
\newcommand{\A}{\ensuremath{\mathcal{A}}}    
\renewcommand{\AA}{\ensuremath{\mathscr{A}}}    
\newcommand{\W}{\ensuremath{\mathcal{W}}}    
\newcommand{\Wm}{\ensuremath{\mathcal{W}_m}}    

\newcommand{\Wdyn}{\ensuremath{\mathcal{W}^\text{dyn}}}    
\newcommand{\Wmdyn}{\ensuremath{\mathcal{W}^\text{dyn}_m}}    

\newcommand{\BUOmega}{\ensuremath{{\mathcal{BU}\big(\Omega^1_0(\M)\big)}}}    
\newcommand{\BUmz}{\ensuremath{\mathcal{BU}_{m,0}}}    
\newcommand{\BUmj}{\ensuremath{{\mathcal{BU}_{m,j}}}}    
\newcommand{\BUmjdyn}{\ensuremath{\mathcal{BU}_{m,j}^\text{dyn}}}    
\newcommand{\BUmzdyn}{\ensuremath{\mathcal{BU}_{m,0}^\text{dyn}}}    
\newcommand{\BU}{\ensuremath{\mathcal{BU}}}    

\newcommand{\K}{\ensuremath{\mathcal{K}}}    
\newcommand{\D}{\ensuremath{\mathcal{D}}}    
\newcommand{\Dzs}{\ensuremath{\mathcal{D}_{0}(\Sigma)}}    
\newcommand{\Dzsp}{\ensuremath{\mathcal{D}_{0}(\Sigma')}}    

\newcommand{\PPM}{\ensuremath{\mathscr{P}_{\mathcal{M}}}}    
\newcommand{\PPN}{\ensuremath{\mathscr{P}_{\mathcal{N}}}}    

\newcommand{\Green}[2]{\ensuremath{\mathcal{G}_m(#1,#2)}}    
\newcommand{\GreenM}[2]{\ensuremath{\mathcal{G}_{m,\mathcal{M}}(#1,#2)}}    
\newcommand{\GreenN}[2]{\ensuremath{\mathcal{G}_{m,\mathcal{N}}(#1,#2)}}    
\newcommand{\Ker}[1]{\ensuremath{\mathrm{ker}{}(#1)}}    
\newcommand{\IMG}[1]{\ensuremath{\mathrm{img}{}(#1)}}    
\newcommand{\Span}[1]{\ensuremath{\textrm{span}\big\{#1\big\}}}   
\newcommand{\Real}[1]{\ensuremath{\textrm{Re}\big(#1\big)}}   

\newcommand{\C}{\ensuremath{\mathcal{C}}}    
\newcommand{\J}{\ensuremath{\mathcal{J}}}    

\newcommand{\IMDYN}{\ensuremath{\mathcal{I}_m^\text{\,dyn}}}    
\newcommand{\IMJDYN}{\ensuremath{\mathcal{I}_{m,j}^\text{\,dyn}}}    
\newcommand{\IMZDYN}{\ensuremath{\mathcal{I}_{m,0}^\text{\,dyn}}}    
\newcommand{\IMCCR}{\ensuremath{\mathcal{I}_m^\text{\,CCR}}}    
\newcommand{\ICCR}{\ensuremath{{\mathcal{I}}_{\sim}^\text{\,CCR}}}    
\newcommand{\ICCRS}{\ensuremath{{\mathcal{I}}_{\sim}^{\text{\,CCR},\Sigma}}}    
\newcommand{\ICCRSP}{\ensuremath{{\mathcal{I}}_{\sim}^{\text{\,CCR},\Sigma'}}}    
\newcommand{\IM}{\ensuremath{\mathcal{I}_m}}    
\newcommand{\IMJ}{\ensuremath{\mathcal{I}_{m,j}}}    
\newcommand{\IMZ}{\ensuremath{\mathcal{I}_{m,0}}}    
\newcommand{\JMDYN}{\ensuremath{\mathcal{J}_m^\text{dyn}}}    
\newcommand{\JM}{\ensuremath{\mathcal{J}_m}}    

\newcommand{\KERN}[1]{\ensuremath{\textrm{ker}{\left(#1\right)}}}   
\newcommand{\SPAN}[1]{\ensuremath{\textrm{span}{\left\{#1\right\} }}}   

\newcommand{\TsS}{\ensuremath{T^*\Sigma}}    

\DeclareDocumentCommand\Gm{ m g }{%
	{\ensuremath{\mathcal{G}_m %
		\IfNoValueF {#2} {(#1, #2)}%
	}%
}
}

\DeclareDocumentCommand\Gmz{ m g }{%
	{\ensuremath{\mathcal{G}_{m_0} %
			\IfNoValueF {#2} {(#1, #2)}%
		}%
	}
}

\DeclareDocumentCommand\Em{ m g }{%
	{\ensuremath{\mathcal{E}_m %
			\IfNoValueF {#2} {(#1, #2)}%
		}%
	}
}

\DeclareDocumentCommand\Ez{ m g }{%
	{\ensuremath{\mathcal{E}_{0} %
			\IfNoValueF {#2} {(#1, #2)}%
		}%
	}
}

\usepackage{xfrac}
\newcommand{\Quotient}[2]{\ensuremath{\sfrac{#1}{#2}}}    

\newcommand{\Quotientscale}[2]{\ensuremath{{\scalebox{1.2}{\Quotient{#1}{#2}}}}}    

\newcommand{\comp}{\mathbin{\mathchoice
  {\xcirc\scriptstyle}
  {\xcirc\scriptstyle}
  {\xcirc\scriptscriptstyle}
  {\xcirc\scriptscriptstyle}
}}
\newcommand{\xcirc}[1]{\vcenter{\hbox{$#1\circ$}}}

\newcommand{\rhoz}{\ensuremath{\rho_{(0)}}}    
\newcommand{\rhon}{\ensuremath{\rho_{(n)}}}    
\newcommand{\rhod}{\ensuremath{\rho_{(d)}}}    
\newcommand{\rhodelta}{\ensuremath{\rho_{(\delta)}}}    

\newcommand{\Az}{\ensuremath{{A_{(0)}}}}    
\newcommand{\Azp}{\ensuremath{{A'_{(0)}}}}    
\newcommand{\An}{\ensuremath{{A_{(n)}}}}    
\newcommand{\Ad}{\ensuremath{{A_{(d)}}}}    
\newcommand{\Adelta}{\ensuremath{{A_{(\delta)}}}}    

\newcommand{\supp}[1]{\ensuremath{\text{\upshape{supp}}\left(#1\right)}}    
\newcommand{\dvolg}{d\textrm{vol}_g} 
\newcommand{\dvolh}{d\textrm{vol}_h} 

\newcommand{\detk}{\sqrt{\abs{k_{\mu\nu}}}}
\newcommand{\dvolk}{d\textrm{vol}_k} 

\newcommand{\name}[1]{\text{#1}}    

\newcommand{\formspace}{\;}

\renewcommand{\i}{\ensuremath{\text{\upshape{i}}}}

\newcommand{\Spac}{\ensuremath{\mathsf{Spac}}}  
\newcommand{\SpacCurr}{\ensuremath{\mathsf{SpacCurr}}}  
\newcommand{\Alg}{\ensuremath{\mathsf{Alg}}}  

\newcommand{\Ldens}{\ensuremath{\mathcal{L}}}    

\renewcommand{\vector}[1]{\ensuremath{\bm{#1}}}  

\newcommand\restr[2]{{
  \left.\kern-\nulldelimiterspace 
  #1 
  \vphantom{\big|} 
  \right|_{#2} 
  }}

\newbox\usefulbox

\makeatletter
\def\getslant #1{\strip@pt\fontdimen1 #1}

\def\skoverline #1{\mathchoice
 {{\setbox\usefulbox=\hbox{$\m@th\displaystyle #1$}%
    \dimen@ \getslant\the\textfont\symletters \ht\usefulbox
    \divide\dimen@ \tw@
    \kern\dimen@
 \hspace{2pt}   \overline{\mkern -3mu \kern-\dimen@ \box\usefulbox\kern\dimen@ \mkern -1mu}\kern-\dimen@ }}
 {{\setbox\usefulbox=\hbox{$\m@th\textstyle #1$}%
    \dimen@ \getslant\the\textfont\symletters \ht\usefulbox
    \divide\dimen@ \tw@
    \kern\dimen@
 \hspace{2pt}  \overline{\mkern -3mu \kern-\dimen@ \box\usefulbox\kern\dimen@ \mkern -1mu}\kern-\dimen@ }}
 {{\setbox\usefulbox=\hbox{$\m@th\scriptstyle #1$}%
    \dimen@ \getslant\the\scriptfont\symletters \ht\usefulbox
    \divide\dimen@ \tw@
    \kern\dimen@
 \hspace{2pt} \overline{\mkern -3mu \kern-\dimen@ \box\usefulbox\kern\dimen@ \mkern -1mu}\kern-\dimen@ }}
 {{\setbox\usefulbox=\hbox{$\m@th\scriptscriptstyle #1$}%
    \dimen@ \getslant\the\scriptscriptfont\symletters \ht\usefulbox
    \divide\dimen@ \tw@
    \kern\dimen@
 \hspace{2pt}  \overline{\mkern -3mu \kern-\dimen@ \box\usefulbox\kern\dimen@ \mkern -1mu}\kern-\dimen@ }}%
 {}}
\makeatother

\usepackage[toc,nopostdot,nogroupskip,style=long]{glossaries}
\makeindex
\makeglossaries

\newglossaryentry{causalfuturepast}
{
	name=\ensuremath{J^\pm (x)},
	description={causal future/past of point $x$ of manifold $\N$,},
	sort={Jpm}
}

\newglossaryentry{domainofdependence}
{
	name=\ensuremath{D (S)},
	description={domain of dependence of subset $S \subset \N$ of manifold $\N$,},
	sort={Domainofdependence}
}

\newglossaryentry{futurepastdomainofdependence}
{
	name=\ensuremath{D^\pm (S)},
	description={future/past domain of dependence of subset $S \subset \N$ of manifold $\N$,},
	sort={Domainofdependencefuturepast}
}

\newglossaryentry{TxN}
{
	name=\ensuremath{T_x \N},
	description={tangent space at point $x$ of manifold $\N$,},
	sort={Tangentspace}
}

\newglossaryentry{TsxN}
{
	name=\ensuremath{T^*_x \N},
	description={co-tangent space at point $x$ of manifold $\N$,},
	sort={TangentspacestarA}
}

\newglossaryentry{TN}
{
	name=\ensuremath{T \N},
	description={tangent bundle  of manifold $\N$,},
	sort={Tangentspacebundle}
}

\newglossaryentry{TsN}
{
	name=\ensuremath{T^*\N},
	description={co-tangent bundle of manifold $\N$,},
	sort={TangentspacestarB}
}

\newglossaryentry{vectorbundle}
{
	name=\ensuremath{\mathfrak{X}},
	description={smooth vector bundle, usually over the spacetime $\M$,},
	sort={X}
}

\newglossaryentry{whitneysum}
{
	name=\ensuremath{\mathfrak{X} \oplus \mathfrak{Y}},
	description={Whitney sum of the vector bundles $\mathfrak{X}$ and $\mathfrak{Y}$,},
	sort={XplusY}
}

\newglossaryentry{outerproductbundle}
{
	name=\ensuremath{\mathfrak{X} \boxtimes \mathfrak{Y}},
	description={outer product of the vector bundles $\mathfrak{X}$ and $\mathfrak{Y}$,},
	sort={XtimesY}
}

\newglossaryentry{gammax}
{
	name=\ensuremath{\Gamma(\mathfrak{X})},
	description={space of smooth sections of the vector bundle $\mathfrak{X}$,},
	sort={GammaX}
}

\newglossaryentry{gammaxzero}
{
	name=\ensuremath{\Gamma_0(\mathfrak{X})},
	description={space of compactly supported smooth sections of the vector bundle $\mathfrak{X}$,},
	sort={GammaXz}
}

\newglossaryentry{sigma}
{
	name=\ensuremath{\Sigma},
	description={Cauchy surface with Riemannian metric $h$,},
	sort={Sigma}
}

\newglossaryentry{sigmapm}
{
	name=\ensuremath{\Sigma^\pm},
	description={causal future/past of the Cauchy surface $\Sigma$,},
	sort={Sigmapm}
}

\newglossaryentry{DSz}
{
	name=\ensuremath{\Dzs},
	description={space of initial data with respect to Proca's equation,},
	sort={DSz}
}

\newglossaryentry{N}
{
	name=\ensuremath{\N},
	description={manifold of dimension $N$,},
	sort={N}
}

\newglossaryentry{M}
{
	name=\ensuremath{\M},
	description={globally hyperbolic, four-dimensional spacetime with metric $g$ and Cauchy surface $\Sigma$,},
	sort={M}
}

\newglossaryentry{nabla}
{
	name=\ensuremath{\nabla},
	description={Levi-Civita connection (covariant derivative),},
	sort={nabla}
}

\newglossaryentry{omega}
{
	name=\ensuremath{\Omega(\N)},
	description={exterior algebra of differential forms on the manifold $\N$,},
	sort={OmegaN}
}

\newglossaryentry{omegaoz}
{
	name=\ensuremath{\Omega^1_0(\N)},
	description={space of test one-forms on the manifold $\N$,},
	sort={Omegaonezero}
}

\newglossaryentry{omegapz}
{
	name=\ensuremath{\Omega^p_0(\N)},
	description={space of compactly supported $p$-forms on the manifold $\N$,},
	sort={Omegapzero}
}

\newglossaryentry{omegap}
{
	name=\ensuremath{\Omega^p(\N)},
	description={space of $p$-forms on the manifold $\N$,},
	sort={Omegap}
}

\newglossaryentry{omegapdelta}
{
	name=\ensuremath{\Omega^p_{\delta}(\N)},
	description={space of co-closed $p$-forms on the manifold $\N$,},
	sort={Omegapdelta}
}

\newglossaryentry{omegapd}
{
	name=\ensuremath{\Omega^p_{d}(\N)},
	description={space of closed $p$-forms on on the manifold $\N$,},
	sort={Omegapd}
}

\newglossaryentry{deltaomegap}
{
	name=\ensuremath{\delta \Omega^{p+1}(\N)},
	description={space of co-exact $p$-forms on the manifold $\N$,},
	sort={dOmegapafter}
}

\newglossaryentry{domegap}
{
	name=\ensuremath{d \Omega^{p-1}(\N)},
	description={space of exact $p$-forms on the manifold $\N$,},
	sort={dOmegap}
}

\newglossaryentry{delta}
{
	name=\ensuremath{\delta},
	description={interior derivative,},
	sort={delta}
}

\newglossaryentry{d}
{
	name=\ensuremath{d},
	description={exterior derivative,},
	sort={d}
}

\newglossaryentry{dalembert}
{
	name=\ensuremath{\square},
	description={d'Alembert operator,},
	sort={deltaafter}
}

\newglossaryentry{hodge}
{
	name=\ensuremath{*},
	description={Hodge star operator,},
	sort={Hodgedual}
}

\newglossaryentry{innerprod}
{
	name=\ensuremath{\langle \cdot , \cdot \rangle_{\N}},
	description={bilinear map on $\Omega^p(\N)$,},
	sort={innerprod}
}

\newglossaryentry{wedge}
{
	name=\ensuremath{\wedge},
	description={exterior product of differential forms,},
	sort={wedge}
}

\newglossaryentry{levicivita}
{
	name=\ensuremath{\epsilon_{\mu_1 \mu_2 ... \mu_N}},
	description={the fully anti-symmetric tensor of rank $N$ (Levi-Civita symbol) on a $N$-dimensional manifold,},
	sort={epsilon}
}

\newglossaryentry{dvolk}
{
	name=\ensuremath{\dvolk},
	description={volume element with respect to the metric $k$,},
	sort={dvolk}
}

\newglossaryentry{alg}
{
	name=\ensuremath{\Alg},
	description={category of unital *-algebras,},
	sort={Alg}
}
\newglossaryentry{algprime}
{
	name=\ensuremath{\Alg'},
	description={subcategory of unital *-algebras with injective morphisms,},
	sort={Algp}
}
\newglossaryentry{spaccurr}
{
	name=\ensuremath{\SpacCurr},
	description={category of globally hyperbolic spacetimes and background currents,},
	sort={Algp}
}

\newglossaryentry{bualgebra}
{
	name=\ensuremath{\BU(V)},
	description={Borchers-Uhlmann algebra over the vector space $V$,},
	sort={BUalgebra}
}
\newglossaryentry{bumj}
{
	name=\ensuremath{\BUmj},
	description={Borchers-Uhlmann algebra of the Proca field,},
	sort={BUalgebramj}
}

\newglossaryentry{phimf}
{
	name=\ensuremath{\phi_{m,j}(F) },
	description={quantum Proca field of mass $m$,},
	sort={phimf}
}

\newglossaryentry{Km}
{
	name=\ensuremath{K_m },
	description={BU-algebra homeomorphism,},
	sort={Km}
}
\newglossaryentry{Xim}
{
	name=\ensuremath{\Xi_m },
	description={BU-algebra homeomorphism,},
	sort={Xim}
}
\newglossaryentry{Lambdam}
{
	name=\ensuremath{\Lambda_m },
	description={BU-algebra homeomorphism,},
	sort={Lambdam}
}

\newglossaryentry{Gammamj}
{
	name=\ensuremath{\Gamma_{m,j,\varphi}},
	description={BU-algebra homeomorphism,},
	sort={Gammamj}
}
\newglossaryentry{Psimj}
{
	name=\ensuremath{\Psi_{m,j,\varphi}},
	description={BU-algebra homeomorphism,},
	sort={Psimj}
}

\newglossaryentry{involution}
{
	name=\ensuremath{^*},
	description={involution,},
	sort={involution}
}

\newglossaryentry{mass}
{
	name=\ensuremath{m},
	description={mass,},
	sort={mass}
}

\newglossaryentry{current}
{
	name=\ensuremath{j},
	description={external current,},
	sort={j}
}

\newglossaryentry{rhoz}
{
	name=\ensuremath{\rhoz},
	description={initial data mapping operator, pullback,},
	sort={rho1}
}
\newglossaryentry{rhod}
{
	name=\ensuremath{\rhod},
	description={initial data mapping operator, forward normal derivative,},
	sort={rho2}
}
\newglossaryentry{rhodelta}
{
	name=\ensuremath{\rhodelta},
	description={initial data mapping operator, pullback of the divergence,},
	sort={rho3}
}
\newglossaryentry{rhon}
{
	name=\ensuremath{\rhon},
	description={initial data mapping operator, forward normal,},
	sort={rho4}
}

\newglossaryentry{inclusionmap}
{
	name=\ensuremath{i},
	description={inclusion operator, usually $i:\Sigma \hookrightarrow \M$,},
	sort={inclusion}
}

\newglossaryentry{Epm}
{
	name=\ensuremath{E_m^\pm},
	description={retarded/advanced fundamental solution of $(\square + m^2)$,},
	sort={Epm}
}

\newglossaryentry{Em}
{
	name=\ensuremath{E_m},
	description={advanced minus retarded fundamental solution of $(\square + m^2)$,},
	sort={Epm2}
}

	\newglossaryentry{Ezcurly}
{
		name=\ensuremath{\mathcal{E}_0},
		description={propagator of $\square$,},
		sort={Epm2}
}

\newglossaryentry{Gpm}
{
	name=\ensuremath{G_m^\pm},
	description={retarded/advanced fundamental solution of $(\delta d + m^2)$,},
	sort={Gpm}
}

\newglossaryentry{Gm}
{
	name=\ensuremath{G_m},
	description={advanced minus retarded fundamental solution of $(\delta d + m^2)$,},
	sort={Gpm2}
}

\newglossaryentry{Gmcurly}
{
	name=\ensuremath{\mathcal{G}_m},
	description={propagator of $(\delta d + m^2)$,},
	sort={Gpm3}
}

\newglossaryentry{E}
{
	name=\ensuremath{E},
	description={real vector space of real-valued test one-forms,},
	sort={E}
}

\newglossaryentry{WmF}
{
	name=\ensuremath{W_m(F)},
	description={Weyl element for $F \in E$,},
	sort={WmF}
}

\newglossaryentry{Wmtilde}
{
	name=\ensuremath{\widetilde{\W}_m},
	description={span of Weyl elements,},
	sort={Wmtilde}
}

\newglossaryentry{Wm}
{
	name=\ensuremath{\W_m},
	description={Weyl algebra, for mass $m$, generated over $(E,\Gm{})$,},
	sort={Wm}
}
\newglossaryentry{CEGm}
{
	name=\ensuremath{\mathcal{C}(E,\Gm{})},
	description={convex set of states on the pre-symplectic space $(E,\Gm{})$,},
	sort={CEGm}
}
\newglossaryentry{omegaC}
{
	name=\ensuremath{\omega_C},
	description={state on the span $\widetilde{W}_m$ of Weyl elements corresponding to a state $C$ on the pre-symplectic space $(E,\Gm{})$,},
	sort={OmegaC}
}

\usepackage[backend=biber, style=numeric, citestyle=numeric, language=english, sortlocale=de_DE, natbib=true, url=false, doi=false, eprint=false, isbn=false, maxbibnames=99, sorting=nyt, firstinits=true]{biblatex}  
\renewbibmacro{in:}{}
\DeclareFieldFormat[article, book, misc, incollection]{title}{\textit{#1}}
\DeclareFieldFormat[article, incollection]{journaltitle}{#1}
\DeclareFieldFormat[article, incollection]{pages}{#1}
\DeclareFieldFormat[article, incollection]{volume}{\textbf{#1}}
\DeclareFieldFormat[incollection]{booktitle}{In: \textit{#1}}

\begin{document}

\begin{titlepage}
{\sffamily \centering
{
{\LARGE  {Universität Leipzig}}\\[2mm]
{\large {Fakultät für Physik und Geowissenschaften}}\\[2mm]
{\large {Institut für Theoretische Physik}}\\[2cm]
\textcolor{black}{\hrule height 1pt} \vspace{2.5mm}
}
{\LARGE \noindent
{Quantization of the Proca field in curved spacetimes\\  - \\A study of mass dependence and the zero mass limit}
\vspace{3.5mm}
\textcolor{black}{\hrule height 1pt}
\par}
{
\vspace{2cm}
{\huge \textbf {Masterarbeit}}\\[3mm]
{\large {zur Erlangung des akademischen Grades}}\\[2mm]
{\large {Master of Science (M. Sc.)}}\\[2mm]
{\large {eingereicht am 31. Mai 2016}}
}

\vspace{9cm}

  {
\begin{tabularx}{\textwidth}{lXr}
eingereicht von: & & Gutachter: \\
Maximilian Schambach & & Prof. Dr. Stefan Hollands \\
geboren am 16.02.1990 & & Dr. Jacobus Sanders \\
in Kassel & &
\end{tabularx}
  }
}
\end{titlepage}

\clearpage{\pagestyle{empty}\cleardoublepage}
\newpage
\clearpage
\pagenumbering{roman}
\onehalfspacing
\section*{Abstract}
In this thesis we investigate the Proca field in arbitrary globally hyperbolic curved spacetimes. We rigorously construct solutions to the classical Proca equation, including external sources and without restrictive assumptions on the topology of the spacetime, and investigate the classical zero mass limit. We formulate necessary and sufficient conditions for the limit to exist in terms of initial data. We find that the limit exists if we restrict the class of test one-forms, that we smear the distributional solutions to Proca's equations with, to those that are co-closed, effectively implementing a gauge invariance by exact distributional one-forms of the vector potential. In order to obtain also the Maxwell dynamics in the limit, one has to restrict the initial data such that the Lorenz constraint is well behaved. With this, we naturally find conservation of current and the same constraints on the initial data that are independently found in the investigation of the Maxwell field by Pfenning.\par
For the quantum problem we first construct the generally covariant quantum Proca field theory in curved spacetimes in the framework of Brunetti, Fredenhagen and Verch and show that the theory is local. Using the Borchers-Uhlmann algebra and an initial value formulation, we define a precise notion of continuity of the quantum Proca field with respect to the mass. With this notion at our disposal we investigate the zero mass limit in the quantum case and find that, like in the classical case, the limit exists if and only if the class of test one-forms is restricted to co-closed ones, again implementing a gauge equivalence relation by exact distributional one-forms. It turns out that in the limit the fields do not solve Maxwell's equation in a distributional sense. We will discuss the reason from different perspectives and suggest possible solutions to find the correct Maxwell dynamics in the zero mass limit.

\newpage
\clearpage{\pagestyle{empty}\cleardoublepage}
\newpage
\clearpage \singlespacing
\tableofcontents
\newpage
\clearpage{\pagestyle{empty}\cleardoublepage}
\newpage
\clearpage
\pagestyle{fancy}
\pagenumbering{arabic}


\onehalfspacing
\hypersetup{linkcolor=linkblue}

\section{Introduction and Motivation}\label{chpt:introduction}
As of now there exist two very well tested theories describing two highly diverse realms of the vast landscape of physical phenomena: That is, on the one hand the theory of gravitation, called \emph{General Relativity} (GR), and on the other hand the \emph{Standard Model of Particle Physics}, describing the remaining three of the four known fundamental interactions, namely the electromagnetic, weak and strong interaction. \par
GR is a \emph{classical} field theory and describes gravitational large scale phenomena, as for example observed in astronomy, and provides our current understanding of the universe as an increasingly expanding one originating from a Big Bang. It was introduced by \name{Einstein} in the early twentieth century and has since been intensively tested, in its scope of application, and confirmed to be valid up to astonishing accuracy\footnote{Just this year, 2016, one of the last predictions of GR lacking experimental confirmation, gravitational waves, have been confirmed by the LIGO Scientific and Virgo Collaboration \cite{grav_waves_detection}.}. GR is a generalization of Einstein's theory of special relativity, which itself generalizes the principles of Newton's classical mechanics and was needed to account for the experimentally confirmed principle that the speed of light has the same constant value for all observers, even when moving relatively to each other. This counter intuitive fact changed the physical perception of space and time. In GR, gravitation is indirectly described via the curvature of spacetime, a four-dimensional space consisting of the observed three spatial dimensions together with one dimension describing time. According to GR, mass and energy, which are considered equivalent, curve the initially flat spacetime, similar to a rubber surface being deformed when putting masses on it. The connection between the curvature of spacetime and mass, or, more precisely, between the Einstein tensor and the stress-energy tensor, is described by Einstein's field equations.\par
The Standard Model on the other hand is a \emph{quantum} field theory (QFT) and unifies the description of electromagnetic interaction (quantum electrodynamics), weak interaction (quantum flavourdynamics) and strong interaction (quantum chromodynamics). It describes short scale and subatomic physical phenomena and has also been confirmed to very high accuracy. At very short scales, matter behaves very differently to what we are used to from our own perception of our surrounding world. In particular, experiments regarding the spectra of excited gases, the photoelectric effect and the so called Rutherford scattering led to a quantum description of matter, which includes a probabilistic behavior of observables.  In QFT, matter is described by quantum \emph{fields}, fulfilling non-trivial commutation relations that were abstracted from the earlier theory of Quantum Mechanics, which was also introduced in the early twentieth century. Under certain circumstances\footnote{For example in the case of free fields or the asymptotic ``in" and ``out" states of scattering processes.}, one can think of these fields as particles, called elementary particles. Even though the Standard Model really is a theory of fields rather than particles, one often uses the two words equivalently. In that sense, there are two classes of particles, the \emph{fermions} (for example electrons or neutrinos) with half integer spin, making up most of the matter around us, and the \emph{bosons} (for example photons or gluons) with integer spin. In the Standard Model, the gauge bosons are the transmitter of the field interactions, for example the photon (the ``light particle") is the transmitter of the electromagnetic interaction. The corresponding quantum field is a quantized version of Maxwell's electromagnetic field describing classical electromagnetism.\par
Even though both theories by themselves have been tremendously successful, it is a priori clear that neither of them describes all of the physical phenomena. While in most (terrestrial) microscopic scenarios gravity, being almost {32} orders of magnitude weaker then the weak interaction, can be neglected, it should play a role in extreme astronomical situations, for example near black holes or at the very early times in the beginning of the universe. Moreover, since matter is responsible for gravitation and is itself made up of elementary particles, there should exist a quantum description of gravitation. Also, there are observed phenomena that both of the theories cannot explain: investigating the rotation speeds of galaxies one finds that the observable mass in the universe cannot account for the measured speeds alone. It turns out that actually about one third of the gravitational matter is \emph{not observable}, that is, only interacting via gravitation and none of the other known interactions. This so called \emph{dark matter} is not described by the Standard Model.
Physicists have therefore tried to find a \emph{Theory of Everything} (ToE) for example by \emph{unifying} the Standard Model and the theory of gravitation into one theory describing all known interactions. So far, all attempts on formulating a quantum description of gravity and unifying it with the Standard Model have failed, ranging from early work by Kaluza \cite{kaluza1921}, Klein \cite{Klein1926} and Bronstein \cite{bronstein1936} to work in the 1960's and 1970's where it became clear that GR, as a QFT, is non-renormalizable \cite{tHooft1973,tHooft1974,Deser1974}, that is, simplifying, it yields unphysical infinite measurement results which cannot be brought under control.  There are also some alternative approaches, not based on QFT, to find a ToE, for example theoretical frameworks collected under the name \emph{string theories}, that seemed promising at first but failed to provide a consistent description of physical phenomena. Even though some important physicists, most prominently \name{Edward Witten}, claim that string theory, or the different parts of the underlying M-theory, is the correct theory to describe all observable physical phenomena, there is a lot of criticism against it. Many critics, prominent figures being \name{Lee Smolin} and \name{Peter Woit}, claim that, while string theory provides elegant and beautiful ideas about physics and mathematics, it lacks a clear description as a theory. It is said to provide only some fragmental descriptions and ideas and, most severely, lacks to be a physical theory a priori as it cannot be falsified: As there is an infinite number of possibilities to compactify the excrescent dimensions\footnote{The mathematical formalism of M-theory only works in {10} rather then the four dimensions that we observe.} and there is no preferred principle, string theory provides a description of \emph{all} possible physical theories and can always be adapted when in conflict with experiments and therefore cannot provide any insight or predictions at all.\par
Instead of constructing a ToE, one might therefore take a step back and try to approximately describe scenarios in which quantum matter is under the influence of gravitation, or find a quantum description of gravitation without unifying it with the other fundamental interactions. In doing so, one hopes to find and understand basic underlying principles that a ToE ought to have. In this thesis we will investigate quantum fields in curved spacetimes, that is, quantum fields under the influence of gravitation, and neglect the influence that the quantum fields themselves have on gravitation. Early investigations of quantum fields in curved spacetimes include the investigation of the influence of an expanding universe on quantum fields and its connection to particle creation by Parker \cite{Parker1969}, the study of radiating black holes, most successfully by Hawking in the 1970's \cite{Hawking1975}, and the description of what is now called the Unruh effect \cite{Unruh1976}. Investigating quantum fields on curved rather than flat spacetime as one does in the Standard Model, one is forced to rethink the underlying principles of QFT. In particular, QFT usually relies heavily on symmetries of the underlying spacetime, such as time translation and Lorentz invariance, implementing the special relativistic effects of quantum mechanics. Searching for Hilbert space representations of the canonical commutation relations (CCR) together with a unitary representation of the Lorentz group, one finds many (unitary) equivalent possibilities and picks out a convenient one specified by a vacuum state - the unique state that is Lorentz invariant.  In a general spacetime, such a global symmetry is of course not present. Hence, the different possibilities of the Hilbert space representation are not equivalent anymore and there is no preferred vacuum state. One therefore takes a different approach and formulates the theory purely algebraically, independent of any Hilbert space representation. This algebraic description of quantum field theory (originally on flat spacetime), AQFT, was studied and axiomatized by \name{Haag} and \name{Kastler} \cite{HaagKastler1964}. From the algebraic description one can construct the corresponding Hilbert space formulation via the so called GNS construction, named after Gelfand, Naimark and Segal. Dyson \cite{dyson1972} realized that the algebraic approach to QFT is suitable for a generalization to account for general covariance. Together with a generalization of the spectrum condition, known as the microlocal spectrum condition \cite{Brunetti1996}, the framework has then been further refined, leading to a categorical formulation of Quantum Field Theory on Curved Spacetimes (QFTCS) by \name{Brunetti, Fredenhagen} and \name{Verch} \cite{Brunetti_Fredenhagen_Verch}. Details on the principles and development of QFTCS can be found in the literature \cite{wald_QFT,baer_ginoux_pfaeffle,wald_hollands_review}.\par
In this thesis we investigate the Proca field in curved spacetimes. The Proca field is a massive vector\footnote{That is, it has spin one.} field first studied by \name{Proca} \cite{proca_original} as the most straightforward massive generalization of the electromagnetic field. Since the photon associated with the electromagnetic field has mass zero, Proca's theory is also called \emph{massive electrodynamics}. On a classical level, it can be used experimentally to find a lower bound of the photon mass. Assuming Proca's equation to describe electromagnetism, one finds that the corresponding electric potential is of Yukawa rather then Coulomb type as it is for a massless photon. Experimentally, one finds at very high accuracy at many orders of length-scales that the electric potential is indeed of Coulomb and not of Yukawa type. With these and other sophisticated methods the photon mass has been determined to be smaller then $4 \times 10^{-51}$\,kg (see \cite[Section I.2]{jackson})\footnote{More recent studies even suggest the bound to be lowered to $1.5 \times 10^{-54}$\,kg \cite{photon_mass}.}. Besides the photon there are several other elementary particles that are described by vector fields. In fact, all gauge bosons in the Standard Model are vector bosons. While the photon and the gluons\footnote{The gluons are the transmitters of the strong interaction.} are massless, the gauge bosons of the weak interaction, the W- and Z-bosons, are massive vector fields and may be described using Proca's equation\footnote{Of course, this is not the case in the Standard Model, as the gauge bosons are by construction massless and only appear massive by their interaction with the Higgs field.}. Further examples of massive vector fields include certain mesons, for example the $\omega$- or the $\varphi$-meson. It is thus desirable to study Proca's equation in a curved spacetime. This was first done by Furlani \cite{FURLANI} in the case of vanishing external sources and under a restrictive assumption on the topology of the spacetime\footnote{In particular, Furlani assumes the Cauchy surface of the spacetime to be compact.}. We are going to formulate the theory as general as possible, including external sources and without topological restrictions. More importantly, we are interested in the zero mass limit of the theory. As we shall see at many points, the massive (Proca-) and the massless (Maxwell-) theory differ enormously in detail. Most severely, the massless theory possesses a gauge invariance while the massive theory does not. While there have been several studies regarding the Maxwell field in curved spacetimes \cite{Sanders,pfenning,dimock1992quantizedEM,Dappiaggi2012}, there are questions regarding locality and the choice of gauge that remain open for discussion. In flat spacetimes, these questions do not arise as the topology is trivial, therefore, in particular, all closed $p$-forms are exact, as we will discuss later in more detail. In curved spacetimes, choosing the vector potential as the fundamental physical entity rather than the field strength tensor, it is a priori not clear if the gauge invariance by closed distributional one-forms is too general to account for all physical phenomena. As argued in \cite{Sanders}, implementing a gauge invariance by closed distributional one-forms rather than exact ones, one cannot capture experimentally established phenomena like the Aharonov-Bohm effect. Furthermore, one finds that the quantum Maxwell theory is not local, as opposed to the quantum Proca theory, and one might look for alternative implementations of locality in the theory. Recent proposals, with emphasis on the question of formulating the same physics in all spacetimes, are discussed in \cite{Fewster2012, Fewster2016}.
One reason to look at the massless limit of the Proca theory is to find answers to these questions naturally arising in the limiting procedure. In the zero mass limit, we indeed find a natural gauge invariance by exact rather than closed distributional one-forms. Questions concerning locality in the limit remain open for further investigations as they are not discussed in detail in this thesis. As a first step, our investigation will be purely based on observables, states are not included in the description. It should in principle be possible to extend the presented framework to include states. Throughout this thesis we work in natural units, that is, in particular we set $\text c = 1 = \hbar$.\par
The structure of this thesis is as follows:
In Chapter \ref{chpt:preliminiaries} we will recap some basic mathematical notations and definitions. Most of the discussion is kept rather brief as it is expected that the reader is familiar with the basics of differential geometry as it is the mathematical framework of GR. We will recap some notions regarding the spacetime geometry and vector bundles. In a bit more detail, we discuss differential forms as it is usually not part of the curriculum for physicists and the used formulation relies heavily on it. Furthermore, we give a brief overview of hyperbolic partial differential operators and their connection to global hyperbolicity. We conclude the first chapter by introducing some basics of category theory and a summary of the chosen sign conventions. \par
In Chapter \ref{chpt:classical} we will investigate the classical problem. We will find solutions to the classical Proca equation including external sources as a generalization of the work by Furlani \cite{FURLANI} by decomposing Proca's equation into a hyperbolic differential equation and a Lorenz constraint. We will then solve the hyperbolic equation and implement the constraint by restricting the initial data. As a foundation of understanding the quantum problem, we will investigate the classical zero mass limit. As it turns out, the existence of the zero mass limit in the quantum case is deeply connected to the classical one. \par
In Chapter \ref{chpt:quantum} we study the quantum problem. First, we will construct the generally covariant quantum Proca field theory in curved spacetimes in the categorical framework of Brunetti, Fredenhagen and Verch and show that the obtained theory is local, as opposed to the quantization of Maxwell's field (see \cite{Dappiaggi2012,Sanders}). Using the Borchers-Uhlmann algebra, we will define a appropriate notion of continuity of quantum fields with respect to the mass and will ultimately investigate the zero mass limit of the quantum Proca theory. It turns out that in the zero mass limit, the quantum fields do not solve Maxwell's equation in a distributional sense. We will discuss the reason from several perspectives and possible solutions. \par
A conclusion and outlook is presented in Chapter \ref{chpt::conclusion}. Our previous attempts that we formulated using a C*-algebraic approach to find a notion of continuity of the quantum Proca theory are presented in Appendix \ref{app:weyl-algebra}. We will discuss that this approach is not suited for the investigation of the zero mass limit but nevertheless present the results obtained along the way as they contain some mathematical results on continuous families of pre-symplectic spaces that have to our knowledge not been discussed in the literature. For clarity, some of the mathematical work needed along the investigation is put in Appendix \ref{app:lemmata} - despite their crucial importance for the results. Finally, a list of used symbols and references can be found at the very end of this thesis.

\section{Preliminaries}\label{chpt:preliminiaries}
In this chapter we will introduce some of the mathematical background and notation needed for this thesis. In particular, we will shortly introduce the differential geometric description of spacetime in Section \ref{sec:spacetime_geometry} and give an introduction to the notion of global hyperbolicity and its connection to Green- and normally-hyperbolic operators in Section \ref{sec:global_hyperbolicity}. In a bit more detail, we will introduce the notion of differential forms and give explicit definitions, also in terms of an index based notation, in Section \ref{sec:differential_forms}. For completeness, in Section \ref{sec:cat-theory}, we present basic definitions of category theory. The reader familiar with these topics can safely skip this chapter and refer to it when interested in the chosen conventions.
%
%
%
%
%
%
%
\subsection{Spacetime geometry}\label{sec:spacetime_geometry}
In GR, the universe is mathematically described as a four dimensional \emph{spacetime}, consisting of a smooth, four dimensional manifold \gls{M} (assumed to be Hausdorff, connected, oriented, time-oriented and para-compact) and a Lorentzian metric $g$. We will assume the signature of the Lorentzian metric $g$ to be $(-,+,+,+)$. The Levi-Civita connection on $(\M,g)$ is as usual denoted by \gls{nabla}.
Throughout this thesis, we treat spacetime as fixed, implementing a gravitational background determined classically by Einstein's field equations. Hence, we neglect any back-reaction of the fields on the metric, both in the quantum and the classical case. In that sense, we treat the fields as \emph{test fields}.\par
For the basic mathematical theory regarding Lorentzian manifolds, we refer to the literature: An introduction to the topic with an emphasis on the physical application in GR is for example given in \cite{wald_GR} and \cite{carroll_spacetime-and-gr}.
Here, we will shortly recap the notion of a tangent space and tangent bundle and generalize to the notion of a vector bundle which we will use in the general description of normally hyperbolic operators and differential forms.
In the following, we generalize the setting to an arbitrary smooth manifold $\N$ of dimension $N$ with either Lorentzian or Riemannian metric $k$.\par
A \emph{tangent vector} $v_x$ at point $x \in \N$ is a linear map $v_x : C^\infty(\N , \IR) \to \IR$ that obeys the Leibniz rule, that is, for $f,g \in C^\infty (\N,\IR)$ it holds $v_x(fg) = f(x)v_x(g) + v_x(f)g(x)$.
We define the \emph{tangent space} \gls{TxN} of $\N$ at $x$ as the real $N$-dimensional vector space of all tangent vectors at point $x$.
The disjoint union of all tangent spaces is called the \emph{tangent bundle} \gls{TN} of $\N$ and is itself a manifold of dimension $2N$. A \emph{vector field} is a map $v: \N \to T\N$ such that $v(x) \in T_x\N$.
The respective dual spaces, that is the space of all linear functionals, the \emph{co-tangent space} and the \emph{co-tangent bundle}, are denoted by \gls{TsxN} and \gls{TsN} respectively.\par
For Lorentzian manifolds, we call a tangent vector $v$ at $x \in \N$ \emph{timelike} if $k_{\mu \nu} v^\mu v^\nu < 0$, \emph{spacelike} if $k_{\mu \nu} v^\mu v^\nu > 0$ and \emph{null} (or lightlike) if $k_{\mu \nu} v^\mu v^\nu = 0$. At every point $x \in \N$, we define the set of all \emph{causal}, that is, either timelike or null, tangent vectors in the tangent space at $x$. This set is called the \emph{light cone} at $x$ and it is split up into two distinct parts, one that we call the future light cone, and one that we call the past light cone at $x$. Since we assume the manifold to be time orientable, there exists a smooth vector field $t$ that is timelike at every $x \in \N$. Given this time orientation, we identify the future (past) light cone with the set of tangent vectors $v \in T_x\N$ such that $k_{\mu\nu} v^\mu t^\nu < 0$ (respectively $> 0$). Therefore, a tangent vector $v$ at $x$ is called \emph{future directed} (past directed) if it lies in the future (past) light cone at $x$.\\
Accordingly, a curve $\gamma : I \to \N$ is called timelike (spacelike, null, causal, future or past directed) if its tangent vector $\dot{\gamma}$ is timelike (spacelike, null, causal, future or past directed) at every $x \in \N$.  For every point $x \in \N$ we define the \emph{causal future/past} \gls{causalfuturepast} of $x$ as the set of all points $q \in \N$ that can be reached by a future directed causal curve originating in $x$. For any subset $S \in \N$ we define $J^\pm (S) = \bigcup_{x \in S} J^\pm(x)$ and $J(S) = J^+(S) \cup J^- (S)$. Finally, the future/past domain of dependence $\gls{futurepastdomainofdependence}$ of a set $S \subset \N$ is the set of all points $x \in \N$ such that every inextendible causal curve through $x$ intersects $S$. The \emph{domain of dependence} \gls{domainofdependence} of $S$ is the union of the future and past domain of dependence of the set $S$.
For more details on the causal structure of spacetime we refer to for example \cite[Chapter 8]{wald_GR}.\par
The notion of tangent bundles can be generalized to the notion of a vector bundle. Instead of ``attaching'' the vector spaces $T_x \N$ to every point $x$ of the manifold, we allow for the occurrence of arbitrary vector spaces, called the fibres of the vector bundle. A vector bundle then consists of the base manifold, in our case $\N$, the total space and a map $\pi$ from the total space to the base manifold, that can be locally trivialized. At each point of the base manifold, the pre-image of $\pi$ is the fibre of the vector bundle. To be precise we define, following \cite{rudolph_schmidt}:
\begin{definition}[Vector bundle]
	A smooth \emph{vector bundle} over $\N$ is a tuple $\gls{vectorbundle} = (E,\N, \pi)$, where $E$ is a smooth manifold and $\pi : E \to \N$ is a smooth surjective map satisfying:
	\begin{enumerate}
		\item For every $x \in \N$, $\pi^{-1}(x)$ is a vector space, called the fibre of the bundle at point $x$.
		\item There exists a finite dimensional vector space $F$, an open covering $\left\{ U_\alpha\right\}_\alpha$ of $\N$ and a family of diffeomorphisms $\chi_\alpha : \pi^{-1}(U_\alpha) \to U_\alpha \times F$ such that for all $\alpha$ it holds $\chi_\alpha \comp \text{pr}_1 =  \restr{\pi}{\pi^{-1}(U_\alpha)}$ and for every $x \in \N$ the map $\text{pr}_2 \comp \restr{\chi_\alpha}{\pi^{-1}(x)} : \pi^{-1}(x) \to F$ is linear.
	\end{enumerate}
\end{definition}
Here, the maps $\text{pr}_1$ and $\text{pr}_2$ denote the projection onto the first respectively second component of an element in $U_\alpha \times F$. The properties graphically mean that \emph{locally}, the vector bundle ``looks like" the product of the base manifold with the fibre. The tuples $(U_\alpha, \chi_\alpha)$ are called \emph{local trivializations} of the vector bundle. Like for vector spaces, we can define the sum and product of vector bundles, by using the according vector space definitions on the fibres of the bundle.\par
Let $\mathfrak{X}, \mathfrak{Y}$ be vector bundles over $\N$ with fibres $X_x$ and $Y_x$ at $x \in \N$. We denote by \gls{whitneysum} the \emph{Whitney sum} of the two vector bundles - the vector bundle over $\N$ whose fibres are given by the direct sum $X_x \oplus Y_x$. Similarly, one obtains the local trivializations of the Whitney sum from the trivializations of $\mathfrak{X}, \mathfrak{Y}$ and direct sums.\par
Accordingly, let $\mathfrak{X}, \mathfrak{Y}$ be vector bundles over $\N$ and $\widetilde{\N}$, with fibres $X_x$ and $Y_{\tilde{x}}$ at $x \in \N$, $\tilde{x} \in \widetilde{\N}$ respectively. We denote by \gls{outerproductbundle} the \emph{outer product} of the two vector bundles - the vector bundle over $\N \times \widetilde{\N}$ whose fibres are given by the tensor products $X_x \otimes Y_x$. Similarly, one obtains the local trivializations of the outer product from the trivializations of $\mathfrak{X}, \mathfrak{Y}$ and tensor products. \par
Finally, we generalize the notion of vector fields:
\begin{definition}[Sections of vector bundles]
Let $\mathfrak{X}=(E,\N,\pi)$ be a vector bundle with fibres $X_x=\pi^{-1}(x)$ at $x \in \N$. A \emph{smooth section} of the vector bundle is a smooth map $\gamma : \N \to E$ such that $\gamma(x) \in X_x$ for all $x \in \N$. The \emph{vector space of smooth sections} of $\mathfrak{X}$ is denoted by \gls{gammax}, the one with compactly supported sections is as usual denoted by \gls{gammaxzero}.
\end{definition}
In this language, a vector field $v$ is just a smooth section of the tangent bundle of a manifold, $v \in \Gamma(T\N)$. One may therefore identify the physical notion of fields with smooth sections of vector bundles. This point of view will be used to define the notion of differential forms in Section \ref{sec:differential_forms}.\par
In this thesis, we usually are interested in complex valued functions (or sections in general). Therefore, we view all occurring vector bundles as complex, in the sense that we take two distinct copies of the vector bundle, one representing the real, one the imaginary part of the bundle. A section of that complex vector bundle is just a pair of two sections of the real vector bundle under consideration. From now, if not specified explicitly, we will view all vector bundles, including the tangent bundle $T\N$, as complex vector bundles. Accordingly, smooth sections of those bundles will in general be complex valued.
%
%
%
%
%
%
%
%
%
%
%
\subsection{Partial differential operators and global hyperbolicity}\label{sec:global_hyperbolicity}
When dealing with field theories, whether classical or quantum, one is, of course, interested in the dynamics of the fields. These are usually described by some partial differential equation, often of second order. In the following, we give a short introduction to the theory of certain partial differential operators acting on smooth sections of a vector bundle over the spacetime $(\M,g)$.\par
As we have seen, these smooth sections are generalizations of the notion of a field.  In the following, let $\mathfrak{X}$ denote a vector bundle over the manifold $\M$ and let $P: \Gamma(\mathfrak{X}) \to \Gamma(\mathfrak{X})$ be a partial differential operator acting on smooth sections of the bundle. As in the case of flat spacetime, we are interested in basic questions regarding the differential equation $Pf = j$, for example: Can we formulate a (globally) well posed initial value problem? Does the differential equation possess (unique) solutions? To answer these questions, we will now restrict to the case where $P$ is linear and of second order, as it is often the case in physical applications. One can show that for a certain class of such operators, namely normally hyperbolic partial differential operators of second order, we can rigorously treat these questions.\par
Choosing local coordinates $x=(x_\mu)$ on $\M$ and a local trivialization of $\mathfrak{X}$, a linear partial differential operator of second order is called \emph{normally hyperbolic} if it takes the form
\begin{align}
	P = - \sum_{\mu,\nu} g^{\mu \nu} \partial_\mu \partial_\nu + \sum_{\alpha} A_\alpha (x) \partial_\alpha + B(x) \formspace,
\end{align}
where $A_\alpha$ and $B$ are matrix-valued coefficients depending smoothly on the coordinate $x$ (see. \cite[Chapter 1.5]{baer_ginoux_pfaeffle}). One can also formulate a coordinate independent definition in terms of the principal symbol, which we will not present here (see for example \cite[Section 1.5]{baer_ginoux_pfaeffle} ). \par
Normally hyperbolic operators possess unique fundamental solutions (see for example the fundamental solutions to the wave operator as noted in Lemma \ref{lem:fundamental_solution_wave_operator}). These fundamental solutions fulfill certain physically important properties, such as a finite propagation speed smaller than the speed of light. Furthermore, specifying the initial data on some space-like hypersurface $X \in  \M$ specifies a unique solution on the domain of dependence $D(X)$ of $X$. Due to these properties, one often calls normally hyperbolic operators just \emph{wave operators}. But to state a \emph{globally} well posed initial value problem for a wave equation, we need to restrict the class of spacetimes $\M$ under consideration to those that possess space-like hypersurfaces $X$ whose domain of dependence is all of the spacetime, $D(X) = \M$. This leads to the notion of \emph{globally hyperbolic} spacetimes:
\begin{definition}[Global Hyperbolicity]
	A spacetime $\M$ is called \emph{globally hyperbolic} if there exists a Cauchy surface $\gls{sigma}$ in $\M$.
\end{definition}
\noindent Here, a Cauchy surface is a space-like hypersurface $\Sigma \subset \M$ such that every inextendible causal curve $\gamma$ intersects $\Sigma$ exactly once. One can show that Cauchy surfaces fulfill the desired property mentioned above, that is,  $D(\Sigma) = \M$. Furthermore, one can show that any globally hyperbolic spacetime $\M$ is foliated by a one-parameter family $\left\{ \Sigma_t \right\}_t$ of Cauchy surfaces (see for example \cite[Theorem 8.3.14]{wald_GR}). \par
In physical applications, one often finds the dynamics of a theory to be described by wave operators. Most prominently, the Klein-Gordon operator $(\square + m^2)$ acting on scalar fields, or its generalization, the wave operator acting on differential forms introduced in Section \ref{sec:differential_forms}, is normally hyperbolic. But there are also important physical field theories that are not described by wave operators, such as the Proca field treated in this thesis. It turns out that the Proca operator (see Definition \ref{def:proca_operator}) is a so called \emph{Green-hyperbolic} operator. These are again partial differential operators $P$ of second order acting on smooth sections of some vector bundle, such that $P$ (and its dual $P'$) posses fundamental solutions. Obviously, normally hyperbolic operators are Green-hyperbolic, but the opposite is not true. One can generalize some results obtained by studying normally hyperbolic operators to Green-hyperbolic operators. An introduction to this topic is given in \cite{baer_green-hyperbolic}, where it is also shown that the Proca operator is Green-hyperbolic but not normally hyperbolic.\par
For our application, the notion of Green-hyperbolicity is not of vast importance, but it is worth mentioning that there exists a more detailed mathematical background on the treatment of such operators.
A very detailed description of normally hyperbolic operators on Lorentzian manifolds, including proofs of the above statements regarding the initial value problem and the existence of fundamental solutions, is given in \cite{baer_ginoux_pfaeffle}, also with an overview of quantization. A shorter introduction to the topic is for example treated in \cite{baer-ginoux_classical-and-quantum-fields}, also with a description of quantization.
%
%
%
%
%
%
%
%
%
%
%
%
\subsection{Differential forms}\label{sec:differential_forms}
Differential forms provide an elegant, coordinate independent description of calculus on smooth manifolds. In particular, they generalize the notion of line- and volume-integrals that are known from analysis. Differential forms play a remarkable role in physics, as one can argue that they indeed describe fundamental physical entities. As an example, instead of viewing a classical force as a vector, one can think of it, more closely related to experiments, as a differential one-form that assigns a scalar to a tangent vector of a curve. This scalar is the (infinitesimal) work associated with the force along the curve. Also, differential forms allow for an elegant geometric description of field theories, for example the Maxwell and Proca field theories that we encounter in this thesis. In Maxwell's classical theory of electromagnetism, instead of viewing the electric and magnetic field (which are conceptually just forces) as the fundamental physical entities, one introduces the \emph{vector potential}, a one-form, consisting of the scalar electric potential and the vector potential associated with the magnet field. Experiments like the Aharonov-Bohm experiment allow for an interpretation of the vector potential as the fundamental physical object, rather than the associated electromagnetic field. \\
Even more fundamentally, the two main theories of physics, General Relativity and the Standard Model of particle physics, are field theories. They are deeply connected to a geometric interpretation and can be elegantly described using differential forms. \par
Despite of all this, differential forms are usually not part of the standard curriculum of physicists. We shall therefore introduce the basic aspects and definitions regarding differential forms that are used in this thesis. For a more detailed introduction we refer to the literature: For example \cite[Chapter 2 and 4]{rudolph_schmidt} or \cite[Appendix B]{wald_GR} provide introductions to the topic.\par
In the following, let $\N$ denote a smooth $N$-dimensional manifold, assumed to be Hausdorff, connected, oriented and para-compact, with either Lorentzian or Riemannian metric $k$ and Levi-Civita connection $\nabla$. For a Lorentzian manifold we use the sign convention $(-,+,\dots,+)$ of the metric $k$. The number of negative eigenvalues of $k$ is denoted by $s$, so $s=0$ for a Riemannian manifold and, in our convention, $s=1$ for a Lorentzian manifold.
Later, we will specify to a four dimensional (globally hyperbolic) spacetime consisting of a four dimensional manifold $\M$ with Lorentzian metric $g$ and Cauchy surface $\Sigma$ with induced Riemannian metric $h$.
We define:
\begin{definition}[Differential form]
	Let $p\in \{0,1,\dots,N\}$. A \emph{differential form} $\omega$ of degree $p$, or $p$-form for short, on the manifold $\N$ is an anti-symmetric tensor field of rank $(0,p)$. That is, at every point $x \in \N$, $\omega_x$ is an anti-symmetric multi-linear map
	\begin{align}
	\omega_x : \underbrace{T_x \N \times T_x \N \times \cdots \times T_x \N}_{p\text{-times}} \to \IR \formspace.
	\end{align}
	We denote the vector space\footnote{Naturally, addition and scalar multiplication are defined point-wise.} of $p$-forms on $\N$ by $\gls{omegap}$, the space with compactly supported ones by \gls{omegapz}.
\end{definition}
As an example, a zero-form $f \in \Omega^0(\N)$ is just a $C^\infty$-function from $\N$ to $\IR$, hence we can identify $\Omega^0(\N) = C^\infty (\N, \IR)$. A one-form $A \in \Omega^1(\N)$ is nothing more than a co-vector field and in a physical context usually denoted in local coordinates by $A_\mu$. Note, that alternatively one can directly define a $p$-form as a smooth section of the $p$-th exterior product of the co-tangent bundle and hence identify $\Omega^p(\N) = \Gamma \big( \largewedge^k T^*\N\big)$. As mentioned in Section \ref{sec:spacetime_geometry}, we view the tangent bundle as a complex bundle. Therefore, the sections of that bundle will be complex valued functionals. In that fashion, we will usually view the spaces $\Omega^p(\N)$ as complex valued differential forms.\par
Next we define the basic operations, besides addition and scalar multiplication, that one can perform on differential forms.
\begin{definition}[Exterior product]
	Let $A \in \Omega^p(\N)$ be a $p$-form and  $B\in \Omega^q(\N)$ a $q$-form on $\N$. \\
	The \emph{exterior product} $\gls{wedge}:\Omega^p(\N) \times \Omega^q(\N) \to \Omega^{p+q} (\N)$ is defined by
	\begin{align}
	(A \wedge B)_{\mu_1\dots\mu_p \nu_1\dots\nu_q} = \frac{(p+q)!}{p!q!}\, A_{[\mu_1 \dots \mu_p} B_{\nu_1\dots\nu_q]} \formspace,
	\end{align}
	where the anti-symmetrization of a tensor $T$ is given through
	\begin{align}
	T_{[\mu_1\dots\mu_p]} = \frac{1}{p!} \sum\limits_{\sigma\in S_N }\textrm{sgn}(\sigma) T_{\sigma(\mu_1)\dots\sigma(\mu_p)} \formspace.
	\end{align}
\end{definition}
Here, $S_N$ denotes the symmetric group\footnote{Usually the symmetric group is defined as the set of permutations of $\{1,2,\dots,N\}$ but we chose the index to run over $\{0,1,\dots,N-1\}$, identifying the time component with zero rather then one.} of degree $N$, consisting of permutations of the set $\{0,1,\dots,N-1\}$.
With this notion of multiplication, point-wise addition and scalar multiplication, the space $\gls{omega} \coloneqq \bigoplus_{p = 0}^\infty \Omega^p(\N) = \bigoplus_{p = 0}^N \Omega^p(\N)$ becomes an algebra, usually called the Grassmann- or \emph{exterior algebra} of differential forms on $\N$. We have used that obviously $\Omega^k(\N) =0$ for $k >N$ due to the anti-symmetrization.\par
Furthermore, we find a notion of how to \emph{pullback} differential forms on manifolds to another manifold, for example the pullback of a differential form on the spacetime $\M$ to differential forms on its Cauchy surface $\Sigma$. Given a $C^\infty$-map $\psi: \widetilde{\N} \to \N$, where $\N, \widetilde{\N}$ are manifolds, we can naturally define the pullback of a function $f \in \Omega^0(\N)$ to a function $(\psi^* f) \in \Omega^0(\widetilde{\N})$ by composing $f$ with $\psi$:
\begin{align}
\psi^* f \coloneqq f \comp \psi \formspace.
\end{align}
\newpage
With the pullback of functions defined, we can define how to \emph{push forward}, or carry along, vector fields on $\widetilde{\N}$ to vector fields on $\N$: Let $f\in \Omega^0(\N)$ and $\tilde{v} \in \Gamma(T\widetilde{\N})$ and $\tilde{x} \in \widetilde{\N}$. Then
\begin{align}
(\psi_* \tilde{v})_{\psi(\tilde{x})} (f) \coloneqq \tilde{v}_{\tilde{x}}(\psi^* f)
\end{align}
defines the vector field $(\psi_* v) \in \Gamma(T\N)$. With these basic operations at hand, we can generalize to define the pullback of differential forms:
\begin{definition}[Pullback]\label{def:pullback}
	Let $\N, \widetilde{\N}$ be manifolds of dimension $N,\widetilde{N}$ respectively, and let $\psi: \widetilde{\N} \to \N$ be a smooth map. Then, $\psi$ defines an algebra homomorphism $\psi^* : \Omega(\N) \to  \Omega(\widetilde{\N})$,
	called the \emph{pullback} of differential forms. For $\omega \in \Omega^p(\N)$, $\tilde{x} \in \widetilde{\N}$ and $\tilde{v}_i \in T_x \widetilde{\N}$, $i=1,2,\dots,p$, it is defined by
	\begin{align}
	\left( \psi^* \omega \right)_{\tilde{x}}  (\tilde{v}_1,\tilde{v}_2,\dots,\tilde{v}_p) \coloneqq \omega_{\psi(\tilde{x})} (\psi_* \tilde{v}_1, \dots , \psi_* \tilde{v}_p) \formspace.
	\end{align}
\end{definition}
On the exterior algebra we find a duality, provided by the Hodge operator:
\begin{definition}[Hodge dual]
	The hodge star operator $\gls{hodge}: \Omega^p(\N) \to \Omega^{N-p}(\N)$ is defined through
	\begin{align}
	B \wedge *A = \frac{1}{p!} B^{\mu_1\dots\mu_p}A_{\mu_1\dots\mu_p} \dvolk \formspace,
	\end{align}
	which yields the coordinate representation
	\begin{align}
	(*A)_{\mu_{p+1}\dots\mu_N} = \frac{\detk}{p!} \, \epsilon_{\mu_1\dots\mu_N} A^{\mu_1\dots\mu_p} \formspace.
	\end{align}
\end{definition}
Here, \gls{levicivita} denotes the fully antisymmetric tensor of rank $N$ (Levi-Civita symbol) satisfying $\epsilon_{12,\dots,N} =1$ and the \emph{volume element} \gls{dvolk} is defined by
\begin{align}
\left( \gls{dvolk} \right)_{\alpha_1\dots\alpha_N} = \detk \, \epsilon_{\alpha_1\dots\alpha_N} \formspace.
\end{align}
In a sense, the volume element describes how the curvature of the manifold deforms a unit volume.
The duality follows from the important property of the Hodge operator as stated in the following lemma:
\begin{lemma}
	Let $*$ denote the Hodge star operator on the exterior algebra $\Omega(\N) $. It holds that
	\begin{align}
	** = (-1)^{s+p(N-p)} \, \mathbbm{1} \formspace,
	\end{align}
	which is trivially equivalent to $*^{-1} = (-1)^{s+p(N-p)} \, *$.
\end{lemma}
\begin{proof}
	Let $A \in \Omega^p(\N)$ be a $p$-form on $\N$. Then:
	\begin{align}
	(*{*A})_{\mu_1 \dots \mu_p}
	&= \frac{\detk \, \detk}{p! \, (N-p)!} \; \epsilon_{\alpha_{p+1}\dots\alpha_N \mu_1 \dots \mu_p}\;\epsilon^{\alpha_{1}\dots\alpha_N}\;A_{\alpha_1\dots\alpha_p} \notag\\
	&= (-1)^{p(N-p)} \frac{\detk \, \detk}{p! \, (N-p)!} \; \epsilon_{\alpha_{p+1}\dots\alpha_N \mu_1 \dots \mu_p}\;\epsilon^{\alpha_{p+1}\dots\alpha_{N}\alpha_1\dots\alpha_p}\;A_{\alpha_1\dots\alpha_p}  \notag\\
	&= (-1)^{s+p(N-p)} \delta\indices{^{[\alpha_{1}}_{\mu_{1}}}\, \dots \, \delta\indices{^{\alpha_p ] }_{\mu_p}} \;A_{\alpha_1\dots\alpha_p} \notag\\
	&=  (-1)^{s+p(N-p)}\;A_{\mu_1\dots\mu_p} \formspace
	\end{align}
	We have used Lemma \ref{lem:epsilon_contraction} and, in the last step, that the anti-symmetrization is absorbed by contraction because $A$ is antisymmetric.
\end{proof}
Furthermore, we can equip the exterior algebra with a differentiable structure, introducing the notion of the exterior derivative.
\begin{definition}[Exterior derivative]
	The \emph{exterior derivative} $\gls{d}:\Omega^p(\N) \to \Omega^{p+1} (\N)$ is defined by the following properties:
	\begin{enumerate}
		\item $d$ is linear
		\item $d$ obeys a graded Leibniz rule: Let $A \in \Omega^p(\N)$ and  $B\in \Omega^q(\N)$, then
		\begin{align}
		d(A \wedge B) = dA \wedge B + (-1)^p \, A \wedge dB
		\end{align}
		\item $d$ is nilpotent, that is,  $d^2 = 0$.
	\end{enumerate}
	In local coordinates, this is equivalent to the representation
	\begin{align}
	(dA)_{\mu \alpha_1\dots\alpha_p} = (p+1)\, \nabla_{[\mu}A_{\alpha_1\dots\alpha_p]} \formspace.
	\end{align}
\end{definition}
An important property of the exterior derivative is that it commutes (or rather intertwines its action) with pullbacks (see \cite[Proposition 4.1.7]{rudolph_schmidt}).
A $p$-form $\omega \in \Omega^p(\N)$ is called \emph{exact} if there is a $(p-1)$-form $\alpha \in \Omega^{p-1}(\N)$ such that $\omega = d\alpha$. We call $\omega$ \emph{closed} if $d \omega =0$. Accordingly, the space of closed $p$-forms is denoted by \gls{omegapd}, the space of exact ones by \gls{domegap}. As usual, the ones with compact support are denoted by a subscript zero. Note, that every exact form is closed, using that $d$ is by definition nilpotent, but the reverse is in general not true. It does hold, however, on certain manifolds with trivial topology, such as Minkowski spacetime. This is expressed in the so called Poincar\'e-Lemma (see for example \cite[Chapter 4]{bott_tu}) based on the study of de Rham cohomology.\par
Moreover, $N$-forms can naturally be integrated. Using local coordinates and a partition of unity, we define the integral of $N$-forms via the well known integration on $\IR^N$:
\begin{definition}[Integration on manifolds]
	Let $\left\{U_\alpha, \psi_\alpha\right\}_\alpha$ be an atlas of the manifold $\N$ and $\left\{\chi_\alpha\right\}_\alpha$ a partition of unity subordinate to the locally finite open cover $\left\{U_\alpha\right\}_\alpha$. Let $x^\mu_{(\alpha)}$ be a coordinate basis of $\psi$ on $U_\alpha$. For any $N$-form $\omega \in \Omega^N_0(\M)$ we define the integral
	\begin{align}
	\int\limits_{\N} \omega &\coloneqq \sum_{\alpha} \int\limits_{\psi_\alpha (U_\alpha)} w(x_{(\alpha)}^0,\dots,x_{(\alpha)}^1)\; dx_{(\alpha)}^0 \cdots dx_{(\alpha)}^{N-1} \formspace,
	\end{align}
	where $w$ are the components of $\omega$ in the coordinates $x_{(\alpha)}^\mu$, that is $\omega = w dx_{(\alpha)}^0 \wedge \cdots \wedge dx_{(\alpha)}^{N-1}$.
	This definition is independent of the choice of the atlas and the partition of unity (see \cite[Proposition 3.3]{bott_tu}).
\end{definition}
With integration at our disposal, we present an important theorem regarding the integration of exact differential forms:
\begin{theorem}[Stoke's Theorem]\label{thm:stokes}
	Let $\N$ be an oriented manifold of dimension $N$ and let its boundary $\partial \N$ be endowed with the induced orientation. Let $\gls{inclusionmap} : \partial \N \hookrightarrow \N$ be the inclusion operator.
	Let $\omega \in \Omega^{N-1}_0(\N)$ be a compactly supported $(N-1)$-form on $\N$. Then it holds
	\begin{align}
	\int\limits_\N d\omega = \int\limits_{\partial \N} i^*\omega \formspace.
	\end{align}
\end{theorem}
\begin{proof}
	A proof is given in most of the introductory literature on differential geometry (see for example \cite[Chapter 17, Theorem 2.1]{lang}).
	Note that one can equivalently formulate Stoke's theorem on a \emph{compact} manifold but for {arbitrary} (that is, in general not compactly supported) $(N-1)$-forms on the manifold (see for example \cite[Theorem 4.2.14]{rudolph_schmidt}). This will be of importance in later calculations.
\end{proof}
Furthermore, we can define a bilinear map on $\Omega^p(\N)$ using the integration of $N$-forms:
\begin{definition}
	Let $A,B \in \Omega^p(\N)$ such that their supports have a compact intersection. Define the bilinear map $\gls{innerprod} : \Omega^p(\N) \times \Omega^p(\N) \to \IC$ by
	\begin{align}
	\langle A, B \rangle_\N \coloneqq  \int_{\N } A \wedge * B = \int_{\N } A_{\mu_1 \dots \mu_p}B^{\mu_1 \dots \mu_p}\,\dvolk \formspace.
	\end{align}
\end{definition}
Since by definition $A \wedge * B$ is a compactly supported $N$-form, this is well defined. We may sometimes refer to $\langle \cdot , \cdot \rangle_\N$ as an inner product for simplicity, even though it is not positive definite.
Using the exterior derivative, we define the interior or co-derivative:
\begin{definition}[Interior derivative]
	The \emph{interior derivative} $\gls{delta} : \Omega^p(\N) \to \Omega^{p-1}(\N)$ is defined by
	\begin{align}
	\delta \coloneqq (-1)^{s+1+N(p-1)}\, {*{d*}} \formspace.
	\end{align}
	From the defining properties of $d$ and $*$ it follows $\delta^2 =0$.
\end{definition}
Here, $s$ again denotes the number of negative eigenvalues of the metric $k$ of $\N$. In accordance with our nomenclature, we call a $p$-form $\omega$ co-exact if there exists a $\alpha \in \Omega^{p+1}(\N)$ such that $\omega = \delta \alpha$ and co-closed if $\delta \omega = 0$. Accordingly, the spaces of co-closed and co-exact $p$-forms are denoted by \gls{omegapdelta} and \gls{deltaomegap} respectively.\par
Using the exterior and interior derivative we define the partial differential operator:
\begin{definition}[D'Alembert Operator]
	The d'Alembert (or Laplace - de Rham) operator $\gls{dalembert}: \Omega^p(\N) \to \Omega^{p}(\N)$ is defined by
	\begin{align}
	\square \coloneqq \delta d +d \delta \formspace.
	\end{align}
\end{definition}
By definition of the exterior and interior derivative, it is easy to show that $\square$ commutes with both $d$ and $\delta$:
\begin{align}
\square d &= (\delta d + d \delta )d \notag \\
&= d \delta d \notag \\
&= d (\delta d + d \delta) \formspace,
\end{align}
and analogously for $\delta$.
The d'Alembert operator, and its generalization to $(\square + m^2)$ for some constant $m > 0$, are important examples for a normally hyperbolic differential operators (see Section \ref{sec:global_hyperbolicity}) and we may therefore sometimes just refer to them as \emph{wave operators}.\par
The sign convention in the definition of the exterior derivative is chosen such that on any Lorentzian or Riemannian manifold the interior derivative is formally adjoint to the exterior derivative, that is,  for $A \in \Omega^{p}(\N)$ and $B \in \Omega^{p+1}(\N)$ it holds that
\begin{align}
\langle dA , B \rangle_{\N} = \langle A , \delta B \rangle_\N \formspace,
\end{align}
which leads to a representation in local coordinates of the Manifold given by:
\begin{align}
(\delta A)_{\mu_2\dots\mu_p} = - \nabla^{\mu_1}A_{\mu_1\dots\mu_p} \formspace.
\end{align}
To see that this is consistent, let $A \in \Omega^{p-1}(\N)$ and $B \in \Omega^{p}(\N)$ such that their supports have compact intersection.
We obtain, using Stoke's Theorem \ref{thm:stokes}:
\begin{align}
0 &= \int \limits_{\partial \N} i^* (A \wedge *B) \notag\\
&= \int \limits_{\N} d(A \wedge *B)  \notag\\
&= \int \limits_{\N} dA \wedge *B + (-1)^{p-1} A \wedge d{*B} \notag\\
&= \int \limits_{\N} dA \wedge *B + (-1)^{p-1} A \wedge *{*^{-1}}\underbrace{d{*B}}_{\textrm{is a } (N-p+1) \textrm{ form.}} \notag\\
&= \int \limits_{\N} dA \wedge *B + (-1)^{p-1}(-1)^{s+(N-p+1)(N-N+p-1)} A \wedge *{*d{*B}} \notag\\
&= \int \limits_{\N} dA \wedge *B + (-1)^{p+(1-p)(p-1)} A \wedge *\delta B \formspace.
\end{align}
It can easily be proven by induction that $\big(p+(1-p)(p-1)\big)$ is odd for any $p \in \IN$, which yields the result
\begin{align}
\langle dA , B \rangle_{\N} = \langle A , \delta B \rangle_\N \formspace.
\end{align}
The definitions stated above thus fulfill the requirement of formal adjointness of the exterior and interior derivate on an arbitrary Lorentzian or Riemannian manifold $\N$.
In local coordinates we use a partial integration to obtain
\begin{align}
\langle dA , B \rangle_\N &= \int \limits_{\N} dA \wedge * B \notag\\
&= \int \limits_{\N}  \frac{p}{p!} \nabla^{[\alpha_1}A^{\alpha_2\dots\alpha_p]}\,B_{\alpha_1 \dots \alpha_p} \, \dvolk \notag\\
&= \int \limits_{\N}  \frac{1}{(p-1)!} \nabla^{\alpha_1}A^{\alpha_2\dots\alpha_p}\,B_{\alpha_1 \dots \alpha_p} \, \dvolk \notag\\
&= - \int \limits_{\N}  \frac{1}{(p-1)!} A^{\alpha_2\dots\alpha_p}\, \nabla^{\alpha_1}B_{\alpha_1 \dots \alpha_p} \, \dvolk \notag\\
&= \langle A, \delta B \rangle_\N \formspace,
\end{align}
which yields
\begin{align}
-\nabla^{\alpha_1}B_{\alpha_1 \dots \alpha p} = (\delta B)_{\alpha_2 \dots \alpha_p}\formspace.
\end{align}
On the four dimensional spacetime $(\M,g)$ the definitions of the Hodge star operator and the interior derivative simplify, such that
\begin{align}
*_{(\M)}*_{(\M)} &= (-1)^{p+1} \mathbbm{1} \\
\delta_{(\M)} &= *_{(\M)}{d_{(\M)}*_{(\M)}} \formspace ,
\end{align}
holds on the spacetime $(\M,g)$ and
\begin{align}
*_{(\Sigma)}*_{(\Sigma)} &= \mathbbm{1} \\
\delta_{(\Sigma)} &= (-1)^p *_{(\Sigma)}{d_{(\Sigma)}*_{(\Sigma)}}
\end{align}
holds on  $(\Sigma,h)$. In the following we will drop the subscript ${(\M)}$, since we will perform all the calculations on a four dimensional spacetime, except when explicitly noted (for example with a subscript $(\Sigma)$).
%
%
%
%
%
%
%
%
\subsection{Category theory}\label{sec:cat-theory}
The description of Quantum Field Theory on Curved Spacetimes (QFTCS) in the framework of \name{Brunetti}, \name{Fredenhagen} and \name{Verch} \cite{Brunetti_Fredenhagen_Verch} is based on category theory. In this thesis, we will not go into detail on those categorical aspects, however we will need some basic definitions to formulate the theory rigorously, that is namely the notion of a category and that of covariant functors, since, in the used framework, the generally covariant QFTCS is a functor.\par
Here, we present definitions given in \cite[Appendix A.1]{baer_ginoux_pfaeffle} and refer to the appropriate literature for details. We define:
\begin{definition}[Category]
	A \emph{category} $\mathsf{Cat}$ consists of the following:
	\begin{enumerate}
		\item a class $\mathsf{Obj}_\mathsf{Cat}$ whose members are called \emph{objects},
		\item a set $\mathsf{Mor}_\mathsf{Cat}(A,B)$, for any two objects $A,B \in \mathsf{Obj}_\mathsf{Cat}$, whose elements are called \emph{morphisms},
		\item for any three objects $A,B,C \in \mathsf{Obj}_\mathsf{Cat}$ there is a map
		\begin{align}
\mathsf{Mor}_\mathsf{Cat}(B,C) \times \mathsf{Mor}_\mathsf{Cat}(A,B) &\to \mathsf{Mor}_\mathsf{Cat}(A,C) \notag\\
(\psi,\phi) &\mapsto \psi \comp \phi
		\end{align}
		called the composition of morphisms subject to the relations:\vspace{4mm}
		\begin{enumerate}[label=(\arabic*)]
			\item for non equal pairs $(A,B)$, $(A',B')$ of objects, the sets $\mathsf{Mor}_\mathsf{Cat}(A,B)$ and $\mathsf{Mor}_\mathsf{Cat}(A',B')$ are disjoint,
			\item for every object $A$ there exists a morphism $\text{id}_A \in \mathsf{Mor}_\mathsf{Cat}(A,A)$ such that it holds for all objects $B$, morphisms $\psi \in \mathsf{Mor}_\mathsf{Cat}(B,A)$ and $\phi \in \mathsf{Mor}_\mathsf{Cat}(A,B)$
			\begin{align}
				\text{id}_A \comp \psi &= \psi \quad \text{and}\\
				\phi \comp \text{id}_A &= \phi \quad,
			\end{align}
			\item the composition law is associative, that is for an objects $A,B,C,D$ and any morphisms $\psi \in \mathsf{Mor}_\mathsf{Cat}(A,B)$, $\phi \in \mathsf{Mor}_\mathsf{Cat}(B,C)$ and $\chi \in \mathsf{Mor}_\mathsf{Cat}(C,D)$ it holds
			\begin{align}
				(\chi \comp \phi) \comp \psi = \chi \comp (\phi \comp \psi) \formspace.
			\end{align}
		\end{enumerate}
	\end{enumerate}
\end{definition}
\begin{definition}[Functor]
	Let $\mathsf{Cat1}$ and $\mathsf{Cat2}$ be categories. A \emph{covariant functor} $\mathscr{A}: \mathsf{Cat1} \to \mathsf{Cat2}$ consists of the map $\mathscr{A} : \mathsf{Obj}_\mathsf{Cat1} \to \mathsf{Obj}_\mathsf{Cat2}$ and maps $\mathscr{A}: \mathsf{Mor}_\mathsf{Cat1}(A,B) \to \mathsf{Mor}_\mathsf{Cat2}\big(\mathscr{A}(A),\mathscr{A}(B)\big)$ for any two objects $A,B \in \mathsf{Obj}_\mathsf{Cat1}$ such that
	\begin{enumerate}
		\item {the composition is preserved, that is for all objects $A,B,C \in \mathsf{Obj}_\mathsf{Cat1}$ and for any morphisms $\psi \in \mathsf{Mor}_\mathsf{Cat1}(A,B)$ and $\phi \in \mathsf{Mor}_\mathsf{Cat1}(B,C)$ it holds
		\begin{align}
			\mathscr{A}(\phi \comp \psi) = \mathscr{A}(\phi) \comp \mathscr{A}(\psi) \formspace,
		\end{align}}
		\item{
			$\mathscr{A}$ maps identities to identities, that is for any object $A \in \mathsf{Obj}_\mathsf{Cat1}$ it holds
			\begin{align}
				\mathscr{A}(\text{id}_\mathsf{A}) = \text{id}_{\mathscr{A}(A)} \formspace.
			\end{align}
			}
	\end{enumerate}
\end{definition}
%
%
%
%
%
%
%
%
%
%
%
%
%
%
\subsection{Sign conventions}\label{sec:sign_conventions}
At certain points throughout this chapter we have had a freedom of choice regarding the signs of some entities, in particular the sign of the signature of the Lorentzian metric $g$ and that of the interior derivative $\delta$. Though at this stage the choice can be made arbitrarily, we want to make it in a way that in the end allows us to make certain physical interpretations on some parameters. More precisely, we want to interpret the parameter $m$ of the Klein-Gordon equation\footnote{or its generalization on $p$-forms} $(\square + m^2) f = 0$ for a zero-form $f \in \Omega^0(\M)$ as a mass in the physical sense. With the chosen sign convention for $\delta$ we find, using ${\delta}f = 0$:
\begin{align}
	\square f
	&= (\delta d + d \delta) f \notag\\
	&= \delta d f \notag\\
	&= - \nabla^\mu \nabla_\mu f \formspace.
\end{align}
In the following heuristic (local) argument we see
\begin{align}
	\square + m^2
	&= -\nabla^\mu \nabla_\mu + m^2 \notag\\
	&\sim \partial_t^2 + \sum_i \partial_i^2 + m^2\notag\\
	&\sim -E^2 + \abs{\vector{p}}^2 + m^2
\end{align}
which yields the correct relativistic relation of energy, momentum and mass according to $E^2 = \abs{\vector{p}}^2 + m^2$.
A similar calculation holds for the Klein-Gordon operator generalized to act on one-forms. If we had found a ``wrong'' relation between energy, momentum and mass, we would have had to adapt the chosen signs. Usually one chooses the sign of the metric and the interior derivative such that they are in some sense mathematically convenient (although one might disagree with another one's choice). We have made the choice of the metric, such that the Cauchy surfaces become Riemannian rather that ``anti-Riemannian'' (with an all minus signature), which seems more natural to some. Also, a lot of the used references on spacetime geometry (in particular the book by \name{Wald} \cite{wald_GR}) use this sign convention, which makes the application of certain formulas easier. As mentioned, the sign of the interior derivative was chosen such that it is formally adjoint to the exterior derivative (with respect the specified inner product) on all Lorentzian and Riemannian manifolds. It seemed convenient for the actual calculations to fix the sign regardless of the signature of the metric of the underlying manifold. One could equivalently have fixed the opposite sign, yielding the two derivatives to be skew-adjoint, which is also done in the literature. However, in the end, one has one freedom left to make the energy-momentum-mass relation work: that is the sign in front of the mass in the Klein-Gordon equation and all other wave equations accordingly. Hence, one regularly also finds the Klein-Gordon equation to be defined with a flipped sign of the mass term. But for our case, we want the mass $m$ in any wave equation to appear with a positive sign.

\section{The Classical Problem}\label{chpt:classical}
In this chapter we will examine Proca's equation at a classical level in an arbitrary globally hyperbolic spacetime. Using differential forms, the formulation will be mostly coordinate independent. The goal is to find a solution to Proca's equation, including external classical sources and without restrictive assumptions on the topology of the spacetime, in terms of fundamental solutions of the Proca operator. Already at this classical stage we will emphasize on similarities and, partially crucial, differences of Proca's equation and Maxwell's equations.\par
We will start by finding the equations of motion from the Lagrangian\footnote{We could equivalently start by imposing the equations of motion directly, but the Lagrangian yields a more familiar comparison to the Maxwell field.}. In order to solve the equations of motion, it is then crucial to decompose the equations in a set of a hyperbolic differential equation and a constraint. After discussing the initial value problem in detail, we will determine the solution of the equations of motion of the Proca field in terms of fundamental solutions of the Proca operator.
Having found solutions to the classical Proca equation, we will investigate the classical zero mass limit as it will be the basis of understanding the according limit in the quantum case.\par
In the following, let $(\M,g)$ denote a globally hyperbolic four dimensional spacetime, consisting of a smooth manifold $\M$, assumed to be Hausdorff, connected, oriented, time-oriented and para-compact,  and a Lorentzian metric $g$, whose signature is chosen to be $(-,+,+,+)$. The Cauchy surface of the spacetime is denoted by $\Sigma$, with an induced Riemannian metric $h$. The Levi-Civita connection on $(\M,g)$ will as usual be denoted by $\nabla$, the one on $\Sigma$ by $\nabla_{(\Sigma)}$.
\subsection{Deriving the equations of motion of the Proca field from the Lagrangian}
Let $A, j \in \Omega ^1 (\M)$ be smooth one-forms on $\M$, $j \in \Omega^1(\M)$ a external source and $m >0$ a positive constant. We will call $A$ the \emph{vector potential}, \gls{mass} the \emph{mass} and \gls{current} denotes an \emph{external current}. The Lagrangian of the Proca field reads:
\begin{align}
\Ldens = -\frac{1}{2} \, dA  \wedge  *dA + A \wedge *j - \frac{1}{2} \, m^2 A \wedge *A \formspace.
\end{align}
In local coordinates this can equivalently be expressed, defining the field-strength tensor $F = dA$, as:
\begin{align}
\Ldens = \left( -\frac{1}{4}\, F\indices{_{\alpha\beta}} F\indices{^{\alpha\beta}} + A\indices{_\mu} j^\mu - \frac{1}{2} \, m^2 A\indices{_\nu} A\indices{^\nu} \right) \dvolg \formspace. \label{eqn:Proca_Lagrangian_coordinates}
\end{align}
At this stage, the similarity to the Lagrangian of the Maxwell field is obvious. Setting $m=0$ in the Proca Lagrangian\footnote{Even though we defined the Proca Lagrangian for non-zero masses only, at this stage setting $m=0$ is not a problem. The restriction to strictly positive masses becomes important later.} yields the Maxwell Lagrangian, that is, the Maxwell field is a massless Proca field.
The Euler-Lagrange-equations for a Lagrangian depending only on the field and the field's first derivative, $\Ldens = \Ldens(A\indices{_\mu}, \partial\indices{_\nu}A\indices{_\mu})$, are:
\begin{align}
0 = \frac{\partial \Ldens}{\partial A\indices{_\mu}} - \partial\indices{_\nu} \frac{\partial \Ldens}{\partial(\partial\indices{_\nu} A\indices{_\mu})} \formspace.
\end{align}
In local coordinates the first summand of the Euler-Lagrange-equations is easily obtained from equation (\ref{eqn:Proca_Lagrangian_coordinates}):
\begin{align}
\frac{\partial \Ldens}{\partial A\indices{_\mu}} = j^\mu - m^2 A^\mu \formspace.
\end{align}
The second term is most easily calculated using the coordinate representation of $F\indices{_{\alpha\beta}} = 2 \nabla\indices{_{[\alpha}} A\indices{_{\beta ]}} = \partial_\alpha A_\beta - \partial_\beta A_\alpha$ (all curvature dependent terms drop out due to the symmetry of the Christoffel symbol in it's lower two indices):
\begin{align}
\frac{\partial \Ldens}{\partial(\partial\indices{_\nu} A\indices{_\mu})}
=& -\frac{\partial }{\partial(\partial\indices{_\nu} A\indices{_\mu})} \left( \frac{1}{4} \left(  \partial_\alpha A_\beta - \partial_\beta A_\alpha \right) \left(  \partial^\alpha A^\beta - \partial^\beta A^\alpha \right)   \right) \notag \\
=&  -\frac{1}{4} \Big( \left( \delta\indices{^\nu_\alpha}  \delta\indices{^\mu_\beta} -  \delta\indices{^\nu_\beta}  \delta\indices{^\mu_\alpha} \right ) F\indices{^{\alpha\beta}} + F\indices{_{\alpha\beta}} \left( g\indices{^{\nu \alpha}} g\indices{^{\mu \beta}} - g\indices{^{\nu \beta}} g\indices{^{\mu \alpha}}  \right ) \Big) \notag  \\
 =& -\frac{1}{4} \,4 \, F\indices{^{\nu\mu}} =- F\indices{^{\nu\mu}}  \formspace
\end{align}
In that last step we have used the anti-symmetry of the two-form $F$, that is, in local coordinates $F\indices{_{\alpha\beta}}  = - F\indices{_{\beta\alpha}}$.
Combining the two results we obtain the equations of motion of the Proca field:
\begin{align}
0 &=  j^\mu - m^2 A^\mu  + \partial_{\nu} F\indices{^{\nu\mu}}\notag \\
\iff 0 &=  j_\mu - m^2 A_\mu  + \partial^{\nu} F\indices{_{\nu\mu}} \formspace.
\end{align}
Going back to an coordinate independent description, identifying $\partial^{\nu} F\indices{_{\nu\mu}} = -(\delta F)_\mu = -(\delta d A)_\mu$, we have found Proca's equation for a smooth one-form $A$ on a curved spacetime:
\begin{definition}[Proca's equation]\label{def:proca_operator}
For smooth one-forms $A,j \in \Omega^1(\M)$ and constant $m>0$, the Proca equation reads:
\begin{align}
\left( \delta d + m^2 \right) A = j \formspace.\label{eqn:Proca}
\end{align}
\end{definition}
Accordingly, the Proca operator is defined by $(\delta d + m^2)$. Again, setting $m=0$ in Proca's equation  yields the (homogeneous) Maxwell equations. It is worth noting that, unlike Maxwell's equation, the Proca equation does not posses a gauge symmetry. In the Maxwell case, two one-forms $A, A'$ that differ by an exact form, that is, $A' = A + d \chi$ for some zero-form $\chi$, yield the same equation of motion; if $A$ solves the Maxwell equation, so does $A'$:
\begin{align}
	\delta d A' = \delta d(A + d \chi) = \delta d A =j \formspace,
\end{align}
using the defining property $d^2 = 0$ of the exterior derivative. In the Proca case no such symmetry has to be accommodated when finding a solution to the equations of motion.
\subsection{Solving Proca's equation}\label{sec:solving_procas_equation}
The goal of this section is to find a solution to Proca's equation (\ref{eqn:Proca}) on an arbitrary globally hyperbolic curved spacetime. The procedure presented in this section is a generalization of \cite{FURLANI}, where \name{Furlani} tackles the same problem but with some topological restrictions on the manifold (he supposes the Cauchy surface to be compact) and excluding external sources. Most of the notation used in this chapter is adopted from \name{Furlani} based on previous work by Dimock \cite{dimock1992quantizedEM}.\par
The problem that one encounters when trying to solve Proca's equation in curved spacetimes is that the Proca operator is not normally hyperbolic\footnote{For the definition of normally hyperbolic operators see Section \ref{sec:global_hyperbolicity}.}. For non normally hyperbolic operators, we do not know a priori of the existence fundamental solutions or if the initial value problem is even well posed.  Fortunately, we are able to decompose Proca's equation into a normally hyperbolic second order partial differential equation and a constraint. For the normally hyperbolic equation on a globally hyperbolic spacetime, we have a well understood theory at our disposal that ensures in particular the well-posedness of the initial value problem, the existence and uniqueness of a solution to the differential equation and the existence and uniqueness of fundamental solutions to the differential operator. Finally, it turns out that the Proca operator is Green-hyperbolic, that is, it possesses unique advanced and retarded fundamental solutions.\par
We will start with decomposing Proca's equation, which is stated in the following lemma:
\begin{lemma}
Let $A, j \in \Omega^1(\M)$ be one-forms on the manifold $\M$ and let $m>0$ be a positive constant. Then Proca's equation is equivalent to a set of equations consisting of a wave equation and a constraint. That is:
\begin{subnumcases}
{(\delta d  +m^2)A = j \iff}
\left(\square +m^2 \right) A = j + \frac{1}{m^2} \, d \delta j \label{eqn:classical_wave_eqation} \\
\delta A = \frac{1}{m^2} \delta j \label{eqn:classical_constraint}
\end{subnumcases}
\label{lem:Proca_wave_equiv}
\end{lemma}
\begin{proof}
1.) $\Rightarrow$ - direction: Let $A, j \in \Omega^1(\M)$ and let $A$ satisfy Proca's equation with $m >0$. We obtain the constraint (\ref{eqn:classical_constraint}) of the lemma by applying $\delta$ to Proca's equation and using $\delta^2 = 0$:
\begin{align}
(\delta d +m^2) A 		&= j \notag \\
\implies 			\delta(m^2 A) 				&= \delta j  \notag \\
\iff 		\delta A 							&= \frac{1}{m^2} \, \delta j \formspace.
\end{align}
By adding $0=d(0) = d(\delta A - \frac{1}{m^2}\delta j) = d\delta A - \frac{1}{m^2}d\delta j$ to Proca's equation one obtains the wave equation (\ref{eqn:classical_wave_eqation}) of the lemma:
\begin{align}
j =& \, (\delta d +m^2) A \notag \\
=& \, (\delta d + d \delta +m^2) A -\frac{1}{m^2} d \delta j \notag \\
\iff (\square +m^2)A =& \, j+ \frac{1}{m^2}d\delta j \formspace,
\end{align}
which completes the first direction of the proof.  \par
2.) $\Leftarrow$ - direction:
Let $A, j \in \Omega^1(\M)$ and let $A$ satisfy the wave equation (\ref{eqn:classical_wave_eqation}) of the lemma:
\begin{align}
(\square +m^2)A =& \, j+ \frac{1}{m^2}d\delta j \formspace.
\end{align}
Inserting the constraint (\ref{eqn:classical_constraint}) into the wave equation we obtain
\begin{align}
 (\delta d + d \delta +m^2) A &= \, j + \frac{1}{m^2} d \delta j \notag  \\
 \implies (\delta d +m^2) A + \frac{1}{m^2} d \delta j  &= j + \frac{1}{m^2} d \delta j \notag \\
 \iff (\delta d +m^2) A  &= j  \formspace,
\end{align}
which completes the proof.
\end{proof}
\noindent As argued in Section \ref{sec:sign_conventions}, the positive sign of the mass term in the wave equation (\ref{eqn:classical_wave_eqation}) is consistent with our conventions. Before we proceed with solving Proca's equation, the above lemma allows for a another brief comparison of Maxwell's and Proca's theory. It is well known that the Maxwell field possesses two independent degrees of freedom, known as the two independent polarization modes of the electromagnetic field. Starting with four independent components of the vector potential $A$ one reduces the degrees of freedom to three by implementing the Lorenz constraint $\delta A = 0$. One can show that this does not completely fix the gauge of the theory: As presented in the previous section, Maxwell's theory is independent of a gauge by exact forms\footnote{We will briefly discuss the different possibilities of choices of gauge of Maxwell's theory in manifolds with non-trivial topology at the end of Section \ref{sec:zero-mass-limit-existence-classical-vanishing-source}.}, that is, it is independent of the transformation $A \to A + d\chi$ for some zero-form $\chi$. Therefore, the exterior derivative transforms as $\delta A \to \delta A + \delta d \chi = \delta A + \square \chi$ as for every zero form it holds $\delta \chi = 0$. So the Lorenz constraint does not fix the gauge completely, as the theory is gauge invariant under addition of $\chi$ fulfilling the wave equation $\square \chi=0$. One can therefore further reduce the degrees of freedom to two. In an appropriate reference frame, the remaining polarization modes will be \emph{transversal}. As shown in the above lemma, in the Proca case one also has a Lorenz constraint (\ref{eqn:classical_constraint}) to solve, quite similar to the Lorenz constraint in Maxwell's theory, but no gauge equivalence relation to accommodate. Therefore, the Proca field has three independent polarization modes, which, in an appropriate reference frame, consist of two transversal and one longitudinal mode.\par
The procedure to solve Proca's equation is as follows: we have already successfully decomposed the equation into a wave equation with external sources and a Lorenz constraint. Next we will solve the wave equation in terms of fundamental solutions and then restrict the initial data in such a way that the constraint (\ref{eqn:classical_constraint}) of Lemma \ref{lem:Proca_wave_equiv} is also fulfilled. This will work as follows:
Assume $A\in \Omega^1(\M)$ solves the wave equation, $(\square +m^2)A =  j + \frac{1}{m^2} d \delta j$. Now, we observe
\begin{align}
(\square +m^2) \delta A
 &= \delta (\square + m^2) A \notag\\
&= \delta \left( j + \frac{1}{m^2} d \delta j \right) \notag\\
&= \delta  j + \frac{1}{m^2} \delta d \delta j + \frac{1}{m^2} d \delta \delta j \notag\\
&= (\square + m^2 ) \left (\frac{1}{m^2} \delta j \right)  \notag \\
\iff 0 &= (\square + m^2 )\left (\delta A - \frac{1}{m^2} \delta j \right)   \label{eqn:waveeqn_constraint}
\formspace,
\end{align}
where again we have used $\delta ^2 = 0$. The solution to the wave equation for $A$ therefore yields a Klein-Gordon equation for $(\delta A - \frac{1}{m^2} \delta j)$. Imposing initial data of $(\delta A - \frac{1}{m^2} \delta j)$ with respect to the Klein-Gordon equation that \emph{vanish} on the Cauchy surface $\Sigma$ implies a globally vanishing solution to the Klein-Gordon equation\footnote{Note that for any normally hyperbolic operator acting on sections of a vector bundle over a globally hyperbolic Lorentzian manifold, specifying vanishing initial data on a Cauchy surface yields a globally vanishing solution to the corresponding homogeneous differential equation (see \cite[Corollary 3.2.4]{baer_ginoux_pfaeffle}).} (\ref{eqn:waveeqn_constraint}) which is equivalent to $A$ globally fulfilling the Lorenz constraint (\ref{eqn:classical_constraint}).
To obtain a more self contained result, we then will re-express the vanishing initial data of $(\delta A -\frac{1}{m^2} \delta j)$ with respect to the Klein-Gordon equation in terms of initial data of $A$ with respect to the wave equation.\par
In conclusion: instead of \emph{globally} constraining $\delta A = \frac{1}{m^2} \delta j$ for a solution $A$ of the wave equation (\ref{eqn:classical_wave_eqation}) one can \emph{equivalently} state vanishing initial data of $(\delta A -\frac{1}{m^2} \delta j)$ on the Cauchy surface $\Sigma$ (with respect to a Klein-Gordan equation (\ref{eqn:waveeqn_constraint})) and restrict the initial data of $A$ (with respect to the wave equation (\ref{eqn:classical_wave_eqation})), such that they are compatible with the vanishing initial data of $\left(\delta A - \frac{1}{m^2} \delta j\right)$.
\newpage
\subsubsection{Solving the wave equation}
Before solving the wave equation it is useful to define some operators, which will map a one-form $A \in \Omega^1(\M)$ to its initial data on $\Sigma$ with respect to the wave equation. In this form, these operators were first introduced by Dimock \cite{dimock1992quantizedEM} but are also used by Furlani \cite{FURLANI} and Pfenning \cite{pfenning}, allowing for a direct comparison of the obtained results, in particular when taking the zero mass limit.
\begin{definition}\label{def:cauchy_mapping_operators}
 Let $\gls{inclusionmap} : \Sigma \hookrightarrow \M$ the inclusion map of the Cauchy surface $\Sigma$ with pullback $i^*$.\\
 The operators $\rhoz, \rhod : \Omega^1(\M) \to \Omega^1(\Sigma)$ and $\rhon, \rhodelta : \Omega^1(\M) \to \Omega^0(\Sigma)$ are defined as:
 \begin{subequations}
  \begin{align}
  \gls{rhoz} 		&= i^* 										&\textrm{pullback,} \\
  \gls{rhod} 		&= -*_{(\Sigma)} i^* * d 		&\textrm{forward normal derivative,} \\
  \gls{rhodelta} 	&= i^*\delta 						&\textrm{pullback of the divergence,} \\
  \gls{rhon} 		&= -*_{(\Sigma)} i^* *  		&\textrm{forward normal.}
 \end{align}
 \end{subequations}
\end{definition}
These operators can be extended to act on arbitrary $p$-forms, as will be important when we will deal with the constraint (\ref{eqn:classical_constraint}). As we will see, the operators $\rho_{(\cdot)}$ not only map a one-form to its initial data with respect to the wave equation, but also map a zero-form to its initial data with respect to the Klein-Gordon equation. This will allow us to elegantly re-express the vanishing initial data of the constraint as mentioned in the previous section.
But first, we define a set of differential forms on $\Sigma$, which will turn out to be equivalent to the initial data of a solution $A$ to the wave equation.
\begin{definition}\label{def:cauchy_data_wave_eq}
 Let $A \in \Omega^1(\M)$ be a one-form on $\M$. The differential forms $\Az, \Ad \in \Omega^1(\Sigma)$ and $\An, \Adelta \in \Omega^0(\Sigma)$ are defined as:
 \begin{subequations}
  \begin{align}
  \Az &= \rhoz A \formspace,\\
  \Ad &= \rhod A \formspace, \\
  \An &= \rhon A \formspace, \\
  \Adelta &= \rhodelta A \formspace.
  \end{align}
 \end{subequations}
\end{definition}
Specifying these $p$-forms is equivalent to imposing the initial data $A_\mu$ and $n^\alpha \nabla_\alpha A_\mu$ on the Cauchy surface $\Sigma$ with future pointing unit normal vector field $n \in \Gamma(T\M)$ (see \cite[Chapter III]{FURLANI}). Therefore, in the following we will view $\Az,\Ad, \An, \Adelta$ as the initial data of $A$ with respect to the wave equation (\ref{eqn:classical_wave_eqation}).\par
The main key to finding a solution to the wave equation is the following lemma.
\begin{lemma}
 Let $A \in \Omega^1(\M)$ be a one-form, $F \in \Omega^1_0(\M)$ a test one-form and $m \geq 0$.
 Let $\Sigma$ denote a Cauchy surface of $\M$ with causal future/past  $\gls{sigmapm} \coloneqq J^\pm (\Sigma)$.
 Then it holds
 \begin{align}
  \int\limits_{\Sigma^\pm} &\Big[ A \wedge * \left( \square + m^2 \right) F - F \wedge *  \left( \square + m^2 \right) A \Big]  \label{eqn:greens_identity} \\
  &= \pm \Big\{ \langle \Az , \rhod F\rangle_\Sigma
  +\langle \Adelta , \rhon F\rangle_\Sigma
  -\langle \An , \rhodelta F\rangle_\Sigma
  -\langle \Ad , \rhoz F\rangle_\Sigma \Big\}   \formspace. \notag
 \end{align}
\label{lem:greens_identity}
\end{lemma}
\begin{proof}
 The proof is based on a proof outline from \cite[Appendix A]{FURLANI}. We start with Stoke's theorem for a $(p+1)$-dimensional sub-manifold $\mathcal{O} \subset \M$ with boundary $\partial \mathcal O$. Let $i : \partial\mathcal{O} \hookrightarrow \mathcal{O}$ be the inclusion operator. Then for any compactly supported $p$-form $\omega \in \Omega^p_0 (\M)$ it holds that (see Theorem \ref{thm:stokes})
 \begin{align}
  \int \limits_\mathcal O d\omega = \int \limits_{\partial \mathcal O} i^* \omega \formspace.
  \end{align}
Let $A \in \Omega^1(\M)$ and $F \in \Omega_0^1 (\M)$. We  define the compactly supported three-forms
\begin{subequations}
\begin{align}
 H' &= \delta F \wedge *A - \delta A \wedge *F \\
 H'' &= A \wedge * dF - F \wedge *dA \\
 H &=  H' -  H'' \formspace.
\end{align}
\end{subequations}
We will apply Stoke's theorem to the three-form $H$. For that we first need to calculate $dH$:
\begin{align}
 dH &=\; d\big( \delta F \wedge *A - \delta A \wedge *F \big) - d \big(A \wedge * dF - F \wedge *dA \big)  \notag \\
  &= d\delta F \wedge *A + (-1)^0 \delta F \wedge d{*A} - d\delta A \wedge *F - (-1)^0 \delta A \wedge d{*F} \notag \notag \\
   &\phantom{M}- dA \wedge *dF - (-1)^1 A \wedge (d{*{dF}}) + dF \wedge *{dA} + (-1)^1 F \wedge d{*{d A}} \notag \\
  &= A \wedge *d\delta F + \delta F \wedge d {* A} - F \wedge *{d\delta A} - \delta A \wedge d {* F }\notag \notag \\
   &\phantom{M}- {dA \wedge *dF} + A \wedge d{*{d F}}  + {dF \wedge * {dA}} -F \wedge d{*{dA }}
\end{align}
where we have used the linearity and graded Leibniz rule for the exterior derivative $d$, and the property $\alpha \wedge * \beta = \beta \wedge * \alpha$ of the wedge product and the Hodge star operator. With this, the terms $-{dA \wedge *dF}$ and ${dF \wedge * {dA}}$ cancel. Next, we use that for the test one-form $F$ it holds that $d{*dF} = (-1)^4 {*{*{d*}}} dF =  *\delta d F$ and analogously for $A$. We therefore get
\begin{align}
 dH &= A \wedge * ( d \delta F + \delta d F) - F \wedge * (d\delta A + \delta d A ) \notag \\
  &\phantom{M}+ \delta F \wedge d{*A} - \delta A \wedge d{* F} \label{eqn:greens_identity_zwischenergebnis}  \\
  &= A \wedge * \square F - F \wedge * \square A \formspace.
\end{align}
We have used in the last step that
\begin{align}
 \delta F \wedge d{* A}\; &=  \delta F \wedge {* *^{-1}}d{*A} \notag\\
 &=  - \delta F \wedge * \delta A \notag\\
 &=  - \delta A \wedge * \delta F \notag\\
 &=  \delta A \wedge {**^{-1}}{d*}F \notag\\
 &= \delta A \wedge d{*F} \formspace,
\end{align}
therefore the last two terms in equation (\ref{eqn:greens_identity_zwischenergebnis}) cancel. Furthermore we can add a vanishing term $0 = A \wedge * m^2 F - F \wedge * m^2 A$ to $dH$ to find the wanted relation by Stoke's theorem:
\begin{align}
 \int\limits_\mathcal O A \wedge * (\square + m^2) F - F \wedge * (\square + m^2) A
 =& \int\limits_{\partial O} i^*\left( \delta F \wedge *A + F \wedge *dA\right) \notag \\
 &- \int\limits_{\partial O} i^*\left( A \wedge *dF + \delta A \wedge *F\right)
\end{align}
Now, we specify $\mathcal O = \Sigma^\pm - \Sigma \Rightarrow \partial \mathcal O = \Sigma$ and note that with respect to integration on $\M$, the Cauchy surfaces denotes a set of measure zero, $\int_{\Sigma^\pm - \Sigma} = \int_{\Sigma^\pm}$. Furthermore, for Stoke's theorem to hold, we need to choose a consistent orientation of the boundary $\Sigma$ with respect to $\Sigma^\pm$. Following \cite[Appendix B.2]{wald_GR}, the orientation of $\Sigma^\pm$ induces a natural orientation on the boundary $\Sigma$. In the case of a Cauchy surface, coordinates in the neighborhood can be parametrized by one parameter $t$ such that $x_\mu=(t, x_1,x_2,x_3)$ where $x_i$ are (right handed) coordinates on $\Sigma$. For $\Sigma^+$ we get a right handed coordinate system in the natural way. For $\Sigma^-$ we get a right handed system if we flip the orientation of $\Sigma$, therefore we will get a relative sign for Stoke's theorem applied to $\mathcal O = \Sigma^-$, where we choose the natural standard orientation (induced by $\M$) on $\Sigma$. Therefore,  Stoke's Theorem reads $\int_{\Sigma^{\pm}} d\omega = \pm \int_{\Sigma} i^* \omega$. Specifying $\omega = H$ we find:
\begin{align}
 \int\limits_{\Sigma^\pm} A \wedge * (\square + m^2) F &- F \wedge * (\square + m^2) A   \\
 = \pm  \Bigg\{ &\int\limits_{\Sigma} i^* \delta F \wedge i^* *A + \int\limits_{\Sigma}  i^* F \wedge i^* *dA \notag \\
 &- \int\limits_{\Sigma} i^* A \wedge i^* *dF - \int\limits_{\Sigma}  i^* \delta A \wedge i^* *F  \Bigg\} \notag \\
 = \pm  \Bigg\{ &\int\limits_{\Sigma} \rhodelta F \wedge *_{(\Sigma)} *_{(\Sigma)} i^* *A + \rhoz F \wedge *_{(\Sigma)} *_{(\Sigma)} i^* *dA \notag \\
 &- \int\limits_{\Sigma} \Az \wedge *_{(\Sigma)} *_{(\Sigma)} i^* *dF -\int\limits_{\Sigma} \Adelta \wedge *_{(\Sigma)} *_{(\Sigma)} i^* *F  \Bigg\} \notag \\
 = \pm  \Bigg\{- &\int\limits_{\Sigma} \rhodelta F \wedge *_{(\Sigma)} \An -  \int\limits_{\Sigma} \rhoz F \wedge *_{(\Sigma)} \Ad \notag \\
 &+ \int\limits_{\Sigma} \Az \wedge *_{(\Sigma)} \rhod F +\int\limits_{\Sigma}  \Adelta \wedge *_{(\Sigma)} \rhon F  \Bigg\} \notag \\
 = \pm  \Big\{ &\langle \Az , \rhod F \rangle_\Sigma +\langle  \Adelta , \rhon F \rangle_\Sigma
 - \langle \rhodelta F , \An \rangle_\Sigma -  \langle \rhoz F , \Ad \rangle_\Sigma \Big\}  \notag
\end{align}
In the last steps we have made use of the Definitions \ref{def:cauchy_mapping_operators} and  \ref{def:cauchy_data_wave_eq}.
\end{proof}
To write down a solution to the wave equation (\ref{eqn:classical_wave_eqation}) we need to introduce the notion of fundamental solution of partial differential operators, in particular of the wave operator and, later, the Proca operator.
\begin{lemma}[Fundamental solutions of the wave operator]\label{lem:fundamental_solution_wave_operator}
Let $m \geq 0$. The wave operator  $(\square + m^2) : \Omega^p(\M) \to \Omega^p(\M)$ has unique advanced ($-$) and retarded ($+$) fundamental solutions $\gls{Epm} : \Omega^p_0(\M) \to \Omega^p(\M)$, which fulfill
\begin{subequations}
\begin{align}
 E_m^\pm (\square + m^2) &= \mathbbm 1 = (\square + m^2) E_m^\pm \formspace \text{\upshape{\quad\quad and}} \label{def:fundamental_solution} \\
 \supp{E_m^\pm F } &\subset J^\pm \big( \supp{F } \big) \formspace, \label{def:fundamental_solution_support}
\end{align}
\end{subequations}
where $F \in \Omega^p_0(\M)$ is a test $p$-form.
Furthermore, the fundamental solutions commute (or intertwine their action) with the exterior and interior derivative:
\begin{align}
E^\pm_m \delta &= \delta E^\pm_m \\
E^\pm_m d &= d E^\pm_m \formspace.
\end{align}
The advanced minus retarded fundamental solution is denoted by
\begin{align}
 \gls{Em} = E^-_m - E^{+}_{m} \formspace.
\end{align}
\end{lemma}
\begin{proof}
The properties (\ref{def:fundamental_solution_support}) and (\ref{def:fundamental_solution}) for the fundamental solutions for any normally hyperbolic operator acting on smooth sections in a vector bundle over a Lorentzian manifold are proven in \cite[Corollary 3.4.3]{baer_ginoux_pfaeffle}.
We will here give a proof of the commutation with the exterior and interior derivative:
	\begin{figure}[h]
		\begin{center}
			\scalebox{0.9}{
				\begin{tikzpicture}
				\node at (0,0) {\includegraphics[scale=0.9]{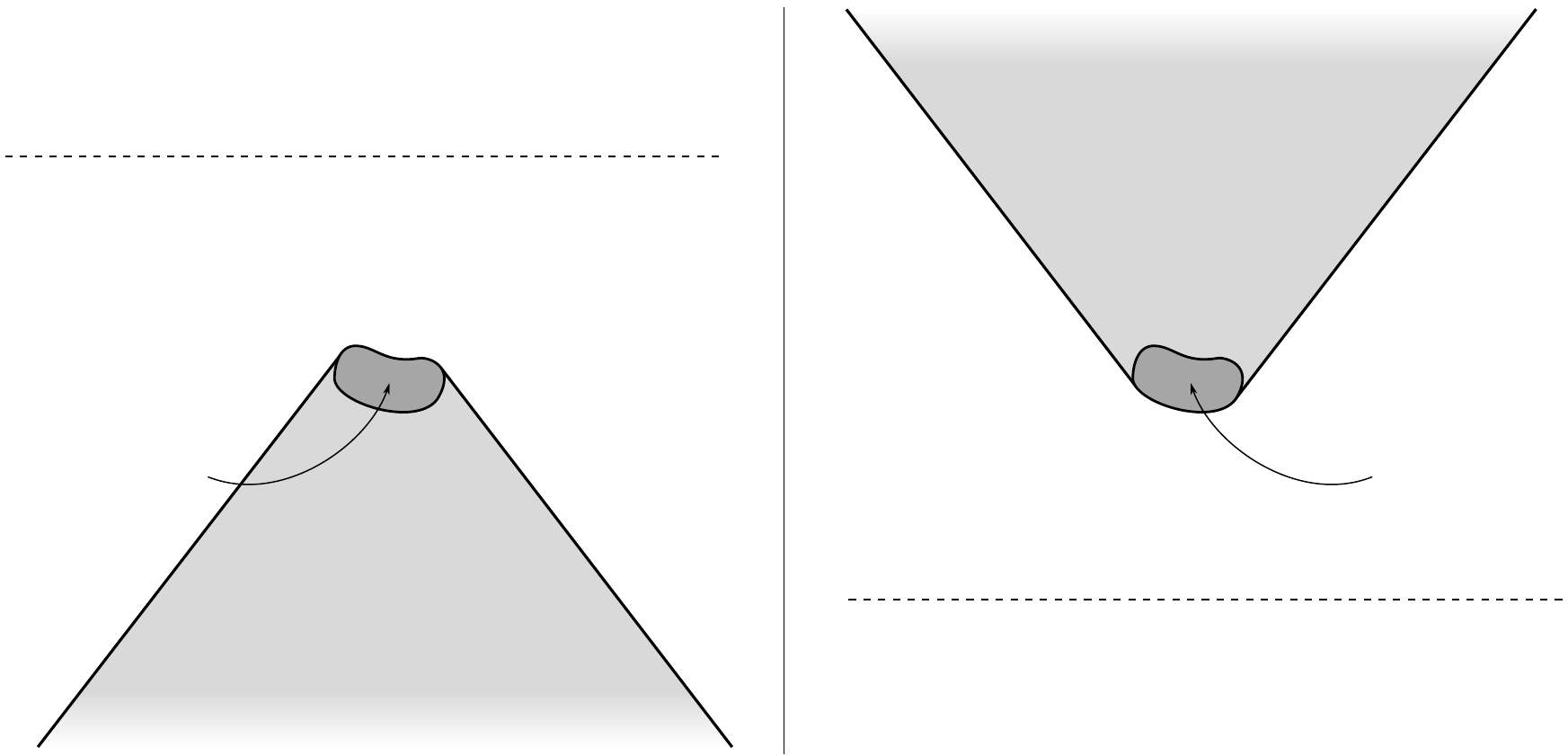}};
				\node at (4,3.5)  {$J^+\big( \supp{F} \big)$};
				\node at (-4,-3.5)  {$J^-\big( \supp{F} \big)$};
				\node at (-8.3,2.2)  {$\Sigma_+$};
				\node at (8.3,-2.3)  {$\Sigma_-$};
				\node at (7,-0.9)  {$\supp{F}$};
				\node at (-7,-0.9)  {$\supp{F}$};
				\end{tikzpicture}
			}
		\end{center}
		\caption{Illustrating the setup of the proof of Lemma \ref{lem:fundamental_solution_wave_operator}: The Cauchy surfaces $\Sigma_\mp$ are chosen such that they have vanishing intersection with $J^\pm\big( \supp{F}\big)$ and thus $(\delta E^\pm_m - E^\pm_m \delta) F $ specifies vanishing initial data on the corresponding Cauchy surface.}
		\label{fig:fund_solutions}
	\end{figure}
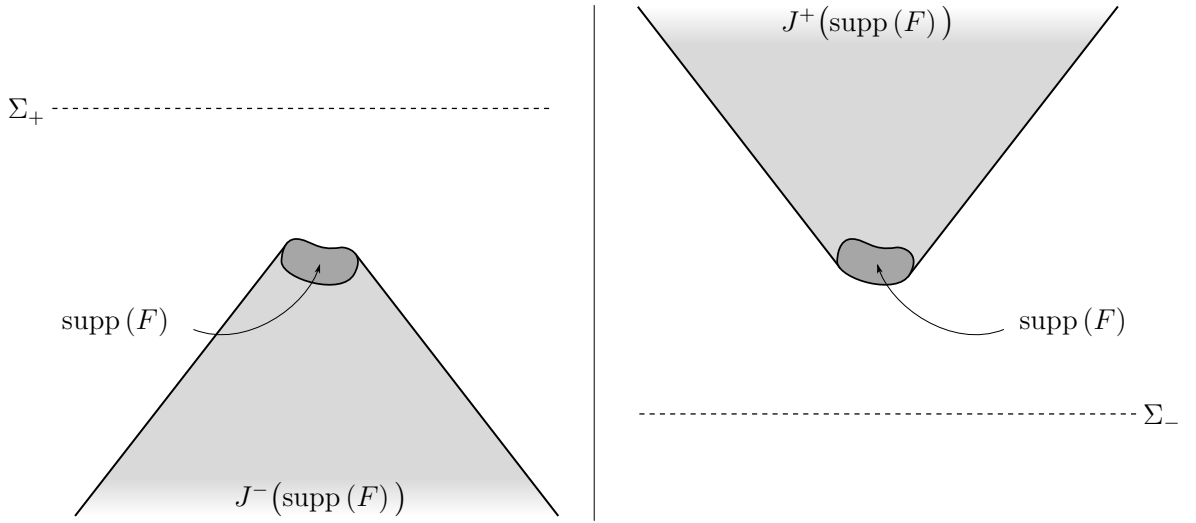
Let $F \in \Omega^p_0(\M)$ be a test $p$-form and $E^\pm_m$ the fundamental solutions to the wave operator $(\square + m^2)$.
Then, from $\delta (\square + m^2) =(\square + m^2) \delta $ it follows:
\begin{align}
(\square + m^2)  E^\pm_m \delta F
&= \delta  F \notag \\
&= \delta (\square +m^2) E^\pm_m F \notag\\
&= (\square + m^2) \delta E^\pm_m F \notag \\
\implies (\square + m^2)(\delta E^\pm_m - E^\pm_m \delta) F &= 0 \formspace.
\end{align}
Since derivatives do not enlarge the support of a function they are acting on, we know
\begin{align}
\supp{\delta E^\pm_m F } &\subset \supp{E^\pm_m F } \subset J^\pm(\supp{F }) \formspace \textrm{and} \\
\supp{E^\pm_m \delta F } & \subset J^\pm(\supp{\delta F }) \subset J^\pm(\supp{F })\formspace.
\end{align}
Now, on a Cauchy surface $\Sigma_\mp$ in the past/future of $\supp{F }$ we specify initial data of $(\delta E^\pm_m - E^\pm_m \delta) F $ with respect to the wave operator $(\square + m^2)$  for the plus and minus sign  respectively. Because of the support property mentioned above, we know by construction that these initial data vanish. This is also illustrated in Figure \ref{fig:fund_solutions}. With respect to the homogeneous differential equation, specifying vanishing initial data on a Cauchy surface yields a globally vanishing solution (c.f. \cite[Corollary 3.2.4]{baer_ginoux_pfaeffle}), therefore for all $F \in \Omega^p_0(\M)$ it holds
\begin{align}
E^\pm_m \delta F= \delta E^\pm_m F\formspace.
\end{align}
Since also the exterior derivative $d$ does not extend the support of a function and commutes with the wave operator $(\square + m^2)$, the proof for the commutativity follows in complete analogy.
\end{proof}
With the notion of the fundamental solutions we can state a solution to the wave equation (\ref{eqn:classical_wave_eqation}) in form of the following theorem:
\begin{theorem}[Solution to the wave equation]
 Let $\Az,\Ad \in \Omega^1(\Sigma)$ and $\An,\Adelta \in \Omega^0(\Sigma)$ specify initial data of the one-form $A\in \Omega^1(\M)$ on the Cauchy surface $\Sigma$. Let $F \in \Omega^1_0(\M)$ be a test one-form and $\kappa \in \Omega^1(\M)$ be an external source.\\
 Then
 \begin{align}
  \langle A, F \rangle_\M = \sum\limits_\pm \langle E_m^\mp F , \kappa \rangle_{\Sigma^\pm}
  &+ \langle \Az , \rhod E_m F \rangle_\Sigma
  + \langle \Adelta , \rhon E_m F \rangle_\Sigma \notag \\
  &- \langle \An , \rhodelta E_m F \rangle_\Sigma
  - \langle \Ad , \rhoz E_m F \rangle_\Sigma
 \end{align}
specifies the unique smooth solution of the wave equation $(\square + m^2)A= \kappa$, where $m \geq 0$, with the given initial data. Furthermore, the solution depends continuously on the initial data. \label{thm:solution_wave_equation}
\end{theorem}
\begin{proof}
The proof follows directly from Lemma \ref{lem:greens_identity} by adding the equations for the plus and minus sign and specifying $F ' = E_m^\mp F $ (even though, for a compactly supported one form $F $,  $F '$ does not have compact support, we will see below that the integrals are still well defined and Lemma \ref{lem:greens_identity} is applicable): \\
1.) Let $A\in \Omega^1(\M)$ be a smooth solution to the wave equation $(\square + m^2)A = \kappa$ and $F ' \in \Omega^1(\M)$ a test one-form. Then the LHS of equation (\ref{eqn:greens_identity}) reads
\begin{align}
 \int\limits_{\Sigma^\pm} \Big[ A \wedge * \left( \square + m^2 \right) F ' - F ' \wedge * \left( \square + m^2 \right) A \Big]
&=  \int\limits_{\Sigma^\pm} \Big[ A \wedge * \left( \square + m^2 \right) F ' - F ' \wedge * \kappa \Big] \notag \\
&= \int\limits_{\Sigma^\pm} \Big[ A \wedge * \left( \square + m^2 \right) E_m^\mp F - E_m^\mp F \wedge * \kappa \Big]  \notag \\
&= \int\limits_{\Sigma^\pm} \Big[ A \wedge *F - E_m^\mp F \wedge * \kappa \Big]  \formspace,
\end{align}
where we have  substituted $F ' = E_m^\mp F $ for the regions $\Sigma^\pm$ respectively, such that, because $\supp{E_m^\mp F } \subset J^\mp\big(  \supp{F }\big)$, $\supp{E_m^\mp F } \cap \Sigma^\pm$ is compact  and therefore the integrals are well defined and Lemma \ref{lem:greens_identity} is indeed applicable\footnote{ We have also commented on this in the proof of Stoke's Theorem \ref{thm:stokes}.}. If we now add the left hand side equations for the plus and the minus sign respectively, we obtain
\begin{align}
\textrm{LHS}^+ + \textrm{LHS}^- &= \int\limits_{\Sigma^+}  A \wedge *F
- \int\limits_{\Sigma^+}  E_m^- F \wedge * \kappa
+\int\limits_{\Sigma^-}  A \wedge *F
- \int\limits_{\Sigma^-}  E_m^+ F \wedge * \kappa \notag \\
&= \int\limits_{\M}  A \wedge *F - \sum\limits_\pm \int\limits_{\Sigma^\pm}  E_m^\mp F \wedge * \kappa \formspace.
\end{align}
We have used that, with respect to an integration over the manifold $\M$, the Cauchy surface $\Sigma$ denotes a set of zero measure and therefore it holds for any compactly supported $4$-form $\alpha$
\begin{align}
\int \limits_{\M} \alpha = \int \limits_{\Sigma^+} \alpha  + \int\limits_{\Sigma^-} \alpha \formspace.
\end{align}
2.) For the right hand side we identify $E_m = E_m^- - E_m^+$ and directly obtain, after summing the equations for the two signs:
\begin{align}
\textrm{RHS}^+ + \textrm{RHS}^-
=   &\,\langle \Az , \rhod E_m F \rangle_\Sigma
  + \langle \Adelta , \rhon E_m F \rangle_\Sigma \notag \\
  & -\langle \An , \rhodelta E_m F \rangle_\Sigma
  - \langle \Ad , \rhoz E_m F \rangle_\Sigma \formspace.
\end{align}
Bringing the factor containing $\kappa$ from the left to the right hand side one finds the wanted relation. Uniqueness of the solution, in a distributional sense, and continuous dependence on the initial data follows from \cite[Theorem, 3.2.12]{baer_ginoux_pfaeffle} which is generalized for non-compactly supported initial data in \cite[Theorem 2.3]{Sanders}.
\end{proof}
\subsubsection{Implementing the Lorenz constraint}\vspace{-5mm}
Now that we have found a solution to the wave equation (\ref{eqn:classical_wave_eqation}) we would like to restrict the initial data, such that the Lorenz constraint (\ref{eqn:classical_constraint}) is fulfilled and we therefore have a solution to Proca's equation. Again, we do this by restricting the initial data of $A$ such that $(\delta A - \frac{1}{m^2} \delta j)$ has vanishing initial data on a Cauchy surface $\Sigma$. It is useful to express the initial data of the zero-form $(\delta A - \frac{1}{m^2} \delta j)$ in terms of the operators $\rho_{(\cdot)}$. For that we need to extend these operators to act on zero-forms. By using the same Definition \ref{def:cauchy_mapping_operators} but letting the operators act on zero-forms we find that:
\begin{lemma}\label{lem:klein_gordon_cauchy_data}
Let $\Sigma$ be a Cauchy surface with unit normal vector field $n$. For any smooth zero-form $f\in \Omega^0(\M)$ it holds that
\begin{subequations}
\begin{align}
\rhon f &= 0 \formspace, \\
\rhodelta f &= 0 \formspace, \\
\rhoz f &= \restr{f}{\Sigma} \formspace, \\
\rhod f &= \restr{(df)(n)}{\Sigma} = \restr{\left(n^\alpha \nabla_\alpha f\right)}{\Sigma} \formspace.
\end{align}
\end{subequations}
Therefore, with respect to the Klein Gordon equation, $\rhoz f$ and $\rhod f$ specify initial data on $\Sigma$.
\end{lemma}
\begin{proof}
1.) We begin with the first equation and chose local coordinates such that  the natural inclusion $i: \Sigma \hookrightarrow \M$ maps $i: (\vector x) \mapsto (0, \vector x)$. In those coordinates the pullback of a $p$-form $T$ is given by $i^*\left( T\indices{_{\alpha_1 \dots \alpha_p}}\right) = T\indices{_{a_1 \dots a_p}}$, where as usual Greek letters take values in $\{0,1,2,3\}$ and Latin letters in $\{1,2,3\}$. From this we find
\begin{align}
\rhon f
&= - *_{(\Sigma)} i^* * f \notag \\
&= - *_{(\Sigma)} i^* \left( \sqrt{\abs{g\indices{_{\mu \nu}}}}\, f\, \epsilon_{\alpha \beta \gamma \delta} \, dx^\alpha \otimes dx^\beta \otimes dx^\gamma \otimes dx^\delta \right) \notag \\
&=  *_{(\Sigma)}  \sqrt{\abs{g\indices{_{\mu \nu}}}}\, f\, \underbrace{\epsilon_{ijkl} \, dx^i \otimes dx^j \otimes dx^k \otimes dx^l}_{=0} \notag \\
&= 0 \formspace .
\end{align}
2.) The second equation follows from, using that for any zero-form $f$ it holds ${\delta f} = 0$,
\begin{align}
\rhodelta f
&= i^* \delta f \notag \\
&= 0 \formspace.
\end{align}
3.) It follows directly from the definition of $\rhoz$ and the inclusion operator $i$ that
\begin{align}
\rhoz f = i^* f = \restr{f}{\Sigma} \formspace.
\end{align}
4.) The action of $\rhod$ on a zero form $f$ is equivalent to the action of $\rhon$ on the one form $df$. For any one-form $A$, the action of $\rhon$ is shown by Furlani \cite[Appendix A]{FURLANI} to be $\rhon A = n^\alpha A_\alpha$. Therefore, we find:
\begin{align}
\rhod f
&= - *_{(\Sigma)} {i^* * }(d f) \notag  \\
&= \rhon df \notag \\
&= n^\alpha (df)_\alpha \notag \\
&= n^\alpha \nabla_\alpha f \formspace.
\end{align}
\end{proof}
In this section we need to introduce some additional properties regarding the properties of the normal vector field $n$ of the Cauchy surface $\Sigma$.
\begin{lemma}[Gaussian Coordinates] \label{lem:normal_vectors}
Let $\Sigma$ be a Cauchy surface of $\M$ with future pointing unit normal vector field $n$. We can extend $n$ to a neighborhood of $\Sigma$ such that the following holds:
\begin{align}
n^\alpha \nabla_\alpha n^\beta = 0 \formspace, \label{eqn:gauss_coordinates_geodesic}\\
dn = 2 \nabla_{[\mu} n_{\nu]} = 0 \formspace. \label{eqn:gauss_coordinates_symmetric}
\end{align}
\end{lemma}
\begin{proof}
	An introduction to Gaussian (normal) coordinates is for example given in \cite[pp. 42,43]{wald_GR} or \cite[pp. 445,446]{carroll_spacetime-and-gr} where equation (\ref{eqn:gauss_coordinates_geodesic}) is shown to hold by construction. Equation (\ref{eqn:gauss_coordinates_symmetric}) can be derived by using Frobenius' theorem (see for example \cite[Theorem B.3.1 and B.3.2]{wald_GR}) as argued in \cite[Section 2.3.2 Equation 5]{Sanders}.
\end{proof}
With this normal vector field we can write the metric $g$ of the spacetime $\M$ in a neighborhood of the Cauchy surface as (see \cite[Equation 10.2.10]{wald_GR})
\begin{align}
g\indices{_{\mu\nu}} = - n_\mu n_\nu + h\indices{_{\mu\nu}} \formspace,
\end{align}
where $h$ denotes the induced metric on $\Sigma$.
Having introduced this notion, we can state the final result of this chapter with the following theorem:
\begin{theorem}[Solution of Proca's equation - constrained version] \label{thm:solution_proca_constrained}
 Let $\Az,\Ad \in \Omega^1(\Sigma)$ and $\An, \Adelta \in \Omega^0(\Sigma)$ specify initial data on the Cauchy surface $\Sigma$. Let $F \in \Omega^1_0(\M)$ be a test one-form, $j \in \Omega^1(\M)$ an external source and $m > 0$ a positive constant.\\
 Restrict the initial data by specifying
 \begin{subequations}\label{eqn:proca_constraints}
  \begin{align}
 \Adelta &= \frac{1}{m^2} \rhodelta j \\
 m^2 \, \An - \delta_{(\Sigma)}\Ad &= \rhon j  \formspace.
 \end{align}
 \end{subequations}
 Then
 \begin{align}
  \langle A, F \rangle_\M = \sum\limits_\pm \langle j + \frac{1}{m^2}d\delta j, E_m^\mp F   \rangle_{\Sigma^\pm}
  &+ \langle \Az , \rhod E_m F \rangle_\Sigma
  + \langle \Adelta , \rhon E_m F \rangle_\Sigma \notag \\
  &- \langle \An , \rhodelta E_m F \rangle_\Sigma
  - \langle \Ad , \rhoz E_m F \rangle_\Sigma \label{eqn:solution_proca_constrained}
 \end{align}
specifies the unique smooth solution of Proca's equation $\left( \delta d + m^2 \right) A = j$ with the given initial data. Furthermore,  the solution depends continuously on the initial data.
\end{theorem}
\begin{proof}
From Theorem \ref{thm:solution_wave_equation} we know that equation (\ref{eqn:solution_proca_constrained}) specifies the unique smooth solution to the wave equation (\ref{eqn:classical_wave_eqation}). We are left with showing that the specified constraints on the initial data are equivalent to the vanishing initial data of $(\delta A - \frac{1}{m^2} \delta j)$ on the Cauchy surface which, as we have seen, is equivalent to the constraint (\ref{eqn:classical_constraint}).
We have to look at the initial data with respect to the Klein-Gordon equation as stated in Lemma \ref{lem:klein_gordon_cauchy_data}. In the following, $i : \Sigma \hookrightarrow \M$ denotes again the inclusion operator.\\
1.) The vanishing of initial value yields
\begin{align}
0 &= \rhoz \left( \delta A - \frac{1}{m^2}\, \delta j\right) \notag\\
&= i^* \delta A - \frac{1}{m^2}\, i^* \delta j \notag \\
&= \rhodelta A - \frac{1}{m^2} \rhodelta j \notag \\
\Leftrightarrow \Adelta &= \frac{1}{m^2} \rhodelta j \formspace.
\end{align}
We have used the linearity of the pullback and Definition \ref{def:cauchy_mapping_operators} of the initial data mapping operator $\rhodelta$. \par
2.) We will calculate the vanishing of the normal derivative in Gaussian normal coordinates and in the end turn back to a coordinate independent notation:
\begin{align}
0 &= \rhod \left( \delta A - \frac{1}{m^2}\, \delta j\right) \notag \\
 &= \rhod \delta A  - \frac{1}{m^2} \rhod  \delta j \notag \\
 &= \restr{\Big[n^\alpha \nabla_\alpha  \delta A \Big]}{\Sigma} - \frac{1}{m^2} \rhod  \delta j  \formspace. \label{eqn:cauchy_normal_derivative_tmp}
\end{align}
We will take a separate look at the first summand:
\begin{align}
n^\alpha \nabla_\alpha  \delta A
&= n^\alpha \left( d\delta A \right)_{\alpha} \notag\\
&= n^\alpha \square A_\alpha - n^\beta \left( \delta  d A \right)_\beta \notag\\
&= n^\alpha \kappa_\alpha -{m^2} \, n^\mu A_\mu + n^\beta \nabla^\nu \nabla_{[\nu} A_{\beta]} \notag\\
&= n^\alpha \kappa_\alpha - {m^2} \, n^\mu A_\mu +  \nabla^\nu \left( n^\beta  \nabla_{[\nu} A_{\beta]} \right) \formspace ,
\end{align}
where we have used that $\left( \nabla^\nu n^\beta \right) \left(  \nabla_{[\nu} A_{\beta]} \right) = 0 $, since from Lemma  \ref{lem:normal_vectors} we know that $\nabla^\nu n^\beta $ is symmetric, and every contraction of a fully symmetric with a fully antisymmetric tensor vanishes. Next we express the metric $g$ as $g\indices{_{\mu\nu}} = -n_\mu n_\nu + h\indices{_{\mu\nu}}$ and use the expansion of the normal vectors as geodesics as stated in Lemma \ref{lem:normal_vectors}. Therefore we further find:
\begin{align}
n^\alpha \nabla_\alpha  \delta A
&= n^\alpha \kappa_\alpha - {m^2} \, n^\mu A_\mu +  g\indices{^{\sigma\nu}} \nabla_\sigma \left( n^\beta  \nabla_{[\nu} A_{\beta]} \right) \notag\\
&= n^\alpha \kappa_\alpha - {m^2} \, n^\mu A_\mu +  \left(- n^\sigma n^\nu +  h\indices{^{\sigma\nu}} \right) \nabla_\sigma \left( n^\beta  \nabla_{[\nu} A_{\beta]} \right) \notag \\
&= n^\alpha \kappa_\alpha - {m^2} \, n^\mu A_\mu  + \nabla_{(\Sigma)}^{\nu} \left( n^\beta  \nabla_{[\nu} A_{\beta]} \right) \formspace.
\end{align}
Here we have made use of the identification $ h\indices{^{\sigma\nu}} \nabla_\sigma B_\mu= \nabla_{(\Sigma)}^\nu  B_\mu$ for any one-form $B$ (see \cite[Lemma 10.2.1]{wald_GR}).
Now, we can identify $n^\beta  \nabla_{[\nu} A_{\beta]} = - n^\beta  \nabla_{[\beta} A_{\nu]} =- A_{(d)\nu}$ and use the local coordinate representation of the exterior derivative of a one-form $B$ on the Cauchy surface $\Sigma$, $\delta_{(\Sigma)} B =  -\nabla_{(\Sigma)} ^\alpha B_\alpha$, and finally obtain
\begin{align}
n^\alpha \nabla_\alpha  \delta A
&=n^\alpha \kappa_\alpha - m^2 \, n^\mu A_\mu  + \nabla_{(\Sigma)}^\nu \left( n^\beta  \nabla_{[\nu} A_{\beta]} \right) \notag \\
&=\rhon \kappa - m^2 \An  + \delta_{(\Sigma)} \Ad \formspace.
\end{align}
In the last step we have taken the restriction to the Cauchy surface, as we are interested in the initial data on the Cauchy surface.
Now, looking back at the normal derivative (\ref{eqn:cauchy_normal_derivative_tmp}) and inserting the definition of the source term $\kappa = j + \frac{1}{m^2}d \delta j$, we find
\begin{align}
 m^2 \An  - \delta_{(\Sigma)} \Ad
 &= \rhon \kappa - \frac{1}{m^2} \rhod \delta j \notag\\
&= \rhon j + \frac{1}{m^2} \rhon d \delta j - \frac{1}{m^2} \rhod \delta j \notag\\
&=  \rhon j \formspace,
\end{align}
since for any $p$-form $B$ it holds by definition $\rhon dB = \rhod B$.\\
By Lemma \ref{lem:Proca_wave_equiv} we have thus found the solution of Proca's equation.
\end{proof}
For further calculations it is useful to develop an unconstrained solution to Proca's equation from the previous theorem.
Basically, instead of considering the full space of initial data for the solution and constraining them, we just take initial data living in a subspace of all initial data and change the dependency of the solution on the data in such a way that this subspace automatically solves the constraints.
Before we state that theorem, we need to introduce fundamental solutions for the Proca operator $\delta d + m^2$ and relate them to the fundamental solutions of the wave operator that we have encountered so far.
\begin{lemma}[Fundamental solutions of the Proca operator]\label{lem:fundamental_solution_proca_operator}
The Proca operator  $(\delta d + m^2) : \Omega^p(\M) \to \Omega^p(\M)$, for $m > 0$,  has unique advanced ($-$) and retarded ($+$) fundamental solutions $\gls{Gpm} : \Omega^p_0(\M) \to \Omega^p(\M)$ that are given by
\begin{align}
G_m^\pm = \left( \frac{d \delta}{m^2} + 1\right) E_m^\pm \formspace,
\end{align}
where $E^\pm_m$ are the advanced and retarded fundamental solutions to the wave operator $(\square + m^2)$. \\
They fulfill the properties
\begin{subequations}
\begin{align}
 G_m^\pm (\delta d + m^2) &= \mathbbm 1 = (\delta d + m^2) G_m^\pm \formspace \text{\upshape{\quad\quad and}} \label{def:fundamental_solution_proca_operator} \\
 \supp{G_m^\pm F } &\subset J^\pm \supp{F } \label{def:fundamental_solution_proca_support} \formspace.
\end{align}
\end{subequations}
The advanced minus retarded fundamental solution is denoted by
\begin{align}
\gls{Gm} = G_m^- - G_m^+ \formspace.
\end{align}
\end{lemma}
\begin{proof}
	First note that existence of $G_m^\pm$ follows from existence of $E_m^\pm$ as stated in Lemma \ref{lem:fundamental_solution_wave_operator}. The quasi inverse property is easily proven by using $\delta ^2 = 0 = d^2$ and the properties for the fundamental solutions $E_m^\pm$ as stated in Lemma \ref{lem:fundamental_solution_wave_operator}. Let $F \in \Omega^p_0(\M)$, then
\begin{align}
(\delta d +m^2) G^\pm_m F
&= (\delta d +m^2)\left( \frac{d \delta}{m^2} + 1\right) E_m^\pm F \notag\\
&= (\delta d +d \delta +m^2) E_m^\pm F \notag\\
&= F
\end{align}
and
\begin{align}
G^\pm_m (\delta d +m^2)  F
&=\left( \frac{d \delta}{m^2} + 1\right) E_m^\pm (\delta d +m^2) F \notag \\
&=\left( \frac{d \delta}{m^2} + 1\right) (\delta d +m^2)E_m^\pm F\notag \\
&= (\delta d +d \delta +m^2) E_m^\pm F \notag\\
&= F \formspace.
\end{align}
The support property follows directly from the support property of the fundamental solutions to $(\square + m^2)$ as stated in Lemma \ref{lem:fundamental_solution_wave_operator}. Again, we use that derivatives do not increase the support of a function. Therefore:
\begin{align}
\supp{G^\pm_m  F }
&= \supp{\left( \frac{d \delta}{m^2} + 1\right) E_m^\pm  F } \notag\\
&\subset \supp{E_m^\pm F } \notag\\
&\subset J^\pm(\supp{F }) \formspace.
\end{align}
For the proof of uniqueness, let $G^\pm_m$ and $G'^\pm_m$ denote two fundamental solutions to the Proca operator. Let $F \in \Omega^p_0(\M)$ and define $A^\pm = (G^\pm_m - G'^\pm_m)F $.
Then:
\begin{align}
(\delta d + m^2) A^\pm = 0 \notag\\
\implies \delta A^\pm = 0 \notag\\
\implies (\square + m^2)A^\pm = 0.
\end{align}
Since $\supp{A^\pm} \subset J^\pm(\supp{F })$, we know that initial data of $A^\pm$ vanishes on every Cauchy surface in the past/future of $\supp{F }$ for the $+/-$ sign respectively. By the same argument as used in the proof of Lemma \ref{lem:fundamental_solution_wave_operator}, we find a globally vanishing solution (again, see \cite[Corollary 3.2.4]{baer_ginoux_pfaeffle}), that is,
\begin{align}
A^\pm &= 0 \notag\\
\implies G^\pm_m &= G'^\pm_m \formspace,
\end{align}
since $F $ is arbitrary. This completes the proof.
\end{proof}
Now we are ready to state the final, unconstrained version of the solution of Proca's equation.This following theorem will by the cornerstone of the calculations in the next chapter:
\begin{theorem}[Solution of Proca's equation - unconstrained version]\label{thm:solution_proca_unconstrained}
 Let $\Az,\Ad \in \Omega^1(\Sigma)$ specify a subset of initial data on the Cauchy surface $\Sigma$. Let $F \in \Omega^1_0(\M)$ be a test one-form, $j \in \Omega^1(\M)$ an external source and $m >0$ a positive constant.
 Then
 \begin{align}
  \langle A, F \rangle_\M = \sum\limits_\pm \langle j , G_m^\mp F   \rangle_{\Sigma^\pm}
  + \langle \Az , \rhod G_m F \rangle_\Sigma
 - \langle \Ad , \rhoz G_m F \rangle_\Sigma \label{eqn:solution_proca_unconstrained}
 \end{align}
specifies the unique smooth solution of Proca's equation $\left( \delta d + m^2 \right) A = j$ with the given subset of initial data. Furthermore,  the solution depends continuously on the initial data.
\end{theorem}
\begin{proof}
The theorem follows from Theorem \ref{thm:solution_proca_constrained} by inserting the constraints (\ref{eqn:proca_constraints}) into the expression (\ref{eqn:solution_proca_constrained}).
We find
 \begin{align} \label{eqn:tmp_proca_solution_constraint_inserted}
  \langle A, F \rangle_\M =
  & \sum\limits_\pm \langle j + \frac{1}{m^2}d\delta j, E_m^\mp F   \rangle_{\Sigma^\pm}
  + \langle \Az , \rhod E_m F \rangle_\Sigma
  + \frac{1}{m^2} \langle  \rhodelta j, \rhon E_m F \rangle_\Sigma \notag \\
  & - \langle \frac{1}{m^2} \delta_{(\Sigma)} \Ad , \rhodelta E_m F \rangle_\Sigma
  - \frac{1}{m^2} \langle \rhon j, \rhodelta E_m F \rangle_\Sigma
  - \langle \Ad, \rhoz E_m F \rangle_\Sigma \formspace.
 \end{align}
Now, for clarity's sake, we take a look at the appearing terms separately:\\
1.) To get rid of the appearing divergence of $\Ad$, we use some basic identities, that is, in particular the formal adjointness of $\delta$ and $d$ and the commutativity of $d$ with the pullback $i^*$:
\begin{align}
\langle \delta_{(\Sigma)} \Ad, \rhodelta E_m F \rangle_\Sigma
=& \langle \Ad, d_{(\Sigma)} \rhodelta E_m F  \rangle_\Sigma \notag\\
=& \langle \Ad, d_{(\Sigma)} i^* \delta  E_m F \rangle_\Sigma \notag\\
=& \langle \Ad,  i^* d \delta  E_m F \rangle_\Sigma \notag \\
=& \langle \Ad, \rhoz d \delta  E_m F \rangle_\Sigma \formspace,
\end{align}
which, together with $\langle \Ad, \rhoz E_m F \rangle_\Sigma$,  combines to $\langle \Ad , \rhoz \left( \frac{d \delta}{m^2} +1 \right) E_m F \rangle_\Sigma$. \par
2.) Next, we have a look at a part of the sum term and use Stoke's theorem (again, we get a sign due to the orientation of $\Sigma$ with respect to $\Sigma^\pm$) for a formal partial integration, at the cost of some boundary terms:
\begin{align}
\sum\limits_\pm \langle d \delta j ,  E_m^\mp F   \rangle_{\Sigma^\pm}
&= \sum\limits_\pm  \int\limits_{\Sigma^\pm}  \left( d \delta j \wedge {* E_m^\mp} F \right) \notag \\
&= \sum\limits_\pm  \int\limits_{\Sigma^\pm} \left\{ d \left(  \delta j \wedge {* E_m^\mp} F \right) - \delta j \wedge d {* E_m^\mp} F \right\} \notag\\
&= \sum\limits_\pm \left\{ \phantom{\pm} \int\limits_{\Sigma^\pm}  d \left(  \delta j \wedge * E_m^\mp F \right) + \int\limits_{\Sigma^\pm} \delta j \wedge {**} {d {* E_m^\mp}} F \right\} \notag\\
&= \sum\limits_\pm \left\{  \pm \int\limits_{\Sigma}  i^*{\left(  \delta j \wedge * E_m^\mp F \right) }+ \int\limits_{\Sigma^\pm} \delta j \wedge * \delta E_m^\mp F \right\} \notag\\
&= \sum\limits_\pm \left\{  \pm \int\limits_{\Sigma}  i^*{\left(  \delta j \wedge * E_m^\mp F \right) }+ \int\limits_{\Sigma^\pm} \delta E_m^\mp F \wedge {* *}{d*}j \right\} \notag\\
&= \sum\limits_\pm \left\{  \pm \int\limits_{\Sigma}  i^*{\left(  \delta j \wedge * E_m^\mp F \right)} - \int\limits_{\Sigma^\pm} \delta E_m^\mp F \wedge {d*}j \right\} \notag \\
&= \sum\limits_\pm \left\{  \pm \int\limits_{\Sigma}  i^*{\left(  \delta j \wedge * E_m^\mp F \right) }- \int\limits_{\Sigma^\pm} \Big( d \left( \delta E_m^\mp F \wedge *j \right) - d \delta E_m^\mp F \wedge *j \Big) \right\} \notag \\
&= \sum\limits_\pm \left\{  \pm \int\limits_{\Sigma}  i^*{\left(  \delta j \wedge * E_m^\mp F \right)} \mp \int\limits_{\Sigma} i^*{\left( \delta E_m^\mp F \wedge *j \right) }+ \int\limits_{\Sigma^\pm} j \wedge *d \delta E_m^\mp F  \right\} \notag \\
&= \sum\limits_\pm \langle  j , d \delta E_m^\mp F \rangle_{\Sigma^\pm}  + \int\limits_{\Sigma}  i^{*}{\left(  \delta j \wedge * E_m F \right)} - \int\limits_{\Sigma} i^*{\left(j \wedge * \delta E_m F \right)} \label{eqn:tmp_proca_boundary_terms}
\end{align}
Lastly, we will see that the remaining source dependent terms of (\ref{eqn:tmp_proca_solution_constraint_inserted}) will cancel with the boundary terms obtained from partial integration above:
\begin{align}
 \langle  \rhodelta j, \rhon E_m F \rangle_\Sigma \notag  -  \langle \rhon &j, \rhodelta E_m F \rangle_\Sigma \notag\\
 &= - \langle  i^* \delta  j, *_{(\Sigma)}{ i^* *}E_m F \rangle_\Sigma
 +  \langle *_{(\Sigma)} {i^* *}  j, i^* \delta  E_m F \rangle_\Sigma  \notag\\
 &= -\int\limits_{\Sigma}  i^* \delta  j \wedge *_{(\Sigma)} {*_{(\Sigma)}{ i^* {*E_m F }}}
 +  \int\limits_{\Sigma}  i^* \delta  E_m F \wedge  {*_{(\Sigma)}  {*_{(\Sigma)}{ i^* {*  j} }}}    \notag\\
  &=-\int\limits_{\Sigma}  i^* {\left( \delta  j \wedge {*E_m F } \right) }  +  \int\limits_{\Sigma}  i^* {\left( \delta  E_m F \wedge  {*  j}   \right) }\formspace.
\end{align}
These terms cancel out the boundary terms in (\ref{eqn:tmp_proca_boundary_terms}) (note that all of them have a prefactor $\frac{1}{m^2}$ that was not carried along the calculation for simplicity). Therefore, we obtain the result
 \begin{align}
  \langle A, F \rangle_\M = \sum\limits_\pm \langle j , \left( \frac{d \delta}{m^2} +1 \right) E_m^\mp F   \rangle_{\Sigma^\pm}
  &+ \langle \Az , \rhod E_m F \rangle_\Sigma \notag \\
  &- \langle \Ad , \rhoz \left( \frac{d \delta}{m^2} +1 \right) E_m F \rangle_\Sigma \formspace.
 \end{align}
 Now, to complete the proof, we make use of the identity $G_m = \left( \frac{d \delta}{m^2} +1 \right) E_m$, see Lemma \ref{lem:fundamental_solution_proca_operator}, and a simple calculation additionally gives, using $d^2 =0$:
 \begin{align}
 \rhod G_m = - {*_{(\Sigma)} {i^* *}} d\left( \frac{d \delta}{m^2} +1 \right) E_m = - {*_{(\Sigma)}{i^* *}} E_m = \rhod E_m \formspace,
 \end{align}
 which completes the proof.
\end{proof}
At this stage, one might wonder how this result compares to the discussed fact that the Proca field only possesses three independent degrees of freedom as discussed at the beginning in Section \ref{sec:solving_procas_equation}. In the formalism that we work with, the counting of degrees of freedom is a bit subtle. We have started with a solution to a wave equation, clearly possessing four independent degrees of freedom expressed in the initial data formulation by $\Az, \Ad, \Adelta$ and $\An$. To obtain a solution to Proca's equation, we have implemented a Lorenz constraint by restricting the initial data. In the last step, concluding in Theorem \ref{thm:solution_proca_unconstrained}, we have effectively eliminated the initial zero-forms $\Adelta$ and $\An$ on the Cauchy surface $\Sigma$. But those two zero forms can be viewed as initial data to a scalar Klein-Gordon field! In that sense, we have eliminated one scalar degree of freedom, and are left with the ``correct'' three independent degrees of freedom in Proca's theory.
\subsection{The zero mass limit}\label{sec:zero-mass-limit-classical}
As a basis for understanding the zero mass limit in the quantum case, we will now investigate the corresponding classical limit. The question is for which test one-forms the zero mass limit of distributional solutions to the Proca equation exists, or more precisely:
\begin{center}\textit{
		Let $A_m$ be a solution to the Proca equation with mass $m$,\\ and $\Az,\Ad \in \Omega^1(\Sigma)$ its initial data.\\ For which $F \in \Omega^1_0(\M)$, if any, does the limit $\lim\limits_{m\to 0^+} \langle A_m , F \rangle_\M = \lim\limits_{m \to 0^+}\Big( \sum\limits_\pm \langle j , G_m^\mp F   \rangle_{\Sigma^\pm} +\langle \Az , \rhod G_m F \rangle_\Sigma
		- \langle \Ad , \rhoz G_m F \rangle_\Sigma \Big)$ exist?}
\end{center}
We have used the explicit form of distributional solutions to the Proca equation as presented in Theorem \ref{thm:solution_proca_unconstrained}.
Before we can answer this question, we need to make an assumption regarding the continuity of the propagators of the Proca operator with respect to the mass. As we have specified the propagator for the Proca operator in terms of the propagator of the Klein-Gordon operator, we will state it in the following way:
\begin{assumption}\label{ass:propagator_continuity}
	Let $m\geq 0$ and $E_m^\pm$ the fundamental solutions to the Klein-Gordon operator $(\delta d + d \delta + m^2)$. Then, for a fixed $F\in \Omega^1_0(\M)$,
	\begin{align}
		m \mapsto E^\pm_m F
	\end{align}
	is continuous. Therefore
		\begin{align}
		m \mapsto E_m F
		\end{align}
	is continuous and
	\begin{align}
	\lim\limits_{m \to 0^+} E^\pm_m F &= E^\pm_0 F \quad \text{and} \\
		\lim\limits_{m \to 0^+} E_m F &= E_0 F \formspace.
	\end{align}
\end{assumption}
This assumption remains unproven in the context of this thesis. With this, continuity of $G_m F$ for a fixed test one-form $F$ follows directly for $m>0$. Using the assumption, we can now investigate the zero mass limit.
We will split this up into the case of vanishing external sources, $j=0$, and the general case with sources for clarity.
\subsubsection{Existence of the limit in the current-free case}\label{sec:zero-mass-limit-existence-classical-vanishing-source}
Let $A_m$ specify a solution to Proca's equation with mass $m$ and vanishing external sources $j=0$. Recall that by Theorem \ref{thm:solution_proca_unconstrained} a solution to Proca's equation with mass $m$ is uniquely determined by initial data $\Az, \Ad \in \Omega^1(\Sigma)$ by
 \begin{align}
\langle A_m, F \rangle_\M = \langle \Az , \rhod G_m F \rangle_\Sigma
 - \langle \Ad , \rhoz G_m F \rangle_\Sigma \formspace,
 \end{align}
where $G_m = \frac{1}{m^2} \left( d \delta + m^2 \right) E_m$. From this expression, we want to find necessary and sufficient conditions for the limit $\lim\limits_{m\to 0^+} \langle A_m , F \rangle_\M$ to exist. This is a rather tricky task, because it is not clear how to link the continuity in this distributional sense to test one-forms. We therefore have to tighten the request to the existence of the corresponding limits of initial data of $G_m F$, that is:
\begin{center}\textit{
		For which $F \in \Omega^1_0(\M)$, if any, do the limits\\ $\lim\limits_{m \to 0^+} \rhoz G_m F$ and  $\lim\limits_{m \to 0^+} \rhod G_m F$ exist?}
\end{center}
Clearly, the existence of these limits is sufficient for the existence of the limit in the distributional sense as stated above. To answer this tightened question, we make use of the following lemma:
\begin{lemma}\label{lem:limit_existence_classical_equivalence}
	Let $j=0$, $F \in \Omega^1_0(\M)$ fixed and assume Assumption \ref{ass:propagator_continuity} holds. The following statements are equivalent:
	\begin{enumerate}
		\item {The limits $\lim\limits_{m \to 0^+} \rhoz G_m F$ and  $\lim\limits_{m \to 0^+} \rhod G_m F$ exist. \\}
		\item {The limit $\lim\limits_{m \to 0^+} G_m F$ exists. \\}
		\item {The limit $\lim\limits_{m \to 0^+} \frac{1}{m^2}E_m d \delta  F $ exists. }
	\end{enumerate}
\end{lemma}
\begin{proof}
1.) We show (i) being equivalent to (ii): \\
a) (ii) $\implies$ (i) is trivial since, if $\lim\limits_{m \to 0^+} G_m F$ exists, clearly the limits $\lim\limits_{m \to 0^+} \rhoz G_m F$ and  $\lim\limits_{m \to 0^+} \rhod G_m F$ exist, as the operators $\rho_{(\cdot)}$ are continuous and do not depend on the mass $m$.  \par
b) (i) $\implies$ (ii):\\
Assume  $\lim\limits_{m \to 0^+} \rhoz G_m F$ and  $\lim\limits_{m \to 0^+} \rhod G_m F$ exist. We know that $\rho_{(\cdot)} G_m F$ specify initial data to the solution $G_m F \equiv B_m$ of the source free Proca equation: Specifying $\rhoz B_m$, $\rhod B_m$, $\rhon B_m$, $\rhodelta B_m$ is equivalent to specifying $B_0 = \restr{B_m}{\Sigma}$, $B_1 = \restr{\nabla_n B_m}{\Sigma}$, for some future pointing timelike unit normal field $n$ of the Cauchy surface $\Sigma$, as shown in \cite[pp. 2613]{FURLANI}. Furthermore, we know that the solution depends continuously on the initial data $B_0$ and $B_1$: Since $ B_m= G_m F$ has compact spacelike support, $B_0$ and $B_1$ will be compactly supported on $\Sigma$. For the case of compactly supported initial data, continuous dependence of the solution on the data is shown in \cite[Theorem 3.2.12]{baer_ginoux_pfaeffle}, which generalizes to arbitrarily supported initial data \cite[Theorem 2.3]{Sanders}.
We conclude that the solution $B_m$ depends continuously on $\rhoz B_m$, $\rhod B_m$, $\rhon B_m$, $\rhodelta B_m$ with respect to the topology of $\Omega^1(\M)$ and restricting $\rhon B_m$ and $\rhodelta B_m$ in terms of $\rhoz B_m$ and $\rhod B_m$ will not change the continuous dependence of the solution on $\rhoz B_m$ and $\rhod B_m$. A more direct approach to this statement is shown in \cite[Proposition 2.5]{pfenning}.
Therefore, $G_m F$ is continuous in $m$ and the corresponding limit exists.\par
2.) It remains to show the equivalence of (ii) and (iii):\\
	a.) (iii) $\implies$ (ii):\\ Assume that $\lim\limits_{m \to 0^+} \frac{1}{m^2}d \delta E_m F $ exists. Then
	\begin{align}
		\lim\limits_{m \to 0^+} G_m F
		&= \lim\limits_{m \to 0^+} \left( \frac{d \delta}{m^2} + 1\right) E_m F \notag\\
		&= \lim\limits_{m \to 0^+} \frac{1}{m^2}E_m d \delta  F  + \lim\limits_{m \to 0^+} E_m F
	\end{align}
	exists, using Assumption $\ref{ass:propagator_continuity}$ and that $d$, respectively $\delta$, commutes with $E_m$.\\
	b.) (ii) $\implies$ (iii):\\ Assume that $\lim\limits_{m \to 0^+} G_m F  $ exists. Then
	\begin{align}
		\lim\limits_{m \to 0^+} \frac{1}{m^2}E_m d \delta  F
		&= \lim\limits_{m \to 0^+} \frac{1}{m^2}d \delta E_m   F \notag \\
		&= \lim\limits_{m \to 0^+} \left( G_m F - E_m F \right) \notag \\
		&= \lim\limits_{m \to 0^+}  G_m F - \lim\limits_{m \to 0^+}  E_m F
	\end{align}
	exists, again using Assumption $\ref{ass:propagator_continuity}$.\par This completes the proof.
\end{proof}
With this result, the existence of the desired limit is purely determined by the existence of the zero mass limit of the propagator of the Proca operator. This can be quite easily formulated in terms of conditions on the test one-forms that the propagator acts on:
\begin{lemma}\label{lem:mass-zero-limit-existence-classical_weak}
	Let $F\in \Omega^1_0(\M)$ and $A_m \in \Omega^1(\M)$ be a solution to Proca's equation with vanishing external sources. Then,
	\begin{center}
		the limits $\lim\limits_{m \to 0^+} \rhoz G_m F$ and  $\lim\limits_{m \to 0^+} \rhod G_m F$ exist\\ if and only if $F = F' + F''$,\\
	\end{center}
	where $F', F'' \in \Omega^1_0(\M)$ such that $dF' = 0 = \delta F''$.\\
	Then also the limit $\lim\limits_{m \to 0^+} \langle A_m , F \rangle_\M$ exists.
\end{lemma}
\begin{proof}
	Let $F \in \Omega^1_0(\M)$. Using Lemma \ref{lem:limit_existence_classical_equivalence}, the existence of the desired limit is equivalent to the existence of the limit $\lim\limits_{m \to 0^+} \frac{1}{m^2}d \delta E_m F $.\par
	1.) We begin by finding the necessary conditions that this limit exists:\\
	 Assume $\lim\limits_{m \to 0^+} \frac{1}{m^2}d \delta E_m F $ exists. Rewriting $E_m d \delta F = m^2 \left( \frac{1}{m^2} d \delta E_m F\right)$ and using the assumption of the existence of the limit of the terms in brackets, we directly find
	\begin{align}
		E_0 d \delta F
		&= \lim\limits_{m \to 0^+} m^2 \left( \frac{1}{m^2}d \delta E_m F \right) \notag\\
		&= 0 \formspace.
	\end{align}
	For a compactly supported $F$ this yields the existence of a compactly supported $F' = E^+_0 d \delta F = E^-_0 d \delta F \in \Omega^1_0(\M)$, such that $d \delta F = (\delta d + d \delta)F'$ (see e.g. \cite[Proposition 2.6]{Sanders}). From the definition of $F'$ we immediately find $dF' = 0$.\\
	%
	%
	Moreover, we obtain an additional condition:
	\begin{align}
		0
		&= E_0 d \delta F\notag  \\
		&= E_0 (d \delta + \delta d)F - E_0 \delta d F \notag \\
		&= - E_0 \delta d F \formspace.
	\end{align}
	By the same argument as before, this yields the existence of a one-form $F'' \in \Omega^1_0(\M)$, $F'' = E^+_0 \delta d F = E^-_0 \delta dF$ which yields $\delta F'' = 0$. \\
	Finally, by definition we find
	\begin{align}
		F' + F '' = E_0^+ (d \delta + \delta d) F = F \formspace.
	\end{align}
	Therefore, as a necessary condition for the limit to exist, $F$ has to be the sum of a closed and a co-closed compactly supported one-form: $F = F' + F''$, where $F',F'' \in \Omega^1_0(\M)$  such that $dF' = 0 = \delta F''$. \par
	2.) In the next step, we show that the condition is also sufficient:\\
	Let $F = F' + F''$, where $F', F'' \in \Omega^1_0(\M)$ and $dF' = 0 = \delta F''$. Then
	\begin{align}
		\lim\limits_{m \to 0^+} \frac{1}{m^2} E_m d \delta F
		&=	\lim\limits_{m \to 0^+} \frac{1}{m^2} E_m d \delta (F' + F '') \notag\\
		&=  \lim\limits_{m \to 0^+} \frac{1}{m^2} E_m d \delta F'  \notag\\
		&= \lim\limits_{m \to 0^+} \Big( \frac{1}{m^2} E_m (d \delta + \delta d + m^2) F' - E_m F'\Big)\notag\\
		& = - \lim\limits_{m \to 0^+}  E_m F' \formspace,
	\end{align}
	which exists by assumption \ref{ass:propagator_continuity}. \\
	This completes the proof.
\end{proof}
From a formal point of view, we have completely classified those test one-forms, for which the zero mass limit exists. But it turns out that we can tighten the result even more, by observing that closed one-forms $F \in \Omega^1_0(\M)$, such that $dF=0$, do not contribute to the observable $\langle A_m , F \rangle_\M$ in the source free case. That is, for those $F$ it holds:
\begin{align}
	G_m F
	&= E_m  \left( \frac{d \delta}{m^2} + 1\right) F \notag \\
	&= \frac{1}{m^2}E_m (d \delta + \delta d + m^2) F \notag\\
	&= 0 \formspace,
\end{align}
which yields $\langle \Az, \rhod G_m F \rangle_\M = 0 = \langle \Ad, \rhoz G_m F \rangle_\M$ and hence $\langle A_m , F \rangle_\M= 0$. Due to the linearity of the fields, two test one-forms that differ by a closed compactly supported one-form give rise to the same physical observable. We may therefore restrict the class of test one-forms that we smear the fields $A_m$ with to the test one-forms modulo closed test one-forms. This yields the final result of this section:
\begin{theorem}[Existence of the zero mass limit in the source free case]\label{thm:limit_existence_sourcefree_classical}
	Let $F, F' \in \Omega^1_0(\M)$ such that $[F] = [F']$, that is, there is a $\chi \in \Omega^1_{0,d}(\M)$ such that $F = F' + \chi$. Let $A_m$ be a solution to Proca's equation with vanishing external source $j=0$. Then,
	\begin{align}
		G_m F &= G_m F' \eqqcolon G_m [F] \quad \text{and} \\
		\langle A_m, F \rangle_\M &=\langle A_m, F' \rangle_\M \eqqcolon \langle A_m, [F] \rangle_\M \formspace,
	\end{align}
	and
	\begin{center}
		the limits $\lim\limits_{m \to 0^+} \rhoz G_m [F]$ and  $\lim\limits_{m \to 0^+} \rhod G_m [F]$ exist\\[2mm] 	 if and only if there exists a representative $\tilde{F}$ of $[F]$ with ${\delta \tilde{F}} = 0 $.
	\end{center}
	Then, also the limit $\lim\limits_{m \to 0^+} \langle A_m, [F] \rangle_\M$ exists.
\end{theorem}
\begin{proof}
	Let $F, F' \in \Omega^1_0(\M)$ such that $[F] = [F']$. Let $A_m$ be a solution to the source free Proca equation.
	We have already seen that for a closed test one-form $\chi \in \Omega^1_{0,d}(\M)$ it holds $G_m \chi =0$. It follows directly that $G_m F = G_m F'$ and hence $G_m [F]$ is well defined using a representative of the equivalence class $[F]$. Using Theorem \ref{thm:solution_proca_unconstrained} and the linearity of $\langle \cdot, \cdot \rangle_\M$, it directly follows $\langle A_m, F \rangle_\M =\langle A_m, F' \rangle_\M$ and hence $\langle A_m, [F] \rangle_\M$ is well defined. Therefore, we can, without losing any observables, divide out the test one-forms that are closed.
	By Lemma \ref{lem:mass-zero-limit-existence-classical_weak} we know that the limit exists if and only if $F$ is a sum of a closed and a co-closed test one-form. Hence, the limits $\lim\limits_{m \to 0^+} \rhoz G_m [F]$ and  $\lim\limits_{m \to 0^+} \rhod G_m [F]$ exists if and only if
	\begin{align}
		[F] \in \frac{\Omega^1_{0,d}(\M) + \Omega^1_{0,\delta}(\M)}{\Omega^1_{0,d}(\M)}\formspace.
	\end{align}
	Here, $\Omega^1_{0,d}(\M)$, $\Omega^1_{0,\delta}(\M)$ denotes the set of closed and co-closed test one-forms respectively.\\
	We will now show that in fact it holds
	\begin{align}
		\frac{\Omega^1_{0,d}(\M) + \Omega^1_{0,\delta}(\M)}{\Omega^1_{0,d}(\M)} \cong \Omega^1_{0,\delta}(\M) \formspace.
	\end{align}
	Let $F \in \Omega^1_{0,d}(\M) + \Omega^1_{0,\delta}(\M)$, that is, $F = F' + F''$ such that $dF' = 0 = \delta F''$. It directly follows $[F] = [F' + F''] = [F'']$. Indeed, $F'' \in \Omega^1_{0,\delta}(\M)$ is the \emph{unique} co-closed representative of the equivalence class $[F]$: Assume there exists a $\tilde{F} \in \Omega^1_{0,\delta}(\M)$ such that $[\tilde{F}] = [F] = [F'']$. From this it follows $[\tilde{F} - F''] = 0$, that is, $\tilde{F}$ and $F''$ differ by a closed test one-form, therefore we conclude $d(\tilde{F} - F'') = 0$.	By construction, it additionally holds $\delta(\tilde{F} - F'') =0$. Therefore $(\tilde{F} - F'')$ solves a source free massless wave equation:
	\begin{align}
		(\delta d + d \delta)(\tilde{F} - F'') = 0 \formspace.
	\end{align}
	Since $(\tilde{F} - F'')$ is compactly supported, it follows from \cite[Corollary 3.2.4]{baer_ginoux_pfaeffle} that $(\tilde{F} - F'') =0$. Hence $F''$ is the unique co-closed representative of $[F]$. \\
	The map
	\begin{align}
		\gamma : \frac{\Omega^1_{0,d}(\M) + \Omega^1_{0,\delta}(\M)}{\Omega^1_{0,d}(\M)} &\to \Omega^1_{0,\delta}(\M)\\
		[F] = [F' + F''] &\mapsto F'' \notag
	\end{align}
	is therefore well defined. Clearly, $\gamma$ is linearly bijective.
\end{proof}
Note, that this ``gauge'' by closed test one-forms is only present in the source free theory and not a real gauge freedom of the theory. We will therefore drop the explicit notation of the equivalence classes in the source free case and keep in mind that closed test one-forms do not contribute to the observables in the source free theory.\\
We find that it is sufficient as well as necessary for the mass zero limit to exist in the source free case to restrict to co-closed test one-forms. What is the interpretation of this?
In fact, this can be quite easily understood under the duality $\langle \cdot , \cdot \rangle_\M$. One finds that $\Quotientscale{\mathcal{D}^1(\M)}{d\mathcal{D}^{0}(\M)}$ is dual to $\Omega^1_{0,\delta}(\M)$ (see \cite[Section 3.1]{Sanders}). Here, $\mathcal{D}^1(\M)$ denotes the set of distributional one-forms (in a physical sense, these are classical vector fields) and $\Omega^1_{0,\delta}(\M)$ denotes the set of all co-closed test one-forms. Therefore, restricting to co-closed test one-forms is equivalent to implementing the gauge equivalence $A \to A + d\chi$, for $A \in \mathcal{D}^1(\M)$ and $ \chi \in \mathcal{D}^0(\M)$, in the theory! This dual relation is easily checked for $A' = A + d\chi$ dual to $F \in \Omega^1_{0,\delta}(\M)$
		\begin{align}
			\langle A', F \rangle
			&= \langle A, F \rangle + \langle d\chi, F \rangle \notag\\
			&= \langle A, F \rangle + \langle \chi, \delta F \rangle \notag\\
			&= \langle A, F \rangle \formspace.
		\end{align}
This is a nice result, since the gauge equivalence is naturally present in the Maxwell theory. And due to the non trivial topology on a general spacetime, it is a priori not clear how to implement the gauge equivalence in Maxwell's theory on curved spacetime: Maxwell's equation $\delta d A = 0$ yields that two solutions that differ by a closed one-form give rise to the same observable. For Minkowski spacetime this yields the familiar gauge equivalence $A \to A + d\chi$ since all closed one-forms are exact due to the trivial topology\footnote{For Minkowski spacetime, this follows from the Lemma of Poincar\'e, see e.g. \cite[Corollary 4.3.11]{rudolph_schmidt}.}. This does not hold for arbitrary spacetimes $\M$. One can argue that the gauge equivalence class given by the gauge equivalence of closed rather then exact one-forms is too large as it does not capture all physical phenomena of the theory: As presented in \cite[p. 626]{Sanders}, the Aharonov-Bohm effect \emph{does} distinguish between forms that differ by a form that is closed but not exact, so the gauge equivalence by closed one-forms cannot be the true physical gauge equivalence class. Hence, arguing with physical properties is needed to find the ``right'' gauge equivalence class for the Maxwell theory in curved spacetimes. With our result, this gauge equivalence by exact forms comes naturally in the limit process! \\
 Hence, we have already captured one important feature of the Maxwell theory also in the massless limit of the Proca theory! It remains to check whether also the dynamics are ``well behaved'' in the massless limit. But first, we investigate the zero mass limit for the general theory with sources.
\subsubsection{Existence of the limit in the general case with current}\label{sec:zero-mass-limit-existence-classical-general-source}
The question of interest is analogous to the one presented in the previous section but now including external sources $j \neq 0$. Again, a solution to Proca's equation with initial data $\Az, \Ad \in \Omega^1(\Sigma)$ is uniquely determined  by
\begin{align}
	\langle A_m , F \rangle_\M = \Big( \sum\limits_\pm \langle j , G_m^\mp F   \rangle_{\Sigma^\pm} +\langle \Az , \rhod G_m F \rangle_\Sigma
	- \langle \Ad , \rhoz G_m F \rangle_\Sigma \Big) \formspace.
\end{align}
In order to classify those test one-forms $F$ for which the zero mass limit exists, we have to tighten the main question as posed in the beginning of Section \ref{sec:zero-mass-limit-classical} which was formulated in a distributional sense. Just as in the source free case we again demand that $\lim\limits_{m \to 0^+} \rhoz G_m F$ and  $\lim\limits_{m \to 0^+} \rhod G_m F$ exist. Furthermore, we need a tightened  condition for the limit $\lim\limits_{m \to 0^+} \sum_\pm \langle j , G_m^\mp F   \rangle_{\Sigma^\pm} $ to exist. First, we note that there are mainly three situations that can occur regarding this sum of integrals. Either, the support of $F$ lies in the future of the Cauchy surface $\Sigma$ in which case $\supp{G_m^+ F} \cap \Sigma^- = \emptyset$ and $\langle j , G_m^+ F   \rangle_{\Sigma^-} = 0$. Similarly, if the support of $F$ lies in the past of $\Sigma$, then $\langle j , G_m^- F   \rangle_{\Sigma^+} = 0$. Or, the intersection of the support of $F$ and the Cauchy surface $\Sigma$ is non-empty in which case both terms appear. But since we want the existence of the limit to be independent of the choice of the Cauchy surface $\Sigma$ we conclude that the limit $\lim\limits_{m \to 0^+} \sum_\pm \langle j , G_m^\mp F   \rangle_{\Sigma^\pm} $ exists if and only if $\lim\limits_{m \to 0^+} \langle j , G_m^+ F   \rangle_{\Sigma^-}$ and $\lim\limits_{m \to 0^+} \langle j , G_m^- F   \rangle_{\Sigma^+}$ exist separately. In the same fashion as for the initial data terms, we therefore want the limits $\lim\limits_{m \to 0^+} G_m^\pm F$ to exist. In this sense of the existence of the limit, the question of interest is now a slightly generalized version of what was stated in the previous section:
\begin{center}\textit{
		For which $F \in \Omega^1_0(\M)$, if any, do the limits \\ $\lim\limits_{m \to 0^+} G_m^\pm F$, $\lim\limits_{m \to 0^+} \rhoz G_m F$ and  $\lim\limits_{m \to 0^+} \rhod G_m F$ exist?}
\end{center}
We have already classified the existence of the latter two limits. And with similar calculations, also the first term is quite easy to handle. We find the following result:
\begin{theorem}[Existence of the zero mass limit in the general case]\label{thm:limit_existence_general_classical}
	Let $F \in \Omega^1_0(\M)$, $A_m$ a solution to Proca's equation with external source $j\neq0$. Then,
	\begin{center}
		the limits $\lim\limits_{m \to 0^+} G_m^\pm F$, $\lim\limits_{m \to 0^+} \rhoz G_m F$ and  $\lim\limits_{m \to 0^+} \rhod G_m F$ exist\\[2.5mm]
		if and only if ${\delta F} = 0 $.
	\end{center}
	Then, also the limit $\lim\limits_{m \to 0^+} \langle A_m, F \rangle_\M$ exists.
\end{theorem}
\begin{proof}
Let $F \in \Omega^1_0(\M)$, $A_m$ a solution to the source free Proca equation. \\
1.)  For the limits $\lim\limits_{m \to 0^+} \rhoz G_m F$ and  $\lim\limits_{m \to 0^+} \rhod G_m F$ to exist we have found in Lemma \ref{lem:mass-zero-limit-existence-classical_weak} that it is sufficient and necessary for $F$ to be the sum of a closed and a co-closed test one-form, $F = F' + F''$, $F',F'' \in \Omega^1_0(\M)$ such that $dF' = 0 = \delta F''$. \par
2.) For the limits $\lim\limits_{m \to 0^+} G_m^\pm F$ we find existence if and only if $\lim\limits_{m \to 0^+} \frac{1}{m^2}E_m^\pm d \delta F0$ exists by a calculation analogous to the one presented in the proof of Lemma \ref{lem:limit_existence_classical_equivalence}.\par
a) Assume $\lim\limits_{m \to 0^+} \frac{1}{m^2}E_m^\pm d \delta F$ exists. We conclude $E_0^\pm d \delta F = 0$, following the calculations in the proof of Lemma \ref{lem:mass-zero-limit-existence-classical_weak} again replacing $G_m$, respectively $E_m$, with $G_m^\pm$, respectively $E_m^\pm$. Using $E_m^\pm (\delta d + d \delta )F = F$ we find
\begin{align}
	0
	&=  E_0^\pm d \delta F  \notag\\
	&= F - E_0^\pm \delta d F
\end{align}
and hence $F = E_0^\pm \delta d F$ which is compactly supported. It clearly follows that $F$ is co-closed using that $\delta$ commutes with $E_m^\pm$:
\begin{align}
	\delta F = \delta E_0^\pm \delta d F = E_0 \delta \delta d F = 0 \formspace.
\end{align}
b) Assuming $F \in \Omega^1_0(\M)$ being co-closed, $\delta F = 0$, we easily conclude that $G^\pm_m F = E^\pm_m F$ and hence the limits $\lim\limits_{m \to 0^+} G_m^\pm F$ exist. \par
We have therefore shown that the limits $\lim\limits_{m \to 0^+} G_m^\pm F$ exist if and only if $F$ is co-closed. Combining this with the existence of the remaining limits $\lim\limits_{m \to 0^+} \rhoz G_m F$ and  $\lim\limits_{m \to 0^+} \rhod G_m F$, we find the desired result, as it is necessary and sufficient for the limits $\lim\limits_{m \to 0^+} G_m^\pm F$, $\lim\limits_{m \to 0^+} \rhoz G_m F$ and  $\lim\limits_{m \to 0^+} \rhod G_m F$ to exist that $F$ is co-closed. This completes the proof.
\end{proof}
Therefore, also in the general case with currents we find existence of the zero mass limit of the Proca field if and only if we implement the gauge equivalence, as a restriction on the dual space of test one-forms, before taking the limit. This was already discussed in the previous Section \ref{sec:zero-mass-limit-existence-classical-vanishing-source}. We can now discuss the dynamics of the fields in the zero mass limit.
%
%
%
%
%
%
%
%
%
%
%
%
%
\subsubsection{Dynamics and the zero mass limit}\label{sec:limit_dynamics_classical}
We have found that both in the source free and the general case, the zero mass limit of the classical Proca theory exists if we restrict the test one-forms that we smear the classical fields with to the ones that are co-closed. So the question regarding the dynamics of the theory in the limit, that is, the behavior of $\langle A_0 , \delta d F \rangle_\M= \lim\limits_{m \to 0^+} \langle A_m, \delta d F \rangle_\M$, is well posed for any $F \in \Omega^1_0(\M)$ since naturally $\delta d F$ is co-closed using $\delta$ being nilpotent. For the Maxwell theory, we expect $\langle A_0 , \delta d F \rangle_\M= \langle j, F \rangle_\M$ as the field $A_0$ should solve Maxwell's equation in that distributional sense. But defining the field as a zero mass limit of the Proca theory, we find
\begin{align}
\langle A_0 , \delta d F \rangle_\M
&\coloneqq  \lim\limits_{m \to 0^+} \langle A_m, \delta d F \rangle_\M \\
&=  \lim\limits_{m \to 0^+}\Big( \sum\limits_\pm \langle j , G_m^\mp \delta d F   \rangle_{\Sigma^\pm} +\langle \Az , \rhod G_m \delta d F \rangle_\Sigma
- \langle \Ad , \rhoz G_m \delta d F \rangle_\Sigma \Big) \formspace.\notag
\end{align}
Recalling $G^\pm_m = \frac{1}{m^2}(d \delta + m^2) E^\pm_m$, we find $G^\pm_m \delta d F = E^\pm_m \delta d F$, and using $E_m^\pm (\delta d F) = F - E_m^\pm (d\delta + m^2) F$ we obtain
\begin{align}
	\langle &A_0 , \delta d F \rangle_\M  \\
&=  \lim\limits_{m \to 0^+}\Big( \sum\limits_\pm \big(  \langle j , F\rangle_{\Sigma^\pm} - \langle j ,  E_m^\mp d\delta F   \rangle_{\Sigma^\pm} \big)
		 -\langle \Az , \rhod E_m d\delta  F \rangle_\Sigma
		+ \langle \Ad , \rhoz E_m d\delta  F \rangle_\Sigma  \notag \\
		 &\phantom{=I} - m^2 \big(
		 \sum\limits_\pm \langle j , E_m^\mp F\rangle_{\Sigma^\pm}
				 + \langle \Az , \rhod E_m  F \rangle_\Sigma
				  -\langle \Ad , \rhoz E_m  F \rangle_\Sigma
		 \big)
		\Big)	\notag \\
&=	\langle j , F \rangle_\M - \lim\limits_{m \to 0^+}\Big(\sum\limits_\pm \langle j ,  E_m^\mp d\delta F   \rangle_{\Sigma^\pm}
+\langle \Az , \rhod E_m d\delta  F \rangle_\Sigma
- \langle \Ad , \rhoz E_m d\delta  F \rangle_\Sigma \Big) \formspace.		\notag
\end{align}
We have used $\sum\limits_\pm \langle j , F\rangle_{\Sigma^\pm} = \langle j , F\rangle_\M$ and that the terms proportional to $m^2$ are continuous by Assumption \ref{ass:propagator_continuity} and bounded and hence vanish in the limit. Furthermore, we find by definition that $\rhod E_m d\delta  F = - *_{(\Sigma)}i^* * d E_m d \delta F = 0$ since $d$ and $E_m$ commute. Concluding, we have calculated
\begin{align} \label{eqn:dynamics_limit_classical_unconstraint}
\langle A_0 , &\delta d F \rangle_\M =  \langle j , F \rangle_\M - \lim\limits_{m \to 0^+}\Big(\sum\limits_\pm \langle j ,  E_m^\mp d\delta F   \rangle_{\Sigma^\pm}
- \langle \Ad , \rhoz E_m d\delta  F \rangle_\Sigma \Big) \formspace.
\end{align}
The second term though will not vanish in general. Ergo, the fields $A_0$ defined as the zero mass limit of the Proca field $A_m$ will not fulfill Maxwell's equation in a distributional sense. While this might seem surprising at first, it is quite easy to understand when we recall how we have found solutions to Proca's equation: instead of finding solutions directly, we have specified solutions to the massive wave equation (\ref{eqn:classical_wave_eqation}) and then restricted the initial data such that the Lorenz constraint (\ref{eqn:classical_constraint}) is fulfilled. Only then we also have found a solution to Proca's equation. And similarly, one solves Maxwell's equation by specifying a solution to the massless wave equation $(\delta d + d \delta )A_0 = j$ and restricting the initial data such that the Lorenz constraint $\delta A_0 = 0$ is fulfilled. Only then, the solution also solves Maxwell's equation. And it is with the constraint where the problem in the limit lies. Recall from Theorem \ref{thm:solution_proca_constrained} that, in order to implement the Lorenz constraint, we have restricted the initial data by
\begin{align}
	\Adelta &= \frac{1}{m^2}\rhodelta j \; , \quad \text{and} \\
	\An &= \frac{1}{m^2}\left( \rhon j  + \delta_{(\Sigma)} \Ad \right) \formspace.
\end{align}
It is obvious that, in general, this is not well defined in the zero mass limit. So in order to keep the dynamics in the zero mass limit, we need to make sure that the constraints are well behaved in the limit. Since we do not want the external source or the initial data to be dependent of the mass, we have to specify
\begin{align}
	\delta j &= 0 \quad \; , \quad \text{and} \\
	\delta_{(\Sigma)} \Ad &= -\rhon j \quad \implies \An =0 \formspace.
\end{align}
This corresponds exactly to the constraints on the initial data in the Maxwell case to implement the Lorenz constraint (see \cite[Theorem 2.11]{pfenning})! With these constraints, we can now look at the remaining term of $\langle A_0 , \delta d F\rangle_\M$ in equation (\ref{eqn:dynamics_limit_classical_unconstraint}). We do this separately for the two summands. Using that $d$ commutes with pullbacks and inserting the constraints on the initial data, we find
\begin{align}
\langle \Ad , \rhoz E_m d\delta  F \rangle_\Sigma
	&= \langle \Ad , d_{(\Sigma)} \rhoz E_m \delta  F \rangle_\Sigma \notag\\
    &=	\langle \delta_{(\Sigma)}\Ad ,  \rhoz E_m \delta  F \rangle_\Sigma \notag\\
    &= -\langle \rhon j ,  \rhoz E_m \delta  F \rangle_\Sigma \notag\\
    &= -\int\limits_{\Sigma} i^* E_m \delta F \wedge *_{(\Sigma)} (- *_{(\Sigma)}  i^* *)j \notag\\
    &=  \int\limits_{\Sigma} i^* E_m \delta F \wedge i^* * j \notag\\
    &=  \int\limits_{\Sigma} i^*\left(  E_m \delta F \wedge * j \right) \formspace.
\end{align}
For the first summand $\sum_\pm \langle j ,  E_m^\mp d\delta F   \rangle_{\Sigma^\pm}$ we use the partial integration that we have already calculated in the proof of Theorem \ref{thm:solution_proca_unconstrained} and find, using the constraint $\delta j = 0$ found above,
\begin{align}
\sum_\pm \langle j ,  E_m^\mp d\delta F   \rangle_{\Sigma^\pm}
&= \sum_\pm \langle d \delta j ,  E_m^\mp  F   \rangle_{\Sigma^\pm} + \int\limits_{\Sigma} i^*\left(  j \wedge * E_m \delta F \right) -  \int\limits_\Sigma i^*(\delta j \wedge *EF ) \notag\\
&= \int\limits_{\Sigma} i^*\left(  j \wedge * E_m \delta F \right)   \formspace.
\end{align}
Using $j \wedge * E_m \delta F = E_m \delta F \wedge * j$ we find that the remaining terms of equation (\ref{eqn:dynamics_limit_classical_unconstraint}) vanish when restricting the initial data such that they are well defined in the zero mass limit. We therefore obtain the correct dynamics
\begin{align}
\langle A_0 , \delta d F \rangle_\M
&=  \langle j , F \rangle_\M - \lim\limits_{m \to 0^+}\Big(\sum\limits_\pm \langle j ,  E_m^\mp d\delta F   \rangle_{\Sigma^\pm}
		- \langle \Ad , \rhoz E_m d\delta  F \rangle_\Sigma \Big) \notag\\
&= \langle j , F \rangle_\M \formspace.
\end{align}
Concluding, when keeping the constraints that implement the Lorenz constraint in the limit well behaved, we indeed end up with the correct dynamics of the Maxwell theory. Furthermore, we also obtain conservation of the external current $\delta j = 0$ as a necessity to get the correct dynamics. This is not surprising as, opposed to Proca's theory, the current in Maxwell's theory is always conserved by the equations of motion! We conclude this in the final theorem of this chapter:
\begin{theorem}[The zero mass limit of the Proca field]
	Let $F\in \Omega^1_0(\M)$ be a test one-form and $j \in \Omega^1(\M)$ an external current. \\
	Let $A_m$ be a solution to Proca's equation specified by initial data $\Az, \Ad \in \Omega^1_0(\Sigma)$ via Theorem \ref{thm:solution_proca_unconstrained}.	\\
	Defining the zero mass limit $\langle A_0 , F \rangle_\M = \lim\limits_{m \to 0^+} \langle A_m, F \rangle_\M$ of the Proca field, the following holds:
	\begin{enumerate}
		\item The limit exists if and only if $\delta F = 0$, effectively implementing the gauge equivalence of the Maxwell theory.
		\item The field $A_0$ is a Maxwell field, that is, it solves Maxwell's equation in a distributional sense if and only if $\delta j = 0$, implementing the conservation of current, and $\rhon j = - \delta_{(\Sigma)} \Ad$ , implementing the Lorenz condition.
	\end{enumerate}
\end{theorem}
\begin{proof}
	The proof follows directly from Theorem \ref{thm:limit_existence_general_classical} and the calculations presented in the above Section \ref{sec:limit_dynamics_classical}
\end{proof}

\section{The Quantum Problem}\label{chpt:quantum}
Having established a good understanding of the classical theory, we will now investigate the quantum Proca field in curved spacetimes. In particular, we are going to construct a generally covariant quantum field theory of the Proca field in the framework of Brunetti Fredenhagen Verch \cite{Brunetti_Fredenhagen_Verch} in Section \ref{sec:generally_covariant_QFTCS} and show that the theory is local. In Section \ref{sec:BU-algebra} we study the Borchers-Uhlmann algebra as the field algebra and rigorously construct an initial data formulation of the quantum Proca field theory. This will allow us to define a notion of continuity of the Proca field with respect to the mass and, finally, to study the mass dependence and the zero mass limit of the theory in Section \ref{sec:mass_depenence_and_limit}.
\subsection{Construction of the generally covariant quantum Proca field theory in curved spacetimes}\label{sec:generally_covariant_QFTCS}
The quantization of the Proca field in a generally covariant way will follow the framework of Brunetti, Fredenhagen and Verch \cite{Brunetti_Fredenhagen_Verch} as well as some natural modifications needed for background source dependent theories which are made analogous to \cite{Sanders}. In the framework of \cite{Brunetti_Fredenhagen_Verch}, a \emph{generally covariant quantum field theory} is mathematically described as a functor between the category $\Spac$, consisting of globally hyperbolic spacetimes as objects and orientation preserving isometric hyperbolic embeddings as morphisms, and the category $\Alg$, consisting of unital $^*$-algebras as objects and unit preserving $^*$-algebra-homomorphisms as morphisms. If these $^*$-algebra-homomorphisms are injective, the theory is said to be \emph{local} (rigorous definitions are given below). To accommodate the given background source $j$ in the theory, we will generalize $\Spac$ to a category whose objects also contain the given background source. \par
In this section, we want to give the necessary definitions and construct this functor explicitly, that is, we give a detailed definition on how to map globally hyperbolic spacetimes with a given background current to an algebra of observables, and how to map the morphisms onto each other. The main work is then to show that these maps are well defined. It is then rather trivial to show that we have obtained a functor. Throughout this section, the mass $m$ as well as the external current $j$ are assumed to be fixed.
\newpage
We begin by defining the necessary objects and morphisms of the two categories.
\begin{definition}[Orientation preserving isometric hyperbolic embedding]
Let $(\M,g)$ and $(\N,g_\N)$ be two globally hyperbolic spacetimes. \\
A  map $\psi : (\M,g) \to (\N,g_N)$ is called a \emph{orientation preserving isometric hyperbolic embedding} if
\begin{enumerate}
\item $\psi$ is a diffeomorphism, that is it is smoothly bijective,
\item $\psi$ preserves orientation and time orientation,
	\item $\psi$ is an isometry, that is $\psi^* g_\N = g$, and
	\item $\psi(\M)$ is causally convex, that is for $p \in \M$ it holds
\begin{align*}
J_\M^\pm(p) = \psi^{-1} \Big( J_\N^\pm \big( \psi(p)\big) \Big) \formspace .
\end{align*}
\end{enumerate}
\end{definition}
\begin{definition}[The categories $\SpacCurr$, $\Alg$ and $\Alg'$]\label{def:categories_alg_spaccurr}
The category \gls{spaccurr} consists of triples $(\M,g,j_\M)$ as objects, where $(\M,g)$ is a globally hyperbolic spacetime  and $j_\M \in \Omega^1(\M)$ corresponds to the background current of the theory,
and morphisms $\psi$, where $\psi : (\M,g) \to (\N,g_\N)$ is an orientation preserving isometric hyperbolic embedding such that $\psi^* j_\N = j_\M$. \par
The category \gls{alg} consists of unital $^*$-algebras as objects and unit preserving  $^*$-algebra-homomorphisms as morphisms. \par
The category \gls{algprime} is a subcategory of $\Alg$ consisting of the same objects but only injective morphisms.
\end{definition}
With the notion of these two categories we are able to define:
\begin{definition}[Generally (locally) covariant quantum field theory with background source]\label{def:generally-coveriant-qftcs}
A \emph{generally covariant quantum field theory with background source} is a covariant functor between the categories $\SpacCurr$ and $\Alg$. \\
The theory is called \emph{local} or \emph{locally covariant} if and only if the range of the functor is contained in $\Alg'$.
\end{definition}
To construct this functor for the Proca field, we will first define how to map a globally hyperbolic spacetime to a unital $^*$-algebra and how to map the morphisms onto each other. Most of the work is to show that those maps are well defined and injective. Then, it is not hard to show that we have obtained a functor and thus the desired generally locally covariant quantum field theory.
\begin{definition}\label{def:algebra-A(M)}
Let $M=(\M,g,j_\M) \in \mathsf{Obj}_\SpacCurr$ be an object of $\SpacCurr$, $F,F' \in \Omega^1_0(\M)$ be test one-forms and $c_1, c_2 \in \IC$ be constants. \\
Let \gls{Gmcurly} be the propagator of the Proca operator with integral kernel $G_m$, that is, $\Green{F}{F'} = \langle F, G_mF' \rangle_\M$.\\
Define  $\AA : \mathsf{Obj}_\SpacCurr \to \mathsf{Obj}_\Alg$,  where $\AA(M)$ is the unital $^*$-algebra obtained from the free algebra, generated by $\mathbbm{1}$ and the objects $\A(F)$, factoring by the relations
 \begin{subequations}  \label{def:ideal_generators}
  \begin{align}
\text{(i)}\; &\A(c_1 F + c_2 F') = c_1 \A(F) + c_2 \A(F') 														&\textrm{linearity,} \\
\text{(ii)}\; &\A(F)^* = \A(\skoverline{F}\,) 																															&\textrm{real field,} \\
\text{(iii)}\; &\A\big( (\delta d + m^2) F \big) = \langle j_\M , F \rangle_\M \cdot \mathbbm{1} 	&\textrm{equation of motion,} \\
\text{(iv)}\; &[\A(F) , \A(F') ] = \i \Green{F}{F'} \cdot \mathbbm{1}															&\textrm{commutation relations}.
 \end{align}
 \end{subequations}
\end{definition}
To be mathematically more precise, the algebra is obtained as the quotient algebra from the free algebra $\PPM$ dividing out the (two-sided) ideal $\mathcal{J}_\M$ that is generated by the relations  (\ref{def:ideal_generators}). As an example, a sub-ideal of $\mathcal{J}_\M$ implementing (\ref{def:ideal_generators}b) is defined as $\widetilde{\mathcal{J}}_\M = \big\{ a\big(\A(F)^* - \A(\skoverline{F}\,) \big)b \;\;\vert\;\; a,b \in \PPM, F \in \Omega^1_0(\M) \big\}$. One obtains an algebra of equivalence classes $\AA(M) = {\Quotientscale{\PPM}{ \mathcal{J}_\M}}$.
For this to be well defined, it suffices to show that the obtained algebra $\AA(M)$ is not trivial, that is, not the zero algebra. Therefore, we need to show that the ideal $\mathcal{J}_\M$ is not the full free algebra $\PPM$.
Clearly, for a suitable test one-form $F$, that is, in particular a one-form that is not of the form $F = (\delta d + m^2)H$ for some test one-form $H$,  $\A(F)$ will not be an element of $\mathcal{J}_\M$, and therefore $\AA(M)$ is not trivial.\par
Next, we define the action of the map $\AA$ on morphisms of $\SpacCurr$.
\begin{definition}\label{def:morphism_alpha_psi}
Let $M,N\in \mathsf{Obj}_\SpacCurr$, where $M=(\M,g,j_\M)$ and $N=(\N,g_\N,j_\N)$, be objects and $\psi \in \textsf{Mor}_\SpacCurr(M,N), \psi: (\M,g,j_\M) \to (\N,g_\N,j_\N)$ be a morphism of the category $\SpacCurr$.
Define $\AA(\psi) \equiv \alpha_\psi : \AA(M) \to \AA(N)$ as a unit preserving $^*$-algebra-homomorphism whose action on elements of $\AA(M)$ is then fully determined by the action on the generators $\A_\M(F)$ :
\begin{align}
\alpha_\psi \big(\A_\M(F)\big) = \A_\N(\psi_*(F)) \formspace.
\end{align}
\end{definition}
We need to show that this is well defined\footnote{Note that in Definition \ref{def:pullback} we have only defined the \emph{pullback} $\psi^* F$ of a one-form $F$. Since here $\psi$ is assumed to be a diffeomorphism, the pushforward of one-forms on $\M$ to one-forms on $\N$ can be defined as the pullback with respect to $\psi^{-1}$.}, in particular that it is compatible with the algebra relations in $\AA(\N)$. To be more precise, the proceeding is as follows:\\
Let $\PPM,\PPN$ be the free unital *-algebras as defined above. We define a morphism $\pi : \PPM \to \PPN$ as a unit preserving $^*$-algebra homomorphism such that $\A_\M(F) \mapsto \A_\N(\psi_* F)$. We need to show that $\pi(\mathcal{J}_\M) \subset \mathcal{J}_\N$, so that if we divide out the ideal $\mathcal{J}_\M$, $\pi$ descends to the wanted unit preserving $^*$-algebra homomorphism $\AA(\psi) \equiv \alpha_\psi : \AA(M) \to \AA(N)$. We do this step by step, showing that each generator of $\mathcal{J}_\M$ maps to a corresponding generator of $\mathcal{J_\N}$. In the following let  $F,F' \in \Omega^1_0(\M)$ and $c_1,c_2 \in \IC$ be constants.\par
\textit{1.) Linearity:}\\
The generator of the corresponding ideal\footnote{Note, that actually we are interested in $^*$-ideals, so we would need to add (or subtract) the hermitian adjoint to that expression, but since $\pi$ is defined as a $^*$-algebra homomorphism this would not change anything in the calculations and is therefore neglected for clarity.} is $\big( \A_\M (c_1\,F+ c_2\,F') -   c_1\, \A_\M (F) - c_2\, \A_\M (F')  \big)$. We calculate:
\begin{align}
\pi \big( \A_\M (c_1\,F+ c_2\,F') &-   c_1\, \A_\M (F) - c_2\, \A_\M (F')  \big) \notag\\
&=  \pi \big( \A_\M (c_1\,F+ c_2\,F') \big)  -  c_1\, \pi \big(  \A_\M (F) \big) - c_2\, \pi \big( \A_\M (F')  \big) \notag\\
&= \A_\N (c_1\, \psi_* F+ c_2\,\psi_*F') -  c_1\,  \A_\N (\psi_*F) - c_2\,  \A_\N (\psi_*F') \formspace.
\end{align}
We have used  the homomorphism property of $\pi$ and that $\psi_*$ is linear and naturally commutes with scalars.
The result clearly is an element of the corresponding ideal in $\PPN$ specifying $\tilde{F}, \tilde{F}' \in \Omega^1_0(\N) $ by $\tilde{F} = \psi_* F$ and $\tilde{F}' = \psi_* F'$. \par
\textit{2.) Real field:} \\
The corresponding generator is $\big( \A_\M( F)  - \A_\M( \skoverline{F}\,)^* \big)$. We obtain:
\begin{align}
\pi \big( \A_\M( \skoverline{F}\,)  - \A_\M( F)^* \big)
&= \pi \big( \A_\M( \skoverline{F}\,) \big)   - \pi \big( \A_\M( F)^* \big) \notag\\
&= \pi \big( \A_\M( \skoverline{F}\,) \big)   - \pi \big( \A_\M( F) \big)^* \notag\\
&=  \A_\N( \psi_* \skoverline{F}\,)    -  \A_\N(\psi_*  F)^* \formspace.
\end{align}
Again, the result clearly is an element of the corresponding ideal in $\PPN$.\newpage
\textit{3.) Equations of motion:}\\
The generator of interest is $\big( \A_\M\big( (\delta d + m^2) F\big )  -  \langle j_\M , F \rangle_\M \, \mathbbm{1}_{\PPM} \big)$.
First, note that
\begin{align}
(\delta d + m^2)\psi_* = \psi_* (\delta d + m^2)
\end{align}
since $\psi_*$ is linear and commutes with $d$. Also, because $\psi$ is an orientation preserving isometry and therefore preserves the volume form, $\psi_*$ commutes with the Hodge star  and thus it also commutes  with the interior derivative $\delta$.
It then follows that
\begin{align}
\pi  \big( \A_\M\big( (\delta d + m^2) F\big )  & -  \langle j_\M , F \rangle_\M \, \mathbbm{1}_{\PPM} \big)  \notag\\
&= \A_\N\big(\psi_*  (\delta d + m^2) F\big )  -  \langle j_\M , F \rangle_\M \, \mathbbm{1}_{\PPN} \notag\\
&= \A_\N\big(  (\delta d + m^2) \psi_* F\big )  -  \langle j_\N , \psi_* F \rangle_\N \, \mathbbm{1}_{\PPN} \formspace.
\end{align}
In the last step it was used that, since $\psi$ is an isometry and $\psi^* j_\N = j_\M$:
\begin{align}
\langle j_\M , F \rangle_\M
&= \langle \psi^* j_\N , F \rangle_\M \notag\\
&= \langle  j_\N , \psi_* F \rangle_\N \formspace,
\end{align}
which yields the wanted generator in $\PPN$.\par
\textit{4.) Commutation relation:}\\
The generator is $\big ( \big[ \A_\M(F) , \A_\M( F') \big] - \i \GreenM{F}{F'}\,\mathbbm{1}_{\PPM} \big)$. We calculate:
\begin{align}
\pi   \big ( \big[ \A_\M(F) , \A_\M( F') \big] & - \i \GreenM{F}{F'}\,\mathbbm{1}_{\PPM} \big) \notag\\
&= \big[ \A_\N(\psi_* F) , \A_\N(\psi_* F') \big] - \i \GreenM{F}{F'}\,\mathbbm{1}_{\PPN}  \notag\\
&= \big[ \A_\N(\psi_* F) , \A_\N(\psi_* F') \big] - \i \GreenN{\psi_* F}{\psi_* F'}\,\mathbbm{1}_{\PPN}
\end{align}
In the last step we have used the uniqueness of the fundamental solutions and the properties that $\psi$ is an isometry and that $\psi(\M)$ is causally convex.
Together, this implies $G_{m,\M} F = \psi^* G_{m,\N} \psi_* F $ (see \cite[Chapter 4.3]{Sanders}) and therefore
\begin{align}
\GreenM{F}{F'}
&= \langle F, G_{m,\M} F' \rangle_\M \notag\\
&= \langle F, \psi^* G_{m,\N} \psi_* F' \rangle_\M\notag \\
&= \langle \psi_* F ,  G_{m,\N} \psi_* F' \rangle_\N \notag\\
&= \GreenN{\psi_*F}{\psi_* F'} \formspace.
\end{align}
Altogether, we have shown that $\pi(\mathcal{J}_\M)  \subset \mathcal{J}_\N$ and therefore $\pi$ descends to the wanted unit preserving $^*$-algebra homomorphism $\alpha_\psi$, having divided out the ideal $\mathcal{J}_\M$.\par
\newpage
Now, to obtain a \emph{locally covariant} QFT from these definitions, we need the defined homomorphism to be injective. We do this the following way:
To show that $\AA(\psi) \equiv \alpha_\psi$ is injective, we can equivalently show that the algebra $\AA(M)$ is simple, that is, there is no non-trivial two-sided ideal\footnote{For a general algebra to be simple one also needs that the multiplication operation is not uniformly zero. Since we deal with unital algebras this is trivially fulfilled.} in $\AA(M)$. It turns out that for $\AA(M)$ to be simple, it is sufficient that $\Green{\cdot}{\cdot}$ is non-degenerate. The basic algebraic work necessary for these arguments is put in Appendix \ref{app:lemmata}.
\begin{lemma}\label{lem:propagator-non-degenerate}
Let $F,F' \in \Omega^1_0(\M)$ be two test one-forms and let $\gls{DSz} = \Omega^1_0(\Sigma) \oplus  \Omega^1_0(\Sigma)$ be the space of initial data on some Cauchy surface $\Sigma$ with respect to Proca's equation.
Then,
\begin{center}
$\Green{F}{F'}$, viewed as a map $\mathcal{G}_m : \Dzs \oplus \Dzs \to \IC$\\ on the space of initial data, is a symplectic form,
\end{center}
 that is it is bilinear, anti-symmetric and non-degenerate.
\end{lemma}
\begin{proof}
First, we want to see how to view $\mathcal{G}_m$ as a map on initial data. Let $F,F' \in \Omega^1_0(\M)$, then by definition
\begin{align}
\Green{F}{F'}
&= \langle F , G_m F' \rangle_\M \notag\\
&= \langle G_m F' , F \rangle_\M  \formspace.
\end{align}
Note that $G_m F'$ is a solution to the source free Proca equation, that is,
\begin{align}
(\delta d + m^2)G_m F' = 0 \formspace.
\end{align}
By Theorem \ref{thm:solution_proca_unconstrained}, setting $j=0$, we find for a test one-form $F$,
\begin{align}
 \langle G_m F' , F \rangle_\M =  \langle (G_m F')_{0} , \rhod GF \rangle_\Sigma - \langle (G_m F')_{d} , \rhoz G_m F \rangle_\Sigma  \formspace.
\end{align}
In the same way, $G_m F$ is a solution to Proca's equation and therefore, since $\rhoz$ and $\rhod$ map a solution to its initial data, we can view $\rhoz G_m F$ and  $\rhod G_m F$ as initial data of $G_m F$.
For short hand we will write for the initial data $\rhoz G_m F = (G_m F)_0 = \varphi$ and  $\rhod G_m F = (G_m F)_0 = \pi$ and analogously for $G_m F'$.
Therefore we obtain
\begin{align}
\Green{F}{F'} &= \langle \varphi' , \pi \rangle_\Sigma  - \langle \pi' , \varphi \rangle_\Sigma \notag \\
&= \langle \pi , \varphi' \rangle_\Sigma  - \langle \varphi , \pi' \rangle_\Sigma   \formspace.
\end{align}
So instead of viewing $\mathcal{G}_m$ as a map on one forms, we view it as map on the space of initial data
\begin{align}
\mathcal{G}_m : \Dzs \oplus \Dzs &\to \IC\formspace \\
\big( (\varphi , \pi) , (\varphi' , \pi')  \big) &\mapsto \langle \pi' , \varphi \rangle_\Sigma - \langle \varphi' , \pi \rangle_\Sigma \formspace. \notag
\end{align}
Now, it is straightforward to show that $\mathcal{G}_m$ is a symplectic form: \par
\textit{1.) Bilinearity: \\}
Bilinearity follows trivially from the bilinearity of $\langle \cdot , \cdot \rangle$. \par
\textit{2.) Alternating: \\}
Let $\psi = (\varphi , \pi) \in \Dzs$. Then
\begin{align}
\Gm{\psi}{\psi}
&=  \langle \pi , \varphi \rangle_\Sigma - \langle \varphi , \pi \rangle_\Sigma \notag\\
&= 0 \formspace.
\end{align}
Therefore, $\mathcal{G}_m$ is anti-symmetric: Specifying $\psi = \psi' + \tilde{\psi}$ and using $0=\Gm{\psi}{\psi}$ together with bilinearity  yields anti-symmetry.\par
\textit{3.) Non-degeneracy: \\}
Let $\psi'= (\varphi' , \pi') \in \Dzs$ specify initial data. Assume, for all $\psi = (\varphi, \pi) \in \Dzs$ it holds that
\begin{align}
0
&= \Gm{\psi}{\psi'} \notag\\
&= \langle \pi , \varphi' \rangle_\Sigma - \langle \varphi , \pi' \rangle_\Sigma \formspace \notag\\
\implies \varphi' &= 0 = \pi' \notag\\
\iff (\varphi' , \pi') &= 0 \in \Dzs\formspace.
\end{align}
Hence, $\Gm{\cdot}{\cdot}$ is non-degenerate.
\end{proof}
\newpage
We are now ready to show that the defined morphism $\alpha_\psi$ is injective:
\begin{lemma}\label{lem:alpha_is_injective}
Let $\AA(M)$ be the unital $^*$-algebra as defined in Definition \ref{def:algebra-A(M)} and $\alpha_\psi$ the unit preserving $^*$-algebra homomorphism as defined in Definition \ref{def:morphism_alpha_psi}. It then holds that
\begin{center}
$\alpha_\psi$ is injective.
\end{center}
\end{lemma}
\begin{proof}
Since by Lemma \ref{lem:propagator-non-degenerate} the propagator $\Gm{\cdot}{\cdot}$ is non-degenerate when viewed as a map on initial data\footnote{We will show in Section \ref{sec:BU-algebra} that the dynamical field algebra is homeomorphic to the field algebra of initial data. We can therefore safely view the propagator as a map on initial data.}, the algebra $\AA(M)$ is simple (c.f. \cite[Scholium 7.1]{baez_segal_zhou}). Using Lemma \ref{lem:injective_mor_simple_algebra} this implies that all unit preserving $^*$-algebra homomorphisms are injective.
\end{proof}
With this, we are ready to state the first major result of this chapter, that is, we have constructed a generally locally covariant quantum theory of the Proca field in curved spacetimes including external sources,  in form of the following theorem:
\begin{theorem}
Let $\SpacCurr$ and $\Alg'$ be the categories as defined in Definition \ref{def:categories_alg_spaccurr}.\\
Let $\AA : \SpacCurr \to \Alg'$ be as defined in Definition \ref{def:algebra-A(M)} (action on objects) and Definition \ref{def:morphism_alpha_psi} (action on morphisms). Then
\begin{center}
$\AA$ is a covariant functor,
\end{center}
that is it describes the generally covariant quantum theory of the Proca field in (globally hyperbolic) spacetimes.
Furthermore, the image of the functor is contained in $\Alg'$. Therefore, the theory is local.
\end{theorem}
\begin{proof}
We have already proven most of the statement by showing that $\AA$ is well defined and that $\alpha_\psi \equiv \AA(\psi)$ is injective (see Lemma \ref{lem:alpha_is_injective}).\\
For $\AA$ to be a (covariant) functor, we need to show that it behaves well under composition of morphisms, that is for any objects $M,N,K \in \mathsf{Obj}_\SpacCurr$, where $M=(\M,g,j_\M)$, $N=(\N,g_\N,j_\N)$ and $K=(\K,g_\K,j_\K)$, and morphisms $\psi \in \textsf{Mor}_\SpacCurr(M,N)$ and $\phi \in \textsf{Mor}_\SpacCurr(N,K)$, it holds that
\begin{align}
\alpha_{\phi \,\comp\, \psi} = \alpha_\phi \comp \alpha_\psi \formspace,
\end{align}
and that it maps the identity $\text{id}_M \in \mathsf{Hom}_\SpacCurr(M,M)$ to the identity $\text{id}_{\AA(M)} \in  \mathsf{Hom}_\Alg\big(\mathscr{A}(M),\mathscr{A}(M)\big)$, which it trivially does by definition.
The behavior under composition follows directly from the definition of $\alpha_\psi$. Let $\psi : (\M,g,jJ_\M) \to (\N, g_\N ,j_\N)$ and $\varphi : (\N, g_\N ,j_\N) \to (\mathcal{K}, g_\K, j_\mathcal{K})$ two orientation preserving hyperbolic isometric embeddings. Then
\begin{align}
\alpha_{\phi\, \comp \, \psi} \big( \A_\M(F) \big)
&= \A_\mathcal{K} \big( (\phi \comp \psi )_* F \big) \notag \\
&= \A_\mathcal{K} \big(\phi_*(\psi_* F) \big) \notag\\
&= \alpha_\phi \big(\A_\N(\psi_* F) \big) \notag\\
&= \alpha_\phi \Big( \alpha_\psi \big(\A_\M (F) \big) \Big) \notag\\
&= (\alpha_\phi \comp \alpha_\psi)\big( \A_\M (F) \big)  \formspace.
\end{align}
Therefore, $\AA$ is a covariant functor.
\end{proof}
\subsection{The Borchers-Uhlmann algebra as the field algebra}\label{sec:BU-algebra}
One algebra that can be used to describe quantum fields in curved spacetimes is the \emph{Borchers-Uhlmann algebra} (BU-algebra), which was studied by \name{Borchers} \cite{borchers} and
\name{Uhlmann} \cite{uhlmann} in 1962. First applied to the case of quantum field theory on Minkowski spacetime, in particular in connection to the Wightman $n$-point functions, the BU-algebra is well suited to generalize to the curved spacetime case. The BU-algebra can be constructed over any vector space which in our case will be the space $\Omega^1_0(\M)$ of compactly supported test one-forms. The construction presented here follows \cite[Chapter 4.1]{verch_sahlman} and \cite[Chapter 8.3.2]{haag}.\par
The BU-algebra $\BUOmega$ is defined as the tensor algebra of the vector space $\Omega^1_0(\M)$. That means elements $f \in \BUOmega$ are tuples $f= (f^{(0)}, f^{(1)}, f^{(2)}, \dots)$, where $f^{(0)} \in \IC$ and for $i>0, f^{(i)} \in \big(\Omega^1_0(\M)\big)^{\otimes i}$, such that only finitely many $f^{(i)}$'s are non-vanishing\footnote{As noted in \cite[Chapter 3.3]{verch_sahlman}, we can alternatively view the $f^{(i)}$'s as smooth sections of the $i$-fold outer product bundle $T^*\M \boxtimes T^*\M \boxtimes \cdots \boxtimes T^*\M$ over $\M ^i$.}. Here, $\otimes$ denotes the algebraic tensor product, without taking any topological completion. We will call the component $f^{(i)}$ the \emph{degree-i-part} of $f$.\newpage
Formally, we summarize in the following definition:
\begin{definition}[Borchers-Uhlmann algebra]
Let $V$ be a vector space. The Borchers-Uhlmann algebra over $V$ is defined as
\begin{align}
	\gls{bualgebra} = \IC \oplus \bigoplus\limits_{n= 1}^\infty V^{\otimes n} \formspace.
\end{align}
\end{definition}
\noindent
Again, $V \otimes V \otimes \cdots \otimes V$ denotes the algebraic tensor product without any topological completion.
This definition makes $\BUOmega$ a vector space (addition and scalar multiplication is defined component wise), which can be endowed with a $*$-algebraic structure: \\
For two elements $f, g$ we define the product $(f\cdot g)$ by defining the degree-$n$-part to be
\begin{align}
	(f\cdot g)^{(n)} = \sum_{i+j = n} f^{(i)} \otimes g^{(j)} \formspace,
\end{align}
which is equivalent to specifying
\begin{align}
(f\cdot g)^{(n)} (p_1,p_2, \dots , p_n) = \sum_{i+j = n} f^{(i)} (p_1, p_2, \dots , p_i) g^{(j)} (p_{i+1} , \dots , p_n) \formspace,
\end{align}
where $p_i \in \M$. The involution is defined by complex conjugation and reversing the order of the arguments in every degree:
\begin{align}
	(f^*)^{(n)} (p_1, \dots , p_n) = \bar{f}^{(n)}(p_n, p_{n-1}, \dots , p_1) \formspace.
\end{align}
With these two additional operations, $\BUOmega$ is a *-algebra. Furthermore, defining a unit element $\mathbbm{1}_{\BU(\Omega^1_0(\M))} = (1, 0, 0, \dots)$, $\BUOmega$ becomes a unital *-algebra. The reason why the BU-algebra is well suited for our problem is that it can be endowed with a topology which will later allow us to define a notion of continuously varying the mass $m$ and compare the corresponding quantum fields with each other.
\subsubsection{The topology of the Borchers-Uhlmann algebra}
In the following we present the basic ideas to defining a (locally convex) topology on the BU-algebra over the space of compactly supported one-forms without going into too much of the necessary details since this procedure is well understood. At every step we make explicit references to the missing details or proofs.
As a first step, we need to find a topology on the space $\Omega^1_0(\M)$ of compactly supported one-forms which will be used for the construction of the topology on the BU-algebra.  The construction follows \cite[Chapter 17.1 to 17.3]{dieudonne_3} where all presented statements are made rigorous.
First we define a locally convex topology on the space of smooth functions $C^\infty(U)$ for an open $U \subset \IR^n$. A locally convex topology is induced by a family of semi-norms\footnote{See \cite[Theorem 12.14.3]{dieudonne_2} or for a general introduction to locally convex topological vector spaces \cite[Chapter 7, in particular page 63 ]{treves} and \cite[Chapter 12.14]{dieudonne_2}.}. Let $\left\{ K_l \right\}_l$ be an increasing family of compact sets of $U$, such that $\left\{ K_l \right\}_l$ is a covering of $U$. We will call such a family a \emph{fundamental sequence}. We define the semi-norm for a class $C^\infty$ function $f: U \to \IC$
\begin{align}
	p_{s,l}(f) = \sup\Big\{ \abs{D^\nu f(x)} : x \in K_l, \abs{\nu} \leq s\Big\} \formspace,
\end{align}
where $\nu \in \IN^n$ is a multi-index, $\abs{\nu} = \sum_{i=1}^n \nu_i$ and the derivative operator is defined as
\begin{align}
D^\nu f (x_1, \dots , x_n) = \frac{\partial^{\nu_1} \, \partial^{\nu_2} \dots\, \partial^{\nu_n}  }{\partial_{x_1}^{\nu_1} \,\partial_{x_2}^{\nu_2} \dots \, \partial_{x_n}^{\nu_n}} \formspace.
\end{align}
Actually, the above definition also endows the space of $n$-times differentiable functions $C^n(U)$ with a locally convex topology, which we will use for the generalization in a few steps.
This procedure generalizes to compactly supported differential forms on a smooth mani\-fold. Even more general, let $\mathfrak{X}=(E, \M,\pi)$ be a vector bundle over $\M$ of rank $n$. We will define a topology on the space $\Gamma(E,U)$ of smooth sections of $E$ over $U$. The basic idea is to use local charts of the manifold and a local trivialization of the vector bundle to define a family of semi-norms on $\Gamma(E,U)$ using the semi-norms we have specified above for functions $f\in C^\infty(U)$. Let $\left\{U_\alpha\right\}_\alpha$ be a locally finite covering of $U \subset \M$ such that there are local charts $(\psi_\alpha, U_\alpha)$ of the manifold $\M$. For each $\alpha$, the map $z \mapsto \Big( \psi_\alpha\big(\pi(z)\big) , v_{1\alpha} , v_{2 \alpha} , \dots , v_{n \alpha} \Big)$ from the fibers $\pi^{-1} (U_\alpha)$ to $\psi_\alpha (U_\alpha) \times \IC^n$ is a diffeomorphism. Such linear diffeomorphisms $v_{i\alpha}$ exist for any open neighborhood of $\M$ as they can be defined as the components of a local trivialization of the bundle $E$ over the neighborhood $U_\alpha$.  Let $u_\alpha$ be the restriction of a section $u \in \Gamma(E,U)$ to $U_\alpha$. Finally, define for every $\alpha$ the fundamental sequence  $\left\{ K_{l\alpha} \right\}_l$ of compact sets in $\psi_\alpha(U_\alpha)$ and denote the semi-norms on $C^n(\psi_\alpha(U_\alpha))$ as specified above by $\tilde{p}_{s, l, \alpha}$. Then
\begin{align}
	p_{s,l,\alpha}(u) = \sum_{j=1}^{n} \tilde{p}_{s, l, \alpha} \left( v_{j\alpha} \comp u_\alpha \comp \psi^{-1}_\alpha \right)
\end{align}
defines a family of semi-norms on $\Gamma(E,U)$ (see \cite[Equation 17.2.1]{dieudonne_2}). Now we specify to the vector bundles $\wedge^p T^* \M$ whose smooth sections are smooth $p$-forms. The space $\Gamma(\wedge^p T^* \M , \M)$, together with the family of semi-norms specified above, is a locally convex topological vector space. For every compact $K \subset \M$, the space $\Gamma(\wedge^p T^* \M,K)$ denotes the set of $p$-forms that are compactly supported in $K$. This space is a closed subspace of $\Gamma(\wedge^p T^* \M,\M)$. The union of $\Gamma(\wedge^p T^* \M,K)$ over all compact subspaces $K \subset \M$ is a locally convex topological vector space, the space of compactly supported $p$-forms\footnote{To be precise, for compact $K \subset \M$, the spaces $\Gamma(\wedge^p T^* \M, K)$ are Fr\'echet spaces \cite[Theorem 17.2.2]{dieudonne_3} and the space of compactly supported $p$-forms is defined as the inductive limit of the spaces $\Gamma(\wedge^p T^* \M,K)$, making it a LF-space which in particular is locally convex.} which we have already denoted by $\Omega^p_0 (\M)$. This in particular yields a locally convex topology on the space $\Omega^1_0(\M)$ that we are interested in.\par
Having found a topology on $\Omega^1_0(\M)$ it is a straightforward procedure to endow $\big(\Omega^1_0(\M)\big)^{\otimes n}$ with a locally convex topology for every $n$: We equip $\big(\Omega^1_0(\M)\big)^{\otimes n}$ with the \emph{projective topology} which is induced by formal tensor products of the semi-norms on $\Omega^1_0(\M)$ (see \cite[Definition 43.2 and Chapter 43]{treves} for details).  For every $n \in \IN$ we define
\begin{align}
	\BU_n = \IC \oplus \bigoplus_{i=1}^n \big(\Omega^1_0(\M)\big)^{\otimes i} \formspace,
\end{align}
which yields a family $\left\{\BU_n\right\}_n$ of locally convex topological vector spaces, where the topology on each $\BU_n$ is given as the direct sum topology\footnote{For a definition, see for example \cite[515]{treves}.}.
The BU-algebra is then endowed with the so called inductive limit topology\footnote{Again, for a definition see \cite[514]{treves}.} of the family $\left\{\BU_n\right\}_n$ (see \cite[Appendix B]{verch_sahlman}). With this construction, the BU-algebra is a locally convex topological *-algebra with unit (c.f. \cite[Lemma 4.1]{verch_sahlman}).
\subsubsection{Dynamics, commutation relations and the field algebra}
So far, the constructed BU-algebra, which we would like to use as a field algebra, does neither incorporate any dynamics, in our case are given by the Proca equation, nor the desired quantum commutation relations. We want to identify quantum fields $\phi$ as elements $\phi(F) = (0, F, 0, 0, \dots)$ of an appropriate field algebra. To endow the algebra with dynamics,  the fields $\phi$ have to solve the Proca equation in a distributional sense. Furthermore, we incorporate the canonical commutation relations (CCR) in the field algebra. We will do this, for reasons that will become clear in the next section, in a two step procedure. Throughout this section, the mass dependence is again made explicit in the notation, but the mass $m$, as well as the external source $j$, are assumed to be fixed.\par
First we will divide out the two-sided ideal $\IMJDYN$ in $\BUOmega$ that is generated by elements
\begin{align}
\big(-\langle j, F \rangle_\M, (\delta d + m^2)F,0,0,\dots\big) \in \BUOmega\formspace,
\end{align}
for $F \in \Omega^1_0(\M)$, to implement the dynamics. That means, by definition, an element $f \in \IMJDYN$ can be written as a finite sum
\begin{align}
f = \sum_i g_i \cdot \left(-\langle j, F_i \rangle_\M, (\delta d + m^2)F_i,0,0,\dots\right) \cdot h_i \formspace,
\end{align}
for some $F_i \in \Omega^1_0(\M)$ and $g_i, h_i \in \BUOmega$. We define
\begin{align}
	\BUmjdyn \coloneqq {\Quotientscale{\BUOmega}{\IMJDYN}} \formspace.
\end{align}
Elements $f \in \BUmjdyn$ are then equivalence classes $ f = \left[ g \right]_m^\text{dyn}$, $g \in \BUOmega$, where the equivalence relation is given for any $g,h \in \BUOmega$ by
\begin{align}
g \sim_m h :\iff g-h \in \IMJDYN \formspace.
\end{align}
Now, in the second step, we incorporate the CCR by dividing out the two-sided ideal $\IMCCR$ that is generated by elements
\begin{align}
\Big[ \big(-\i \Gm{F}{F'}, 0 , F \otimes F' - F' \otimes F, 0 , 0 , \dots\big) \Big]_m^\text{dyn} \in \BUmjdyn
\end{align}
to obtain the final field algebra $\BUmj$ as specified by the following definition:
\begin{definition}[Field algebra and quantum Proca fields]
The \emph{Borchers-Uhlmann field algebra} $\BUmj$, for some fixed $m>0$, is defined by
	\begin{align}
     \gls{bumj} \coloneqq {\Quotientscale{\BUmjdyn}{\IMCCR}} \formspace.
	\end{align}
We will sometimes equivalently write $\BUmj = {\Quotientscale{\BUOmega}{\IMJ}}$, where $\IMJ$ is generated by both of the wanted relations for short hand.
A \emph{quantum Proca field} is then an element
\begin{align}
	\gls{phimf} \coloneqq [(0,F,0,0,\dots)]_{m,j} \in \BUmj \formspace,
\end{align}
where the equivalence class $[\cdot]_{m,j}$ is taken w.r.t. $\IMJ$.
\end{definition}
By construction, the quantum Proca fields fulfill the wanted dynamical and commutation relations
\begin{align}
	\phi_{m,j}( (\delta d +m^2) F ) 				&= \langle j, F \rangle_\M \cdot \mathbbm{1}_{\BUmj} \formspace,\\
	\big[ \phi_{m,j}(F) , \phi_{m,j}(F') \big] 	   &= \i \Gm{F}{F'} \cdot  \mathbbm{1}_{\BUmj} \formspace.
\end{align}
We still have to endow the field algebra $\BUmj$ with a topology.
\subsubsection{Topology, initial data and the field algebra}\label{sec:field-algebra-topology}
In this chapter, as a preparation for the investigation of the mass dependence and the zero mass limit of the theory, we investigate the topology of the field algebra and its connection to a field algebra of initial data. The mass $m$ as well as the external source $j$ again remain fixed.\par
The straightforward way to obtain a topology on $\BUmj$ is from the topology on $\BUOmega$ as the quotient topology, assuming the ideal $\IMJ$ is closed\footnote{We will argue shortly that, at least in the case of $j=0$, it is.}:
If $\IMJ$ is a closed subspace of $\BUOmega$, and $\BUOmega$ is endowed with a locally convex topology, also $\BUmj$ is locally convex \cite[Theorem 12.14.8]{dieudonne_2}. Indeed, the quotient topology on $\BUmj$ coincides with the topology induced by the semi-norms
\begin{align}
	q_{m, \alpha}( [f]_m ) = \inf\big\{ p_\alpha(g) : g \in [f]_m \big\}
\end{align}
where $\left\{ p_\alpha\right\}_\alpha $ is a family of semi-norms on $\BUOmega$ that induces its topology (see \cite[Lemma 12.14.8]{dieudonne_2}).
While this topology allows us to define continuous families $\left\{ \phi_{m,j,n}(F_n)\right\}_n$ of fields at a fixed mass, it is not suitable, as we will discuss in the next Section \ref{sec:mass_depenence_and_limit}, to define a notion of continuity with respect to the mass $m$.
In principal, the problem is that at different masses the fields live in different algebras that we are a priori unable to compare with each other. The idea to solve this is simple: we will for every fixed mass find a topological algebra that is \emph{homeomorphic to the field algebra} $\BUmj$ which does not depend on the mass (the homeomorphism of course does). We are then able to map a family $\{ \phi_{m,j} (F_m)\}_m$ of fields into only \emph{one} topological algebra where we have a natural sense of continuity given by the topology. This mass independent algebra will be the Borchers-Uhlmann algebra of \emph{initial data}.\par
To start the construction, we first look at the  topological vector space $\Omega^1_0(\M)$. Let $\Sigma$ be a Cauchy surface. As done at the classical level, we introduce the short notation
$\Dzs = \Omega^1_0(\Sigma) \otimes \Omega^1_0(\Sigma)$ for the space of initial data with respect to a fixed Cauchy surface $\Sigma$. We define the map
\begin{align}
	\kappa_m: \Omega^1_0(\M) &\to \Dzs  \\
					F &\mapsto (\rhoz G_m F , \rhod G_m F) \formspace, \notag
\end{align}
where the operators $\rho_{(\cdot)}$ were introduced in Definition \ref{def:cauchy_mapping_operators}. The map $\kappa_m$ maps a test one-form $F$ to the initial data of $G_m F$ on the Cauchy surface $\Sigma$. In the notation we omit the dependence of the map on the Cauchy surface. By construction, the map $\kappa_m$ is continuous for a fixed value of $m$ with respect to the topology on $\Dzs$ that is induced by the topology of $\Omega^1_0(\Sigma)$.
Since $\kappa_m$ is continuous, we know that $\KERN{\kappa_m}$ is closed (see \cite[34-36 ]{treves}). By construction the kernel of the map $\kappa_m$ is the set
\begin{align}
	\KERN{\kappa_m} = \JMDYN \coloneqq \big\{ (\delta d + m^2) F , F \in \Omega^1_0(\M) \big\} \formspace.
\end{align}
This is useful since at the field algebra level we want to divide out fields where the degree-one-part is of the above form to incorporate the dynamics of the theory.
We would like to find a homeomorphism between the quotient space ${\Quotientscale{\Omega^1_0(\M)}{\JMDYN}}$ and $\Dzs$.
From a standard construction (see \cite[ibid.]{treves}) we find the map $\xi_m : {\Quotientscale{\Omega^1_0(\M)}{\JMDYN}} \to \IMG{\kappa_m}$ which is the unique bijective map such that $\xi_m([F]_m) = \kappa_m(F)$ (see \cite[16]{treves}). The construction is illustrated in Diagram \ref{dia:homeomorphism_one_particle_level}.
\begin{table}
\begin{displaymath}
\xymatrix @R=20mm @C=30mm
{
	\Omega^1_0 (\M)  \ar[r]^{\kappa_m}   \ar[dr]_{[\cdot]_m}  			&     \IMG{\kappa_m} 													\ar@<-.5ex>[d]_{\xi_m^{-1}} \ar@{^{(}->}[r]^i &   \Dzs \\
	&      {{\Quotientscale{\Omega^1_0(\M)}{\JMDYN}}}          \ar@<-.5ex>[u]_{\xi_m}
}
\end{displaymath}
\caption{Illustrating the construction of the homeomorphism $\xi_m$ of the space of dynamical test one-forms and the space of initial data.}
\label{dia:homeomorphism_one_particle_level}
\end{table}
Now we would like to show that $\IMG{\kappa_m} = \Dzs$ and that $\Dzs$ and ${\Quotientscale{\Omega^1_0(\M)}{\IMDYN}}$ are homeomorphic, that is, we need to show that $\xi_m$ and $\xi_m^{-1}$ are continuous. We will state this in form of the following lemma:
\begin{lemma}\label{lem:one-particle-homeomorphism}
	Let $\Sigma$ be a Cauchy surface and $\Dzs$ be the space of initial data on $\Sigma$.
	\begin{center}
	The spaces ${\Quotientscale{\Omega^1_0(\M)}{\JMDYN}}$ and $\Dzs$ are homeomorphic.
	\end{center}
\end{lemma}
\begin{proof}
1.) First we will show that $\kappa_m$ is surjective, that is, for every initial data $(\varphi, \pi)$ we find a corresponding $F \in \Omega_0^1(\M)$ such that $(\varphi , \pi) = (\rhoz G_m F , \rhod G_m F)$.
For this, we explicitly construct a map $\vartheta : \Dzs \to \Omega^1_0(\M)$ that maps any pair $(\varphi , \pi)$ to a corresponding $F$.\par
Let $(\varphi,\pi) \in \Dzs$ specify initial data on a Cauchy surface $\Sigma$. Then by Theorem \ref{thm:solution_proca_unconstrained} there exists a unique solution $A \in \Omega^1(\M)$ to the source free Proca equation $(\delta d + m^2) A = 0$ with the given data. Furthermore, $A$ depends continuously on $(\varphi,\pi)$.\\
We need to construct a \emph{compactly supported} one-form $F$ from $A$, such that $A$ and $G_mF$ determine the same initial data. First, we note that (see \cite[Theorem 3.2.11]{baer_ginoux_pfaeffle})
\begin{align}
	\supp{A} \subset J\big( \supp{\varphi} \cup \supp{\pi} \big) \formspace.
\end{align}
We choose a $\chi \in \Omega^1(\M)$ such that
\begin{align}
	\chi =
	\begin{cases}
	1 , &\text{in the future of some Cauchy surface } \Sigma_+ \\
	0 , &\text{in the past of some Cauchy surface } \Sigma_-
	\end{cases}
\end{align}
where $\Sigma_\pm$ are Cauchy surfaces in the future/past of the Cauchy surface $\Sigma$.\\
Then, by construction,
\begin{align}
	F \coloneqq -(\delta d + m^2) \chi A \eqqcolon \vartheta(\varphi,\pi)
\end{align}
is a compactly supported one-form, since $F = 0$ on $J^+(\Sigma_+)$ and $J^-(\Sigma_-)$, hence $\supp{F} \subset J\big( \supp{\varphi} \cup \supp{\pi}\big) \cap  J^-(\Sigma_+) \cap J^+(\Sigma_-)$ which is compact. Furthermore we observe that, since $A$ depends continuously on the initial data, $\vartheta$ is continuous. The setup is illustrated in Figure \ref{fig:init_data_chi}.
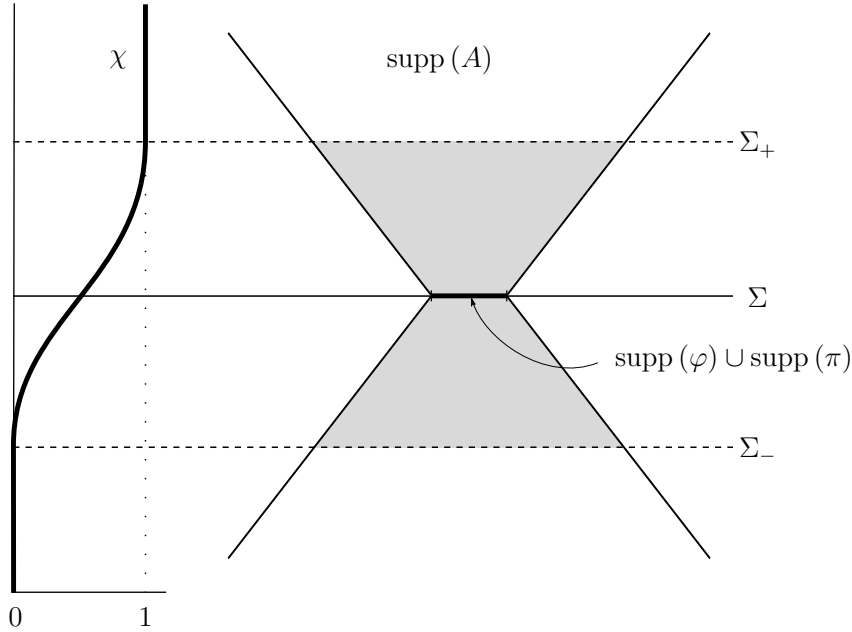
\begin{figure}[]
	\begin{center}
		\scalebox{0.9}{
		\begin{tikzpicture}
		\node at (0,0) {\includegraphics[scale=0.9]{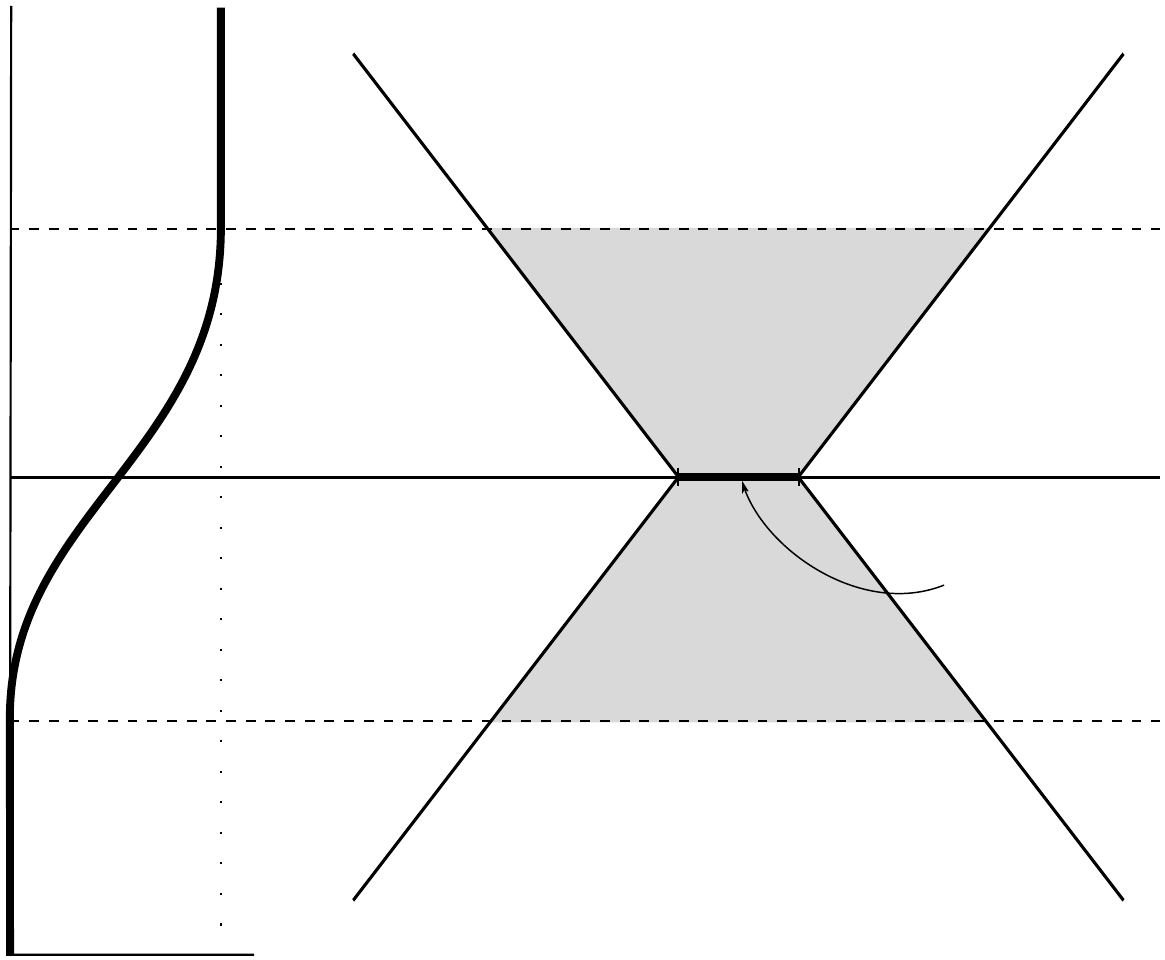}};
		\node at (1,3.5)  {$\supp{A}$};
		\node at (-3.7,3.5)  {$\chi$};
		\node at (-5.2,-4.7)  {$0$};
		\node at (-3.3,-4.7)  {$1$};
		\node at (5.65,2.25)  {$\Sigma_+$};
		\node at (5.65,0.01)  {$\Sigma$};
		\node at (5.65,-2.25)  {$\Sigma_-$};
		\node at (5.3,-0.9)  {$\supp{\varphi} \cup \supp{\pi}$};
		\end{tikzpicture}
		}
	\end{center}
	\caption{Illustrating the setup of the proof of Lemma \ref{lem:one-particle-homeomorphism}: On the Cauchy surface $\Sigma$, $\varphi$ and $\pi$ denote initial data of the solution $A$ to Proca's equation. The area marked grey contains the support of the one-form $F = -(\delta d + m^2) \chi A$.}
	\label{fig:init_data_chi}
\end{figure}
We finally want to show that $G_m F = A$, since then they in particular have the same initial data $(\varphi, \pi)$ which would complete the proof of surjectivity.
A priori, the domain of $G_m^\pm$ is $\Omega^1_0(\M)$ but we can extend its action to one-forms that are supported in the future (or past) of some Cauchy surface, as it is the case for $(\chi A)$ which is supported in the future of $\Sigma_-$. We extend the action by defining
\begin{align}
	G_m^+(\chi A) = \sum\limits_i  G_m^+ \alpha_i \chi A
\end{align}
where $\alpha_i$ is a partition of unity and at any point $x \in \M$ only for finitely many $i$'s $\left( G_m^+ \alpha_i \chi A\right)(x)$ is non-zero since $J^- (x) \cap J^+(\Sigma_-)$ is compact (and for the retarded propagator it holds $\supp{G_m^+ F} \subset J^+\big( \supp{F}\big)$). In the same fashion we extend the action of $G_m^-$ to act on one-forms that are supported to the past of some Cauchy surface as it is the case for $(1-\chi)A$:
\begin{align}
	G_m^-(1-\chi)A = \sum\limits_i  G_m^+ \alpha_i (1-\chi) A \formspace.
\end{align}
With this notion we find
\begin{align}
	G_m^+ F
	&= - G_m^+ (\delta d + m^2) \chi A \notag\\
	&= -  \sum_i G_m^+ \alpha_i (\delta d + m^2) \chi A\notag \\
	&= -  \sum_i G_m^+  (\delta d + m^2) \alpha_i \chi A \notag\\
	&= -  \sum_i \alpha_i \chi A\notag \\
	&= - \chi A \formspace.
\end{align}
Also, we observe that
\begin{align}
	G_m^- (\delta d + m^2) (1-\chi) A
	&=  G_m^- \underbrace{(\delta d + m^2)  A}_{= 0} - G_m^- (\delta d + m^2)\chi A \notag\\
	&= - G_m^- (\delta d + m^2)\chi A \notag\\
	&= G_m^- F \formspace,
\end{align}
and therefore we find in the same fashion as above
\begin{align}
G_m^- F
&=  G_m^- (\delta d + m^2) (1-\chi) A \notag\\
&=  \sum_i G_m^- \alpha_i (\delta d + m^2) (1-\chi) A \notag\\
&=  (1-\chi) A \formspace.
\end{align}
We therefore find the result
\begin{align}
G_m F
&= (G_m^- - G_m^+)F \notag\\
&= (1- \chi) A + \chi A \notag\\
&= A
\end{align}
which completes the proof of surjectivity. That is, we have found $\IMG{\kappa_m} = \Dzs$.\par
2.) Now we are left to show that both $\xi_m$ and $\xi_m^{-1}$ are continuous so that indeed $\xi_m$ is a homeomorphism.\par
\emph{i)} By construction, $\xi_m$ is continuous if and only if $\kappa_m$ is continuous \cite[Proposition 4.6]{treves}, which, as we have argued, is the case.\par
\emph{ii)} The inverse is by construction given by $\xi_m^{-1} = [\cdot]_m \comp \,\vartheta$. As we have argued, the map $\vartheta$ is continuous. Also by construction, the map $[\cdot]_m$ is continuous which yields that indeed $\xi^{-1}_m$ is continuous.
This completes the proof.
\end{proof}
\par We will now generalize these ideas, by explicit use of the constructed maps, to the field algebra $\BUmj$.
First, we set the external source to vanish, $j = 0$. We make this explicit in the notation by indexing the effected elements with $0$ instead of $j$. In a second step we will then generalize to non vanishing external sources.\par
On the level of $\BUOmega$, to implement the dynamics, we divide out the ideal generated by $\big(0,(\delta d + m^2)F',0,0,\dots\big)$  where $F'\in \Omega^1_0(\M)$. This is equivalent to the ideal generated by $= \big(0,F,0,0,\dots\big)$, $F \in \JMDYN$. As we did on the degree-one level, we would like to find a map $K_m : \BUOmega \to \BU(\Dzs)$ such that $\KERN{K_m} = \IMZDYN$ and then show that $\BUmzdyn$ is homeomorphic to $\BU(\Dzs)$. We do this by \emph{lifting} the map $\kappa_m$ to the BU-algebra:
We define the map $\gls{Km} : \BUOmega \to \BU(\Dzs) $ as a BU-algebra-homomorphism which is then completely determined by its degree wise action:
\begin{align}
	K_m : \BUOmega &\to \BU(\Dzs) \\
	(f^{(0)},0,0,\dots) &\mapsto 	(f^{(0)},0,0,\dots) \notag \\
	\big(0,F,0,0,\dots\big) &\mapsto 	\big(0, \kappa_m (F),0,0,\dots\big)\notag \\
		\left(0,0, \sum\nolimits_i F_i \otimes F_i',0,0,\dots\right) &\mapsto 	\left(0,0, \sum\nolimits_i \kappa_m (F_i) \otimes \kappa_m (F'_i) ,0,0,\dots\right) \notag\\
		&\dots\notag
\end{align}
We will call this definition the \emph{lifting} of a map from the space of test one-forms to the BU-algebra.
Before we can proceed to show homeomorphy we need to investigate some properties of the map $K_m$.
\begin{lemma}
	Let $\Sigma$ be a Cauchy surface, $\Dzs$ be the space of initial data on $\Sigma$ and the map $K_m : \BUOmega \to  \BU\big(\Dzs\big) $ be as defined above. Then:
	\begin{align}
	\IMG{K_m} &=  \BU\big(\Dzs\big)  \formspace \text{and}\\
		\KERN{K_m} &= \IMZDYN \formspace.
	\end{align}
\end{lemma}
\begin{proof}
1.) As we have already shown in the proof of Lemma \ref{lem:one-particle-homeomorphism}, $\kappa_m$ is surjective. Therefore, it directly follows $\IMG{K_m} = \BU\big(\Dzs\big) $.\par
2.) We will show the equivalence of the kernel and the ideal by a two way inclusion.\par
\emph{i)} It is obvious by construction that $\IMZDYN \subset \KERN{K_m} $:
Let $f \in \IMZDYN$ be arbitrary. By definition, we can write $f$ as a finite sum
\begin{align}
f = \sum\limits_i  g_{i} \cdot \left(0,F_i,0,0,\dots \right) \cdot h_{i} \formspace,
\end{align}
where $F_i \in \JMDYN$ and $g_i, h_i \in \BUOmega$. Since $K_m$ was constructed as a homomorphism and $\KERN{\kappa_m}  = \JMDYN$, we find
\begin{align}
K_m(f)
= \sum\limits_i  K_m\left(g_{i}\right) \cdot \underbrace{\left(0,\kappa_m(F_i),0,0,\dots \right)}_{=0} \cdot K_m\left(h_{i}\right)
= 0 \formspace.
\end{align}
Hence we find $f \in \KERN{K_m}$. \par
\emph{ii)} The non-trivial part is to show $\KERN{K_m} \subset \IMZDYN$:\\
Let $f = (f^{(0)}, f^{(1)}, f^{(2)}, \dots , f^{(N)} , 0 , 0 ,\dots) \in \KERN{K_m}$, $f^{(k)} \in \left(\Omega^1_0(\M)\right)^{\otimes k}$, be arbitrary.
With a slight abuse of notation\footnote{We abuse the notation by applying $\kappa_m$ to a tensor product of differential forms, that is, write for shorthand $\kappa_m(\sum_i F_i \otimes G_i) =\sum_i \kappa_m(F_i) \otimes \kappa_m(G_i)$.} we formulate the condition on $f$ being an element of $\KERN{K_m}$ by stating $\kappa_m\left(f^{(k)}\right)  = 0$ for all $k=1,\dots,N$.
We will show by induction in $n$ that an arbitrary homogeneous element $(0,0, \dots , f^{(n)} , 0 ,0, \dots)$ of degree $n$ with $\kappa_m\left(f^{(n)}\right)  = 0$ can be written in the form $\sum_i  g_{i} \cdot \left(0,F_i,0,0,\dots \right) \cdot h_{i}$ for some $g_{i}, h_{i} \in \BUOmega$, that is, it is an element of $\IMZDYN$. Since $f$ can be written as a sum of those elements, it then follows that $f \in \IMZDYN$.
We begin the induction with the base case:\par
\emph{a)} A homogeneous element $f'$ of degree zero in the kernel of $K_m$ can be trivially shown to be an element of $\IMZDYN.$ Since it holds $\kappa_m(f'^{(0)}) = 0$ for all $f'^{(0)} \in \IC$ it follows by definition $f'^{(0)} = 0$ and, trivially, $(0,0,\dots) \in \IMZDYN$.\par
\emph{b)} For a homogeneous element $f'$ of degree one in the kernel of $K_m$ we find the condition $\kappa_m(f'^{(1)}) = 0$, or equivalently $f'^{(1)}  \in \JMDYN$. Therefore, also $(0, f'^{(1)} , 0 ,0 ,\dots) \in \IMDYN$. \par
We can now make the induction step.
The assumption is that for some $n$ it holds for an arbitrary $f^{(n)} \in \left(\Omega^1_0(\M)\right)^{\otimes n}$ with $\kappa_m(f^{(n)}) = 0$ that
$(0,0,\dots,f^{(n)},0,0,\dots) \in \IMZDYN$.
Now we look at an element of the form $(0,0,\dots,0,f^{(n+1)},0,0\dots)$ where $f^{(n+1)} \in \left(\Omega^1_0(\M)\right)^{\otimes (n+1)}$ such that $\kappa_m(f^{(n+1)}) = 0$.
We can write this more explicitly for some $F_i \in \Omega^1_0(\M)$ and some $\F_i^{(n)}\in \left(\Omega^1_0(\M)\right)^{\otimes n}$ as
\begin{align}
(0,0,\dots,0,f^{(n+1)},0,0, \dots)  = (0,0,\dots,\sum\limits_{i=1}^{M} F_i \otimes \F_i^{(n)},0, 0 , \dots) \formspace.
\end{align}
Let $V \coloneqq \SPAN{F_1,F_2, \dots , F_M}$ and $V_{\J} \coloneqq V \cap \JMDYN$ define finite dimensional subspaces of $\Omega^1_0(\M)$. We find a basis $\{ \widetilde{F}_1, \dots ,\widetilde{F}_\mu \}$, $\mu \leq M$, of $V_{\J}$ which we can extend to a basis $\{ \widetilde{F}_1, \dots ,\widetilde{F}_M \}$ of $V$.\\
With the use of this basis we can re-write
\begin{align}
f^{(n+1)} = \sum\limits_{i=1}^{M} F_i \otimes \F_i^{(n)}
&= \sum\limits_{i=1}^{M} \widetilde{F}_i \otimes \widetilde{\F}_i^{(n)} \notag\\
&= \sum\limits_{i=1}^{\mu} \widetilde{F}_i \otimes \widetilde{\F}_i^{(n)}  +  \sum\limits_{i=\mu + 1}^{M} \widetilde{F}_i \otimes \widetilde{\F}_i^{(n)} \notag\\
& \eqqcolon X_1^{(n+1)} + X_2^{(n+1)} \formspace.
\end{align}
Here, each $\widetilde{\F}_i^{(n)}$ can be constructed as linear combinations of the ${\F}_i^{(n)}$'s.\\
Now we first have a look at $X_1^{(n+1)}$. We know by construction for $i=1,\dots,\mu$ that  $\kappa_m ( \widetilde{F_i} ) = 0$. Therefore, $\kappa_m(X_1^{(n+1)}) = 0$. Hence, $X_1^{(n+1)}$ is already of the wanted form, that is, we can choose $g_{i}, h_{i} \in \BUOmega$, such that for $\alpha = 1, \dots , \mu$
\begin{align}
	g^{(0)}_{\alpha} &= 1 \\
	F_\alpha &= \widetilde{F}_\alpha \\
	h^{(n)}_{\alpha} &= \widetilde{\F}_\alpha^{(n)}
\end{align}
and all remaining components of $g_{\alpha}, h_{\alpha}$ vanish. With this we have brought the first part into the wanted form
\begin{align}
(0,0,\dots,X_1^{(n+1)},0,0, \dots)
&= \sum\limits_{\alpha=1}^\mu  \left(0,0,\dots,\widetilde{F}_\alpha \otimes \widetilde{\F}_\alpha^{(n)},0,0,\dots \right) \notag\\
&= \sum\limits_{\alpha=1}^\mu  g_{\alpha} \cdot \left(0, F_\alpha, 0 ,0,\dots \right) \cdot h_{\alpha} \in \IMZDYN \formspace.
\end{align}
Now we have a closer look at the remaining part $X_2^{(n+1)}$.
First, by construction, we observe $\SPAN{ \widetilde{F}_{\mu+1}, \dots ,\widetilde{F}_M } \cap \JMDYN = \{ 0 \}$. This implies that the $\kappa_m(\widetilde{F_i})$'s are linearly independent\footnote{For them to be linearly dependent we would need to find some constants $z_i \in \IC$, for $i=\mu+1,\dots,M$, such that $\sum z_i \kappa_m(\widetilde{F_i}) = 0$. Using the linearity of $\kappa_m$, this yields $\kappa_m(\sum z_i \widetilde{F_i}) =0$. This contradicts the fact that $\SPAN{ \widetilde{F}_{\mu+1}, \dots ,\widetilde{F}_M } \cap \JMDYN = \{ 0 \}$.} for $i=\mu+1,\dots,M$.
Furthermore, since $\kappa_m(X_1^{(n+1)}) = 0$,  $\kappa_m(f^{(n+1)}) = 0$ and due to the linearity of $\kappa_m$, also $\kappa_m(X_2^{(n+1)}) = 0$. Using the linear independence we conclude that $\kappa_m( \widetilde{\F}_i^{(n)}) = 0$. Since $\widetilde{\F}_i^{(n)}$ is of degree $n$ and lies in the kernel of $\kappa_m$, we can apply the induction assumption and find
\begin{align}
\Big(0,0,\dots, \widetilde{\F}_i^{(n)} , 0 , 0 , \dots \Big) \in \IMZDYN \formspace.
\end{align}
Therefore, also
\begin{align}
\Big(0,0,\dots, \widetilde{F}_i &\otimes \widetilde{\F}_i^{(n)} , 0 , 0 , \dots \Big) \notag \\
&= \Big(0, \widetilde{F}_i , 0 , 0 , \dots \Big) \cdot \Big(0,0,\dots, \widetilde{\F}_i^{(n)} , 0 , 0 , \dots \Big) \in \IMDYN \formspace.
\end{align}
Now, bringing the two parts together we successfully have completed the induction:
\begin{align}
\Big(0,0,\dots, f^{(n+1)} , 0 , 0 , \dots \Big) = 	\Big(0,0,\dots, X_1^{(n+1)} + X_2^{(n+1)} , 0 , 0 , \dots \Big) \in \IMZDYN
\end{align}
and, as we have argued before, therefore $f \in \IMZDYN$, which completes the proof.
\end{proof}
With these results we can, analogously to the degree-one-part, construct a homeomorphism.
By the same argument as for the degree-one level, we find a unique map $\gls{Xim} : \BUmzdyn \to \BU\big( \Dzs \big)$ where $\Xi_m ([f]_m) = K_m(f)$, $f \in \BUOmega$. This is again best illustrated in form of a diagram, shown in Diagram \ref{dia:source-free-dynamical-homeo}.
\begin{table}
\begin{displaymath}
\xymatrix @R=20mm @C=30mm
{
	\BUOmega  \ar[r]^{K_m} \ar[dr]_{[\cdot]_m^{\text{dyn}}} 	& \BU\big(\Dzs\big) \ar@<-.5ex>[d]_{\Xi^{-1}_m}  		\\
	&  \BUmzdyn  	       \ar@<-.5ex>[u]_{\Xi_m}
}
\end{displaymath}
\caption{Illustrating the construction of the homeomorphism $\Xi_m$ between the source-free dynamical field algebra and the field algebra of initial data.}
\label{dia:source-free-dynamical-homeo}
\end{table}
We formulate one important result of this thesis in form of the following lemma:
\begin{lemma}\label{lem:field-algebra-homeomorphism}
	Let $m>0$ and $j=0$.
	\begin{center}
The spaces $\BUmzdyn$ and $\BU\big( \Dzs \big)$ are homeomorphic.
	\end{center}
\end{lemma}
\begin{proof}
	1.) To show that $\Xi_m$ is continuous is trivial. As we have argued when we first introduced the map, $\kappa_m$ is continuous. Therefore $K_m$ is continuous and hence $\Xi_m$ is continuous (again, see \cite[Proposition 4.6]{treves}).\par
	2.) As we did on the test one-form level, we directly construct the inverse $\Xi_m^{-1}$ as follows:
	On the degree-one level we have constructed a continuous map $\vartheta_m : \Dzs \to \Omega^1_0 (M)$ to construct $\xi_m^{-1} = [\cdot]_m \comp \vartheta_m$.
	We lift the map $\vartheta_m$ to the map $\Theta_m : \BU\big( \Dzs \big) \to \BUOmega$ in the same fashion as we have lifted the map $\kappa_m$ to $K_m$. Then, by construction, $\Theta_m$ is continuous and $\Theta_m(K_m(f)) = f$.
	We can thus conclude $\Xi^{-1} = [\cdot]_m^\text{dyn} \comp \Theta_m$ which is continuous by construction.
	This completes the proof.
\end{proof}
With this lemma, we have successfully incorporated the dynamics of the Proca field into our field algebra. We are left to also include the quantum nature of the fields, that is we have to divide out the relation that implements the CCR. In $\BUmzdyn$, we need to divide out the two-sided ideal $\IMCCR$ that is generated by elements $\big(-\i \Gm{F}{F'}, 0 , F \otimes F' - F' \otimes F , 0 , 0 , \dots\big)$. As we have already calculated in the proof of Lemma \ref{lem:propagator-non-degenerate}, we can write the action of $\Gm{\cdot}{\cdot}$ in form of initial data $(\varphi, \pi), (\varphi', \pi') \in \Dzs$ as follows: \vspace{-.6cm}
\begin{align}
\Gm{F}{F'} = \langle \pi , \varphi' \rangle_\Sigma - \langle \varphi, \pi'\rangle_\Sigma \formspace,
\end{align}
where $(\varphi, \pi), (\varphi', \pi') \in \Dzs$ are the initial data of $G_m F$, $G_m F'$ respectively.
Therefore, the ideal $\IMCCR$ maps under $\Xi_m$ to the two-sided ideal $\ICCR \subset \BU\big( \Dzs \big)$ that is generated by elements \vspace{-.3cm}
\begin{align}
\big(-\i ( \langle \pi , \varphi' \rangle_\Sigma - \langle \varphi, \pi'\rangle_\Sigma), 0 , (\varphi, \pi) \otimes (\varphi', \pi') - (\varphi', \pi') \otimes (\varphi, \pi) , 0 , 0 , \dots\big)\formspace.
\end{align}
With this, the following theorem follows easily.
\begin{theorem}\label{thm:field_algebra_homeomorphy_source_free}
	Let $m>0$ and $j=0$.
	\begin{center}
The field algebra $\BUmz$ is homeomorphic to ${\Quotientscale{\BU\big( \Dzs \big)}{\ICCR}}$.
	\end{center}
\end{theorem}
\begin{proof}
	We have already argued most of what is necessary for the proof:\\
	$\BUmz$ is obtained from $\BUmzdyn$ by dividing out the ideal $\IMCCR$. This ideal maps under $\Xi_m$ one-to-one to the ideal $\ICCR$ in $\BU\big( \Dzs\big)$. 	Intuitively, it is quite clear that if we have two homeomorphic algebras, where in each algebra we divide out an ideal and the two ideals are homeomorphic, we end up with two algebras that are homeomorphic.
	The construction of the homeomorphism formalizes that idea. Due the the homeomorphy of the two ideals $\IMCCR$ and $\ICCR$, the map $[\cdot]_\sim^\text{CCR} \comp \Xi_m : \BUmzdyn \to {\Quotientscale{\BU\big( \Dzs \big)}{\ICCR}}$ has kernel $\IMCCR$. This yields the unique existence of the continuous map $\gls{Lambdam} : \BUmz \to {\Quotientscale{\BU\big( \Dzs \big)}{\ICCR}}$. The construction of the inverse $\Lambda^{-1}$ follows analogously using the map $[\cdot]_m^\text{CCR} \comp \Xi_m^{-1} : \BUmzdyn \to {\Quotientscale{\BU\big( \Dzs \big)}{\ICCR}}$.\\ This completes the proof.
\end{proof}
To conclude the construction made so far, the situation is illustrated in Diagram \ref{dia:source-free-homeomorphism-final}.
\begin{table}[]
\begin{displaymath}
\xymatrix @R=20mm @C=30mm
{
	\BUOmega  \ar[r]^{K_m} \ar[dr]_{[\cdot]_m^{\text{dyn}}} 	& \BU\big(\Dzs \big) \ar[r]^{[\cdot]^\text{CCR}_{\sim}}\ar@[linkred][dr]^(.35){  \color{linkred}{\;\;\;[\cdot]^\text{CCR}_m \comp \, \Xi_m^{-1}}  } \ar@<-.5ex>[d]_{\Xi^{-1}_m}  		&       {{\Quotientscale{\BU\big( \Dzs\big) }{\ICCR}}} \ar@<-.5ex>[d]_{\Lambda_m^{-1}} \\
	&  \BUmzdyn  		\ar[r]_{[\cdot]^\text{CCR}_{m}} \ar@[linkblue][ur]_(.35){  \color{linkblue}{\;\;\;[\cdot]^\text{CCR}_\sim \comp \, \Xi_m} } 	                                     \ar@<-.5ex>[u]_{\Xi_m}		& 		 \BUmz	\ar@<-.5ex>[u]_{\Lambda_m}\\
}
\end{displaymath}
\caption{Illustrating the construction of the homeomorphy of the source-free field algebra and the field algebra of initial data. Bi-directional arrows represent homeomorphisms.}
\label{dia:source-free-homeomorphism-final}
\end{table}
So, for the case with vanishing external sources, we have found a way to compare the field algebra at different masses with each other. This will be the starting point for the investigation of mass dependence for $j=0$ in section \ref{sec:mass-dependence-j-zero}. Before we turn to this, we first study the case where $m$ is still fixed but we allow for external currents $j \neq 0$.\par
Assume we have given a non vanishing external current $j$. It is clear that the previous construction will not generalize in a trivial fashion, since the ideal that implements the dynamics is generated by elements $\big(-\langle j, F \rangle_\M, (\delta d + m^2)F,0,0,\dots\big)$ and it is not obvious to find a map, similar to $K_m$, such that the dynamical ideal is the kernel of that map. Instead, we will show that actually the field algebra with source dependent dynamics is homeomorphic to the field algebra with vanishing source dynamics, $\BUmzdyn \cong \BUmjdyn$.
\begin{theorem}[Field algebra homeomorphy]\label{thm:field-algebra-homeomorphy}
	Let $m > 0$ and $j \in \Omega^1(\M)$.
	\begin{center}
	The field algebras $\BUmj$ and $\BUmz$ are homeomorphic.
	\end{center}
	In particular, this yields that also $\BUmj$ and ${\Quotientscale{\BU\big(\Dzs\big)}{\ICCR}}$ are homeomorphic.
\end{theorem}
\begin{proof}
	The proof works in two steps:\\
	1.) First we construct a non-trivial homeomorphism
	\begin{align}
	\gls{Gammamj} : \BUOmega &\to \BUOmega \formspace,
	\end{align}
	where $\varphi_{m,j}$ is a fixed solution to the classical source dependent Proca equation $(\delta d +m^2) \varphi_{m,j}= j$,
	by defining $\Gamma_{m,j,\varphi}$ as a *-algebra-homomorphism which is then uniquely determined by its action on homogeneous elements of degree zero and one in $\BUOmega$. That is, we define for any $c \in \IC$ and $F \in \Omega^1_0(\M)$
			\begin{align}
			\Gamma_{m,j,\varphi} :
			(c,0,0,\dots) &\mapsto (c,0,0,\dots) \\
			(0,F,0,0,\dots) &\mapsto (-\langle \varphi , F \rangle_\M , F , 0,0,\dots) \notag
			\end{align}
	which is extended to act on arbitrary elements of $\BUOmega$ by linearity and homomorphy with respect to multiplication.
	The inverse is obviously determined by
	\begin{align}
	\Gamma_{m,j,\varphi}^{-1} : \BUOmega &\to \BUOmega \\
	(c,0,0,\dots) &\mapsto (c,0,0,\dots) \notag\\
	(0,F,0,0,\dots) &\mapsto (+\langle \varphi , F \rangle_\M , F , 0,0,\dots) \formspace.\notag
	\end{align}
	Both $\Gamma_{m,j,\varphi}$ and $\Gamma_{m,j,\varphi}^{-1}$ are continuous:\\
	Trivially, the identity map $c \mapsto c$ and $(0,F,0,0,\dots) \mapsto (0,F,0,0,\dots)$ is continuous. Furthermore $(0,F,0,0,\dots) \mapsto (\pm \langle \varphi_{m,j}, F \rangle_\M , 0,0,\dots )$ is continuous since $\varphi_{m,j}$ is assumed fixed. Since the sum of continuous functions in continuous, $(0,F,0,0,\dots) \mapsto (\pm \langle \varphi_{m,j},F \rangle_\M, F, 0,0,\dots)$ is continuous. Hence, $\Gamma_{m,j,\varphi}$ defines a non-trivial homeomorphism of $\BUOmega$.\par
	2.) In this second step we will show that, with respect to $\Gamma_{m,j,\varphi}$, the ideals $\IMDYN$ and $\IMCCR$ are homeomorphic to $\IMJDYN$ and $\IMCCR$ respectively.
	It suffices to show that the generators of the source-free ideals map under $\Gamma_{m,j,\varphi}$ to the corresponding generators of the source dependent ideals and vice versa. Let $F\in \Omega^1_0(\M)$, then
	\begin{align}
	\Gamma_{m,j,\varphi}&\Big( \big( 0,(\delta d +m^2)F , 0 , 0 , \dots   \big)  \Big) \notag\\
	&= \big( 0,(\delta d +m^2)F , 0 , 0 , \dots   \big) - \langle \varphi_{m,j}, (\delta d + m^2)F \rangle_\M \cdot \mathbbm{1}_\BU \notag\\
	&= \big( 0,(\delta d +m^2)F , 0 , 0 , \dots   \big) - \langle (\delta d + m^2) \varphi_{m,j}, F \rangle_\M \cdot \mathbbm{1}_\BU \notag\\
	&= \big( - \langle j, F \rangle_\M ,(\delta d +m^2)F , 0 , 0 , \dots   \big)  \formspace.
	\end{align}
	So the generators for the dynamics transform in the desired way. For the commutation relations we first decompose:
	\begin{align}
	\big( -\i \Gm{F}{F'},0 , F \otimes F' - F' \otimes F   , 0 ,0, \dots   \big)  &=\big( -\i \Gm{F}{F'}, 0 ,0, \dots   \big) \\
	&\phantom{M}+ \big( 0,F,0,0,\dots  \big)\cdot \big( 0,F',0,0,\dots  \big)  \notag \\
	&\phantom{M}-\big( 0,F',0,0,\dots  \big)\cdot \big( 0,F,0,0,\dots  \big) \notag
	\end{align}
	and therefore obtain
	\begin{align}
	\Gamma_{m,j,\varphi}\Big( \big( -\i \Gm{F}{F'},&0 , F \otimes F' - F' \otimes F   , 0 ,0, \dots   \big)  \Big) \notag\\
	&=\big( -\i \Gm{F}{F'}, 0 ,0, \dots   \big)  \notag \\
	&\phantom{M}+ \big( - \langle \varphi_{m,j},F \rangle_\M,F,0,0,\dots  \big)\cdot \big( - \langle \varphi_{m,j},F' \rangle_\M,F',0,0,\dots  \big)  \notag \\
	&\phantom{M}-\big( - \langle \varphi_{m,j},F' \rangle_\M,F',0,0,\dots  \big)\cdot \big( - \langle \varphi_{m,j},F \rangle_\M,F,0,0,\dots  \big) \notag\\
	&= \big( -\i \Gm{F}{F'},0 , F \otimes F' - F' \otimes F   , 0 ,0, \dots   \big)   \formspace.
	\end{align}
	It is straightforward to check in a completely analogous fashion that the generators of the source-dependent ideal map under $\Gamma_{m,j,\varphi}^{-1}$ to the generators of the source-free ideal.
	In conclusion, we find that, with respect to $\Gamma_{m,j,\varphi}$, the ideals $\IMJ$ and $\IMZ$ are indeed homeomorphic. By the same argument as used in Theorem \ref{thm:field_algebra_homeomorphy_source_free}, that is if we divide out two ideals that are homeomorphic the resulting algebras are homeomorphic, we find that $\BUmj$ and $\BUmz$ are homeomorphic. The implication that then $\BUmj$ and ${\Quotientscale{\BU\big(\Dzs\big)}{\ICCR}}$ are homeomorphic follows trivially using Theorem \ref{thm:field_algebra_homeomorphy_source_free}.
\end{proof}
To conclude all of the construction made so far, the final algebraic structure is illustrated in Diagram \ref{dia:final_structure}. Here, $\gls{Psimj}$ is the homeomorphism constructed implicitly in the proof of Theorem \ref{thm:field-algebra-homeomorphy}.
\begin{table}[]
\begin{displaymath}
\xymatrix @R=20mm @C=30mm
{ 																																	& \BU{\left( \Dzs \right)} 		\ar[r]^{[\cdot]_{\sim}^\text{CCR}}		\ar@<.5ex>[d]^{\Xi^{-1}_{m}}				& {{\Quotientscale{\BU{\left( \Dzs \right)}}{\ICCR}}} \ar@<.5ex>[d]^{\Lambda^{-1}_m}\\
	\BUOmega \ar[r]^{[\cdot]_{m,0}^\text{dyn}}  	\ar@<.5ex>[d]^{\Gamma_{m,j,\varphi}} \ar[ur]^{K_m}	&  \BUmzdyn  \ar[r]^{[\cdot]_{m,0}^\text{CCR}} \ar@<.5ex>[u]^{\Xi_{m}}  \ar@<.5ex>[d] &   \BUmz \ar@<.5ex>[d]^{\Psi_{m,j,\varphi}} \ar@<.5ex>[u]^{\Lambda_m}  \\
	\BUOmega  \ar[r]_{[\cdot]_{m,j}^\text{dyn}}      \ar@<.5ex>[u]^{\Gamma^{-1}_{m,j,\varphi}}    & \BUmjdyn \ar[r]_{[\cdot]_{m,j}^\text{CCR}}	\ar@<.5ex>[u]& 	\BUmj 	\ar@<.5ex>[u]^{\Psi_{m,j,\varphi}^{-1}}     \\
}
\end{displaymath}
\caption{Overview of the final algebraic structure and connections between the field algebras. Bi-directional arrows represent homeomorphisms.}
\label{dia:final_structure}
\end{table}
Given an observable of the source free theory $\phi_{m,0}(F)$, we obtain
\begin{align}
\Psi_{m,j,\varphi}(\phi_{m,0}(F))
&= \big[ \Gamma_{m,j,\varphi}(0,F,0,0,\dots)\big]_{m,j} \notag\\
&= \big[ (-\langle \varphi_{m,j}, F \rangle_\M,F,0,0,\dots)\big]_{m,j} \notag\\
&=  -\langle \varphi_{m,j}, F \rangle_\M \cdot \mathbbm{1}_\BUmj + \phi_{m,j}(F) \formspace.
\end{align}
Hence we can easily check that, given an observable $\phi_{m,0}(F)$ that solves the source-free field equations $\phi_{m,0}\big((\delta d + m^2)F\big)=0$ , and fulfills the commutation relations, we obtain by
\begin{align}
\phi_{m,j}(F) = \langle \varphi_{m,j}, F \rangle_\M \cdot \mathbbm{1}_\BUmj + \Psi_{m,j,\varphi}\big(\phi_{m,0}(F)\big)
\end{align}
an observable that clearly solves the source-dependent field equations and fulfills the commutation relations:
\begin{align}
\phi_{m,j}((\delta d+m^2)F)
&=  \langle \varphi, (\delta d +m^2)F \rangle_\M \cdot \mathbbm{1}_\BUmj + \underbrace{\Psi_{m,j,\varphi}\big(\phi_{m,0}\big((\delta d+m^2)F\big)\big)}_{=0 }\notag\\
&= \langle j, F \rangle_\M \cdot \mathbbm{1}_\BUmj
\end{align}
and
\begin{align}
[\phi_{m,j}(F), \phi_{m,j}(F')]
&= \Psi_{m,j,\varphi} \big([\phi_{m,0}(F), \phi_{m,0}(F')] \big) \notag\\
&= \i \Gm{F}{F'} \cdot \mathbbm{1}_\BUmj \formspace.
\end{align}
This concludes our investigation of the field algebra for a fixed mass. With the results of this section, we are able to define a notion of continuity with respect to the mass of the theory and investigate the zero mass limit.
%
%
%
%
%
%
\subsection{Mass dependence and the zero mass limit}\label{sec:mass_depenence_and_limit}
In this section we will present one of the main results of this thesis, that is, we will formulate a notion of continuity of the field theory with respect to a change of the mass and study the zero mass limit.
We will study \emph{existence} of the limit, split into two parts, one for the case of vanishing external sources in Section \ref{sec:mass-dependence-j-zero} and one for the general case with sources in Section \ref{sec:mass-dependence-j-general}. Then, we study the algebraic relations and the dynamics of the fields in the zero mass limit in Section \ref{sec:zero-mass-limit-quantum-algebra-relations}. At given points, we compare our results with the theory of the quantum Maxwell field in curved spacetimes as studied in \cite{Sanders, pfenning}.\par
In order to investigate the zero mass limit we need Assumption \ref{ass:propagator_continuity} to hold. As in the classical theory, this ensures the continuity of the propagators with respect to the mass. This assumption is also essential for the quantum case.
We observe that for every $m$ we obtain a different field algebra $\BUmj$, since both the dynamics and the commutation relations that we implement depend on the mass. The problem is first to find a notion of comparing the Proca fields at different masses with each other.
As we have hinted in section \ref{sec:field-algebra-topology}, we could use the semi-norms
\begin{align}
q_{m, \alpha}( [f]_m ) = \inf\big\{ p_\alpha(g) : g \in [f]_m \big\}
\end{align}
to define a notion of continuity of the theory with respect to the mass $m$. We could define that a function $\eta: \IR_+ \to \bigcup_m \BUmj$, such that $\eta(m_0) \in \BU_{m_0,j}$, is called continuous if and only if the map $m \mapsto q_{\alpha,m}\big(\eta(m)\big)$ is continuous for all $\alpha$ with respect to the standard topology in $\IR$. While this definition seems appropriate, we could not, despite much time and effort spent, prove that for a fixed $F \in \BUOmega$ the map $m \mapsto \phi_m([F]_m)$ is continuous in the above sense, not even in the simpler case $j=0$. But this is certainly a map that we would like to be continuous: If we fix a test function $F$ that we smear the field $\phi_m$ with and vary the mass continuously, we certainly want the observable $\phi_m([F]_m)$ to vary continuously as well. Even at the more simple one-particle level of $\Omega^1_0(\M)$, where we can formulate an equivalent notion of continuity for compactly supported one-forms that incorporate the dynamics in the source free case, that is, one-forms that are of the form $F= (\delta d +m^2)F'$ for some $F' \in \Omega^1_0(\M)$, we were not able to prove that the equivalence classes $[F]_m$ vary continuously with respect to the continuity formulated using semi-norms on the quotient space. It might be a problem worth tackling with more sophisticated mathematical methods. Another approach of formulating a notion of continuity is to use the C*-Weyl algebra rather then the BU-algebra. Being a normed algebra, the Weyl algebra seems suited for the investigation of the zero mass limit at first but, as it turns out, this is not the case as discussed in Appendix \ref{app:weyl-algebra}.
We will now turn to a formulation of continuity that we were able to show to have the desired behavior for fixed test functions $F$.
\subsubsection{Existence of the limit in the current-free case}\label{sec:mass-dependence-j-zero}
Of course we would like to make use of the constructed homeomorphisms to formulate a notion of continuity. And since we have found for every mass $m$ that $\BUmz$ is homeomorphic to $\Quotientscale{\BU\big(\Dzs\big)}{\ICCR}$, we can map a family of elements $\left\{ f_m \right\}_m$, such that for every $m$ it holds $f_m \in \BUmz$, to a family in the algebra ${\Quotientscale{\BU\big(\Dzs\big)}{\ICCR}}$. Here, we have a natural sense of continuity given by the topology.
We state this important result in form of the following definition:
\begin{definition}[Continuity of a family of observables with respect to the mass]\label{def:field_continuity}
	Let $K_m: \BUOmega \to \BU{\left( \Dzs \right)}$ and $\Lambda_m : \BUmz \to {\Quotientscale{\BU{\left( \Dzs \right)}}{\ICCR}}$ be as defined in section \ref{sec:field-algebra-topology}. We call a function
	\begin{align}
	\eta : \IR_+ &\to \bigcup_m \BUmz\\
	m &\mapsto \eta_m  \in \BUmz \notag
	\end{align}
	continuous if the map
	\begin{align}
	\widetilde{\eta} : \IR_+ &\to {\Quotientscale{\BU\big( \Dzs \big) }{\ICCR}} \\
	m &\mapsto \Lambda_m (\eta_m)\notag
	\end{align}
	is continuous.\\
	Equivalently, one can identify $\eta_m = [f_m]_{m,0}$ for some family $\left\{ f_m \right\}_m \subset \BUOmega$ and define the map $\eta$ to be continuous if the map
	\begin{align}
	\hat{\eta} : \IR_+ &\to {\Quotientscale{\BU\big( \Dzs \big) }{\ICCR}} \\
	m &\mapsto [K_m(f_m)]_\sim^\text{CCR}\notag
	\end{align}
	 is continuous. The latter will actually be the more practical definition.
\end{definition}
As an example, and also to check that this notion of continuity has the desired properties, we check that the map $\eta: m \mapsto \phi_m(F)$ is continuous for a fixed $F \in \Omega^1_0(\M)$:\\
According to the above definition, $\eta$ is continuous if
\begin{align}
\widetilde{\eta}: m \mapsto [(0, \kappa_m (F), 0, 0, \dots)]^\text{CCR}_\sim
\end{align}
is continuous in ${\Quotientscale{\BU\big(\Dzs\big)}{\ICCR}}$. Recall that
\begin{align}
\kappa_m : \Omega^1_0(\M) &\to \Dzs \\
F&\mapsto (\rhoz G_m F , \rhod G_m F)  \notag
\end{align} and
\begin{align}
	G_m = \left( \frac{d \delta}{m^2} +1\right) E_m \formspace.
\end{align}
By Assumption \ref{ass:propagator_continuity} and using that the operators $\rho_{(\cdot)}$ are continuous, we find that $m \mapsto \kappa_m F$ is continuous in $m >0$ for any fixed test one-form $F$. By construction, the map $[\cdot]^\text{CCR}_\sim$ is continuous and does not depend on $m$. Therefore, $\widetilde{\eta}$ is continuous and thus $\eta$ is continuous in the above sense.
We indeed find the desired property of continuously varying quantum fields with respect to the mass! Additionally, we find for a fixed $F \in \Omega^1_0(\M)$ and the corresponding field $\phi_m(F)$ that the notion of continuity is independent of the choice of the Cauchy surface $\Sigma$, since $\kappa_m(F)$ is continuous in $m$ for any Cauchy surface.\par
The notion of continuity defined in Definition \ref{def:field_continuity} therefore seems appropriate, and, at least for a field $\phi_m(F)$, is independent of the choice of the Cauchy surface. We therefore will from now on identify a field $\phi_m$ with its initial data formulation $\phi_m(F) = \big[\big (0, \kappa_m(F),0,0,\dots\big)\big]_\sim^\text{CCR}$ and also view the propagator $\Gm{\cdot}{\cdot}$ as a map on initial data as explained in Lemma \ref{lem:propagator-non-degenerate}. But is the notion in general independent of the choice of the Cauchy surface $\Sigma$? To some extent, we can answer this positively:
\begin{lemma}\label{lem:continuity-independence-source-free}
	Let $\Sigma$, $\Sigma'$ be Cauchy surfaces, $a,b \in \IR^+$ arbitrary. Let $\left\{ f_m\right\}_m$ be a family in $\BUOmega$. Then
	\begin{align*}
		\eta^{(\Sigma)} : [ a , b] &\to {\Quotientscale{\BU(\Dzs  )}{\ICCRS}}\hspace{1mm}, \quad m \mapsto \big[ K_m^{(\Sigma)}(f_m) \big]_\sim^{\text{CCR},\Sigma}
	\end{align*}
		is continuous if and only if
		\begin{align*}
		\eta^{(\Sigma')} : [ a , b] &\to {\Quotientscale{\BU(\Dzsp)}{\ICCRSP}}, \quad m \mapsto \big[ K_m^{(\Sigma')}(f_m) \big]_\sim^{\text{CCR},\Sigma'}
	\end{align*}
	is continuous.
	That means, the notion of continuity defined in Definition \ref{def:field_continuity} is independent of the choice of the Cauchy surface $\Sigma$ for $m$ being an element of a compact set.
\end{lemma}
\begin{proof}
	Let $\Sigma, \Sigma'$ be Cauchy surfaces and $\left\{ f_m \right\}_m$ be a family in $\BUOmega$.\\
In this proof we will make the dependence on the Cauchy surfaces $\Sigma, \Sigma'$ of the maps introduced so far explicit, that is, we will write $\kappa_m^{(\Sigma)}, \kappa_m^{(\Sigma')}$ for the map that maps to initial data on $\Sigma$ and $\Sigma'$ respectively.\par
1.) The $\Rightarrow$-direction: First, we want to show that if the map $m \mapsto \big[  K_m^{(\Sigma)}(f_m) \big]_\sim^{\text{CCR},\Sigma}$ is continuous, then also $m \mapsto \big[  K_m^{(\Sigma')}(f_m) \big]_\sim^{\text{CCR},\Sigma'}$ is continuous. Using the homeomorphisms defined in Section \ref{sec:field-algebra-topology}, recall for example the diagrammatic overview in Diagram \ref{dia:final_structure}, we identify
\begin{align}
	m \mapsto \big[  K_m^{(\Sigma')}(f_m) \big]_\sim^{\text{CCR},\Sigma'} = \Big( \Lambda^{(\Sigma')}_m \comp \left( \Lambda^{(\Sigma)}_m \right)^{-1} \Big) \Big(  \big[ K_m^{(\Sigma)}  (f_m)\big]_\sim^{\text{CCR}, \Sigma}\Big) \formspace.
\end{align}
We split this up into the map
\begin{align}
(m,m') \mapsto \Big( \Lambda^{(\Sigma')}_m \comp \left( \Lambda^{(\Sigma)}_m \right)^{-1} \Big) \Big(  \big[ K_{m'}^{(\Sigma)}  (f_{m'})\big]_\sim^{\text{CCR}, \Sigma}\Big)
\end{align}
and show that it is continuous both in $m$ and $m'$ using the assumption that $m \mapsto  \big[ K_m^{(\Sigma)}(f_m) \big]_\sim^{\text{CCR},\Sigma} $ is continuous. We will then use the Banach-Steinhaus theorem to show joint continuity.\par
a) Let $m$ be fixed.\\
By assumption, $m' \mapsto  \big[ K_{m'}^{(\Sigma)}(f_{m'}) \big]_\sim^{\text{CCR},\Sigma} $ is continuous. Furthermore, both $\Lambda^{(\Sigma')}_{m}$ and $\left( \Lambda^{(\Sigma)}_{m} \right)^{-1}$ are continuous for a fixed $m$ since they were shown to be homeomorphisms, see Lemma \ref{lem:field-algebra-homeomorphism}.\\
Consequently, the map
\begin{align}
	m' \mapsto  \Big( \Lambda^{(\Sigma')}_m \comp \left( \Lambda^{(\Sigma)}_m \right)^{-1} \Big) \Big(  \big[ K_{m'}^{(\Sigma)}  (f_{m'})\big]_\sim^{\text{CCR}, \Sigma}\Big)
\end{align}
is continuous for fixed $m$.\par
b) Now assume $m'$ to be fixed. \\
We identify $\big[ K_{m'}^{(\Sigma)}(f_{m'}) \big]_\sim^{\text{CCR},\Sigma}= \big[ \psi \big]_\sim^{\text{CCR},\Sigma}$ for some $\psi \in \BU{\left( \Dzs \right)}$. We will let that initial data propagate with mass $m$ to the Cauchy surface $\Sigma'$ by
\begin{align}
	[\psi]_\sim^{\text{CCR},\Sigma} \mapsto \Big( \Lambda^{(\Sigma')}_{m} \comp \big( \Lambda^{(\Sigma)}_{m} \big)^{-1} \Big) \big([\psi]_\sim^{\text{CCR},\Sigma}\big) \formspace.
	\end{align}
Recall that we have explicitly constructed the homeomorphisms and their inverses in Section \ref{sec:field-algebra-topology} using the map $\vartheta_m^{(\Sigma)} : \Dzs  \to \Omega^1_0(\M) $ which mapped initial data to a compactly supported test one-form $F$ such that $G_m F$ is a solution to Proca's equation with these data. This map was shown to be continuous. Furthermore, $\vartheta_m^{(\Sigma)}$ is continuous in $m$:\\
For fixed initial data $\psi= (\varphi, \pi) \in \Dzs$, a solution $A_{m} \in \Omega^1(\M)$ to $(\delta d + {m}^2)A_{m} =0$ with the given data was uniquely specified by, see Theorem \ref{thm:solution_proca_unconstrained}:
\begin{align}
	\langle A_{m}, F \rangle_\M =  \langle \varphi, \rhod G_{m} F \rangle_\Sigma - \langle \pi , \rhoz G_{m} F \rangle_\Sigma
\end{align}
which by Assumption \ref{ass:propagator_continuity} depends continuously on $m$. Furthermore, we defined $\vartheta_m^{(\Sigma)} (\varphi,\pi) = - (\delta d + {m}^2) \chi A_{m}$ which is hence continuous in $m$. Thus, the map $\vartheta_m^{(\Sigma)}$ is continuous in $m$. We lifted $\vartheta_m^{(\Sigma)}$ to the map $\Theta_m^{(\Sigma)} : \BU\big( \Dzs\big) \to \BUOmega$ which is then by construction also continuous in $m$ and we find
\begin{align}
\Big( \Lambda^{(\Sigma')}_{m} \comp \big( \Lambda^{(\Sigma)}_{m} \big)^{-1} \Big) \big( [\psi]_\sim^{\text{CCR},\Sigma} \big) = \big[ \big( K_{m}^{(\Sigma')} \comp  \Theta_{m}^{(\Sigma)}\big) \big( \psi \big) \big]_\sim^{\text{CCR},\Sigma'}
\end{align}
is continuous in $m$.\par
Therefore the map $	(m,m') \mapsto \Big( \Lambda^{(\Sigma')}_m \comp \left( \Lambda^{(\Sigma)}_m \right)^{-1} \Big) \Big(  \big[ K_{m'}^{(\Sigma)}  (f_{m'})\big]_\sim^{\text{CCR}, \Sigma}\Big)$ is \emph{separately} continuous. In order to obtain the wanted result, we need to show \emph{joint} continuity. \par
We have shown that $\left\{T_{m} \right\}_{m}$, where
\begin{align}
T_{m} : {\Quotientscale{\BU(\Dzs)}{\ICCRS}} & \to {\Quotientscale{\BU(\Dzsp)}{\ICCRSP}}\\
T_{m} &=  \Lambda^{(\Sigma')}_{m} \comp \big( \Lambda^{(\Sigma)}_{m} \big)^{-1} \notag
\end{align}
 is a family of continuous linear mappings, with continuous linear inverse, which is point-wise continuous in $m$.
 Therefore, for any $a,b \in \IR^+$ the family $\left\{T_{m} \right\}_{m \in [a,b]}$ is \emph{point-wise bounded}, that is for any $\psi \in \BU(\Dzs)$ the set $\left\{ T_m([\psi]_\sim^{\text{CCR},\Sigma}) : m \in [a,b] \right\}$ is bounded.\par
Since $\BU(\Dzs)$ is barreled (see Lemma \ref{lem:BU-algebra-barreled}) and the ideal $\ICCR$ was shown to be a closed subspace in Section \ref{sec:field-algebra-topology}, the quotient ${\Quotientscale{\BU(\Dzs)}{\ICCRS}}$ is barreled (c.f. \cite[Proposition 33.1]{treves}). Using the Banach-Steinhaus theorem (see for example \cite[Theorem 33.1]{treves}), we find that $\left\{T_{m} \right\}_{m \in [a,b]}$ is equicontinuous, that is, for all $m \in [a,b]$ and all open $W \subset {\Quotientscale{\BU(\Dzsp)}{\ICCRSP}}$, there is a open $U \subset {\Quotientscale{\BU(\Dzs)}{\ICCRS}}$ such that $T_m (U) \subset W$.\\ From this, joint continuity follows:
Since we have shown that $m \mapsto \left[ K_m^{(\Sigma)} (f_m)\right] \eqqcolon \tau(m)$ is continuous, we find a open $V \subset [a,b]$ such that $\tau (V) \subset U$. Ergo, we find that for all open $W \subset {\Quotientscale{\BU(\Dzs)}{\ICCRS}}$, there is a open $V \subset [a,b]$ such that $(T \comp \tau)(V) \subset W$. Therefore, the map $m \mapsto (T \comp \tau)(m) = \Big( \Lambda^{(\Sigma')}_m \comp \left( \Lambda^{(\Sigma)}_m \right)^{-1} \Big) \Big(  \big[ K_m^{(\Sigma)}  (f_m)\big]_\sim^{\text{CCR}, \Sigma}\Big)$ is continuous for $m \in [a,b]$. \par
2.) The $\Leftarrow$-direction: Using that $\Lambda^{(\Sigma)}_m$ and $\Lambda^{(\Sigma)}_m$ are homeomorphisms, the map $T_m$ introduced above possesses a continuous linear inverse which depends continuously on $m$, hence, the above proof works analogously in the other direction, interchanging $\Sigma$ and $\Sigma'$.
\end{proof}
With this notion of continuity at our disposal, we are ready to investigate one of the main questions of this thesis: Does the zero mass limit of the quantum Proca field theory in curved spacetimes exist for the case $j=0$? To phrase this more specifically:
\begin{center}\textit{
		Let $j=0$. For which $F \in \Omega^1_0(\M)$, if any, does the limit\\[2mm] $\lim\limits_{m\to 0^+} \big( \phi_{m,0}(F) \big)$\\[2mm] exist with the notion of continuity defined in Definition \ref{def:field_continuity}?}
\end{center}
Using the above notion of continuity, the question of interest is if the limit
\begin{align}
	\lim\limits_{m \to 0^+} \phi_{m,0}(F) = \lim\limits_{m \to 0^+} \big[\big( 0 , \kappa_m(F) , 0 , 0, \dots\big)\big]^\text{CCR}_\sim
	\end{align}
exists. We recall the definitions
\begin{align}
\kappa_m(F) = (\rhoz G_m F, \rhod G_m F)
\end{align} and
\begin{align}
	G_m F = \left( \frac{d \delta}{m^2} + 1\right) E_m F \formspace.
\end{align}
In order to precisely answer this question, we find that the existence of the desired limit is equivalent to the existence of the corresponding classical limit as stated in the following lemma:
\begin{lemma}\label{lem:limit_existence_quantum_equivalence}
	Let $j=0$, $F \in \Omega^1_0(\M)$ be fixed and assume Assumption \ref{ass:propagator_continuity} holds. The following statements are equivalent:
	\begin{enumerate}
		\item {The limit $\lim\limits_{m \to 0^+} \phi_{m,0}(F) = \lim\limits_{m \to 0^+} \big[\big( 0 , \kappa_m(F) , 0 , 0, \dots\big)\big]^\text{CCR}_\sim $ exists. \\}
        \item {The limit $\lim\limits_{m \to 0^+} \big( 0 , \kappa_m(F) , 0 , 0, \dots\big)$ exists. \\}
		\item {The limit $\lim\limits_{m \to 0^+} G_m F $ exists.  \\}
		\item {The limit $\lim\limits_{m \to 0^+} \frac{1}{m^2}d \delta E_m F $ exists. }
	\end{enumerate}
\end{lemma}
\begin{proof}
	The nontrivial part is to show the equivalence of the first two statements:\par
	1.) a) (ii) $\implies$ (i) is obvious, since $[\cdot]_\sim^\text{CCR}$ is continuous and does not depend on the mass. So if $\big( 0 , \kappa_m(F) , 0 , 0, \dots\big)$ is continuous, so is $\big[\big( 0 , \kappa_m(F) , 0 , 0, \dots\big)\big]^\text{CCR}_\sim$. This yields the existence of the corresponding limit.\par
	b) (i) $\implies$ (ii) is highly non-trivial and most of the work needed to show this is put in the appendix in form of Lemma \ref{lem:symmetrization-of-fields}. Assume that $\big[\big( 0 , \kappa_m(F) , 0 , 0, \dots\big)\big]^\text{CCR}_\sim$ specifies a continuous family of algebra elements and that the corresponding zero mass limit exists. This implies that there is a family $\left\{\tilde{g}_m\right\}_m \subset \ICCR$ such that
	\begin{align}
     \big( 0 , \kappa_m(F) , 0 , 0, \dots\big) + \tilde{g}_m
	\end{align}
	specifies a continuous family. Note that this does not imply that neither $\tilde{g}_m$ nor $\big( 0 , \kappa_m(F) , 0 , 0, \dots\big) $ specify continuous families by themselves. In this context, continuity is as always assumed to be continuity in $m$. By Lemma \ref{lem:symmetrization-of-fields}, we can decompose this continuous family into a sum of a continuous family of symmetric elements and a continuous family of elements in the ideal $\ICCR$, that is,
	\begin{align}
     \big( 0 , \kappa_m(F) , 0 , 0, \dots\big) + \tilde{g}_m = f_{m,\text{sym}} + g_m
	\end{align}
	where $g_m \in \ICCR$ and both $g_m$ and  $f_{m,\text{sym}}$ specify continuous families of algebra elements. An element in the BU-algebra of initial data is called symmetric if it is symmetric in every degree. Therefore,
	\begin{align}
\big( 0 , \kappa_m(F) , 0 , 0, \dots\big) - f_{m,\text{sym}} =  ( g_m -\tilde{g}_m)
	\end{align}
	specifies a symmetric element, since the sum, respectively the difference, of two symmetric elements is again symmetric. We conclude that $(g_m - \tilde{g}_m) \in \ICCR$ is symmetric.\\
	By definition, we can write an element in the ideal $\ICCR$ as a finite sum
	\begin{align}
		g_m - \tilde{g_m} = \sum_{i=1}^{k} h_i \cdot \big( -\i \mathcal{G}_m(\psi_i, \psi'_i) , 0 , \psi_i \otimes \psi'_i - \psi'_i \otimes \psi_i , 0 , 0 ,\dots  \big) \cdot \tilde{h}_i \formspace
	\end{align}
	for some $k \in \IN$, $h_i, \tilde{h}_i \in \BU(\Dzs)$ and $\psi_i, \psi'_i \in \Dzs$. Writing $\psi_i = (\varphi_i , \pi_i)$ we have used for shorthand notation $\Gm{\psi_i}{ \psi'_i} = \langle \pi_i , \varphi'_i \rangle_\Sigma - \langle \varphi_i , \pi'_i \rangle_\Sigma$. Furthermore, we have dropped the mass dependence of the summands for clarity, as we are not using it in the following. Since $g_m - \tilde{g}_m$ is symmetric, in particular its highest degree $(g_m - \tilde{g}_m)^{(N)}$, for some $N \in \IN$, is symmetric. We can write the highest degree explicitly, using the above representation,
	\begin{align}
(g_m - \tilde{g}_m)^{(N)} = \sum_{i=1}^{k} h_i^{(N_i)}\, \phi_i^{(2)}\, \tilde{h}_i^{(\tilde{N}_i)} \formspace,
	\end{align}
	such that $N_i + \tilde{N}_i = N-2$ for every $i = 1,2,\dots,k$. Additionally, we have introduced $\phi_i^{(2)} =  \psi_i \otimes \psi'_i - \psi'_i \otimes \psi_i $ as yet another shorthand notation. As usual, $h_i^{(N_i)} \in (\Dzs)^{\otimes N_i}$ denotes the degree-$N_i$ part of $h_i \in \BU(\Dzs)$. By definition, the $\phi^{(2)}_i$'s are fully anti-symmetric, that is, $\phi^{(2)}_i(p,q) = - \phi^{(2)}_i(q,p)$ for any $p,q \in \Sigma$. Since $(g_m - \tilde{g}_m)^{(N)}$ is symmetric, it is invariant under arbitrary permutations of its arguments, that is, for any permutation $\sigma$ of $\{ 1,2, \dots , N\}$ it holds
	\begin{align}
		(g_m - \tilde{g}_m)^{(N)}(p_1,p_2,\dots,p_N) = (g_m - \tilde{g}_m)^{(N)}(p_{\sigma(1)},p_{\sigma(2)},\dots,p_{\sigma(N)}) \formspace.
	\end{align}
	Introducing the permutation operator $P_\sigma : (\Dzs)^{\otimes N} \to (\Dzs)^{\otimes N}$ by defining $(P_\sigma f^{(N)})(p_1,p_2,\dots,p_N) = f^{(N)}(p_{\sigma(1)},p_{\sigma(2)},\dots,p_{\sigma(N)}) $ we conclude in a short hand notation
	\begin{align}
(g_m - \tilde{g}_m)^{(N)}
&= \frac{1}{N!}\sum\limits_\sigma P_\sigma (g_m - \tilde{g}_m)^{(N)} \notag\\
&= \frac{1}{N!}\sum\limits_\sigma P_\sigma  \sum_{i=1}^{k} h_i^{(N_i)}\, \phi_i^{(2)}\, \tilde{h}_i^{(\tilde{N}_i)}\notag \\
&= \frac{1}{N!}  \sum_{i=1}^{k}  \sum\limits_\sigma P_\sigma \big(  h_i^{(N_i)}\, \phi_i^{(2)}\, \tilde{h}_i^{(\tilde{N}_i)} \big) \formspace.
	\end{align}
	In the last step we used that $P_\sigma$ is by construction linear, and that, since both sums are finite, we can rearrange the summands in the desired way. We find for any (fixed) $i=1,2,\dots,k$
	\begin{align}
\sum\limits_\sigma P_\sigma \big(  h_i^{(N_i)}\, \phi_i^{(2)}\, \tilde{h}_i^{(\tilde{N}_i)} \big) = 0
	\end{align}
	since to every permutation $\sigma$ there exist another permutation that differs from $\sigma$ by a swap of the $N_{i+1}$ and $N_{i+2}$ index. Because $\phi^{(2)}_i$ is fully anti-symmetric, the corresponding terms will cancel. Hence, all the summands in the sum over all permutations will cancel pair-wise, using the anti-symmetry of $\phi^{(2)}$. We can thus conclude
	\begin{align}
(g_m - \tilde{g}_m)^{(N)} = 0 \formspace,
	\end{align}
	which is a contradiction, because we have assumed the degree $N$ to be the highest degree of $(g_m - \tilde{g}_m)$. Therefore, it must hold $(g_m - \tilde{g}_m) = 0$. Concluding, we have shown that the only fully symmetric element in the ideal $\ICCR$ is the zero element. From this we can straightforwardly conclude
	\begin{align}
		\tilde{g}_m &= g_m &\text{specifies continuous family,} \\
		\implies  \big( 0, \kappa_m(F),0,0,\dots\big) &= f_{m,\text{sym}} &\text{specifies continuous family,}
	\end{align}
	which completes the proof of 1 b).\par
	2.) The equivalence of (ii) and (iii) was mostly shown in Lemma \ref{lem:limit_existence_classical_equivalence}.\\
	a) (iii) $\implies$ (ii):\\
	Assume $\lim\limits_{m \to 0^+} G_m F $ exists. Then, the limit $\lim\limits_{m \to 0^+} \kappa_m(F) = \lim\limits_{m \to 0^+} (\rhod G_m F, \rhoz G_m F)$ exists, since convergence in the direct sum topology on $\Omega^1_0(\Sigma) \oplus \Omega^1_0(\Sigma)$ means convergence in every component and the operators $\rho_{(\cdot)}$ are continuous. Then, since convergence in the BU-algebra means convergence in every degree, clearly, the limit $\lim\limits_{m \to 0^+} \big( 0 , \kappa_m(F),0,0,\dots \big)$ exists.\par
	b) (ii) $\implies$ (iii):\\ Assume that $\lim\limits_{m \to 0^+} \big( 0 , \kappa_m(F),0,0,\dots \big)$ exists.
	Since convergence in the BU-algebra means convergence in every degree, we conclude that $\lim\limits_{m \to 0^+} \kappa_m(F)$ exists. Recall, $\kappa_m(F) = (\rhoz G_m F , \rhod G_m F) \in \Dzs$, and also here convergence means convergence in every component, hence $\lim\limits_{m \to 0^+}  \rhoz G_m F $ and $\lim\limits_{m \to 0^+}  \rhod G_m F$ exist. According to Lemma \ref{lem:limit_existence_classical_equivalence} this yields that $\lim\limits_{m \to 0^+} G_m F $ exists.\par
	3.) The equivalence of (iii) and (iv) was already proven in the classical chapter, see Lemma \ref{lem:limit_existence_classical_equivalence}.
	\end{proof}
With this lemma, we have found quite a remarkable result: The existence of the zero mass limit of the quantum Proca field theory in curved spacetimes is purely determined by its classical properties!\\
With this result, we can now quite easily find an answer to our original question, that is, explicitly find test functions $F \in \Omega^1_0(\M)$ such that the limit $\lim\limits_{m \to 0^+} \phi_{m,0}(F)$, or equivalently, $\lim\limits_{m \to 0^+} \frac{1}{m^2}d \delta E_m F $ exists. We have already investigated this in the classical section, therefore, we directly find the following lemma which we will later tighten to the main result as stated in Theorem \ref{thm:limit_existence_sourcefree}.
\begin{lemma}\label{lem:mass-zero-limit-existence_weak}
Let $F\in \Omega^1_0(\M)$ and $j=0$. Then,
\begin{center}
the limit $\lim\limits_{m \to 0^+}  \phi_{m,0}(F)$ exists if and only if $F = F' + F''$,\\
\end{center}
where $F', F'' \in \Omega^1_0(\M)$ such that $dF' = 0 = \delta F''$.
\end{lemma}
\begin{proof}
	According to Lemma \ref{lem:limit_existence_quantum_equivalence}, the desired limit $\lim\limits_{m \to 0^+}  \phi_{m,0}(F)$ exists if and only if $\lim\limits_{m \to 0^+} \frac{1}{m^2}d \delta E_m F $ exists. Using Lemma \ref{lem:limit_existence_classical_equivalence} together with Lemma \ref{lem:mass-zero-limit-existence-classical_weak} we find that this is equivalent to $F$ being a sum of a closed and a co-closed test one-form.
\end{proof}
From a formal point of view we have completely classified those test one-forms for which the zero mass limit exists. But it turns out that, just like in the classical case, we can tighten the result even more, by observing that closed test one-forms $F \in \Omega^1_0(\M)$, such that $dF=0$, do not contribute to the observables $\phi_{m,0}(F)$ for the source free case. We therefore may restrict the class of test-functions that we smear the fields $\phi_{m,0}$ with to the test one-forms modulo closed test one-forms, analogously to the classical case. This yields the final result of this section:
\begin{theorem}[Existence of the zero mass limit in the source free case]\label{thm:limit_existence_sourcefree}
Let $j=0$ and $F, F' \in \Omega^1_0(\M)$ such that $[F] = [F']$, that is, there is a $\chi \in \Omega^1_{0,d}(\M)$ such that $F = F' + \chi$.  Then,
	\begin{align*}
	\kappa_m(F) &= \kappa_m (F') \eqqcolon \kappa_m ([F]) \quad\text{and hence}\\
	\phi_{m,0}(F) &= \phi_{m,0}(F') \eqqcolon \phi_{m,0}([F]) \formspace,
	\end{align*}
and
\begin{center}
	the limit $\lim\limits_{m \to 0^+}  \phi_{m,0}([F]) \coloneqq \lim\limits_{m \to 0^+} \big[\big( 0 , \kappa_m([F]) , 0 , 0, \dots\big)\big]^\text{CCR}_\sim $ exists\\[2mm]  if and only if there exists a representative $\tilde{F}$ of $[F]$ with ${\delta{\tilde{F}}} = 0 $.
\end{center}
Furthermore, the existence of the limit is independent of the choice of the Cauchy surface $\Sigma$.
\end{theorem}
\begin{proof}
	Recall from Theorem \ref{thm:limit_existence_sourcefree_classical} that for closed one-forms $F \in \Omega^1_{0,d}(M)$ it holds $G_m F = 0$, which directly yields $\kappa_m(F) = 0$ and hence $\phi_{m,0}(F) = 0$. Ergo, due to the linearity of the fields, two test one-forms that differ by a closed test one-form give rise to the same physical observable.
	Therefore, we can, without losing any observables, divide out the test one-forms that are closed. \\
	By Lemma \ref{lem:mass-zero-limit-existence_weak}  we know that the limit exists if and only if $F$ is a sum of a closed and a co-closed test one-form.
	Hence, the proof of the existence of the limit $\lim\limits_{m \to 0^+}  \phi_{m,0}([F])$ follows in complete analogy to the proof of Theorem \ref{thm:limit_existence_sourcefree_classical}, that is, basically showing that 	\begin{align}
	\frac{\Omega^1_{0,d}(\M) + \Omega^1_{0,\delta}(\M)}{\Omega^1_{0,d}(\M)} \cong \Omega^1_{0,\delta}(\M) \formspace.
	\end{align}
For the second part, let $F \in \Omega^1_{0,\delta}(\M)$ be a co-closed test one-form. Then it is clear the corresponding limit $\lim\limits_{m \to 0^+} \big[\big( 0 , \kappa_m(F) , 0 , 0, \dots\big)\big]^\text{CCR}_\sim $ exists regardless of the choice of the Cauchy surface, since $G_m F = E_m F$ and the $\rho^{(\Sigma)}_{(\cdot)}$ are continuous for any Cauchy surface $\Sigma$. Furthermore, as already argued, for a fixed $F \in \Omega^1_0(\M)$, continuity of the map $m \mapsto \phi_m(F) = \big[\big( 0 , \kappa_m(F) , 0 , 0, \dots\big)\big]^\text{CCR}_\sim$ is independent of the choice of the Cauchy surface since for two Cauchy surfaces $\Sigma, \Sigma'$ both $m \mapsto \kappa_m^{(\Sigma)}(F)$ and $m \mapsto \kappa_m^{(\Sigma')}(F)$ are continuous. Therefore, the statements in Lemma \ref{lem:limit_existence_quantum_equivalence} are independent of the choice of the Cauchy surface $\Sigma$.
\end{proof}
We find that it is sufficient as well as necessary for the mass zero limit to exist in the source free case to restrict to co-closed test one-forms. And just as in the classical case, as discussed in Section \ref{sec:zero-mass-limit-existence-classical-vanishing-source}, this implements the gauge equivalence of the Maxwell theory at the classical level. But also in the quantization of the vector potential of the Maxwell theory, restricting to co-closed test one-forms $F$ is a way to implement the gauge equivalence is the theory as presented in \cite{Sanders} or \cite{fewster_pfenning_quantum_weak}. Hence also in the quantum case, the limit exists only if we implement the gauge equivalence! Before we investigate the algebraic structures of the zero mass limit fields, we discuss the existence of the limit in the general case, including external sources.
\subsubsection{Existence of the limit in the general case with current}\label{sec:mass-dependence-j-general}
In this next step we would like to include non vanishing external currents $j \neq 0$. So far, most of the quantum investigation, in particular the construction of a notion of continuity of the theory with respect to the mass, has been done for the source-free case $j = 0$. So before we can investigate a zero mass limit for the general theory, we need to again find a notion of continuity for the general field algebra with respect to the mass $m$. Fortunately, we can make use of the existent notion, since we have already constructed a homeomorphism of the general field algebra $\BUmj$ and the BU-algebra of initial data $\Quotientscale{\BU\big(\Dzs\big)}{\ICCR}$ in Section \ref{sec:field-algebra-topology}. This homeomorphism was constructed using a classical solution $\varphi_{m,j}$ to the inhomogeneous Proca equation $(\delta d + m^2) \varphi_{m,j} = j$. Recalling the algebraic structure (see Diagram \ref{dia:final_structure}) we can formulate a natural notion of continuity analogously to the source free case.
\begin{definition}[Continuity of a family of observables with respect to the mass]\label{def:field_continuity_general}
	Let $j \in \Omega^1(\M)$ be fixed. Let $K_m$, $\Gamma_{m,j,\varphi}$, $\Psi_{m,j,\varphi}$ and $\Lambda_{m}$ be defined as  in Section \ref{sec:field-algebra-topology}, and let $\left\{ \varphi_{m,j} \right\}_m$ specify a continuous family of classical solutions to the inhomogeneous Proca equation $(\delta d + m^2) \varphi_{m,j} = j$.\\
	We call a function
	\begin{align}
	\eta : \IR_+ &\to \bigcup_m \BUmj\\
	m &\mapsto \eta_m  \in \BUmj\notag
	\end{align}
	continuous if the map
	\begin{align}
	\widetilde{\eta} : \IR_+ &\to {\Quotientscale{\BU\big( \Dzs \big) }{\ICCR}} \\
	m &\mapsto \big( \Lambda_m \comp  \Psi_{m,j,\varphi}^{-1} \big) (\eta_m)\notag
	\end{align}
	is continuous.\\
	Equivalently, one can identify $\eta_m = [f_m]_{m,j}$ for some family $\left\{ f_m \right\}_m \subset \BUOmega$ and define the map $\eta$ to be continuous if the map
	\begin{align}
	\hat{\eta} : \IR_+ &\to {\Quotientscale{\BU\big( \Dzs \big) }{\ICCR}} \\
	m &\mapsto \big[  \big( K_m \comp  \Gamma_{m,j,\varphi}^{-1} \big) (f_m) \big]_\sim^\text{CCR}\notag
	\end{align}
	is continuous. The latter will again be the more practical definition.
\end{definition}
First, we again would like to check if, with that notion of continuity, the map $m \mapsto \phi_{m,j}(F)$ is continuous for a fixed $F \in \Omega^1_0(\M).$
Hence, we need to check whether $m \mapsto \big[ \big( K_m \comp \Gamma^{-1}_{m,j,\varphi}\big)\big( (0,F,0,0,\dots) \big) \big]_\sim^\text{CCR}$ is continuous. Assume we have fixed a classical solution $\varphi_{m,j} \in \Omega^1(\M)$ by specifying vanishing initial data on the Cauchy surface $\Sigma$. Shortly, we will show that the notion of continuity does not depend on the choice of the classical solution $\varphi_{m,j}$, so we can indeed chose arbitrary initial data without loss of generality. According to Theorem \ref{thm:solution_proca_unconstrained}, the solution is then specified by
\begin{align}
	\langle \varphi_{m,j} , F \rangle = \sum\limits_\pm \langle j , G_m^\mp F \rangle_{\Sigma^\pm} \formspace.
\end{align}
For $m > 0$, this clearly depends continuously on $m$. We therefore find
\begin{align}
\left( K_m \comp \Gamma^{-1}_{m,j,\varphi}\right)\big( (0,F,0,0,\dots) \big)
&= K_m \Big( \big(   \sum\limits_\pm \langle j , G_m^\mp F \rangle_{\Sigma^\pm} , F,0,0,\dots\big)  \Big) \notag\\
&=  \big(   \sum\limits_\pm \langle j , G_m^\mp F \rangle_{\Sigma^\pm} , \kappa_m(F),0,0,\dots\big)  \formspace,
\end{align}
which again is continuous in $m$. Since $[\cdot]_\sim^\text{CCR}$ does not depend on the mass, we find that $m \mapsto \phi_{m,j}(F)$ is continuous. Therefore, the notion of continuity seems appropriate!\\
We have yet to check that the notion is independent of the choice of the Cauchy surface $\Sigma$ and the classical solution $\varphi_{m,j}$.
\begin{lemma}
	Let $\Sigma$, $\Sigma'$ be Cauchy surfaces, $\left\{\varphi_{m,j} \right\}_m$ specify a continuous family of classical solutions to the inhomogeneous Proca equation and $a,b \in \IR^+$ be arbitrary. Let $\left\{ f_m\right\}_m$ be a family in $\BUOmega$. Then
	\begin{align*}
	\eta^{(\Sigma)} : [ a , b] \to {\Quotientscale{\BU(\Dzs  )}{\ICCR}}\hspace{1mm}, \quad m \mapsto \big[ \big( K^{(\Sigma)}_m \comp \Gamma^{-1}_{m,j,\varphi} \big) \big(f_m\big) \big]_\sim^{\text{CCR}, \Sigma}
	\end{align*}
	is continuous if and only if
	\begin{align*}
	\eta^{(\Sigma')} : [ a , b] \to {\Quotientscale{\BU(\Dzsp)}{\ICCR}}, \quad m \mapsto\big[ \big( K^{(\Sigma')}_m \comp \Gamma^{-1}_{m,j,\varphi} \big) \big(f_m\big) \big]_\sim^{\text{CCR}, \Sigma'}
	\end{align*}
	is continuous.
	That means, the notion of continuity defined in Definition \ref{def:field_continuity_general} is independent of the choice of the Cauchy surface $\Sigma$ for $m$ being an element of a compact set.
	Furthermore, the notion of continuity is independent of the choice of the classical solutions $\left\{\varphi_{m,j} \right\}_m$.
\end{lemma}
\begin{proof}
	Note that both $K_m$ and $\Gamma_{m,j,\varphi}$ are defined as BU-algebra-homomorphisms. Furthermore, since they act on homogeneous elements of degree zero as the identity map, the statement clearly holds for those elements. Since every BU-algebra element can be decomposed into a sum of products of homogeneous degree zero and degree one elements, it suffices to prove the statement for homogeneous elements of degree one.
	Let $\left\{ F_m\right\}_m$ be a family in $\Omega^1_0(\M)$ specifying a family $\left\{ f_m\right\}_m= \left\{(0, F_m,0,0,\dots)\right\}_m$ of BU-algebra elements. Let $a,b \in \IR^+$ and $\Sigma, \Sigma'$ be Cauchy surfaces. The map $\eta: [a,b] \to \BUmj$, $\quad \eta : m\to \eta_m = [f_m]_{m,j}$ is assumed to be continuous in the sense that
\begin{align}
	m \mapsto  \left[ \left( K^{(\Sigma)}_m \comp \Gamma^{-1}_{m,j,\varphi} \right) (f_m) \right]_\sim^{\text{CCR}, \Sigma}
\end{align}
is continuous. We calculate
\begin{align}
\big[ \big( K^{(\Sigma)}_m \comp &\Gamma^{-1}_{m,j,\varphi} \big) \big( (0,F_m,0,0,\dots) \big) \big]_\sim^{\text{CCR}, \Sigma} \notag\\
&=  \left[  K^{(\Sigma)}_m  \big( (-\langle \varphi_{m,j},F_m\rangle_\M,F_m,0,0,\dots) \big) \right]_\sim^{\text{CCR}, \Sigma} \notag\\
&=  \langle \varphi_{m,j} , F_m \rangle_\M \cdot \big[  \mathbbm{1}_{\BU(\Dzs)} \big]_\sim^{\text{CCR}, \Sigma}   + \left[ \big(0,\kappa_m(F_m),0,0,\dots\big) \right]_\sim^{\text{CCR}, \Sigma} \formspace.
\end{align}
We note that, using Lemma \ref{lem:continuity-independence-source-free}, the continuity of $m \mapsto  \left[ \big(0,\kappa_m(F_m),0,0,\dots\big) \right]_\sim^{\text{CCR}, \Sigma}$ is independent of the choice of the Cauchy surface $\Sigma$ for $m \in [a,b]$. Moreover, it is obvious that the continuity of $m \mapsto \langle \varphi_{m,j} , F_m \rangle_\M \cdot \big[  \mathbbm{1}_{\BU(\Dzs)} \big]_\sim^{\text{CCR}, \Sigma} $ is also independent of the choice of $\Sigma$ since it is determined by the continuity of $m \mapsto \langle \varphi_{m,j} , F_m \rangle_\M$ which is independent of the Cauchy surface $\Sigma$. Ergo, the defined notion of continuity is independent of the choice of the Cauchy surface $\Sigma$ as long as $m \in [a,b]$.\par
Now, we can check if the notion depends on the choice of the classical solution $\varphi_{m,j}$. \\
Let $\left\{\varphi_{m,j} \right\}_m$ and $\left\{\varphi'_{m,j} \right\}_m$ specify continuous families of classical solutions to the inhomogeneous Proca equation and assume that
\begin{align}
	m &\mapsto \left[ \left( K^{(\Sigma)}_m \comp \Gamma^{-1}_{m,j,\varphi} \right) (f_m) \right]_\sim^{\text{CCR}, \Sigma}
	=\left[   \big( \langle \varphi_{m,j} , F_m \rangle_\M ,\kappa_m(F_m),0,0,\dots\big) \right]_\sim^{\text{CCR}, \Sigma}
\end{align}
is continuous. Since trivially$\big( \langle \varphi_{m,j} , F_m \rangle_\M ,\kappa_m(F_m),0,0,\dots\big)$ is a symmetric element, this is equivalent to
\begin{align}
	m \mapsto   \big( \langle \varphi_{m,j} , F_m \rangle_\M ,\kappa_m(F_m),0,0,\dots\big)
\end{align}
being continuous, following the proof of Lemma \ref{lem:limit_existence_quantum_equivalence} part 1 b) in complete analogy. Since in the topology of the BU-algebra continuity is equivalent to continuity in every degree and recalling $\kappa_m(F_m) = (\rhoz G_m F_m, \rhod G_m F_m)$, we find that
\begin{align}
	m &\mapsto \langle \varphi_{m,j} , F_m \rangle_\M \\
	m &\mapsto \rhoz G_m F_m \quad\text{and}\\
	m &\mapsto \rhod G_m F_m
\end{align}
are continuous. Let $\phi_m,\phi_m'$, $\pi_m,\pi_m'$ denote the initial data of $\varphi_{m,j}$ and $\varphi'_{m,j}$ respectively. By assumption these are continuous in $m$. We calculate
\begin{align}
	\langle \varphi'_{m,j} , F_m \rangle_\M
	&= \langle \varphi'_{m,j} - \varphi_{m,j} , F_m \rangle_\M + \langle \varphi_{m,j} , F_m \rangle_\M \\
	&= \langle \phi'_m - \phi_m, \rhod G_m F_m\rangle_\Sigma - \langle \pi'_m - \pi_m, \rhoz G_m F_m\rangle_\Sigma + \langle \varphi_{m,j} , F_m \rangle_\M \notag
\end{align}
using Theorem \ref{thm:solution_proca_unconstrained}. Using the assumptions on continuity stated above we directly find that $m \mapsto \langle \varphi'_{m,j} , F_m \rangle_\M $ is continuous and hence
\begin{align}
	m \mapsto \left[   \big( \langle \varphi'_{m,j} , F_m \rangle_\M ,\kappa_m(F_m),0,0,\dots\big) \right]_\sim^{\text{CCR}, \Sigma}
\end{align}
is continuous. Therefore, the notion of continuity is independent of the choice of the classical solution $\varphi_{m,j}$.
This completes the proof.
\end{proof}
Choosing vanishing initial data of the classical solution $\varphi_{m,j}$ we will therefore, as in the source free case, from now on identify a field $\phi_{m,j}$ with its initial data formulation $\phi_{m,j}(F) =  \big[\big( \sum_{\pm}\langle j , G_m^\mp F \rangle_{\Sigma^\pm} , \kappa_m(F) , 0 , 0, \dots\big)\big]^\text{CCR}_\sim$. Now, with the work done for the quantum source free case and the general classical case, the result for the general quantum case follows easily:
\begin{theorem}[Existence of the zero mass limit in the general case with sources]\label{thm:limit_existence_general}
Let $F \in \Omega^1_0(\M)$ and $j \in \Omega^1(\M)$. Then,
\begin{center}
	the limit $\lim\limits_{m \to 0^+}  \phi_{m,j}(F) \coloneqq \lim\limits_{m \to 0^+} \big[\big( \sum_{\pm}\langle j , G_m^\mp F \rangle_{\Sigma^\pm} , \kappa_m(F) , 0 , 0, \dots\big)\big]^\text{CCR}_\sim $ exists\\[2.5mm] if and only if ${\delta F} = 0 $.
\end{center}
Furthermore, the existence of the limit is independent of the choice of the Cauchy surface $\Sigma$.
\end{theorem}
\begin{proof}
	First, we note that since the notion of continuity is independent of the choice of the classical solution $\varphi_{m,j}$ of Proca's equation, we can choose $\varphi_{m,j}$ such that it has vanishing initial data on the Cauchy surface $\Sigma$. It then holds for all test one-forms $F$, using Theorem \ref{thm:solution_proca_unconstrained}, that $\langle \varphi_{m,j} , F \rangle_\M = \sum_{\pm}\langle j , G_m^\mp F \rangle_{\Sigma^\pm}$. Therefore, by the notion of continuity defined above, the zero mass limit is identified with
	\begin{align}
	\lim\limits_{m \to 0^+}  \phi_{m,j}(F)
	&= \lim\limits_{m \to 0^+} \big[\big( K_m \comp \Gamma^{-1}_{m,j,\varphi} \big)\big(  (0,F,0,0,\dots)   \big)\big]^\text{CCR}_\sim \notag \\
	&= \lim\limits_{m \to 0^+} \big[K_m \big(  (\sum_{\pm} \langle \varphi_{m,j} , F \rangle_\M,F,0,0,\dots)   \big)\big]^\text{CCR}_\sim 	    \notag\\
	&= \lim\limits_{m \to 0^+} \big[\big( \sum_{\pm} \langle j , G_m^\mp F \rangle_{\Sigma^\pm} , \kappa_m(F) , 0 , 0, \dots\big)\big]^\text{CCR}_\sim  \formspace.
	\end{align}
Now, using the same argument as presented in the proof of Lemma \ref{lem:limit_existence_quantum_equivalence}, because $\big( \sum_{\pm} \langle j , G_m^\mp F \rangle_{\Sigma^\pm} , \kappa_m(F) , 0 , 0, \dots\big)$ is symmetric, the limit
\begin{align}
\lim\limits_{m \to 0^+} \big[\big( \sum_{\pm} \langle j , G_m^\mp F \rangle_{\Sigma^\pm} , \kappa_m(F) , 0 , 0, \dots\big)\big]^\text{CCR}_\sim
\end{align} exists if and only the limit
\begin{align}
\lim\limits_{m \to 0^+} \big( \sum_{\pm} \langle j , G_m^\mp F \rangle_{\Sigma^\pm} , \kappa_m(F) , 0 , 0, \dots\big)
\end{align}
exists.
	Since convergence in the BU-algebra means convergence in every degree, we find that the desired limit exists if and only if the limits $\lim\limits_{m \to 0^+}\sum_{\pm} \langle j , G_m^\mp F \rangle_{\Sigma^\pm}$ and $\lim\limits_{m \to 0^+} \kappa_m(F) $ exist. Recalling Lemma \ref{lem:limit_existence_quantum_equivalence} the latter is equivalent to the existence of the limit $\lim\limits_{m \to 0^+} \frac{1}{m^2}E_m d \delta F$. Together, this corresponds exactly to the classical situation and using Theorem \ref{thm:limit_existence_general_classical} completes the proof.
\end{proof}
Just as in the classical case, also with external sources present, the limit exists if and only if we implement the gauge equivalence of Maxwell's theory as discussed in Section \ref{sec:zero-mass-limit-existence-classical-vanishing-source}.
\subsubsection{Algebra relations, dynamics and the zero mass limit}\label{sec:zero-mass-limit-quantum-algebra-relations}
Now that we have classified the existence of the zero mass limit of the quantum Proca field in curved spacetimes we want study the algebra relations of the fields in the zero mass limit. We would like to compare these with the ones obtained from the quantization of the Maxwell field.\par
Identifying the fields $\phi_{0,j}(F)$ with the zero mass limit of
\begin{align*}
\phi_{m,j}(F) = \big[\big( \sum_{\pm}\langle j , G_m^\mp F \rangle_{\Sigma^\pm} , \kappa_m(F) , 0 , 0, \dots\big)\big]^\text{CCR}_\sim \formspace,
\end{align*}
we define the algebra $\mathscr{A_0}$ as the algebra generated by $\mathbbm{1}$ and the symbols $\phi_{0,j}(F)$ for any co-closed test one-form $F$. Using the algebra relations of the fields $\phi_{m,j}$, it is clear that in the zero mass limit the fields are subject to the relations
\begin{align}
	1.)\; &\phi_{0,j}(\alpha F + \beta F') = \alpha \phi_{0,j}(F) + \beta \phi_{0,j}(F') 														 \\
	2.)\; &\phi_{0,j}(F)^* = \phi_{0,j}(\skoverline{F}\,)
\end{align}
for all $F \in \Omega^1_{0,\delta}(\M)$ and $\alpha, \beta \in \IC$, corresponding to the linearity and the real field property.
For the canonical commutation relations we calculate, $F,F' \in \Omega^1_{0,\delta}(\M)$,
\begin{align}
\gls{Ezcurly}(F,F')
&=\langle F, E_0 F' \rangle_\M \notag\\
&= \langle E_0 F' , F \rangle_\M \notag\\
&= \langle \rhoz E_0 F' , \rhod E_0 F  \rangle_\Sigma + \langle \rhodelta E_0 F' , \rhon E_0 F \rangle_\Sigma \notag \notag\\
&\phantom{=I}- \langle  \rhod E_0 F' , \rhoz E_0 F \rangle_\Sigma - \langle \rhon E_0 F', \rhodelta E_0 F \rangle_\Sigma \formspace.
\end{align}
We have used that $E_0 F'$ solves a homogeneous massless wave equation to which the solution is determined by initial data using Theorem \ref{thm:solution_wave_equation}. Using $\rhodelta E_0 F' = i^* \delta E_0 F' = i^* E_0 \delta F' = 0$ and the analogue expression for $F$, we find
\begin{align}
\Ez{F}{F'}
&= \langle \rhoz E_0 F' , \rhod E_0 F  \rangle_\Sigma - \langle  \rhod E_0 F' , \rhoz E_0 F \rangle_\Sigma \notag \\
&= \lim\limits_{m \to 0^+} \Big( \langle \rhoz E_m F' , \rhod E_m F  \rangle_\Sigma - \langle  \rhod E_m F' , \rhoz E_m F \rangle_\Sigma  \Big) \notag \\
&= \lim\limits_{m \to 0^+} \Big( \langle \rhoz G_m F' , \rhod G_m F  \rangle_\Sigma - \langle  \rhod G_m F' , \rhoz G_m F \rangle_\Sigma  \Big) \notag \\
&= \lim\limits_{m \to 0^+}  \Gm{F}{F'} \formspace.
\end{align}
Again we have used that for co-closed test one-forms $F$ it holds $G_m F = E_m F$.
Since for co-closed test one-forms $F \in \Omega^1_{0,\delta}$ the fundamental solutions $E^\pm_0$ of the massless Klein-Gordon operator are actually also fundamental solutions to Maxwell's equation, $E_0^\pm \delta d F = E_0^\pm (\delta d + d \delta) F = F$, we find that the fields in the zero mass limit are subject to the correct canonical commutation relations
\begin{align}
	\big[ \phi_{0,j}(F) ,\phi_{0,j}(F')  \big]
	&= \lim\limits_{m \to 0^+} \big[ \phi_{m,j}(F) ,\phi_{m,j}(F')  \big]  \notag\\
	&=\i \cdot \lim\limits_{m \to 0^+} \Gm{F}{F'}\cdot \mathbbm{1} \notag\\
	&= \i \, \Ez{F}{F'}\cdot \mathbbm{1}
\end{align}
for all $F,F' \in \Omega^1_{0,\delta}(\M)$. So far, this also corresponds perfectly to the relations presented in \cite[Definition 4.5]{Sanders}.
Note that now $\Ez{F}{F'}$, for $F,F' \in \Omega^1_{0,\delta}(\M)$, is in general degenerate, hence the quantum field theory associated with $\phi_{0,j}$ will in general fail to be local in the sense defined in Definition \ref{def:generally-coveriant-qftcs}. This is in detail discussed in \cite{Sanders}. \par
It remains to check the dynamics of the theory. We want to check if the fields solve Maxwell's equation in a distributional sense, that is if $\phi_{0,j}(\delta d F) = \langle j , F\rangle_\M$ holds for all $F\in \Omega^1_0(\M)$. Since it holds for all $F \in \Omega^1_0(\M)$ that $\delta d F$ is co-closed, the limit
\begin{align}
	\phi_{0,j} (\delta d F) =  \lim\limits_{m \to 0^+} \big[\big( \sum_{\pm}\langle j , G_m^\mp \delta d F \rangle_{\Sigma^\pm} , \kappa_m(\delta d F) , 0 , 0, \dots\big)\big]^\text{CCR}_\sim
\end{align}
exists for all test one-forms $F$. We use $G_m \delta d F = E_m \delta d F$ and find
\begin{align}
\phi_{0,j} (\delta d F) = \lim\limits_{m \to 0^+} \big[\big( \sum_{\pm}\langle j , E_m^\mp \delta d F \rangle_{\Sigma^\pm} , (\rhoz E_m \delta d F, \rhod E_m \delta d F) , 0 , 0, \dots\big)\big]^\text{CCR}_\sim \formspace.
\end{align}
Using $E^\pm_0 \delta d F = F - (d \delta + m^2) E^\pm_0 F$, we find
\begin{align}
\phi_{0,j} &(\delta d F) \notag \\
&= \lim\limits_{m \to 0^+}\Big(  \big[\big( \sum_{\pm}\langle j ,  F \rangle_{\Sigma^\pm} - \sum_{\pm}\langle j , E_m^\mp d\delta  F \rangle_{\Sigma^\pm} ,- (\rhoz E_m d \delta F, \rhod E_m d \delta F) , 0 , 0, \dots\big)\big]^\text{CCR}_\sim \notag \\
&\phantom{=I} - m^2 \big[\big( \sum_{\pm}\langle j ,  E_m^\mp F \rangle_{\Sigma^\pm}  , (\rhoz E_m  F, \rhod E_m  F) , 0 , 0, \dots\big)\big]^\text{CCR}_\sim\Big) \formspace.
\end{align}
Since the term proportional to $m^2$ is continuous and bounded in every degree and hence vanishes in the limit and using that $\rhod E_m d \delta F = \rhon d E_m d \delta F =0$, we find
\begin{align}\label{eqn:zero-mass-limit-field-temp}
\phi_{0,j} (\delta d F)
&= \langle j ,  F \rangle_{\M} - \lim\limits_{m \to 0^+}\big[ \big(  \sum_{\pm}\langle j , E_m^\mp d\delta  F \rangle_{\Sigma^\pm} , (d_{(\Sigma)} \rhoz E_m  \delta F, 0) , 0 , 0, \dots\big)\big]^\text{CCR}_\sim \notag\\
&= \langle j ,  F \rangle_{\M} - \big[ \big(  \sum_{\pm}\langle j , E_0^\mp d\delta  F \rangle_{\Sigma^\pm} , (d_{(\Sigma)} \rhoz E_0  \delta F, 0) , 0 , 0, \dots\big)\big]^\text{CCR}_\sim  \formspace.
\end{align}
Note that there appears a remainder $\big[ \big(  \sum_{\pm}\langle j , E_0^\mp d\delta  F \rangle_{\Sigma^\pm} , (d_{(\Sigma)} \rhoz E_0  \delta F, 0) , 0 , 0, \dots\big)\big]^\text{CCR}_\sim$ which, in order for the quantum fields to solve the correct dynamics, should vanish. It turns out that this is in general not the case. We have encountered a similar situation in the investigation of the classical theory in Section \ref{sec:limit_dynamics_classical}. There we could get rid of similar remaining terms by restricting the initial data such that the Lorenz constraint is well behaved in the limit.
In the quantum scenario this will also partly solve the problem, but note that the construction of the quantum theory is in that point quite different from the classical construction as we directly impose the dynamics by dividing out corresponding ideals, rather than first solving a wave equation and restrict to those solutions that fulfill the Lorenz constraint. Therefore, at the quantum level, the Lorenz constraint does not appear directly. It does, however, appear indirectly as we have fixed a classical solution $\varphi_{m,j}$ to map the source dependent theory to the source free theory. And it is with this homeomorphism where one part of the problem lies: \\
Recall that we have mapped the source dependent theory to the source free theory via the homeomorphism $\Gamma_{m,j,\varphi}$, choosing a classical solution to Proca's equation $(\delta d + m^2)\varphi_{m,j} = j$. We have used that the notion of continuity does not depend on the initial data that we choose for $\varphi_{m,j}$ and, for simplicity, we chose vanishing initial data. Using Theorem \ref{thm:solution_proca_unconstrained} for the classical solution $\varphi_{m,j}$ with vanishing initial data, this gave for some $F \in \Omega^1_0(\M)$
\begin{align}
	\Gamma^{-1}_{m,j,\varphi}\big( (0,F,0,0,\dots) \big)
	&= \langle \varphi_{m,j} , F \rangle_\M \cdot \mathbbm{1} + (0,F,0,0,\dots)   \notag \\
	&= \sum_{\pm}\langle j ,G_m^\mp F \rangle_{\Sigma^\pm} \cdot \mathbbm{1} + (0,F,0,0,\dots) \formspace.
\end{align}
For non-zero $m$, this is not a problem. But in the zero mass limit, $\lim\limits_{m \to 0^+} \langle \varphi_{m,j} , F \rangle_\M = \lim\limits_{m \to 0^+}  \sum_{\pm}\langle j ,G_m^\mp F \rangle_{\Sigma^\pm}$ will not specify a solution to Maxwell's equation as the Lorenz constraint expressed by the initial data,
\begin{align}
	\rhodelta \varphi_{m,j} &= \frac{1}{m^2}\rhodelta j \; , \quad \text{and} \\
	\rhon \varphi_{m,j} &= \frac{1}{m^2}\left( \rhon j  + \delta_{(\Sigma)} \rhod \varphi_{m,j} \right) \formspace,
\end{align}
is not well behaved (as we have chosen $\rhod \varphi_{m,j}=0$). As discussed in the classical Section \ref{sec:limit_dynamics_classical}, we need to impose $\delta j=0$ and $\rhon j= -\delta_{(\Sigma)} \pi$, where $\pi = \rhod \varphi_{m,j}$. Choosing vanishing initial data, we have violated the latter for non-zero external sources! Ergo, the homeomorphism between the source free and source dependent theory will only be well behaved in the limit if we set
\begin{align}
	\Gamma^{-1}_{m,j,\varphi} \big( (0,F,0,0,\dots) \big)
	&= \langle \varphi_{m,j} , F \rangle_\M \cdot \mathbbm{1} +(0,F,0,0,\dots)  \\
	&= \big(  \sum_{\pm}\langle j ,G_m^\mp F \rangle_{\Sigma^\pm} - \langle \pi , \rhoz G_m F\rangle_\Sigma \big) \cdot \mathbbm{1} +(0,F,0,0,\dots) \formspace,\notag
\end{align}
such that $\delta j=0$ and $\rhon j= -\delta_{(\Sigma)} \pi$.\\
Concluding, we have to define the zero mass limit field as
\begin{align}
	\phi_{0,j}(F) = \lim_{m \to 0^+} \big[\big( \sum_{\pm}\langle j , G_m^\mp F \rangle_{\Sigma^\pm} - \langle \pi , \rhoz G_m F \rangle_\Sigma, \kappa_0(F) , 0 , 0, \dots\big)\big]^\text{CCR}_\sim \formspace,
\end{align}
such that $\delta j=0$ and $\rhon j= -\delta_{(\Sigma)} \pi$. The limit exists if and only if $F$ is co-closed. With this we obtain
\begin{align}
	\phi_{0,j} (\delta d F)
	&= \langle j ,  F \rangle_{\M} \notag \\
	&\phantom{=I}- \big[ \big(  \sum_{\pm}\langle j , E_0^\mp d\delta  F \rangle_{\Sigma^\pm} - \langle \pi , \rhoz E_0 F \rangle_\Sigma , (d_{(\Sigma)} \rhoz E_0  \delta F, 0) , 0 , 0, \dots\big)\big]^\text{CCR}_\sim \notag \\
	&= \langle j ,  F \rangle_{\M}
	- \big[ \big( 0 , (d_{(\Sigma)} \rhoz E_0  \delta F, 0) , 0 , 0, \dots\big)\big]^\text{CCR}_\sim
\end{align}
since $\sum_{\pm}\langle j , E_0^\mp d\delta  F \rangle_{\Sigma^\pm} = \langle \pi , \rhoz E_0 F \rangle_\Sigma$ as we have shown in Section \ref{sec:limit_dynamics_classical}.\par
At this point, the remainder $\big[ \big( 0 , (d_{(\Sigma)} \rhoz E_0  \delta F, 0) , 0 , 0, \dots\big)\big]^\text{CCR}_\sim $ does not seem to vanish naturally. Its appearance might be explained by gauge equivalence:
Note that $-E_0 \delta dF$ solves the source free Maxwell equation, $-\delta d E_0 \delta d F = \delta d  E_0 d \delta F = 0$, and has initial data $-(\rhoz E_0 \delta d F, \rhod E_0 \delta d F) =  (d_{(\Sigma)} \rhoz E_0  \delta F, 0)$ as we have calculated before. From the classical investigation, see for example \cite{Sanders} or \cite{pfenning} which works in the same initial data formalism as we do here, we know that two solutions $A, A'$ to Maxwell's equation are gauge equivalent if and only if their initial values are gauge equivalent, that is, if  $\Az = \Azp + d\chi$ for some $\chi \in \Omega^0_0(\Sigma)$ (see \cite[Proposition 2.13]{pfenning}). With this we find that the solution $-E_0 \delta dF$ is gauge equivalent to zero and $(d_{(\Sigma)} \rhoz E_0  \delta F, 0) \sim_\text{gauge} 0$. This would give rise to a quantum field with the correct dynamics implemented. But this gauge equivalence on the level of the observables rather than the fields does not seem to come out of the limiting procedure naturally! It might be possible to find natural conditions that hold in the limiting procedure (as we have demanded that the homeomorphism $\Gamma_{m,j,\varphi}$ should be well behaved in the zero mass limit) to obtain the correct dynamics in the limit. Another possibility is to include states in the investigation and formulate a similar limiting process which might give rise to conditions identifying the remaining observables with zero. Alternatively, one could even go as far as concluding that the Proca field is unsuitable to describe massive electrodynamics with a well defined zero mass limit. A recent study \cite{stueckelberg_curvedST} suggests to abandon the investigation of Proca's theory as it is unsuited for the quantum investigation and rather examine Stueckelberg massive electromagnetism. As it is argued in \cite{stueckelberg_curvedST}, Proca's theory is nothing more than Stueckelberg's in a certain gauge that seems unfit for the investigation of the zero mass limit. It might therefore be of interest to apply the construction presented in this thesis to Stueckelberg electromagnetism, but since Stueckelberg's theory involves interaction with a scalar field, it is a priori not clear if that is possible.\par
Unfortunately within the (time) scope of this thesis, these ideas were not further studied but they are worthwhile investigating in future research projects.

\section{Conclusion and Outlook} \label{chpt::conclusion}
In this thesis we have studied the Proca field in curved spacetimes, including external sources and without restrictive assumptions on the topology of the spacetime, in the classical and the quantum case. We have rigorously constructed classical solutions to the Proca equation by decomposing the equation into a wave equation and a Lorenz constraint. After solving the wave equation we implemented the constraint by restricting the initial data. Investigating the classical zero mass limit, we found that the limit exists if we implement a gauge invariance of the distributional solutions to Proca's equation by exact distributional one-forms. This gauge is independently argued to be the correct physical gauge by Sanders, Dappiaggi and Hack \cite{Sanders} as it accounts for phenomena such as the Aharonov-Bohm effect, as opposed to a gauge invariance by closed distributional one-forms. In the zero mass limit, we find the correct Maxwell dynamics if we make sure that the Lorenz constraint is well behaved in the limit. This naturally leads to conservation of current and a restriction of the initial data as found in the Maxwell case by Pfenning \cite{pfenning}.\par
Studying the quantum problem, we first constructed the generally covariant QFTCS in the framework of Brunetti, Fredenhagen and Verch and proved that the theory is local. Choosing the Borchers-Uhlmann algebra as the algebra of observables we rigorously constructed an initial data formulation of the quantum Proca theory in curved spacetimes. With this initial data formulation we were able to define a precise notion of continuity of the Proca fields with respect to the mass: Using specifically constructed BU-algebra homeomorphism, we were able to map a family of Proca fields at different masses, initially elements in different BU-algebras, into one topological space, the BU-algebra of initial data. In the BU-algebra of initial data there is a natural notion of continuity provided by the topology. Using this notion of continuity, we studied the zero mass limit and showed that the limit exists if and only if we restrict the class of test one-forms to those that are co-closed. Analogously to the classical case, this effectively implements a gauge invariance by exact distributional fields. In the zero mass limit, we find that the obtained fields fulfill the basic properties of linearity, the real field property and the correct CCR, that is the same as in the Maxwell case, but we do not find the expected Maxwell dynamics. Unlike in the classical case, this is not caused by ill behaved constraints, since in the quantum case the dynamics are implemented directly by dividing out appropriate dynamical ideals instead of first defining solutions to a wave equation and implement a Lorenz constraint. It is not clear how to obtain the Maxwell dynamics naturally in the zero mass limit. Within the presented framework, it might be possible to find natural additional assumptions, for example demanding continuity of the defined BU-algebra homeomorphisms with respect to the mass, to restore the dynamics. Other approaches, for example an investigation of states and the zero mass limit, are worth considering in future research projects and might lead to a full description of Maxwell's theory as a limit of Proca's theory. As we argue in Appendix \ref{app:weyl-algebra}, a C*-Weyl algebra approach is ill-suited for the investigation of the zero mass limit. A recent argument by Belokogne and Folacci \cite{stueckelberg_curvedST} states that the Proca theory is indeed unfit to study the zero mass limit in the quantum case and should be replaced by Stueckelberg electromagnetism. It is our hope that the presented construction can be adapted to Stueckelberg's theory, but as Stueckelberg's theory includes interaction with a scalar field, it is not clear whether this is possible. \par
Further possible application of the presented construction is the investigation of locality in the zero mass limit. Since Proca's theory is local, as opposed the Maxwell's \cite{Dappiaggi2012,Sanders}, one might gain insight into this issue with the presented initial data formulation and the zero mass limit.\par
Moreover, it is of crucial importance to this thesis that Assumption \ref{ass:propagator_continuity} holds. While this seems reasonable, it is suggested to investigate a proof of the assumption, for example by a use of energy estimates. Within the scope of this thesis a deeper investigation of the assumption was not possible and is left open for future projects.



\newpage
\appendix

\section{The C*-Weyl Algebra as the Field Algebra}\label{app:weyl-algebra}
In this chapter we investigate the C*-Weyl algebra as the field algebra for the quantum Proca field theory.  The Weyl algebra is often used as the field algebra, both in algebraic quantum field theory in Minkowski space and in quantum field theory in curved spacetimes: Field operators obtained by GNS-construction (for details see e.g. \cite[Chapter III.]{haag} or \cite[Chapter III.14]{fragoulopoulou}) from the Weyl algebra are bounded, unlike the ones obtained from the field-algebra constructed from the free algebra. Since the product of such unbounded operators, which appear for example when calculating the commutator of field operators, are in general not well defined, it is useful to go over to bounded ones. Furthermore, the Weyl algebra is a normed algebra which makes it, at first glance, suitable for our purposes of finding a notion of continuity of the Proca field theory with respect to the mass. It turns out that similar investigations have been made in the literature:  In \cite{rieckers_honegger_deformation}, \name{Binz}, \name{Honegger} and \name{Rieckers} investigate the limit $\hbar \to 0$ for a family of Weyl algebras generated over an arbitrary pre-symplectic space, for which they introduce the notion of continuous fields of C*-algebras, which we will adapt to formulate a mass dependence of the Proca field theory.\par
In Section \ref{sec:continuous-families-of-symplectic-forms} we will generalize some mathematical results from \cite{rieckers_honegger_deformation} regarding continuous families of pre-symplectic forms and the corresponding C*-Weyl algebras. At this point, the dynamics of the theory are not implemented. This is done in a second step in Section \ref{sec:weyl-algebra-dynamics}. Unfortunately, it turns out that a notion of continuity with respect to the mass for the Weyl algebra formulation does not work in the wanted generality. This is why the use of the Weyl algebra is unsuited for the investigation of the zero mass limit and was discarded. Nevertheless, we present the results in this appendix to on the one hand illustrate why the ansatz is not suited for our problem and on the other hand to present the mathematical results in Section \ref{sec:continuous-families-of-symplectic-forms} that to our knowledge have not been discussed in the literature.
Most of the notation in this chapter is adapted from \cite{rieckers_honegger_deformation}.
\subsection{On C*-Weyl algebras and continuous family of pre-symplectic forms}\label{sec:continuous-families-of-symplectic-forms}
In this section we will generalize results from in \cite{rieckers_honegger_deformation} and \cite{rieckers_honegger_construction} regarding C*-Weyl algebras generated over a \emph{continuous family} of pre-symplectic spaces.
First, we briefly review the general construction of the C*-Weyl algebra over a pre-symplectic space. This field algebra will depend on the mass, even though we will only implement the commutation relations and not the dynamics yet. We would then like to find a notion of comparing these Weyl algebras at different masses with each other, using the notion of continuous fields of C*-algebras, similar to \cite{rieckers_honegger_deformation}. In our case, these algebras are generated over a \emph{continuous family} of pre-symplectic spaces, which makes it necessary to generalize the results obtained in \cite{rieckers_honegger_deformation} and \cite{rieckers_honegger_construction}. Just as in the discussion of the quantum problem using the BU-algebra, we need the Proca propagator to be continuous in the mass, that is, we need Assumption \ref{ass:propagator_continuity} to hold.
To start, we investigate the Weyl algebra generated over the real \emph{pre}-symplectic space\footnote{Unlike in the previous chapters, here we will view $\Omega^1_0(\M)$ as a real vector space.} $\big(\Omega^1_0(\M),\mathcal{G}_m\big)$ and, for simplicity, we write in this chapter $\gls{E} \equiv \Omega^1_0(\M)$ and keep in mind that, of course, there is an underlying manifold structure. 
Generalizing to a pre-symplectic space, rather than a symplectic one, does not have any effects on the construction (and uniqueness) of the C*-Weyl algebra of observables (see \cite{rieckers_honegger_construction}), but the constructed algebra will not be simple.
\begin{definition}[Weyl elements]
\label{def:weyl_elements}
For any $m >0$ and $F \in E$ we define the linearly independent \emph{Weyl-elements} \gls{WmF}, such that for all $F,F' \in E$ it holds
\begin{align}
\textrm{(i)} \quad&W_m(F)W_m(F') = \e^{-\i \Gm{F}{F'}/2 } W_m(F + F') \formspace,\\
\textrm{(ii)} \quad&W_m(F)^* = W_m(-F) \formspace.
\end{align}
\end{definition}
Condition (i) of the above definition is also known as the \emph{Weyl-form} of the CCR.
From these properties, it immediately follows that
\begin{align} 
W_m(0) W_m(F) 
&= W_m(F) W_m(0)\notag\\
&= \e^{-\i \Gm{F}{0}/2} W_m(F)\notag\\
&= W_m(F) \\  
\implies W_m(0) &= \mathbbm{1} \formspace,
\end{align}
from which it follows that
\begin{align}
\mathbbm{1}
&= W_m(0) \notag\\
&= W_m(F-F)\notag\\
&= W_m(F)^* W_m(F)\notag\\
&= W_m(F) W_m(F)^* \formspace,
\end{align}
that is, the Weyl elements are unitary.
Using these linearly independent Weyl elements, one can define the span of the Weyl elements:
\begin{align}
\gls{Wmtilde} \coloneqq \Span{W_m(F) : F \in E} \formspace.
\end{align}
Together with the above defined (twisted) product and $^*$-operation, $\widetilde{\mathcal{W}}_m$ becomes a unital $^*$-algebra, where the unit is given by $\mathbbm{1} = W_m(0)$.
Now, we would like to endow $\widetilde{\mathcal{W}}_m$ with a (unique) norm to define a C*-algebra as the norm closure of $\widetilde{\mathcal{W}}_m$. For this, we need the notion of a \emph{state} on $\widetilde{\mathcal{W}}_m$:
%
\begin{definition}[States]\label{def:states}
	Define \gls{CEGm} as the convex set of normalized, projectively positive functions $C : E \to \IC$. That is, for $C \in \mathcal{C}(E,\Gm{})$ it holds by definition:
	 \begin{subequations} 
	 	\begin{align}
	 		& C(0) = 1 																															&\textrm{(normalization)} \\
	 		& \sum\limits_{i,j=1}^N \bar{z}_i z_j \,\e^{\i \Gm{F_i}{F_j}/2}\, C(F_j - F_i) \geq 0 		&\textrm{(positiveness)}	 		
	 	\end{align}
	 	for every $N \in \IN$, $z_i \in \IC$  and $F_i \in E$.
	 \end{subequations}
	 Spanning the convex set of these function means that for every linear combination of $C$'s the corresponding coefficients add up to one. This is to ensure that linear combinations of states are also normalized.
	 To each $C \in \mathcal{C}(E,\Gm{})$ we associate a unique positive linear functional $\gls{omegaC} : \widetilde{\mathcal{W}}_m \to \IC $ via
	 \begin{align}
	 \omega_C \big(W_m(F)\big) = C(F)
	 \end{align}
	 for all $F \in E$. By the properties of $C$ it then holds for every $A \in \widetilde{\mathcal{W}}_m$ :
	 	\begin{align}
	 	\quad & \omega_C(\mathbbm{1}) = 1 			\\																												
	 	\quad & \omega_C (A^* A) \geq 0 		\formspace.	 	
	 	\end{align}	 
\end{definition}
In fact, every state $\omega_C$ on $\widetilde{\mathcal{W}}_m $ corresponds to a \emph{unique} $C \in \mathcal{C}(E,\Gm{})$	(see for example \cite[Chapter 3]{rieckers_honegger_deformation}). With this notion of a state we are able to define a C*-norm on the algebra $\widetilde{\mathcal{W}}_m$:
\begin{definition}[Weyl Algebra]
	Let $A \in \widetilde{\mathcal{W}}_m$ and $\omega_C$ be a state corresponding to a $C \in \mathcal{C}(E,\Gm{})$. On the *-algebra $\widetilde{\mathcal{W}}_m$ we define the (unique) C*-norm $\| . \|_m : \widetilde{\mathcal{W}}_m \to \IR$ by
	\begin{align}
		\| A \|_m \coloneqq \sup\left\{\sqrt{\omega_C(A^* A)} : C \in \mathcal{C}(E,\Gm{})\right\} \formspace.
	\end{align}
	The $\|.\|_m$-closure of $\widetilde{\mathcal{W}}_m$ is called the \emph{Weyl algebra} and will be denoted by \gls{Wm}.
\end{definition}
\noindent Since we are interested in a change of the mass parameter $m$ we need to build a mathematical structure in which we are able to compare the Weyl algebras for different masses in a continuous way. This can be done in the bundle $\W = \bigcup\limits_m \Wm$ with the notion of a \emph{continuous field of C*-algebras}. 
\begin{definition}[Sections and continous fields of C*-algebras]\label{def:continous_field_algebra}
	The set of sections $K$ of the bundle $\W$ is defined as
	\begin{align}
		\prod_{m} \Wm = \left\{ K : \IR^+ \ni m \mapsto K(m) \in \Wm \right\} \formspace.
	\end{align}	
	A \emph{continuous field of C*-algebras} is a tuple $(\left\{ \Wm \right\}_m , \K)$ consisting of a family $\left\{ \Wm \right\}_m $ of C*-algebras and a sub *-algebra $\K$ of $\prod_{m} \Wm$ such that
	\begin{enumerate}
		\item The map $\IR^+ \ni m \mapsto \|K(m)\|_m$ is continuous for all $K \in \K$, that is, $K$ is a continuous section.
		\item For every $m \in \IR^+$ the set $\left\{ K(m) : K \in \K \right\}$ is dense in $\Wm$.
		\item Let $K \in \prod_{m} \Wm $. If for every $m_0 \in \IR^+$ and every $\epsilon > 0$ there exists a section $H \in \K$ and a neighborhood $U_0$ of $m_0$ such that for all $m \in U_0$ it holds $\| K(m) - H(m)\|_m < \epsilon$, then $K \in \K$.
	\end{enumerate}
\end{definition}
We will in the following sometimes denote the sections explicitly by $[m \mapsto K(m)] \in \prod_{m} \Wm$.
In order to construct the desired field of C*-algebras, the following lemma is essential. It allows us to construct a specific sub *-algebra of $\prod_{m} \Wm$ containing only continuous sections that are point-wise dense in $\Wm$, which then guarantees the existence of a continuous field of C*-algebras:
\begin{lemma}\label{lem:cont_field_existance}
	Let $\D$ be a sub *-algebra of $\prod_m \Wm$ such that the conditions (i) and (ii) of the above Definition \ref{def:continous_field_algebra} are fulfilled with $\K$ replaced by $\D$. 
	Then 
	\begin{center}
there exists a unique continuous field of C*-algebras $(\left\{ \Wm \right\}_m , \K)$\\ such that $\D \subseteq \K$. 
	\end{center}
	In fact, $\K$ contains only those $K \in \prod_m \Wm$ such that for every $m_0 \in \IR^+$ and every $\epsilon > 0$ there exists a section $H \in \D$ and a neighborhood $U_0$ of $m_0$ such that for all $m_0 \in U_0$ it holds $\| K(m) - H(m)\|_m < \epsilon$.
\end{lemma}
\begin{proof}
	The above definition and lemma in this form are due to \cite[Chapter 2]{rieckers_honegger_deformation}. The lemma follows from \cite[Proposition 10.2.3]{dixmier} (which gives the same statement for continuous fields of Banach spaces) and applying \cite[Proposition 10.3.2]{dixmier} with use of the subset $\D$.
\end{proof}
Now, the procedure is as follows: We will define a specific sub *-algebra of $\prod_m \Wm$ that contains the sections $[m \mapsto W_m(F)]$, that is, it contains all physically interesting observables, and we will then show that this subset fulfills the necessary conditions, such that the above Lemma \ref{lem:cont_field_existance} is applicable. The non trivial part is to show that all the contained sections are continuous. For that, we need a notion on how to compare the norm of a section at different masses $m$ with each other. More specifically, we need to know how to relate states $C_m \in \mathcal{C}(E, \Gm{} )$ to states $C_{m_0} \in \mathcal{C} (E, \Gmz{})$. This is done in the following lemma, which provides some new insight on the Weyl algebra bundle over vector spaces endowed with a \emph{continuous family} of pre-symplectic forms that has, to our knowledge, not yet been discussed in the literature:
\begin{lemma}[On continuous families of pre-symplectic forms]\label{lem:on_contiuous_families}
	Let $E$ be a $\IR$-vector space of arbitrary dimension (including infinite dimensional). Let $I \in \IR$ open be an index set. Let $\left\{ \sigma_i \right\}_{i \in I}$ be a continuous family of pre-symplectic forms (that is, for all $F,F' \in E$ the map $I \ni i \mapsto \sigma_i(F,F')$ is continuous). 
	Then the following holds:
	\begin{enumerate}
		\item $\rho \coloneqq \sigma_i + \sigma_j$, for some $i,j \in I$, defines a pre-symplectic form on $E$. \\If $C_i \in \mathcal{C}(E,\sigma_i)$ and $C_j \in \mathcal{C}(E,\sigma_j)$, then
		\begin{align}
			C = C_i C_j \in \mathcal{C}(E,\rho) \formspace.
		\end{align}
		\item If $\left\{s_i\right\}_{i \in I}$ is a family of symmetric, positive, $\IR$-bilinear forms on $E$, such that for all $i \in I$ and all $F,F' \in E$ it holds $\sigma_i(F,F')^2 \leq s_i(F,F) s_i(F',F')$, then for all $F \in E$ it holds
		\begin{align}
			[F \mapsto \e^{-s_i(F,F)/2}] \in \mathcal{C}(E,\sigma_i)\formspace.
		\end{align}
		Furthermore, if $E$ is finite dimensional, there exists a family $\left\{s_i\right\}_{i \in I}$ with the above properties and such that $I \ni i \mapsto s_i (F,F')$ is continuous for all $F,F' \in E$.
	\end{enumerate}
\end{lemma}
\begin{proof}
(i) The proof is given for fixed pairs of pre-symplectic forms $\sigma_i , \sigma_j$ in \cite[Theorem 2.3]{rieckers_honegger_partially_classical_states}. \par 
(ii) For a fixed $i \in \IN$, the statement is proven in \cite[Theorem 3.4]{petz}\footnote{Actually, the proof in this reference is given to hold for symplectic forms rather than pre-symplectic forms, but generalizes directly to the pre-symplectic case since the non-degeneracy property is not used neither within the proof directly nor in the lemma that is needed for the proof.}. Here, we want to show that, in the finite dimensional case, the family $\left\{ s_i \right\}$ exists and can be chosen to be continuous:
Let $E$ be of finite dimension $N \in \IN$, let $I \in \IR$, open, be an index set and let $\left\{\sigma_i \right\}_{i \in I}$ be a continuous family of pre-symplectic forms on $E$. In the following we will assume to have chosen a basis $\left\{ e_i \right\}_{i=1}^{N}$ of $E$, $\norm{e_i} = 1$,  which induces a scalar product by $\langle e_i , e_j \rangle = \delta_{ij}$.
To clarify the structure of the proof, we start off with a claim: \\ 
\emph{Claim: There is an operator $\Lambda_i : E \to E$ for every $i \in I$ such that $\sigma_i(F,F') = \langle \Lambda_i F , F' \rangle$.} \\
We will use the operator $\Lambda_i$ explicitly to construct the symmetric form $s_i$. To see that the claim holds, define for a fixed $i \in I$ the operator
\begin{align}
\tilde{\sigma}_i :E &\to E^*  \\ 
F &\mapsto \tilde{\sigma}_i (F) \coloneqq \sigma_i (F , \cdot) \formspace, \notag
\end{align}
where $E^*$ denote the dual space of $E$. By Riesz' representation theorem (see for example \cite[Chapter II.16]{akhiezer-linear_ops_on_HS}), there exists a unique dual vector of $\tilde{\sigma}_i(F) \equiv \tilde{\sigma}_{i,F} $ given by $\tilde{\sigma}_i(F) ^*$ such that
\begin{align}
	\tilde{\sigma}_{i,F}(F') = \langle \tilde{\sigma}_i(F)^* , F' \rangle \formspace. 
\end{align}
Now, define the desired operator
\begin{align}
	\Lambda_i : E &\to E  \\
	F &\mapsto \Lambda_i(F) \coloneqq \tilde{\sigma}_i (F)^* \formspace. \notag
\end{align}
The claim then follows by construction: $\sigma_i(F,F') = \tilde{\sigma}_{i,F}(F') = \langle \Lambda_i F , F' \rangle $. \\
Now we choose, using the operator norm $\| \Lambda_i \| = \sup\left\{ \frac{\| \Lambda_i F\|}{\|F \|}: {F \in E}   \right\}$, 
\begin{align}
	s_i (F , F') \coloneqq \norm{\Lambda_i} {\cdot }\langle F , F' \rangle \formspace.
\end{align}
Clearly, this fulfills, using the Cauchy-Schwarz inequality:
\begin{align}
	\sigma_i (F , F')^2 
	&= \langle \Lambda_i F , F' \rangle ^2 \notag \\
	&\leq \langle \Lambda_i F, \Lambda_i F \rangle \, \langle F', F' \rangle\notag \\
	&\leq \| \Lambda_i \|^2 {\cdot} \langle F, F \rangle \, \langle F', F' \rangle \notag\\
	&= s_i(F,F) \, s_i(F' , F') \formspace.
\end{align}
This proves existence. We are left to show that this particular choice of $\left\{ s_i \right\}$ is continuous:
Let $ \left\{ x_i \right\}_{i \in \IN} \subset \IR $ be a continuous sequence in $\IR$ with $\lim\limits_{i \to \infty} \left( x_i -x_0\right) = 0$. It then holds by definition for all fixed $F,F' \in E$ that 
\begin{align}
0&=\lim\limits_{i \to \infty} \Big( \sigma_{x_i}(F,F') -\sigma_{x_0}(F,F') \Big) \notag	\\
\iff 0&= \lim\limits_{i \to \infty} \Big(  \langle \Lambda_{x_i} F, F' \rangle -  \langle \Lambda_{x_0} F, F' \rangle    \Big) \notag \\
\iff 0&= \lim\limits_{i \to \infty}      \langle (\Lambda_{x_i} - \Lambda_{x_0} ) F, F' \rangle \notag \\
\iff 0&=    \langle ( \lim\limits_{i \to \infty}  \Lambda_{x_i} - \Lambda_{x_0} ) F, F' \rangle \formspace.
\end{align}
Specifying $F' = \lim\limits_{i \to \infty}  (\Lambda_{x_i} - \Lambda_{x_0} ) F  )$,  this implies
\begin{align}
\implies 0&= \norm{\lim\limits_{i \to \infty}     (\Lambda_{x_i} - \Lambda_{x_0} ) F }\notag \\
\implies 0&= \lim\limits_{i \to \infty}     \| \left( \Lambda_{x_i} - \Lambda_{x_0} \right) F\|  
\end{align}
holds for all $F \in E$. We have used the continuity of the scalar product. Consequently, in the chosen basis, the matrix elements of $(\Lambda_{x_i} - \Lambda_{x_0})$ converge to zero, that is, for each pair $(i,j) \in I \times I$ it holds: For all $\epsilon > 0$ there exists a $i_0 \in \IN_0$, such that if $i \geq i_0$, then 
\begin{align}
	\abs{(\Lambda_{x_i} - \Lambda_{x_0})_{ij}} < \epsilon / N^2 \formspace.
\end{align}
We now choose $F \in E$ such that $\| F \| = 1$, that is, in the given basis we find $F = \sum_{n=1}^N f_n \, e_n$, where $f_n = \langle e_n , F \rangle$ and $\sum_{n=1}^N f_n ^2 = 1$. Therefore, it holds for all $n= 1,2, \dots, N$ that $\abs{f_n} \leq 1$. We find for $i \geq i_0$ that:
\begin{align}
	\| \Lambda_{x_i} - \Lambda_{x_0} \|
	&= \sup \left\{ \norm{\big( \Lambda_{x_i} - \Lambda_{x_0} \big) F} : F \in E, \norm{F} = 1 \right\} \notag\\
	&= \sup \left\{\norm{  \sum_{m,n=1}^N \big( \Lambda_{x_i} - \Lambda_{x_0} \big)_{mn}\;  f_n \, e_n  } : F \in E, \norm{F} = 1 \right\} \notag\\
	&\leq \sup \left\{ \sum_{m,n=1}^N \abs{\big( \Lambda_{x_i} - \Lambda_{x_0} \big)_{mn}}\cdot  \underbrace{\abs{f_n}}_{\leq 1} \cdot \underbrace{\norm{e_n}}_{=1}   : F \in E, \|F\| = 1 \right\} \notag\\
	&\leq  \sum_{m,n=1}^N \abs{\big( \Lambda_{x_i} - \Lambda_{x_0} \big)_{mn}}\notag\\
	&< N^2 \cdot \epsilon / N^2 \notag\\
	&= \epsilon \formspace.
\end{align}
Hence,  $i \mapsto \| \Lambda_i \| $ is continuous. 
Therefore, by construction, $i \mapsto s_i(F,F') = \norm{ \Lambda_i}{\cdot} \langle F , F' \rangle $ is continuous for all fixed $F,F' \in E$.
\end{proof}
With this lemma at our disposal we are able to compare states at different masses $m$ with each other. This is crucial for the proof of the following theorem:
\begin{theorem}[Continuous Weyl C*-field]\label{thm:cont_field_W}
	There exists a unique continuous field $\left( \left\{ \Wm \right\}_m , \K \right)$ of C*-algebras, such that for every $F \in E$ it holds $\left[ m \mapsto W_m(F) \right] \in \K $.
\end{theorem}
\begin{proof}
	The proof presented here is a generalization of the proof given in \cite[Theorem 5.2]{rieckers_honegger_deformation}. It was even possible, even though we are looking at a more general setup,  to simplify the given proof quite a bit. The main idea presented in \cite{rieckers_honegger_deformation}, that is, to use Lemma \ref{lem:cont_field_existance} and \ref{lem:on_contiuous_families}, was nevertheless essential for the proof.
	As mentioned, we will proof the statement by applying Lemma \ref{lem:cont_field_existance} to a specific sub *-algebra $\D$ of $\W$. For this we define:
	\begin{align}
		\D \coloneqq \Span{ \left[  m \mapsto \e^{-\i \alpha(m)} \, W_m(F)  \right]: F\in E, \alpha : \IR \to \IR_+ \text{ continuous} }  \formspace.
	\end{align}
	Clearly, as desired, we find for every $F \in E$ that $\left[ m \mapsto W_m(F) \right] \in \D $. 
	First, we need to check that $\D$, equipped with point-wise defined algebra relations, really constitutes a sub *-algebra of $\W$: 
	An arbitrary element of $\D$ is of the form
	\begin{align}
		K = \sum\limits_{i=1}^{N} \sum\limits_{j=1}^{M} z_{ij} \left[ m \mapsto \e^{-\i \alpha_{ij}(m) }\,W_m(F_i) \right] \label{def:K_arb_element_D}
	\end{align}
	for some $z_{ij} \in \IC$, $F_i \in E$ and some continuous maps $\alpha_{ij} : \IR_+ \to \IR$. 
	Clearly, for every $m \in \IR_+$ it holds that $K(m) \in \Wm$, and hence $\D \subset \W$. The nontrivial part is to show that $\D$ is closed under the point-wise defined algebra relations.
	Let $K, \tilde{K} \in \D$ be of the form specified above. We find
	\begin{align}
		(K + \tilde{K})(m)
		&\coloneqq K(m) + \tilde{K}(m) \notag \\
		&= \sum\limits_{i,j} z_{ij} \e^{-\i \alpha_{ij}(m) }\,W_m(F_i)   + \sum\limits_{k,l} \tilde{z}_{kl} \e^{-\i \tilde{\alpha}_{kl}(m) }\,W_m(\tilde{F}_k)  \notag \\
		&= \sum\limits_{pq} w_{pq} \e^{-\i \alpha'_{pq}(m) }\,W_m(F'_{p}) 
	\end{align} 
	for some $z'_{pq} \in \IC$, $F'_{p} \in E$ and some continuous $\alpha'_{pq} : \IR_+ \to \IR$ (because the composition of continuous maps is continuous)\footnote{In particular: if for some $i, k$ we find $F_i$ = $\tilde{F_k}$, then we certainly find $w_{pq}$ and continuous maps $\alpha'_{pq} : \IR_+ \to \IR$, such that $\sum_q w_{pq} \e^{-\i \alpha'_{pq}(m) } = \sum_j z_{ij} \e^{-\i \alpha_{ij}(m) } + \sum_l \tilde{z}_{kl} \e^{-\i \tilde{\alpha}_{kl}(m) }$.}.	Therefore, $\D$ is closed with respect to addition. \\
	Next, we find for the multiplication, using the Weyl-relations :
	\begin{align}
		(K \cdot \tilde{K})(m) 
		&\coloneqq K(m) \cdot \tilde{K}(m) \notag \\
		&= \sum\limits_{i,j,k,l}z_{ij}\,\tilde{z}_{kl} \,\e^{-\i \left( \alpha_{ij}(m) + \tilde{\alpha}_{kl}(m)  + \Gm{F_i}{\tilde{F}_k}/2   \right) }\,W_m(F_i + \tilde{F}_k) \formspace.
	\end{align}
	This also is clearly an element in $\D$, since the combination of continuous maps is again continuous\footnote{Here, of course, we need $\Gm{}$ to be continuous.}.
	In the same fashion, it follows trivially that $\D$ is also closed under involution, hence $\D$ does indeed constitute a sub *-algebra of $\W$.\par
	Now, the main effort is to show that $\D$ fulfills the assumptions (i) and (ii) of Lemma \ref{def:continous_field_algebra}: 
	First, it is trivial to see that,  since $\left[ m \mapsto W_m(F) \right] \in \D$ for all $F\in E $ and since for every $m \in \IR_+$ the $W_m(F)$ are dense in $\Wm$ by construction, 
	we find that for every $m \in \IR^+$ the set $\left\{ K(m) : K \in \D \right\}$ is dense in $\Wm$.
	The non trivial part is to show that every section $K \in \D$ is continuous: 
	Again, let $K \in \D$ denote an arbitrary element of the form specified in equation (\ref{def:K_arb_element_D}).
	Further more, since only finitely many $F_i$'s contribute to the construction of $K(m)$, we can equivalently view $K(m)$ as an element of a subspace of $\Wm$ that is generated by only finitely many $W_m(F_i)$'s. In particular, define $E_N = \Span{F_i, i=1,2,\dots,N }$, where the $F_i$'s correspond to the ones contributing to $K(m)$. Then $\Wm(E_N,\Gm{})$ is a sub-C*-algebra of $\Wm(E,\Gm{})$  and we can evaluate the norm of $K(m)$ in terms of states $C \in \C (E_N, \Gm{})$ (see e.g. \cite[Section III.B]{rieckers_honegger_construction}), and the norms of $K(m)$ on $\Wm(E_N,\Gm{})$ and $\Wm(E,\Gm{})$ agree.
		Next, we define the pre-symplectic form $\sigma_m = \Gm{} - \Gmz{}$ on $E$. Since for fixed $F,F' \in E$, $\Gmz{F,F'}$ is constant and $m \mapsto \Gm{F}{F'}$ is continuous by assumption, it is clear that $\left\{\sigma_m \right\}_{m \in \IR_+}$ constitutes a continuous family of pre-symplectic forms on $E$, hence it also constitutes a continuous family of pre-symplectic forms on $E_N$. 
		Moreover, we find
		\begin{align}
			\lim\limits_{m \to m_0} \sigma_m (F,F') = 0
		\end{align}
		for all fixed $F,F' \in E_N$. 
		Now, since by construction $E_N$ is finite dimensional, we can make use of the additional feature of Lemma \ref{lem:on_contiuous_families}, that is, there exists a continuous family $\left\{ s_m \right\}_{m \in \IR_+}$ of symmetric, positive, bilinear forms on $E_N$ such that
			\begin{align}
				E_N \ni F \mapsto \e^{- s_m(F,F) /2} \in \C (E_N,\sigma_m) \formspace.
			\end{align}
	Knowing the existence of such a family $\left\{ s_m \right\}_{m \in \IR_+}$ and the corresponding exponential states is the reason why we go over to a finite dimensional subspace $E_N$ of $E$.
	By construction we moreover know that $\Gm{} = \sigma_m + \Gmz{}$ is a sum of pre-symplectic forms on $E_N$. Making use of the first part of Lemma \ref{lem:on_contiuous_families}, we conclude that for an arbitrary state $C_0 \in \C (E_N , \Gmz{})$ there exists a state $C \in \C (E_N , \Gm{})$ of the form
	\begin{align}
		C (F) = C_0(F) \, \e^{- s_m(F,F) /2} \formspace.
	\end{align}
	With this notion we are able to compare states, and hence the norm of a section of the Weyl algebra bundle, at different masses $m$. But it is worth mentioning that not necessarily \emph{all} states $C \in \C (E_N , \Gm{})$ are of the above product form. This subtlety has to be kept in mind. 
	Let us now look at the expectation value of the section $K$ at point $m$ and at point $m_0$, inserting the product form of the state $C \in \C (E_N , \Gm{})$:
	\begin{align}
		\omega_C\big( K^*(m) K(m) \big)
		&= \sum_{ijkl} \bar{z}_{ij}z_{kl} \, \e^{-\i \big(  \alpha_{kl}(m)  - \alpha_{ij}(m) + \Gm{F_i}{F_k}/2\big)} \notag \\ &\phantom{M}\cdot \e^{-s_m(F_k - F_i , F_k - F_i) /2} \, C_0(F_k - F_i)  \formspace,\\
		\omega_{C_0}\big( K^*(m_0) K(m_0) \big)
		&= \sum_{ijkl} \bar{z}_{ij}z_{kl} \, \e^{-\i \big(  \alpha_{kl}(m_0)  - \alpha_{ij}(m_0) + \Gmz{F_i}{F_k}/2\big)} \, C_0(F_k - F_i) \formspace.
		\end{align}	
Since for all $i,j$ and all $F_i, F_k$, the maps $m \mapsto \Gm{F_k}{F_i}$ and $m \mapsto S_m(F_k,F_i)$ as well as the maps $\alpha_{ij}, \alpha_{kl}$ are continuous, we conclude that
\begin{align}
	\Lambda_{ijkl}(m) \coloneqq \e^{-\i \big(  \alpha_{kl}(m)  - \alpha_{ij}(m) + \Gm{F_i}{F_k}/2\big)} \, \e^{-s_m(F_k - F_i , F_k - F_i) /2}
\end{align}
is continuous. In particular it holds that
\begin{align}
	\lim\limits_{m \to m_0} 	\Lambda_{ijkl}(m)  = \e^{-\i \big(  \alpha_{kl}(m_0)  - \alpha_{ij}(m_0) + \Gmz{F_i}{F_k}/2\big)} \eqqcolon \Lambda_{ijkl}(m_0)  \formspace.
\end{align}
This means, for all $i,k= 1,\dots,N$ and $j,l=1,\cdots,M$ and for every $\epsilon >0$ we find a $\delta > 0 $, such that for every $m \in \IR_+$ with $\abs{m - m_0} < \delta$ it holds
\begin{align}
\abs{ \Lambda_{ijkl}(m) - \Lambda_{ijkl}(m_0) }< \epsilon/\left( M^2 N^2 \,\max\limits_{i,j,k,l}\abs{\bar{z}_{ij} z_{kl}} \right)  \formspace.
\end{align}
With this we can write (here $C$ is a state of product form)
\begin{align}
\omega_{C_0}&\big( K^*(m_0) K(m_0) \big) \notag \\
&= \sum_{ijkl} \bar{z}_{ij}z_{kl} \, \Lambda_{ijkl}(m_0)\, C_0(F_k - F_i) \notag \\
&= \omega_{C}\big( K^*(m) K(m) \big) -  \sum_{ijkl} \bar{z}_{ij}z_{kl} \, \big( \Lambda_{ijkl}(m) - \Lambda_{ijkl}(m_0)\big) \, C_0(F_k - F_i) \formspace,
\end{align}
where in the second step we have just inserted a zero term.  Since by definition $\omega_{C_0}\big( K^*(m_0) K(m_0) \big)  \geq 0$ we conclude that
\begin{align}
\omega_{C_0}&\big( K^*(m_0) K(m_0) \big)\notag  \\
&=\abs{\omega_{C}\big( K^*(m) K(m) \big) -  \sum_{ijkl} \bar{z}_{ij}z_{kl} \, \big( \Lambda_{ijkl}(m) - \Lambda_{ijkl}(m_0)\big) \, C_0(F_k - F_i)}  \notag  \\
&\leq \omega_{C}\big( K^*(m) K(m) \big) + \sum_{ijkl} \abs{\bar{z}_{ij}z_{kl}} \, \abs{\big( \Lambda_{ijkl}(m) - \Lambda_{ijkl}(m_0)\big)} \, \abs{C_0(F_k - F_i)} \notag  \\
&\leq \omega_{C}\big( K^*(m) K(m) \big) + \epsilon \formspace,
\end{align}
where the last step holds for $\abs{m - m_0} < \delta$. We have made use of the estimate 
\begin{align}
\abs{C(F)} \leq 1 \formspace,
\end{align}
as stated in Lemma \ref{lem:states_estimate}. Having found an estimate of $\omega_{C_0}\big( K^*(m_0) K(m_0) \big)$, we can estimate the norm of $K$ at $m_0$, where here $C$ is of the above specified product form:
\begin{align}
	\norm{K(m_0)}^2 
	&= \sup\Big\{ \omega_{C_0}\big( K^*(m_0) K(m_0) \big) : C_0 \in \C(E_N , \Gmz{})\Big\} \notag \\
	&= \sup\Big\{ \omega_{C}\big( K^*(m) K(m) \big) - \omega_{C}\big( K^*(m) K(m) \big) + \omega_{C_0}\big( K^*(m_0) K(m_0) \big)   :\notag \notag \\
	&\phantom{M} C=C_0\, e^{s_m(\cdot, \cdot)/2}, C_0 \in \C(E_N , \Gmz{})\Big\}\notag \\
	&\leq \sup\Big\{ \omega_{C}\big( K^*(m) K(m) \big)+\epsilon   : C=C_0\, e^{s_m(\cdot, \cdot)/2}, C_0 \in \C(E_N , \Gmz{})\Big\} \notag\\	
	& \leq  \sup\Big\{ \omega_{C}\big( K^*(m) K(m) \big) + \epsilon : C \in \C(E_N , \Gm{})\Big\} \notag\\
	&= 	\norm{K(m)}^2 + \epsilon.
\end{align}
The calculation might seem a bit too detailed at first, but one has to be a careful in which order one makes the estimate for the proof to work, due to the mentioned subtlety that not \emph{all} states at mass $m$ are of product form.
In the exact same fashion, interchanging the roles of $C_0$ and $C$, we define (using the same notation as before) the pre-symplectic form
\begin{align}
	\sigma_m = \Gmz{} -\Gm{} \formspace.
\end{align}
In the same line of arguments as before, we can express states at $m_0$ by $C_0 = C\, \e^{s_m(\cdot, \cdot)/2}$, where here actually the symmetric forms $s_m$ are the same as before, since, by construction, they only depend on the norm of $\sigma_m$, and, reversing the roles of $m$ and $m_0$, we only have flipped the sign of $\sigma_m$. 
We find
\begin{align}
	\omega_{C}\big( K^*(m) K(m) \big) 
	&= \Big\lvert\omega_{C_0}\big( K^*(m_0) K(m_0) \big)  \notag \\
	&\phantom{M}-  \sum_{ijkl} \bar{z}_{ij}z_{kl} \, \Big( \Lambda_{ijkl}(m)\,\e^{s_m(F_k - F_i, F_k - F_i)/2} \notag \\
		&\phantom{M}- \Lambda_{ijkl}(m_0) \, \e^{-s_m(F_k - F_i, F_k - F_i)/2}\Big) \, C(F_k - F_i) \Big\rvert\notag\\
	&\leq \omega_{C_0}\big( K^*(m_0) K(m_0) \big)  + \epsilon
\end{align}
for $\abs{m-m_0} < \delta$.\\
By an estimate of the same fashion as before we find for $\abs{m-m_0}<\delta$:
\begin{align}
\norm{K(m)}^2  \leq  \norm{K(m_0)}^2 + \epsilon \formspace.
\end{align}
Combining the two estimates, we find the wanted result, namely that for all $\epsilon >0$ there exists a $\delta > 0$ such that for all $m \in \IR_+$ with $\abs{m-m_0}<\delta$ it holds:
\begin{align}
	\abs{ \norm{K(m)}^2 - \norm{K(m_0)}^2} < \epsilon \formspace,
\end{align}
hence, $m \mapsto \norm{K(m)}^2$ is continuous. Since $0 \leq \norm{K(m)}^2$ it follows that $m \mapsto \norm{K(m)}$ is also continuous. This concludes the proof.
\end{proof}
\subsection{Dynamics and the C*-Weyl algebra} \label{sec:weyl-algebra-dynamics}
Having found a notion of continuity of the fields with respect to the mass, the task is now to implement the dynamics of the system, that is, the Weyl operators should solve the Proca equation in a suitable sense. Motivated by Definition \ref{def:algebra-A(M)} of the field operators $\A(F)$ we want to implement the relation $W_m\big( (\delta d + m^2) F\big) = \e^{\langle j , F' \rangle} \mathbbm{1}$ (the relation follows from Definition \ref{def:algebra-A(M)} by heuristically viewing the Weyl operators as $\e^{\A(F)}$). For now, we specify to the source free case $j = 0$. To formulate the implementation of the above relation mathematically precisely, we define the set 
\begin{align}
	\JM \coloneqq \left\{ F \in E : \exists F' \in E : F = (\delta d + m^2)F'  \right\} \equiv \left\{ F \in E : G_mF=0  \right\} \formspace.
\end{align} 
	The dynamical theory is then implemented in the Weyl algebra $\Wm\left({\Quotientscale{E}{\JM}}, \Gm{}\right)$, where the elements $\left[ F \right]_m \in {\Quotientscale{E}{\JM}}$ are equivalence classes with respect to the equivalence relation $F \sim_m F' :\iff (F-F') \in \JM$. Identifying $[0]_m = \JM$, we have implemented the wanted field equation $W_m([0]_m) = \mathbbm{1}$. Note, that we use the same symbol $W_m$ for the Weyl operators in $\Wm$ and $\Wmdyn$ since it is clear from the context, that is, whether they act on functions $F$ or equivalence classes $[F]_m$, in which algebra the corresponding Weyl elements lie.\par 
	The question is how to apply the results we have gained on the algebra $\Wm$ and $\W$ to the dynamical algebras $\Wmdyn$ and $\Wdyn$. To accomplish this we define the following *-algebra homomorphism:
	\begin{align}
		\tilde{\alpha}_m : \;&\Wm \to \Wmdyn  \\
										& W_m(F) \mapsto W_m([F]_m) \notag
	\end{align}
	and analogously for linear combinations of Weyl elements. It is easily checked, that indeed $\tilde{\alpha}_m$ is a *-algebra homomorphism:
	Let $F, F' \in E$, then
	\begin{align}
		\tilde{\alpha}_m\big( W_m(F) W_m(F') \big) 
		&= \e^{-\i \Gm{F}{F'}} \, \tilde{\alpha}_m\big( W_m(F + F')\big) \notag\\
		&= \e^{-\i \Gm{F}{F'}} \,  W_m\big([F]_m + [F']_m\big) \notag\\
		&= \e^{-\i \Gm{F}{F'}}\e^{\i \Gm{[F]_m}{[F']_m}} \,  W_m\big([F]_m\big) W_m\big([F']_m\big) \notag\\
		&= W_m\big([F]_m\big) W_m\big([F']_m\big) \formspace.
	\end{align}
	Also, it holds that
	\begin{align}
	\tilde{\alpha}_m\big( W_m(F) ^* \big) 
	&=\tilde{\alpha}_m\big( W_m(-F) \big) \notag\\
	&= W_m([-F]_m)  \notag\\
	&=  W_m(-[F]_m)  \notag\\
	 &= W_m([F]_m)^* \formspace.
	\end{align}
	Now, from the first isomorphism theorem, since clearly the image of $\tilde{\alpha}_m$ is the whole algebra $\Wmdyn$, it holds that $\Wmdyn = \tilde{\alpha}_m\left( \Wm \right) \simeq {\Quotientscale{\Wm}{\Ker{\tilde{\alpha}_m}}}$ .
	In particular, $\Ker{\tilde{\alpha}_m}$ is a two sided ideal in $\Wm$ and one can show that this ideal corresponds to the ideal that is generated by the relation $\big( W_m(F) -\mathbbm{1} \big)$ for $F \in \JM$. So the morphism $\tilde{\alpha}$ implements the wanted dynamics. \\[4mm]  
	To see this, let us explicitly compare the two ideals:
	\begin{align}
		\Ker{\tilde{\alpha}_m} &= \Big\{A = \sum z_i W_m(F_i) : 0 = \sum z_i W_m\big([F_i]_m\big)\Big\}  \formspace,\\
		\IM &=  \Big\{A\, \big(W_m(F) - \mathbbm{1}\big)\,  B  : A, B  \in \Wm, F \in E \Big\} \notag  \\
									& = \Big\{A\, \big(W_m(F) - \mathbbm{1}\big)  : A  \in \Wm, F \in E \Big\} \notag\\
									& =  \Big\{\sum z_i  \big( W_m(F_i + H ) - W_m(F_i)\big) :  z_I \in \IC, F_i \in E , H \in \JM\Big\} \formspace.
	\end{align}
	\begin{lemma}
		Let $\tilde{\alpha}_m$ and $\IM$ be defined as above. Then
		\begin{align}
		 \Ker{\tilde{\alpha}_m} = \IM \formspace.
		\end{align}
	\end{lemma}
	\begin{proof}
1.) For the first direction of the proof let $A \in \IM$. Then
\begin{align}
	\tilde{\alpha}_m(A) = \sum z_i  \big( W_m([F_i + H]_m ) - W_m([F_i]_m)\big)
	=0
\end{align}
since $[F_i + H]_m = [F_i]_m + [H]_m = [F_i]_m$, because $H \in \JM$ by assumption. Hence, $\IM \subset \Ker{\tilde{\alpha}_m}$. \par 
2.) For the other direction, let $A \in \Ker{\tilde{\alpha}_m}$. \\
An arbitrary $A \in \Wm$ is of the form $A = \sum z_i \, W_m(F_i)$. We can reorder this finite sum, such that we group together the $F_i$'s that are in the same equivalence class, thus writing $A = \sum_j A_j$, where $A_j = \sum_i z_{ij} W_m(F_{ij})$ with $[F_{ij}] = [F_{i'j'}] \iff j=j'$. With this notation it is easier to classify elements of the kernel of $\tilde{\alpha}$. By construction $\tilde{\alpha}(A) = 0 \iff \tilde{\alpha}(A_j) =0$ for all $j \iff \sum_i z_{ij} = 0$ for all $j$. It therefore suffices to show that $A_j \in \IM$ for all $j$:
\begin{align}
	A_j 
	&= \sum_i z_{ij} W(F_{ij})  \notag\\		
	&= \sum_i z_{ij} \big( \underbrace{W(F_{ij}) - W(F_{i'j})}_{\in \IM} +W(F_{i'j}) \big) , \text{ for some }i' \notag \\
	&= \underbrace{\sum_i z_{ij}}_{=0}  W(F_{i'j}) + \underbrace{X}_{\in \IM}
\end{align}
hence, $A_j \in \IM$ for all $j$, from which we conclude $\Ker{\tilde{\alpha}_m} \subset \IM$ which completes the proof. 
	\end{proof}
We now want to construct a homomorphism from the Weyl algebra bundle $\W$ to $\Wdyn$. We do this by a point-wise definition using the homomorphism $\tilde{\alpha}$. 
Define:
\begin{align}
	\alpha : \;&\W \to \Wdyn  \\
	& K \mapsto \alpha(K), \text{ such that }\big(  \alpha(K) \big) (m) = \tilde{\alpha}_m \big( K(m)\big) \formspace. \notag
\end{align}
Since the algebra relations in $\W$ and $\Wdyn$ are defined point-wise and we have furthermore already seen that $\tilde{\alpha}_m$ is a *-algebra-homomorphism for each $m$ we conclude that $\alpha$ is a *-algebra homomorphism. By the same argument as before we also conclude that $\Wmdyn \simeq {\Quotientscale{\W}{\Ker{\alpha}}}$.\par 
Now we are interested if we can find continuous fields of C*-algebras on $\Wdyn$ such that the physically interesting observables are contained. In particular we would like to have a theorem of the kind of Theorem \ref{thm:cont_field_W} for the algebra $\Wdyn$. Unfortunately it turns out that a theorem of the same generality as for the non-dynamical case is not possible. This is stated in form of the following theorem:
\begin{theorem}
	There cannot exist a continuous field $\big( \left\{\Wmdyn \right\} , \K^\textrm{dyn} \big)$ of C*-algebras such that for all $[F]_m \in {\Quotientscale{E}{\JM}}$ it holds $\big[ m \mapsto W_m([F]_m)\big] \in \K^\textrm{dyn}$.
\end{theorem}
\begin{proof}
	Let $\K^\text{dyn}$ be a sub *-algebra of $\prod_m \Wmdyn$ such that for all $[F]_m \in {\Quotientscale{E}{\JM}}$ it holds $\big[ m \mapsto W_m([F]_m)\big] \in \K^\textrm{dyn}$. In particular all sections $K \in \K$ are continuous in the sense that $m \mapsto \norm{K(m)}_m$ is continuous. It is clear that for all $[F]_m \in {\Quotientscale{E}{\JM}}$ the sections $\left[ m \mapsto W_m([F]_m)\right]$ are continuous. But, since $\K$ is an algebra, also linear combinations of $\left[ m \mapsto W_m([F]_m)\right]$ are elements of $\K$ and are ought to be continuous. This does not hold as we shall see:\\
	Let $K = \Big( \big[ m \mapsto W_m([F]_m) \big] - \big[ m \mapsto W_m([H]_m) \big]  \Big) \in \K$ where we choose $F, H \in E$ such that $[F]_{m_0} = [H]_{m_0}$ for some $m_0 \in \IR$ but $[F]_{m_0} \neq [H]_{m_0}$ for $m \neq m_0$.\\
	Then by construction we find $\norm{K(m_0)}_{m_0} = 0$ but, because for Weyl elements it holds $\norm{z W_m([F]_m) + w W_m([H]_m)}_m = \abs{z} + \abs{w}$ (see e.g. \cite[Proposition 3.10]{rieckers_honegger_construction}), for $m \neq m_0$ it holds $\norm{K(m)}_m = 2$, hence $K(m)$ is not continuous. \Lightning\\
\end{proof}
Note, that this does not mean that there cannot be a continuous field $\K^\textrm{dyn}$ of C*-algebras such that at least for \emph{some} $[F]_m\in {\Quotientscale{E}{\JM}}$ the sections $\big[ m \mapsto W_m([F]_m)\big] \in \K^\textrm{dyn}$. By the argument in the above proof, we certainly need to leave out all the elements that can for some $m_0$ and some $F' \in E$ be written as $F = (\delta d + m_0^2)F'$. In particular this includes all closed test forms $F$: Choosing $m_0 =1$ and $F' = F$ one finds $(\delta d + 1)F' = F$. But there are certainly a lot more $F$ of the above form that we would have to discard. Also, when removing some $\big[ m \mapsto W_m([F]_m)\big]$ from the field of C*-algebras it is not clear, whether we are still able to fulfill the property that the set $\big\{  K(m) : K \in \K^\text{dyn}\big\}$ is dense in $\Wmdyn$. 
We would have to include some sections that take at every $m$ a value in $\Wmdyn$ that cannot be represented by a linear combination of $W_m([F]_m)$'s. These elements exist due to $\Wmdyn$ being defined as the $\norm{\cdot}_m$-closure of the span of the Weyl elements $W_m([F]_m)$. Such elements cannot be written down in a fashion that we have used to characterize sections so far.\par
Due to these reasons, the C*-Weyl algebra does not seem to be suitable for our problem and the ansatz of finding continuous fields of C*-algebras to investigate the zero mass limit is hereby discarded.
\section{Additional Lemmata}\label{app:lemmata}
\begin{lemma}\label{lem:epsilon_contraction}
	Let $\N$ be a $N$ dimensional manifold with metric $k$. It holds for any $j \in {1,2,\dots,N-1}$
	\begin{align}
		\detk ^2 \; \epsilon\indices{^{\alpha_1\dots\alpha_j \alpha_{j+1} \dots \alpha_N }} \epsilon\indices{_{\alpha_1\dots\alpha_j \beta_{j+1} \dots \beta_N} }
		= (-1) ^s \, (N-j)! j! \; \delta\indices{^{[\alpha_{j+1}}_{\beta_{j+1}}}\,\delta\indices{^{\alpha_{j+2}}_{\beta_{j+2}}}\, \dots \, \delta\indices{^{\alpha_N ] }_{\beta_N}} \formspace \notag
	\end{align}
	where $s$ is the number of negative eigenvalues of the metric $k$.
\end{lemma}
\begin{proof}
The above formula is shown to hold in \cite[Equation (B.2.13)]{wald_GR}
\end{proof}
\begin{lemma} \label{lem:injective_hom_trivial_kernel}
	Let $\mathscr{A},\mathscr{B}$ be unital $^*$-algebras and let $\alpha : \mathscr{A} \to \mathscr{B}$ be a unit preserving $^*$-algebra homomorphism. Then:\\
	\begin{center}
		$\alpha$ is injective if and only if the kernel of $\alpha$ is trivial, that is, $\Ker{\alpha} = \{0\}$.
	\end{center}
\end{lemma}
\begin{proof}
	1.) Assume that $\alpha$ is injective, that is, for $a,\tilde{a} \in \mathscr{A}, \alpha(a) = \alpha(\tilde{a}) \implies a = \tilde{a}$.
	Then, because $\alpha$ is a homomorphism, it follows that $\alpha(0) = 0$.
	Let $g \in \Ker{\alpha}$. Then
	\begin{align}
		\alpha(g) &= 0 = \alpha(0) \notag \\
		\implies g &= 0 \quad \textrm{by injectivity} \notag \\
		\implies \Ker{\alpha} &= \{ 0 \} \formspace,
	\end{align}
	which completes the proof of the first direction. \par
2.) Assume that $\Ker{\alpha}=\{ 0 \}$.
	Let $a,\tilde{a} \in \mathscr{A},$ such that $\alpha(a) = \alpha(\tilde{a})$. Then
	\begin{align}
		\alpha(a) - \alpha(\tilde{a}) &= 0 \notag\\
		\implies \alpha(a - \tilde{a} ) &= 0 	\quad\textrm{because }\alpha \textrm{ is homomorphism} \notag \\
		\implies a - \tilde{a} &= 0 					\quad\textrm{because the kernel is trivial} \formspace,
	\end{align}
	therefore, $\alpha$ is injective. This completes the proof.
\end{proof}
\begin{lemma}\label{lem:injective_mor_simple_algebra}
	Let $\mathscr{A},\mathscr{B}$ be unital $^*$-algebras and let $\alpha : \mathscr{A} \to \mathscr{B}$ be a unit preserving $^*$-algebra homomorphism.
	Then: \\
	\begin{center}
		$\alpha$ is injective if $\mathscr{A}$ is simple.
	\end{center}
\end{lemma}
\begin{proof}
	The proof follows straight forward from Lemma \ref{lem:injective_hom_trivial_kernel}:
	First, note that for every homomorphism $\alpha : \mathscr{A} \to \mathscr{B}$, the kernel of $\alpha$ is an ideal in $\mathscr{A}$:\\
	Let $a \in \mathscr{A}$ arbitrary and $g \in \Ker{\alpha}$. Then
	\begin{align}
		\alpha(ag)
		&= \alpha(a) \alpha(g) \notag\\
		&= \alpha(a)\,0 \notag\\
		&= 0 \notag\\
		\implies (ag) &\in \Ker{\alpha}\formspace.
	\end{align}
	And an analogous result follows for $(ga)$ which shows that $\Ker{\alpha}$ is a two sided ideal in $\mathscr{A}$.\\
	Now, since $\mathscr{A}$ is simple, $\alpha$ has either full or trivial kernel, that is, $\Ker{\alpha} = \{0 \}$ or $\Ker{\alpha} = \mathscr{A}$. But since $\alpha(\mathbbm{1}) = \mathbbm{1} \neq 0$ it follows that $\Ker{\alpha} = \{0 \}$ and by Lemma \ref{lem:injective_hom_trivial_kernel} that $\alpha$ is injective.
\end{proof}
\begin{lemma}\label{lem:BU-algebra-barreled}
	Let $\mathfrak{X}$ a complex vector bundle over a smooth differential manifold. Then, the Borchers-Uhlmann algebra $\BU(\Gamma_0(\mathfrak{X}))$ is barreled.
\end{lemma}
\begin{proof}
	Let $\mathfrak{X}$ a complex vector bundle over a smooth differential manifold $\N$. Recall from chapter \ref{sec:BU-algebra} that we have defined the BU-algebra over $\Gamma_0(\mathfrak{X})$ as the tensor algebra $\BU(\Gamma_0(\mathfrak{X})) = \IC \oplus \bigoplus\limits_{n= 1}^\infty \Gamma_0(\mathfrak{X})^{\otimes n}$ and endowed it with the inductive limit topology of the spaces
\begin{align}
\BU_N = \IC \oplus \bigoplus\limits_{n= 1}^N \Gamma_0(\mathfrak{X})^{\otimes n} \formspace.
\end{align}
	 Equivalently, one can densely embed the tensor products $\Gamma_0(\mathfrak{X})^{\otimes n} \subset \Gamma_0( \boxtimes^n \mathfrak{X} )$, where $\boxtimes$ denotes the outer tensor product of vector bundles (see \cite[Chapter 3.3]{verch_sahlman}). The space of compactly supported section of a complex vector bundle is a LF-space, as it is defined as the inductive limit of the Frech\'et spaces of sections with support in some compact $K_l$ where $\left\{ K_l \right\}_l$ is a fundamental sequence of compact $K_l \subset \N$ (see \cite[17.2.2 and 17.3.1]{dieudonne_3}). Since LF-spaces are barreled \cite[Chapter 33, Corollary 3]{treves} and the direct sum of barreled spaces is again barreled \cite[18.11]{kelly-namioka}, we find for any $N \in \IN$ that $\BU_N$ is barreled. Additionally, the inductive limit of barreled spaces is barreled \cite[Chapter V, Proposition 6]{robertson}, hence the BU-algebra over smooth compactly supported sections $\Gamma_0(\mathfrak{X})$ over a complex vector bundle $\mathfrak{X}$ is barreled.
\end{proof}
\begin{lemma}[Symmetrization of fields] \label{lem:symmetrization-of-fields}
	Let $\left\{f_m\right\}_m \subset \BU\left(\Dzs \right)$ be a continuous family in the Borchers-Uhlmann algebra algebra of initial data $\Dzs$.
	Then
	\begin{align}
		f_m = f_{m,\text{sym}} + g_m \formspace,
	\end{align}
	where $\left\{f_{m,\text{sym}}\right\}_m$ is a continuous family of symmetric algebra elements in $\BU\left(\Dzs \right)$, and $\left\{ g_m\right\}_m$ is a continuous family in the ideal $\ICCR$ as defined in section \ref{sec:field-algebra-topology}.
\end{lemma}
\emph{Note:} An element $f_{\text{sym}} \in \BU(\Dzs)$ is symmetric, if it is symmetric in every degree.
\begin{proof}
First, one notes that it suffices to show the statement for any continuous family of \emph{homogeneous} elements of degree $N$, that is, $f_m = (0,0,\dots,0,f^{(N)}_m,0,0,\dots)$, for $f_m^{(N)} \in \left(\Dzs\right)^{\otimes N}$, since an arbitrary continuous family of elements in $\BU\left(\Dzs \right)$ can be written as a finite sum of continuous families of homogeneous elements. The statement then follows, because the sum of symmetric elements remains symmetric, and the ideal $\ICCR$ is a closed sub-*-algebra of $\BU\left(\Dzs \right)$ and hence by definition closed under addition. In the context of this proof, continuous families of elements in the BU-algebra are always meant to be continuous in $m$. 	\par
The main idea for this proof is then to write a (homogeneous) element as the sum of something symmetrized plus something containing commutations of indices. One can then use the commutator to identify an element of degree $N$ with an element of degree $N-2$ plus some terms that lie in the ideal $\ICCR$. Furthermore we need to show that all these operations are continuous in order to get the result for continuous families. To formulate this precisely, we need to introduce some notation.\par
Recall that we have defined test differential one-forms on the Cauchy surface $\Sigma$ as smooth sections of the cotangent bundle $\TsS$ with compact support, that is, $\Omega^1_0(\Sigma) = \Gamma_0(\TsS)$. We can therefore identify the space of initial data $\Dzs$ as
\begin{align}
	\Dzs &= \Omega^1_0(\Sigma) \oplus \Omega^1_0(\Sigma) \notag\\
	&= \Gamma_0(\TsS) \oplus \Gamma_0(\TsS) \notag\\
	&= \Gamma_0(\TsS \oplus \TsS) \formspace,
\end{align}
where $\TsS \oplus \TsS$ denotes the Whitney sum of vector bundles as defined in section \ref{sec:spacetime_geometry}. Furthermore, we can identify tensor products of initial data as
\begin{align}
	\Dzs \otimes \Dzs &= \Gamma_0\left( \TsS \oplus \TsS \right) \otimes \Gamma_0\left( \TsS \oplus \TsS \right) \notag\\
	&\subset \Gamma_0\big( (\TsS \oplus \TsS) \boxtimes (\TsS \oplus \TsS) \big) \formspace,
\end{align}
where the canonical embedding is dense (c.f. \cite[Chapter 3.3]{verch_sahlman}). The outer tensor product $\boxtimes$ of vector bundles is also defined in section \ref{sec:spacetime_geometry}. We may therefore identify an element in $\Dzs \otimes \Dzs$ with an element in $\Gamma_0\big( (\TsS \oplus \TsS) \boxtimes (\TsS \oplus \TsS) \big) $. We can furthermore re-write this in terms of components as
\begin{align}
\Gamma_0\big( &(\TsS \oplus \TsS) \boxtimes (\TsS \oplus \TsS) \big)\notag \\
  &= \underbrace{\Gamma_0\big( (\TsS \boxtimes \TsS) \big) }_{\eqqcolon \Omega_{11}} \oplus \underbrace{\Gamma_0\big( (\TsS \boxtimes \TsS) \big) }_{\eqqcolon \Omega_{12}}  \oplus \underbrace{\Gamma_0\big( (\TsS \boxtimes \TsS) \big)}_{\eqqcolon \Omega_{21}}    \oplus \underbrace{\Gamma_0\big( (\TsS \boxtimes \TsS) \big)   }_{\eqqcolon \Omega_{22}}  \formspace.
\end{align}
We define the corresponding continuous projectors
\begin{align}
	\text{pr}_{ij} : \Gamma_0\big( &(\TsS \oplus \TsS) \boxtimes (\TsS \oplus \TsS) \big) \to \Omega_{ij} \;,\quad i,j=1,2 \formspace.
\end{align}
As an example, let $\psi = (\varphi, \pi) \otimes (\varphi',\pi') \in \Gamma_0\big( (\TsS \oplus \TsS) \boxtimes (\TsS \oplus \TsS) \big)$, then $\text{pr}_{11}(\psi) = \varphi \otimes \varphi'$, $\text{pr}_{12}(\psi) = \varphi \otimes \pi'$, etcetera.\\
We generalize this to higher order tensor products of the space of initial data to obtain for any $N \in \IN$
\begin{align}
	\left(\Dzs\right)^{\otimes N}
	&\subset \Gamma_0\big( (\TsS \oplus \TsS)^{\boxtimes N} \big) \notag\\
	 &= \underbrace{\Gamma_0\big( (\TsS)^{\boxtimes N} \big) }_{\eqqcolon \Omega_{11\dots11}} \oplus \underbrace{\Gamma_0\big( (\TsS)^{\boxtimes N} \big) }_{\eqqcolon \Omega_{11\dots12}}  \oplus \cdots \oplus \underbrace{\Gamma_0\big( (\TsS)^{\boxtimes N} \big)   }_{\eqqcolon \Omega_{22\dots22}} \formspace,
\end{align}
where there are $2^N$ summands in the direct sum. By construction, $ (\TsS \oplus \TsS)^{\boxtimes N}$ is a vector bundle over $\Sigma^N$. Again, we define the continuous projectors $\text{pr}_{a_1 a_2 \dots a_N} : \Gamma_0\big( (\TsS \oplus \TsS)^{\boxtimes N} \big) \to \Omega_{a_1 a_2 \dots a_N}$, $a_i = 1, 2$ for $i=1,2,\dots,N$. To simplify notation for later use, we also define for any $k=1,\dots,N-1$, $i,j=1,2$,
\begin{align}
	\text{pr}_{k,ij} &:  \Gamma_0\big( (\TsS \oplus \TsS)^{\boxtimes N} \big) \to  \Gamma_0(\mathfrak{X}) \notag \\
	\text{pr}_{k,ij} &\coloneqq \bigoplus\limits_{\substack{a_i = 1,2 \\ i \neq k,k+1}} \text{pr}_{a_1 a_2 \dots a_{k-1} i j a_{k+2} \dots a_N}  \formspace,
\end{align}
where naturally, $\text{pr}_{k,ij}$ is a map into $\bigoplus^{2^{N-2}} \Gamma_0 (\TsS^{\boxtimes N})$ that we embed into $\Gamma_0(\mathfrak{X}) = \Gamma_0\big( (\TsS \oplus \TsS)^{\boxtimes (k-1)} \boxtimes \TsS \boxtimes \TsS \boxtimes (\TsS \oplus \TsS)^{\boxtimes(N-k-1)} \big)$. This notation makes clear that basically $\text{pr}_{k,ij}$ corresponds to a projection of the $k$-th and $(k+1)$-th $i$- respectively $j$-component of a tensor of degree $N$.\par
With this notation at our disposal we may continue the actual proof. We want to prove that, for any $N \in \IN$, a continuous family $\big\{ (0,0,\dots,f^{(N)}_m,0,0,\dots) \big\}_m$  of homogeneous elements can be decomposed into a sum of a continuous family $f_{m,\text{sym}}$ of symmetric elements and a continuous family in the ideal $\ICCR$. We will do this by induction in steps of $2$.\par
Base case:\par
1.) $N=0$ \\
Let $\big\{(f_m^{(0)} , 0 , 0 ,\dots)\big\}_m$ be a family of homogeneous elements of degree zero, $f_m^{(0)} \in \IC$. Trivially, the element is symmetric and by assumption continuous in $m$, and since $0 \in \ICCR$ by definition, the statement follows directly. \par
2.) $N=1$\\
Let $\big\{(0,f_m^{(1)} , 0 , 0 ,\dots)\big\}_m$ be a family of homogeneous elements of degree one, $f_m^{(1)} \in \Dzs$. Again, this is already symmetric and continuous in $m$, and the statement follows trivially. \par
We now proceed with the induction step.\\
The induction assumption is that for $N\in \IN$, every continuous family of homogeneous elements $(0,0,\dots,0,f_m^{(N)},0,0,\dots )$, where $f_m^{(N)} \in (\Dzs)^{\otimes N}$, can be written as a sum of a continuous family of fully symmetric elements in the field algebra and a continuous family of elements in the ideal. \par
Now let $f_m^{(N+2)} \in (\Dzs)^{\otimes(N+2)}$, which we identify with an element
\begin{align}
f_m^{(N+2)} (p_1,p_2,\dots,p_{N+2}) \in \Gamma_0\big( (\TsS \oplus \TsS)^{\boxtimes (N+2)} \big) \formspace,
\end{align}specify a continuous family of homogeneous elements $(0,0,\dots,0,f^{(N+2)}_m ,0,0,\dots)$ in the BU-algebra of initial data. In order to re-write this into the desired form, we start by looking at the symmetrized element
\begin{align}
	f_{m,\text{sym}}^{(N+2)}(p_1,p_2,\dots,p_{N+2})
	&\coloneqq \frac{1}{(N+2)!} \sum\limits_\sigma f_m^{(N+2)} (p_{\sigma(1)},p_{\sigma(2)},\dots,p_{\sigma(N+2)}) \formspace,
\end{align}
where the sum is taken all permutations $\sigma$ of $\left\{ 1,2,\dots,N+2\right\}$. Introducing the permutation operator $P_\sigma : \Gamma_0\big( (\TsS \oplus \TsS)^{\boxtimes (N+2)} \big)  \to \Gamma_0\big( (\TsS \oplus \TsS)^{\boxtimes (N+2)} \big) $, such that $(P_\sigma \psi)(p_1,p_2,\dots,p_{N+2}) = \psi (p_{\sigma(1)},p_{\sigma(2)},\dots,p_{\sigma(N+2)}) $, we may write this in the short hand notation as
\begin{align}
f_{m,\text{sym}}^{(N+2)}
&=  \frac{1}{(N+2)!} \sum\limits_\sigma P_\sigma f_m^{(N+2)} \notag\\
&= f_{m}^{(N+2)} +\frac{1}{(N+2)!} \sum\limits_\sigma  (P_\sigma - 1) f_m^{(N+2)} \formspace,
\end{align}
and therefore
\begin{align}\label{eqn:symmetrization_of_field}
f_{m}^{(N+2)} = f_{m,\text{sym}}^{(N+2)} - \frac{1}{(N+2)!} \sum\limits_\sigma  (P_\sigma - 1) f_m^{(N+2)}  \formspace.
\end{align}
Since the topology of $\Gamma_0\big( (\TsS \oplus \TsS)^{\boxtimes (N+2)} \big) $ is invariant under the swapping of variables, we find that if $f_m^{(N+2)} (p_1,p_2,\dots,p_{N+2})$ is continuous in $m$, then also the family associated with $f_m^{(N+2)} (p_{\sigma(1)},p_{\sigma(2)},\dots,p_{\sigma(N+2)})$ is continuous in $m$ for any permutation $\sigma$. Since taking sums is continuous, we conclude that $f_{m,\text{sym}}^{(N+2)}$ as defined above is continuous in $m$ as well as $\sum_\sigma  (P_\sigma - 1) f_m^{(N+2)}$. \par
So, we have successfully decomposed $(0,0,\dots,0,f_m^{(N+2)},0,0,\dots)$ into a sum of a continuous family of symmetric elements, and a continuous family of permutations. We have now left to show that this residual second term of permutations
\begin{align}\label{eqn:permuted_homogenous_element}
	\big(0,0,\dots,0,\sum_\sigma  (P_\sigma - 1) f_m^{(N+2)},0,0,\dots \big)
\end{align}
can be decomposed into the desired form.\par
We note that every permutation $\sigma$ can be written as the composition of transpositions of neighbor indices $\tau_i$, $i=1,2,\dots,l$, for some $l \in \IN$. For example. $\tau(1,2,3,4,5) = (1,2,4,3,5)$. Decomposing $\sigma = \tau_1 \comp \tau_2 \comp \cdots \comp \tau_l$, we naturally find $P_\sigma = P_{\tau_1}\cdot P_{\tau_2}\cdots P_{\tau_l}$. We can therefore write, by expanding the expression into a telescoping series,
\begin{align}\label{eqn:commuted-elements}
(P_\sigma - 1) f_m^{(N+2)} = \sum_{i=1}^l (P_{\tau_i} - 1)\,P_{\tau_{i+1}}\cdots P_{\tau_{l}}\,   f_m^{(N+2)} \formspace.
\end{align}
This is now a sum over differences of elements where only two indices are swapped, for example a difference like
\begin{align}
\psi_m^{(N+2)}(p_1,p_2,\dots,p_i,p_{i+1},\dots,p_{N+2}) - \psi_m^{(N+2)}(p_1,p_2,\dots,p_{i+1},p_{i},\dots,p_{N+2}) \formspace,
\end{align}
where $\psi_m^{(N+2)} = P_{\tau_{i+1}}\cdots P_{\tau_{l}}\,   f_m^{(N+2)}$ for some transpositions $\tau_i$.
We now want to use a generalized notion of the commutator to reduce this to something of degree $N$.\par
We define the map $G^{(N)}_k : \Gamma_0\big( (\TsS \oplus \TsS)^{\boxtimes (N+2)} \big) \to \Gamma_0\big( (\TsS \oplus \TsS)^{\boxtimes N} \big)$, for all $k=1,2,\dots,N-1$, by
\begin{align}
G^{(N)}_k (\psi) =& \int\limits_\Sigma h^{ij}(p) \left( \text{pr}_{k,21} \psi - \text{pr}_{k,12} \psi  \right)_{b_1 b_2 \dots b_{k-1} i j b_{k+2}\dotsb_{N+2}}(\dots,p,p,\dots)\,\dvolh(p) \,\notag \\
&\cdot \; dx^{b_1} \otimes \cdots \otimes  dx^{b_{k-1}}\otimes  dx^{b_{k+2}} \otimes \cdots \otimes  dx^{b_{N+2}}   \formspace,
\end{align}
where for the point $(\dots,p,p,\dots) \in \Sigma^{(N+2)}$ the $p$'s are put in the $k$-th and $(k+1)$-th entry. By definition of the projectors the map $G^{(N)}_k (\psi)$ applied gives an element in $\Gamma_0\big( (\TsS \oplus \TsS)^{\boxtimes N} \big)$. By construction, if $f_m^{(N+2)}$ specifies a continuous family of algebra elements, so does $G_k^{(N)}\left( f_m^{(N+2)} \right)$: Continuity in the space of smooth test one-forms means continuity in all orders of derivatives, using a partition of unity of the compact support and local charts. In particular this yields continuity in the components itself. Since the projectors are continuous, $G^{(N)}_k$ acts as a distribution and is therefore continuous .Using this generalized propagator, we find for an arbitrary $\tilde{f}_m^{(N+2)} \in (\Dzs)^{\otimes(N+2)}$ and an arbitrary $\tau$ that swaps the $k$-th and $(k+1)$-th component of $\psi^{(N+2)}$
\begin{align}\label{eqn:inserting-commutator-zero}
	\big( 0,0,\dots,0, (P_\tau - 1)& \tilde{f}_m^{(N+2)},0,0,\dots \big) \notag \\
&= \left(0,0,\dots,0,-\i G^{(N)}_k\big( \tilde{f}_m^{(N+2)}\big)  , 0, (P_\tau - 1)\tilde{f}_m^{(N+2)} ,0,0,\dots\right) \notag \\
&+ \left( 0,0,\dots,0,\i G^{(N)}_k\big( \tilde{f}_m^{(N+2)}\big)  , 0, 0,\dots \right) \formspace,
\end{align}
where the second term is naturally a homogeneous element of degree $N$ that is continuous in $m$.\\
The first term, also continuous in $m$, may be explicitly worked out by choosing a pure $\psi^{(N+2)} = \psi_{(1)} \otimes \psi_{(2)} \otimes \cdots \otimes \psi_{(N+2)}$, where $\psi_{(i)} = (\varphi_{(i)}, \pi_{(i)} )\in \Dzs$ for $i=1,\dots,N+2$.  By construction, we find
\begin{align}
	G^{(N)}_k &\big( \psi^{(N+2)} \big)  \\
	&= \Big( {\langle \pi_{(k)} , \varphi_{(k+1)} \rangle_\Sigma - \langle \varphi_{(k)} , \pi_{(k+1)} \rangle_\Sigma}\Big) \cdot  \psi_{(1)} \otimes \cdots \otimes \psi_{(k-1)} \otimes \psi_{(k+2)} \otimes \cdots \otimes \psi_{(N+2)} \formspace, \notag
\end{align}
and hence we obtain the product
\begin{align} \label{eqn:big_ideal_element}
\big(0,0,\dots,0,-\i  G^{(N)}_k&\big( \psi^{(N+2)}\big)  , 0, (P_\tau - 1) \psi^{(N+2)} ,0,0,\dots\big) \notag \\
	&=\left(0,0,\dots,0,  \psi_{(1)} \otimes \cdots \otimes \psi_{(k-1)} ,0,0,\dots \right) \notag \\
	&\quad \cdot  \big( -\i G(\psi_{(k)}  , \psi_{(k+1)})  ,  0, \psi_{(k)} \otimes  \psi_{(k+1)} - \psi_{(k+1)} \otimes \psi_{(k)},0,0,\dots   \big)   \notag \\
	&\quad\cdot \left(0,0,\dots,0, \psi_{(k+2)} \otimes \cdots \otimes \psi_{(N+2)} ,0,0,\dots \right) \formspace,
\end{align}
where we have introduced $G(\psi_{(k)}  , \psi_{(k+1)}) = \langle \pi_{(k)} , \varphi_{(k+1)} \rangle_\Sigma - \langle \varphi_{(k)} , \pi_{(k+1)} \rangle_\Sigma$ as a shorthand notation in analogy to the classical propagator $\mathcal G (\cdot,\cdot)$.
Now the construction ``pays off'' and we find that element in the middle of the above product,
\begin{align}
&\big( -\i G(\psi_{(k)}  , \psi_{(k+1)})   ,  0, \psi_{(k)} \otimes  \psi_{(k+1)} - \psi_{(k+1)} \otimes \psi_{(k)},0,0,\dots   \big)   \notag \\
&=  \big( -\i ( \langle \pi_{(k)} , \varphi_{(k+1)} \rangle_\Sigma - \langle \varphi_{(k)} , \pi_{(k+1)} \rangle_\Sigma )  ,  0, \notag \\
&\quad\quad (\varphi_{(k)},\pi_{(k)}) \otimes  (\varphi_{(k+1)},\pi_{(k+1)}) - (\varphi_{(k+1)},\pi_{(k+1)}) \otimes (\varphi_{(k)},\pi_{(k)}),0,0,\dots   \big)  \formspace,
\end{align}
is clearly an element of $\ICCR$ by definition. Hence, using equation (\ref{eqn:big_ideal_element}), we find that $\big(0,0,\dots,0,-\i  G^{(N)}_k\big( \psi^{(N+2)}\big)  , 0, (P_\tau - 1)\psi^{(N+2)} ,0,0,\dots\big)$ is an element of the ideal $\ICCR$. Since an arbitrary element of $(\Dzs)^{\otimes(N+2)}$ can be written as a sum of simple tensor product elements $\psi^{(N+2)}$ and using again that $\ICCR$ is closed under addition, we find that
\begin{align}
\left(0,0,\dots,0,-\i G^{(N)}_k\big( \tilde{f}_m^{(N+2)}\big)  , 0, (P_\tau - 1) \tilde{f}_m^{(N+2)} ,0,0,\dots\right)
\end{align}
is an element of the ideal $\ICCR$. \\
By equation (\ref{eqn:inserting-commutator-zero}) we find that the continuous family of homogeneous elements associated with $(P_\tau -1) \tilde{f}_m^{(N+2)}$ can be written as a sum of a continuous family of elements in the ideal $\ICCR$, and a continuous family of homogeneous elements of degree $N$. By equation (\ref{eqn:commuted-elements}) and again using that the ideal is closed under addition and that the sum of homogeneous elements of degree $N$ is again a homogeneous element of degree $N$, we find that the homogeneous element associated with $\sum_{\sigma} (P_\sigma -1) f^{(N+2)}_m$ can be written as the sum of an element in the ideal $\ICCR$ plus a homogeneous element of degree $N$ and overall using equation (\ref{eqn:symmetrization_of_field}), we find the result that
\begin{align}
	(0,0,\dots,0,f_m^{(N+2)},0,0,\dots) &= (0,0,\dots,0,f_{m,\text{sym}}^{(N+2)},0,0,\dots) \notag \\
																&\phantom{M}+ g_m + (0,0,\dots,0,f_m^{(N)},0,0,\dots) \formspace,
\end{align}
where all the summands give rise to families that are continuous in $m$, $g_m \in \ICCR$ and $f_m^{(N)} \in (\Dzs)^{\otimes N}$. We can now apply the induction hypothesis to the homogeneous element of degree $N$ and write it as a sum of a continuous family of symmetric elements, and a continuous family of elements in the ideal $\ICCR$. Finally, once again using that the sum of symmetric elements is still symmetric, that the ideal $\ICCR$ is closed under addition, and that the sum of continuous families is continuous, we obtain the desired result.
\end{proof}
\begin{lemma}\label{lem:states_estimate}
Let $\W (E,\sigma)$ be the Weyl algebra generated over some real, pre-symplectic space $(E, \sigma)$. For all states $C \in \C(E,\sigma)$ it holds
\begin{align}
	\abs{C(F)} \leq 1 \formspace.
\end{align}
\end{lemma}
\begin{proof}
		Let $z \in \IC$ and $F \in E$. Define $A = \big( \mathbbm{1} + z W(F) \big) \in \W(E,\sigma)$. Then, for any state $C \in \C(E,\sigma)$ we find
		\begin{align}
			\omega_C (A^* A) = \Big(1+ z\bar{z} + 2 \Real{z\,C(F)}   \Big)\formspace.
		\end{align}
		We can now choose $z$ with $\abs{z} =1$, such that $z\,C(F) = - \abs{C(F)}$ and find the result
		\begin{align}
			0 & \leq 	\omega_C (A^* A)  = 2 -2\, \abs{C(F)} \notag\\
			\implies   \abs{C(F)}  & \leq 1 \formspace.
		\end{align}
		This completes the proof.
\end{proof}

\section{Acknowledgement}
I thank Prof. Stefan Hollands for giving me the opportunity to write this thesis in a fascinating, recent research area. I thank my research supervisor, Dr. Ko Sanders, in particular for his constant support and weekly meetings. Without his help in figuring out detailed calculations and his insight in the framework of QFTCS I would have surely been hopelessly lost in the subject. I would also like to thank Prof. Klaus Sibold for sharing his knowledge of Minkowski QFT and his understanding of the zero mass limit. Furthermore, I thank Florian Knoop and Janos Borst for proofreading parts of the manuscript and giving me helpful feedback. Finally, I thank my parents as without their support my carefree studying, and enjoying other pleasant parts of a student's life, would not have been possible. I know that this was quite a privileged situation.

\newpage

\section{Bibliotheks- und Selbstst\"andigkeitserkl\"arung}
\subsection*{Bibliothekserkl\"arung}
Hiermit erkläre ich, Maxmilian Schambach, mein Einverständnis, dass die von mir erstellte Masterarbeit in die Bibliothek der
Universit\"at Leipzig eingestellt werden darf.\\[2ex]

\noindent
Leipzig, den
\vspace{3cm}

\subsection*{Selbstst\"andigkeitserkl\"arung}
Hiermit versichere ich, Maximilian Schambach, dass ich diese Arbeit selbstständig verfasst und keine anderen, als die angegebenen Quellen und Hilfsmittel benutzt und die wörtlich oder inhaltlich übernommenen Stellen als solche kenntlich gemacht habe. \\[2ex]

\noindent
Leipzig, den


\newpage
\printglossaries
\newpage \noindent
\printbibliography[heading=bibintoc,title={References}]
%

\end{document}